\documentclass[11pt]{article}

\usepackage[colorlinks=true, linkcolor=blue, citecolor=blue, urlcolor=blue, hypertexnames=false]{hyperref}
\usepackage{xurl}

\usepackage{graphicx}
\usepackage{amsmath,amsfonts,amssymb,bbm}
\DeclareFontFamily{U}{mathx}{}
\DeclareFontShape{U}{mathx}{m}{n}{<-> mathx10}{}
\DeclareSymbolFont{mathx}{U}{mathx}{m}{n}
\DeclareMathAccent{\widehat}{0}{mathx}{"70}
\DeclareMathAccent{\widecheck}{0}{mathx}{"71}
\usepackage{natbib}
\bibpunct[, ]{(}{)}{,}{a}{}{,}%
\usepackage{rotating}
\usepackage{fancyvrb}
\usepackage[FIGTOPCAP]{subfigure}
\usepackage{hyperref,setspace}
\usepackage{float}
\usepackage{booktabs,multicol,multirow,tabularx}
\usepackage{fix-cm}
\usepackage{cleveref}
\usepackage{threeparttable}
\usepackage{enumerate}
\usepackage{enumitem}
\usepackage{bibunits}
\usepackage{hhline}
\usepackage{fix-cm}
\usepackage{anyfontsize}
\usepackage[T1]{fontenc}

\usepackage{endnotes}
\let\footnote=\endnote
\usepackage{array}
\usepackage[dvipsnames]{xcolor}
\usepackage{cases}
\usepackage{dsfont}
\usepackage{algorithm}
\usepackage{algpseudocode}
\usepackage{etoolbox}

\newcommand{\calS}{{\mathcal{S}}}

\newcommand{{\calY}}{{\mathcal{Y}}}

\newcommand{\calC}{{\mathcal{C}}}

\newcommand{\calL}{{\mathcal{L}}}
\newcommand{\calN}{{\mathcal{N}}}

\newcommand{\calO}{{\mathcal{O}}}

\newcommand{\calV}{{\mathcal{V}}}

\newcommand{\E}{\mathbb{E}}

\newcommand{\dd}{{\rm d}}

\newcommand{\order}{\ensuremath{\calO}}

\newcommand{\indicator}{\mathbbm{1}}
\newcommand{\bm}{\boldsymbol}
\newcommand{\bb}{\mathbb}

\newcommand{\bRb}{\mathbf R^{\leftarrow}_t}

\newcommand{\hatbRb}{\widehat{\mathbf R}^{\leftarrow}_t}

\usepackage{amsmath,amsfonts,amsthm,amssymb,bbm}

\graphicspath{{Figures/}} 

\newtheorem{theorem}{Theorem}

\newtheorem{example}{Example}
\newtheorem{assumption}{Assumption}

\newtheorem{lemma}{Lemma}
\makeatletter

\newcommand{\Rmnum}[1]{\expandafter\@slowromancap\romannumeral #1@}
\makeatother

\newcommand{\Halmos}{\hfill\ensuremath{\square}}

\allowdisplaybreaks

\begin{document}

\title{\huge\bf
Diffusion Factor Models: \\
Generating High-Dimensional Returns with Factor Structure \\ [.25in]
}

\author{Minshuo Chen\thanks{Department of Industrial Engineering and Management Sciences, Northwestern University. \url{minshuo.chen@northwestern.edu} (email).}, \ Renyuan Xu\thanks{Department of Management Science \& Engineering, Stanford University. \url{renyuanxu@stanford.edu} (email).}, \ 
Yumin Xu\thanks{School of Mathematical Sciences, Peking University. \url{xuyumin@pku.edu.cn} (email).}, \  
and 
Ruixun Zhang\thanks{School of Mathematical Sciences, Center for Statistical Science, Laboratory for Mathematical Economics and Quantitative Finance, and National Engineering Laboratory for Big Data Analysis and Applications, Peking University. \url{zhangruixun@pku.edu.cn} (email).}}

\date{First version: April 2025\\
This version: January 2026
}

\maketitle
\thispagestyle{empty}

\centerline{\bf Abstract} \baselineskip 14pt \vskip 10pt
\noindent
Financial scenario simulation is essential for risk management and portfolio optimization, yet it remains challenging especially in high-dimensional and small data settings common in finance. 
We propose a \textit{diffusion factor model} that integrates latent factor structure into generative diffusion processes, bridging econometrics with modern generative AI to address the challenges of the curse of dimensionality and data scarcity in financial simulation. 
By exploiting the low-dimensional factor structure inherent in asset returns, we decompose the score function---a key component in diffusion models---using time-varying orthogonal projections, and this decomposition is incorporated into the design of neural network architectures.
We derive rigorous statistical guarantees, establishing non-asymptotic error bounds for both score estimation at $\widetilde{\mathcal{O}}\left(d^{5/2}n^{-\frac{2}{k+5}}\right)$ and generated distribution at $\widetilde{\mathcal{O}}\left(d^{5/4}n^{-\frac{1}{2(k+5)}}\right)$, primarily driven by the intrinsic factor dimension $k$ rather than the number of assets $d$, surpassing the dimension-dependent limits in the classical nonparametric statistics literature and making the framework viable for markets with thousands of assets. 
Numerical studies confirm superior performance in mean and covariance estimation as well as latent subspace recovery under small data regimes. Empirical analysis demonstrates the economic significance of our framework in constructing mean-variance optimal portfolios and factor portfolios. This work presents the first theoretical integration of factor structure with diffusion models, offering a principled approach for high-dimensional financial simulation with limited data. Our code is available at \url{https://github.com/xymmmm00/diffusion_factor_model}.


\newpage

\thispagestyle{empty}
\singlespacing
\tableofcontents
\newpage

\setcounter{page}{1}
\pagenumbering{arabic}
\setlength{\baselineskip}{1.5\baselineskip}
\onehalfspacing
\setcounter{equation}{0}
\setcounter{table}{0}
\setcounter{figure}{0}

\section{Introduction}
\label{sec: intro}
Financial scenario simulation, central to quantitative finance and risk management, has evolved significantly over recent decades  \citep{alexander2005present,eckerli2021generative, brophy2023generative}. Generating realistic and diverse financial scenarios is crucial not only for institutional traders to better manage their strategy risks, but also for regulators to ensure market stability  \citep{acharya2023climate,schneider2023bank,shapiro2024stress}. The US Federal Reserve evaluates market conditions and releases a series of market stress scenarios on an annual basis  \citep{federal}. Financial institutions are required to apply these scenarios to their portfolios to estimate and mitigate potential losses during market downturns. With the rise of trading automation and technological advancements, there is a pressing need from both parties to simulate more complex and high-dimensional financial scenarios  \citep{reppen2023deep}. This request challenges traditional model-based simulation approaches  \citep{behn2022limits, hambly2023recent}, highlighting the need for sophisticated data-driven techniques.

With the advances in machine learning techniques and computational power, generative AI has become a transformative force and is increasingly popular in financial applications. Its capabilities are now being harnessed for a wide range of tasks, such as generating financial time series  \citep{yoon2019time,cont2022tail,brophy2023generative,acciaio2024time}, modeling volatility surfaces  \citep{vuletic2025volgan}, simulating limit order book dynamics  \citep{coletta2023conditional,cont2023limit,hultin2023generative}, 
forecasting and imputing missing values  \citep{tashiro2021csdi,vuletic2024fin}, and constructing portfolio strategies  \citep{cetingoz2025synthetic}.

In recent years, several families of generative models have been explored in this context, including generative adversarial networks (GANs), autoencoders, and variational autoencoders (VAEs). In financial applications, GANs have been the primary workhorse  \citep{yoon2019time,cont2022tail,liao2024sig,cetingoz2025synthetic,vuletic2025volgan}, but are hindered by several issues, including training instability, mode collapse, high computational costs, and evaluation difficulties \citep{saatci2017bayesian,borji2019pros}. 
In addition, developing a theoretical understanding of GANs is challenging due to their minimax structure and complex training process, which has hindered principled analysis and sustainable improvements since their inception  \citep{creswell2018generative,gui2021review}. Non-adversarial latent-variable models such as autoencoders and VAEs often face non-identifiable problems \citep{locatello2019challenging}, which can yield unstable performance on complex data \citep{saatci2017bayesian,borji2019pros,he2019lagging,dai2020usual}, risk posterior collapse, and residual-risk underestimation \citep{hoffman2016elbo,he2019lagging,dai2020usual}.

More recently, generative diffusion models have gained traction as a more robust alternative, offering significant advantages in financial applications \citep{trilemma2022,kotelnikov2023tabddpm,coletta2024constrained,li2024understanding,barancikova2024sigdiffusions}. Compared to autoencoders, VAEs, and GANs, diffusion models can capture complex data distributions with more robust and stable performance, ease of training, and enhanced stability and efficiency \citep{nichol2021improved,dhariwal2021diffusion,karras2022elucidating}, achieving state-of-the-art results in practice and making them invaluable tools in advancing generative AI for finance. In particular, it is well documented that the diffusion model works well and beats alternatives in the small-data regime \citep{kotelnikov2023tabddpm,li2024understanding}. In addition, the diffusion model framework is grounded in rigorous stochastic and statistical analysis  \citep{chen2024opportunities,gao2024reward,tang2024score}, providing a theoretically sound basis for incorporating domain-specific properties, such as those in finance.

\subsection{Our Work and Contributions}
We develop a deep generative model based on diffusion models to simulate high-dimensional asset returns that follow an {\it unknown} factor structure, which we term the \textit{diffusion factor model}. The returns of the $d$-dimensional assets are explained by the linear combination of $k$ {\it unknown} common factors $(k \ll d)$ and an idiosyncratic noise that varies from asset to asset (see Equation~\eqref{equ: factor model}).%
\footnote{The number of factors $k$ varies from 1 to several dozen, balancing predictive power and economic interpretability  \citep{harvey2016and,giglio2021thousands}.}
We develop the theory for our diffusion factor model and establish statistical guarantees of the error of diffusion-generated returns, which overcomes the curse of dimensionality in the number of assets. We also conduct numerical and empirical studies to demonstrate its practical relevance.

Our generative model is particularly relevant for the high-dimensional small-data setting, a classical challenge for medium- (e.g., daily) to low- (e.g., weekly or monthly) frequency return data in finance.
In empirical applications, the number of assets $d$ often ranges from hundreds to thousands, easily exceeding the number of available observations in a stationary period  \citep{kan2007optimal,nagel2013empirical,gu_etal:2020}.
While machine learning is commonly perceived as a ``big data'' tool, many core finance questions are hindered by the ``small data'' nature of economic time series. Our model offers a methodology to tackle this challenge: fitting the diffusion model on the limited data and then generating additional realistic data for downstream tasks.

As a result, our diffusion factor model can be used to simulate realistic financial data for potential applications in a wide range of important contexts, including asset pricing and factor analysis across stock  \citep{fama_french:1993,fama_french:2015}, option  \citep{buchner2022factor}, bond  \citep{kelly2023modeling,elkamhi2024one}, and cryptocurrency markets  \citep{liu2022common},
large-scale asset covariance estimation  \citep{bickel2008regularized,fan2016overview,ledoit2022power}, robust portfolio construction  \citep{demiguel2009generalized,avramov_zhou:2010,fabozzi2010robust,tu_zhou:2010,jacquier2011bayesian}, and systematic and institutional risk management  \citep{bisias2012survey,cont2022tail,he2022risk}.

Our contributions are multi-fold. First, our diffusion factor model presents the first theoretical integration of factor models with generative diffusion models. It fully exploits the structural property of factor models, using observations of asset returns with heterogeneous idiosyncratic noises, and without requiring prior knowledge of the exact factors. In particular, our framework addresses the curse of dimensionality issue in the ``small data'' regime by achieving a sample complexity that scales exponentially in the desired error, with an exponent that primarily depends on the intrinsic factor dimension $k$ rather than the ambient asset dimension $d$.

Second, the success of the diffusion factor model hinges on accurately estimating Stein's score function, which we achieve by decomposing the score function via a time-varying projection into a subspace component in a $k$-dimensional space and a linear component (Lemma \ref{lem: score decomposition -- heterogeneous}). This decomposition informs our neural network design---integrating denoising scheme, an encoder-decoder structure, and skip connections---to efficiently approximate the score function (Theorem \ref{theorem: score approximation}). We establish a statistical guarantee that the $L^2$ error between the neural score estimator and the ground truth is $\tilde\order(d^{\frac{5}{2}}n^{-\frac{2}{k+5}} )$ (Theorem~\ref{theorem: score estimation}), demonstrating that the dependence on the sample size $n$ is governed primarily by $k$ rather than $d$, effectively mitigating the curse of dimensionality.%
\footnote{It is worth noting that complete independence from $d$ is unattainable due to idiosyncratic noise spanning the full $d$-dimensional space. We achieve a mild polynomial dependence of the estimated score function on the ambient dimension $d$ from the residual noise.}

Third, we establish statistical guarantees for the errors in the generated return distribution as well as the subspace recovery. The output return distribution of our diffusion factor model is close to the true distribution in total variation distance, achieving an error of $\Tilde{\order} ( d^{\frac{5}{4}} n^{-\frac{1-\delta(n)}{2(k+5)}} )$, where $\delta(n)$ vanishes as $n$ grows. By applying singular value decomposition (SVD), we can also achieve latent subspace recovery with an error of order $\Tilde{\order}(d^{\frac{5}{4}}n^{-\frac{1-\delta(n)}{k+5}})$ (Theorem \ref{theorem: distribution estimation}). These results are achieved by the design of our sampling algorithm (Algorithm~\ref{algo: sampling}) and a novel analysis based on matrix concentration inequalities and coupling argument of stochastic processes (Lemmas~\ref{lem: eigenvalue estimation error} and~\ref{lem: backward SDEs L2 error}). Furthermore, our efficient sample complexities hold true under a mild Lipschitz assumption (Assumption~\ref{assumption: Lipschitz}), demonstrating the generality of our analysis.

Fourth, numerical studies with synthetic data confirm that our diffusion factor model is capable of simulating realistic return data. It delivers substantial improvements in mean and covariance estimation as well as in subspace recovery, especially in the ``small data'' regime where the number of samples is small relative to the number of assets. These improvements suggest that our diffusion factor model automatically adapts to the (unknown) underlying factor structure and captures patterns of the data distribution more effectively than direct empirical estimation from limited data. From a statistical perspective, our methodology acts as a form of data-dependent regularization, introducing a modest modeling bias while substantially reducing the estimation variance, thus improving the moment estimation.

Finally, empirical analysis of the U.S. stock market shows that data generated by our diffusion factor model improves both mean and covariance estimation, leading to superior mean-variance optimal portfolios and factor portfolios. Portfolios using diffusion-generated data consistently outperform traditional methods, including equal-weight and shrinkage approaches, with higher mean returns and Sharpe ratios. In addition, factors estimated from the generated data capture interpretable economic characteristics and the corresponding tangency portfolios exhibit higher Sharpe ratios, effectively capturing systematic risk. These results demonstrate that our diffusion factor model can serve as a practical and broadly applicable tool for learning return distributions and constructing robust portfolios from limited, heavy-tailed financial data.

\subsection{Related Literature} \label{sec:related_literature}

Our work is broadly related to two strands of the literature on factor models and diffusion models.
\paragraph{Factor Models.}
There is a vast econometric literature on factor models. Classic factor-based asset pricing models primarily focus on risk premium estimation, time-varying factors, model validity, and factor structure interpretability. Recent methodological advances have pioneered techniques for analyzing large, high-dimensional datasets, incorporating semiparametric estimation, robust inference, machine learning techniques, and time-varying risk premiums  \citep{ferson_harvey:1991, connor_etal:2012, feng_etal:2020, gu_etal:2020, raponi_etal:2020, chen2024deep, feng2024deep, giglio2025test}. We refer interested readers to survey papers such as \cite{fama2004capital}, \cite{giglio2022factor}, \cite{kelly2023financial}, and \cite{bagnara2024asset}.

While we assume the (target) data distribution follows a factor model structure, the implementation and analysis of the diffusion models  {\it do not require observing} the factors.  In fact, our goal is to uncover the latent low-dimensional factor space through the data generation process. This is extremely valuable for financial applications, particularly in identifying effective factors, which is often challenging using traditional methods, see, for example, \cite{chen_etal:1986,jegadeesh_titman:1993,jagannathan_wang:1996,lettau_ludvigson:2001,pastor_stambaugh:2003,yogo:2006,adrian_etal:2014,hou_etal:2015,he_etal:2017,lettau2020estimating} and \cite{lettau2020factors}.

\paragraph{Diffusion Models and Their Theoretical Underpinnings.}

Diffusion models have shown groundbreaking success and quickly become the state-of-the-art method in diverse domains  \citep{yang2023diffusion, cao2024survey,guo2024diffusion, liu2024learning}.

Despite significant empirical advances, the development of theoretical foundations for diffusion models falls behind. Recently, intriguing statistical and sampling theories emerged for deciphering, improving, and harnessing the power of diffusion models. Specifically, sampling theory considers whether diffusion models can generate a distribution that closely mimics the data distribution, given that the diffusion model is well-trained 
 \citep{de2021diffusion,de2022convergence,albergo2023stochastic,chen2022sampling,chen2023restoration,benton2024nearly,huang2024convergence,li2024towards,li2024unified}.

Complementary to sampling theory, statistical theory of diffusion models mainly concerns how well the score function can be learned given finitely many training samples  \citep{yang2022convergence,koehler2022statistical,chen2023score,oko2023diffusion, dou2024optimal,wibisono2024optimal, zhang2024minimax}. Later, end-to-end analyses in \citet{chen2023score, oko2023diffusion, azangulov2024convergence, fu2024unveil, tang2024adaptivity, zhang2024minimax, yakovlev2025generalization} present statistical complexities of diffusion models for estimating nonparametric data distributions. It is worth noting that \cite{chen2023score,oko2023diffusion, azangulov2024convergence, tang2024adaptivity, wang2024diffusion} prove the adaptivity of diffusion models to the intrinsic structures of data---they can circumvent the curse of ambient dimensionality when data are exactly concentrated on a low-dimensional space.

Two works most closely related to ours are \citet{chen2023score} and \citet{wang2024diffusion}, both of which consider subspace-structured data. \citet{chen2023score} assume that each data point $\mathbf X$ lies exactly on a low-dimensional subspace, i.e., $\mathbf X = \mathbf A \mathbf Z$ for some unknown matrix $\mathbf A \in \mathbb{R}^{d \times k}$ and latent variable $\mathbf Z \in \mathbb{R}^k$. In contrast, our factor model (Equation~\eqref{equ: factor model}) relaxes this strict subspace assumption by allowing idiosyncratic noise in the asset returns. The neural network architecture and parts of our analysis are inspired by \citet{chen2023score}, but the presence of high-dimensional idiosyncratic noise introduces substantial technical challenges in our setting. We discuss these technical novelties in detail in Section~\ref{sec:technical_highlights}. \citet{wang2024diffusion} also consider noisy subspace data, but assume that the latent variable $\mathbf Z$ follows a Gaussian mixture distribution. By comparison, we only require that the distribution of the latent variable satisfies a general sub-Gaussian tail condition. During the preparation of this manuscript, \cite{yakovlev2025generalization} generalize the study to noisy nonlinear low-dimensional data structures. They assume that the data follow a transformation on a latent variable, which is uniformly distributed in a hypercube. This is very different from our study on the factor model structure.

\subsection{Notation}
We denote vectors and matrices by bold letters. For a vector $\mathbf v$, we denote $\|\mathbf v\|_2$, $\|\mathbf v \|_{\infty}$ as its $\ell^2$ and $\ell^{\infty}$ norm, respectively. For a matrix $\mathbf M$, we denote ${\rm tr}(\mathbf M)$, $\|\mathbf M\|_{\rm F}$ and $\|\mathbf M\|_{\rm op}$ as its trace, Frobenius norm, and operator norm, respectively. When $\mathbf M$ is symmetric, we denote $\lambda_{\max}(\mathbf M)$ and $\lambda_k(\mathbf M)$ as the maximal and the $k$-th largest eigenvalues. We also denote a matrix-induced norm as $\| \mathbf v \|_{\mathbf M}^2 = \mathbf v^\top \mathbf M \mathbf v$. For two symmetric matrices, we associate a partial ordering $\mathbf M \succeq \mathbf N$ if $\mathbf M - \mathbf N$ is positive semi-definite. For a random vector $\mathbf X$ following distribution $P$, we denote $\|\mathbf X\|_{L^2(P)}^2 = \E[\|\mathbf X\|_2^2]$. We denote $\phi(\cdot; \bm\mu, \bm\Sigma)$ as the Gaussian density function with mean $\bm\mu$ and covariance~$\bm\Sigma$. Throughout the paper, we use uppercase letters (e.g., ${\bf X}$) for random variables and lowercase letters (e.g., ${\bf x}$) for their realizations.

\section{Problem Set-up for Diffusion Factor Models}\label{sec:pre}
Given limited market data, our objective is to design and train a diffusion-based factor model capable of simulating realistic, high-dimensional asset returns. This section introduces the two core components of our approach: generative diffusion models and the underlying factor structure. Section~\ref{sec:diffusio_intro} defines diffusion models and emphasizes the central role of score functions in their construction. Section~\ref{sec:asset returns} presents a framework for modeling high-dimensional asset returns with an unknown low-dimensional latent structure, typically captured by a factor model—an essential feature for enabling efficient and robust modeling in small-data environments. 

\subsection{Generative Diffusion Models}
\label{sec:diffusio_intro}
Diffusion models consist of two interconnected processes: a forward process progressively injects noise into data over time, and a time-reverse process that constructs new data by progressively removing noise  \citep{anderson1982reverse, haussmann1986time, song2019generative, ho2020denoising, song2020score}. The forward process is employed during training, while {\it the time-reverse process is used for data generation}. In the following, we formulate both processes using stochastic differential equations (SDEs) and detail the training methodology for diffusion models.

\paragraph{Forward and Time-Reverse SDEs.} 
For ease of theoretical analysis, we follow the convention in the literature~ \citep{ho2020denoising,song2020improved} and adopt Ornstein-Ulhenbeck (O-U) process for the forward process. In particular, we study a simple O-U process with a deterministic and nondecreasing weight function $ \eta(t) > 0 $ as
\begin{equation}
\label{equ: diffusion forward SDE}
    \dd \mathbf R_t = -\frac{1}{2}\eta(t) \mathbf R_t \dd t + \sqrt{\eta(t)} \dd \mathbf W_t \quad \text{with} \quad  \mathbf R_0\sim P_{\textrm{data}} ~\text{and}~ t \in [0, T],
\end{equation}
where $(\mathbf W_t)_{t \geq 0}$ is a standard Wiener process, $T$ is a terminal time and $P_{\textrm{data}}$ is the data distribution, i.e., the distribution of high-dimensional asset returns. We also denote $P_t$ as the marginal distribution of $\mathbf R_t$ with a corresponding density function $p_t$. Given an initial value $\mathbf R_0=\mathbf r$, at time $t$, the conditional distribution of $\mathbf R_t \,|\, \mathbf R_0=\mathbf r$ is Gaussian, i.e., 
\begin{equation}
\label{equ: distribution of R_t condition on R_0}
    \mathbf R_t \,|\, \mathbf R_0=\mathbf r \,\,\sim \mathcal{N}(\alpha_t \mathbf r, h_t \mathbf I_d),
\end{equation}
where $\alpha_t=\exp\left(-\int_0^t \frac{1}{2}\eta(s)\dd s\right)$ is the shrinkage ratio and $h_t = 1-\alpha^2_t$ is the variance of the added Gaussian noise. For simplicity, we take $\eta(t) = 1$ throughout the paper. Note that the terminal distribution $P_T$ is close to $P_{\infty} = \mathcal{N}(\mathbf 0, \mathbf I_d)$ when $T$ is sufficiently large, since the marginal distribution of an O-U process converges exponentially fast to its stationary distribution  \citep{bakry2014analysis, chen2022sampling}.

To design a procedure to generate new samples, we reverse the forward process in time~ \citep{anderson1982reverse, song2020score}. Under mild regularity conditions~ \citep{haussmann1986time}, this yields a well-defined backward process that transforms (white) noise into data. We denote the time-reversed SDE (backward process) associated with \eqref{equ: diffusion forward SDE} as
\begin{equation}
\label{equ: diffusion backward SDE}
    \dd \bRb = \left( \frac{1}{2}\bRb + \nabla \log p_{T-t}(\bRb) \right) \dd t + \dd \overline{\mathbf{W}}_t \quad \text{with} \quad \mathbf R^{\leftarrow}_{0} \sim Q_0 \text{ and } t \in [0, T],
\end{equation}
where $(\overline{\mathbf{W}}_{t})_{t \ge 0}$ is another Wiener process independent of $({\mathbf{W}}_t)_{t \ge 0}$, $ \nabla\log p_{t}(\cdot) $ is known as the {\it score function} and $Q_{0}$ is the initial distribution of the backward process. If we set $Q_0 = P_T$, under mild assumptions, the time-reverse process has the {\it same marginal distribution} as the forward process in the sense of ${\rm Law}(\bRb) ={\rm Law}(\mathbf R_{T-t})$; see \citet{anderson1982reverse} and \citet{haussmann1986time} for details. In particular, we have ${\rm Law} (\mathbf R_T^{\leftarrow})=P_{\textrm{data}}$, which leads us to recover the data distribution.

In practice, however, \eqref{equ: diffusion backward SDE} cannot be directly used to generate samples from the data distribution $P_{\textrm{data}}$ as both the score function and the distribution $P_{T} $ are {\it unknown}. To train a simulator that generates data (closely) from $P_{\textrm{data}}$, the key is to accurately learn the score function. With a score estimator $\widehat{\mathbf s}$ that approximates $\nabla \log p_t$ and an initial distribution $Q_0:=\mathcal{N}(\mathbf 0, \mathbf I_d)$ that is easy to sample, we specify the following implementable process for data generation
\begin{equation}
\label{equ: learned backward SDE}
    \dd \hatbRb = \left( \frac{1}{2}\hatbRb + \widehat{\mathbf s}\left( \hatbRb, T-t \right) \right) \dd t + \dd \overline{\mathbf W}_t \quad \text{with} \quad \widehat{\mathbf R}_0^{\leftarrow} \sim \mathcal{N}(\mathbf 0, \mathbf I_d).
\end{equation}
For O-U processes, the error introduced by taking $Q_0=\mathcal{N}(\mathbf 0, \mathbf I_d)$ usually decays exponentially with respect to $T$  \citep{chen2022sampling,lee2023convergence,tang2024contractive,gao2023wasserstein}.

\paragraph{Training by Score Matching.}
To learn the score function $\nabla \log p_{t}$ in \eqref{equ: diffusion backward SDE}, a natural method is to minimize a mean-squared error between the estimated and true scores~ \citep{hyvarinen2005estimation}, i.e.,
\begin{equation}\label{equ: score-matching loss function}
    \min_{\mathbf s \in \calS} \int_{t_0}^{T} w(t) \E_{\mathbf R_t} \left[ \left\| \mathbf s(\mathbf R_t, t) - \nabla \log p_t(\mathbf R_t) \right\|_2^2 \right] \dd t,
\end{equation}
where $w(t)$ is a positive weighting function and $\mathbf s$ is a parameterized estimator of the score function from a class $\calS$ such as neural networks.  Here, $t_0 > 0$ is a small early-stopping time to prevent the score function from blowing up as $t \to 0$  \citep{song2019generative,chen2023score}.

A key challenge in minimizing the score-matching loss \eqref{equ: score-matching loss function} is that the target term, $\nabla \log p_t$, is generally intractable—it cannot be computed directly from observed data. Alternatively, one can equivalently minimize the following denoising score matching proposed in \citet{vincent2011connection, song2020sliced},
which utilizes the conditional density of $\mathbf R_t \,|\, \mathbf R_0$ in \eqref{equ: distribution of R_t condition on R_0}:
\begin{equation}
\label{equ: denoising score-matching loss function}
    \min_{\mathbf s \in \calS} \int_{t_0}^{T} w(t)\E_{\mathbf R_0} \left[\E_{\mathbf R_t | \mathbf R_0} \left[ \left\| \mathbf s(\mathbf R_t, t) - \nabla \log \phi(\mathbf R_t ; \alpha_t \mathbf R_0, h_t \mathbf I_d) \right\|_2^2 \right] \right]\dd t.
\end{equation}
Here $\phi$ is the Gaussian density function defined at the end of Section~\ref{sec: intro}. For technical convenience, we choose a uniform weight $w(t) = 1/(T-t_0)$. Note that under the forward dynamics \eqref{equ: diffusion forward SDE},  $\nabla \log \phi(\mathbf r_t ; \alpha_t \mathbf r_0, h_t \mathbf I_d)$ in \eqref{equ: denoising score-matching loss function} has an analytical form, 
$$
    \nabla \log \phi(\mathbf r_t ; \alpha_t \mathbf r_0, h_t \mathbf I_d) = -\frac{\mathbf r_t - \alpha_t \mathbf r_0}{h_t}.
$$
In practice, we can only observe a finite sample of asset returns $\{\mathbf r^{i}\}_{i=1}^n$ from $P_{\textrm{data}}$. Therefore, we train the diffusion model using the following empirical score-matching objective:
\begin{equation}
\label{equ: empirical loss function}
    \min_{\mathbf s \in \calS} \widehat{\calL}(\mathbf s) := \frac{1}{n} \sum_{i=1}^n \ell(\mathbf r^i,\mathbf s) \quad \text{with} \quad \ell(\mathbf r^i, \mathbf s) = \dfrac{1}{T - t_0}\int_{t_0}^T \E_{\mathbf R_t|\mathbf R_0 = \mathbf r^i}\Big\|\mathbf s(\mathbf R_t, t) +\frac{\mathbf R_t - \alpha_t \mathbf r^i}{h_t}\Big\|_2^2 \dd t.
\end{equation}
Henceforth we write the population loss function in \eqref{equ: denoising score-matching loss function} as $\calL(\mathbf s) := \E[\widehat{\calL}(\mathbf s)]$.

\subsection{Asset Returns and Unknown Factor Structure} 
\label{sec:asset returns}
To improve sample efficiency, especially in small-data regimes, the central idea is to incorporate domain knowledge into the diffusion model. Specifically, we leverage a key insight from the finance literature: a relatively small set of latent factors—reflecting both macroeconomic and firm-specific variables—can effectively explain a broad class of asset returns  \citep{ross2013arbitrage,fan2016projected,ait2019principal,giglio2021asset,bryzgalova2023asset,kelly2023principal}. Following these studies, we consider the asset return $\mathbf R\sim P_{\textrm{data}}$ satisfying the following factor model structure: 
\begin{equation}
\label{equ: factor model}
    \mathbf R = \bm\beta \, \mathbf F + \bm\varepsilon,
\end{equation}
where $\mathbf F \in \bb R^k$ are {\it unknown} factors with $k \ll d$, $\bm\beta \in \bb R^{d \times k}$ is a factor loading matrix, and $\bm\varepsilon \in \bb R^d$ is the vector of idiosyncratic residuals.

We want to emphasize that, while we assume the data distribution $P_{\rm data}$ follows a factor model structure \eqref{equ: factor model}, the implementation and analysis of the diffusion models {\it do not require observing} the factors. Instead, our approach is capable of uncovering the latent low-dimensional factor space through the data generation process; see Section \ref{sec: distribution estimation} for more details.

Under the unknown factor scenario, factors and their loadings are identifiable only up to an invertible linear transformation, e.g., rescaling and rotation  \citep{kelly2023financial}. Thus, it is reasonable to assume that $\bm \beta$ has orthonormal columns. Otherwise, one can perform a QR decomposition to write $\bm \beta = \bm \beta' \mathbf H$, where $\bm \beta'\in \mathbb{R}^{d\times k}$ has orthonormal columns and $\mathbf H \in \mathbb{R}^{k \times k}$ is an upper triangular matrix.

In light of the factor structure in \eqref{equ: factor model}, we aim to develop a diffusion model framework that explicitly exploits this low-dimensional representation. Crucially, the statistical guarantees of our approach depend primarily on the number of latent factors $k$, rather than the ambient data dimension $d$. This dimensionality reduction enables the diffusion model to be trained effectively on a limited number of observations, while still generating realistic high-dimensional samples. As a result, the proposed framework addresses two central challenges in modeling financial data: the curse of dimensionality and data scarcity.

\section{Score Decomposition under Diffusion Factor Model}
\label{sec: model}
To simulate high-dimensional asset returns using diffusion factor models, the key challenge is accurately learning the score function via neural networks. However, due to the high dimensionality of asset returns and limited market data, directly estimating the score function is impractical as it suffers from the curse of dimensionality. To overcome this, we analyze the structural properties of score functions under factor models, deriving a tractable decomposition. This decomposition informs a neural network architecture designed to perform effectively in small-data regimes.

\subsection{Score Decomposition}
\label{sec: score decompostion}
With factor model structure in \eqref{equ: factor model}, we show that the score function $\nabla \log p_t$ can be decomposed into a subspace score in a $k$-dimensional space and a complementary component, each possessing distinct properties.

To ensure the decomposition is well-defined, we impose the following assumption.
\begin{assumption}[Factor model] We assume the following conditions on the factor model \eqref{equ: factor model}:
\label{assumption: factor}
    \begin{itemize}
        \item[(i)] The factor loading $\bm\beta \in \bb{R}^{d \times k}$ has orthonormal columns. 
        \item[(ii)] The factor $\mathbf F\in \bb R^k$ follows a distribution that has a density function denoted as $p_{\mathrm{fac}}$ and has a finite second moment, i.e., $\int \|\mathbf f\|_2^2 ~p_{\mathrm{fac}}(\mathbf f) \dd \mathbf f < \infty$.
        \item[(iii)] The residual $\bm\varepsilon$ is Gaussian with density $\phi(\cdot; \mathbf 0, \operatorname{diag}\{\sigma_1^2, \sigma_2^2, \dots, \sigma_d^2\})$ and there exists a positive constant $\sigma_{\max} > 0$ such that 
        $$
        \sigma_{\max} \ge \sigma_1 \ge \sigma_2 \ge \dots \ge \sigma_d > 0.
        $$
        \item[(iv)] $\mathbf F$ and $\bm \varepsilon$ are independent.
    \end{itemize}
\end{assumption}
As a result, $\mathbf R$ has a positive definite covariance matrix, defined as 
\begin{equation}
\label{equ: ground-truth covariance matrix}
\bm\Sigma_0 := \E[\mathbf R \mathbf R^{\top}] - \E[\mathbf R] \E[\mathbf R]^{\top}.
\end{equation}

Next, for an arbitrary time $t \in [0,T]$, we consider a linear subspace $
\calV_t$ spanned by the column vectors of $\bm\Lambda_t^{-\frac{1}{2}} \bm\beta$, with $\bm\Lambda_t$ defined as
\begin{equation}
\label{equ: Lambda_t}
    \bm\Lambda_t := \operatorname{diag}\Big\{h_t + \sigma_1^2 \alpha_t^2, h_t + \sigma_2^2 \alpha_t^2, \dots, h_t + \sigma_d^2 \alpha_t^2\Big\}.
\end{equation}
We further define $\mathbf T_t$ as the matrix of orthogonal projection onto $\calV_t$: 
\begin{equation}
\label{equ: Projection_t}
    \mathbf T_t := \bm\Lambda_t^{-\frac{1}{2}} \bm\beta \bm\Gamma_t \bm\beta^\top \bm\Lambda_t^{-\frac{1}{2}} \quad \text{with} \quad \bm\Gamma_t := (\bm\beta^\top \bm\Lambda_t^{-1} \bm\beta)^{-1}.
\end{equation}
Matrix $\bm\Gamma_t$ is well-defined as $\bm\beta^\top \bm\Lambda_t^{-1} \bm\beta$ is positive definite due to Assumption~\ref{assumption: factor}. The following lemma presents the score decomposition.

\begin{lemma}
\label{lem: score decomposition -- heterogeneous}
Suppose Assumption \ref{assumption: factor} holds. The score function $\nabla \log p_t(\mathbf r)$ can be decomposed into a subspace score and a complement score as
\begin{equation}
\label{equ: score decomposition -- heterogeneous}
    \nabla \log p_t(\mathbf r) = \underbrace{\mathbf T_t \bm\Lambda_t^{\frac{1}{2}} \bm\beta \cdot \nabla \log p_t^{\rm{fac}}(\bm\beta^\top \bm\Lambda_t^{\frac{1}{2}} \mathbf T_t \cdot \bm \Lambda_t^{-\frac{1}{2}} \mathbf r)}_{\textrm{ Subspace score }} \,\,\underbrace{- \bm \Lambda_t^{-\frac{1}{2}} (\mathbf I - \mathbf T_t) \cdot \bm \Lambda_t^{-\frac{1}{2}} \mathbf r}_{\textrm{ Complement score }},
\end{equation}
where $p_t^{\rm{fac}}(\cdot) := \int \phi(\cdot; \alpha_t \mathbf f, \bm\Gamma_t) p_{\rm{fac}}(\mathbf f) \dd \mathbf f$ and $\bm\Lambda_t$, $\bm\Gamma_t$, $\mathbf T_{t}$ are defined in \eqref{equ: Lambda_t} and \eqref{equ: Projection_t}. 
\end{lemma}
For future convenience, we denote the subspace score as $\mathbf s_{\textrm{sub}}: \bb R^{k} \times [0, T] \to \bb R^{d}$ and the complement score as $\mathbf s_{\textrm{comp}}: \bb R^{d} \times [0, T] \to \bb R^{d}$:
\begin{align}
\mathbf s_{\textrm{sub}}(\mathbf z, t) &:= \mathbf T_t \bm\Lambda_t^{\frac{1}{2}} \bm\beta \cdot \nabla \log p_t^{\rm{fac}}(\mathbf z), \quad \text{and} \label{equ: s_parallel} \\
\mathbf s_{\textrm{comp}}(\mathbf r, t) &:= - \bm \Lambda_t^{-\frac{1}{2}} (\mathbf I - \mathbf T_{t}) \bm \Lambda_t^{-\frac{1}{2}} \mathbf r. \label{equ: s_perp}
\end{align}

We defer the proof to Appendix \ref{sec: proof of lemma -- decomposition}. A few discussions are in place.

\paragraph{Motivation and Insights of Score Decomposition.} 
Lemma~\ref{lem: score decomposition -- heterogeneous} is proved using an orthogonal decomposition of the rescaled noisy data, $\mathbf \Lambda_t^{-1/2} \mathbf r = \mathbf T_t \cdot \mathbf \Lambda_t^{-1/2} \mathbf r + (\mathbf I - \mathbf T_t) \cdot \mathbf \Lambda_t^{-1/2} \mathbf r$, with the two decomposed terms serving different roles. Specifically, $\mathbf s_{\rm sub}$ is responsible for recovering the distribution of the low-dimensional factors, while $\mathbf s_{\rm comp}$ progressively adjusts the covariance of the generated returns to match that of the heterogeneous noise.

Furthermore, Lemma~\ref{lem: score decomposition -- heterogeneous} provides key insights into an efficient representation of the score function. As observed, the subspace score depends only on a $k$-dimensional input, while the complement score is linear, suggesting a natural dimension reduction in representing the score. Learning a low-dimensional nonlinear function together with a linear component is substantially more efficient—in terms of both sample complexity and computational cost—than learning a high-dimensional nonlinear function. Indeed, the score network architecture introduced in Section~\ref{sec:score_NN}, along with the subsequent statistical analysis in Sections~\ref{sec: score approximation and estimation} and \ref{sec: distribution estimation}, reflects this critical observation.

\subsection{Choosing Score Network Architecture}\label{sec:score_NN}
When training a diffusion model, we parameterize the score function using neural networks, where a properly chosen network architecture plays a vital role in effective training. The score decomposition in Lemma~\ref{lem: score decomposition -- heterogeneous} suggests a well-informed network architecture design. Before we introduce our network architecture, we briefly summarize our notion of ReLU networks considered in this paper.

Let $\calS_{\rm ReLU}$ be a family of neural networks with ReLU activations determined by a set of hyperparameters $L$, $M$, $J$, $K$, $\kappa$, $\gamma_1$, and~$\gamma_2$. Roughly speaking, $L$ is the depth of the network, $M$ is the width of the network, $J$ is the number of non-zero weight parameters, $K$ is the range of network output, $\kappa$ is the largest magnitude of weight parameters, and $\gamma_1$ as well as $\gamma_2$ are both Lipschitz coefficients as we detail below. Formally, considering that a score network takes noisy data $\mathbf r$ and time $t$ as input, we define $\calS_{\rm ReLU}$ as
\begin{equation}
\label{equ: ReLU network}
    \begin{aligned}
        &\quad \calS_{\rm ReLU}(L, M, J, K, \kappa, \gamma_1, \gamma_2) \\
        &= \Big\{ \mathbf g_{\bm\zeta}(\mathbf r, t) = \mathbf W_L \cdot {\rm ReLU}(\cdots {\rm ReLU}(\mathbf W_1[\mathbf z^\top, t]^\top + \mathbf b_1)\cdots) + \mathbf b_L \textrm{ with } \bm\zeta := \{\mathbf W_{\ell}, \mathbf b_{\ell}\}_{\ell = 1}^{L}: \\
        & \qquad \qquad  \qquad\text{network width bounded by } M, \ \underset{\mathbf r, t}{\sup} ~\| \mathbf g_{\bm\zeta}(\mathbf r, t) \|_2 \le K, \\
        &\qquad \qquad \qquad \max\{ \| \mathbf b_\ell \|_{\infty}, \| \mathbf W_\ell \|_{\infty} \} \le \kappa \text{ for } \ell = 1, \dots, L, ~\sum_{\ell=1}^{L} (\| \mathbf b_\ell \|_0 +  \| \mathbf W_\ell \|_0) \le J, \\
        & \qquad \qquad \qquad \| \mathbf g_{\bm\zeta}(\mathbf r_1, t) - \mathbf g_{\bm\zeta}(\mathbf r_2, t) \|_2 \le \gamma_1 \| \mathbf r_1 - \mathbf r_2 \|_2 \text{ for any } t \in (0, T], \\
        &\qquad  \qquad \qquad\| \mathbf g_{\bm\zeta}(\mathbf r, t_1) - \mathbf g_{\bm\zeta}(\mathbf r, t_2) \|_2 \le \gamma_2 |t_1 - t_2| \text{ for any } \mathbf r \Big\},
    \end{aligned}
\end{equation}
where ${\rm ReLU}$ activation is applied entrywise, and each weight matrix $\mathbf W_{\ell}$ is of dimension $d_{\ell} \times d_{\ell+1}$. Correspondingly, the width of the network is denoted by $M = \max_{\ell} d_\ell$. Here, the uniform bound $\sup_{{\bf r}, t} \|{\bf g}_{\zeta}({\bf r}, t)\|_2 \leq K$ and the sparsity constraint $\sum_{\ell=1}^{L}(\|{\bf b}_{\ell}\|_0 + \|{\bf W}_{\ell}\|_0) \leq J$ are standard and practically assumptions for ReLU networks  \citep{bartlett2017spectrally,louizos2017learning,hoefler2021sparsity,song2020score}.\footnote{The bounded-output condition $\sup_{{\bf r}, t} \|{\bf g}_{\zeta}({\bf r}, t)\|_2 \leq K$ is often enforced in practice by clipping the layer of ReLU networks (e.g., $g(a)=\mathrm{ReLU}(a-R)-\mathrm{ReLU}(a+R)-R$ clips to $[-R,R]$), which is a standard assumption for the score approximation in the score-based models  \citep{vincent2011connection,bartlett2017spectrally,ho2020denoising,song2020score}. The sparsity constraint $\sum_{\ell=1}^{L}(\|{\bf b}_{\ell}\|_0 + \|{\bf W}_{\ell}\|_0) \leq J$ directly controls the complexity of the neural network class and enters our covering-number bound in Lemma~B.4. Empirically, sparsity in neural networks is typically induced by regularization; for example, explicit $\ell_1$ regularization (akin to LASSO) on weight matrices induces sparsity  \citep{srinivas2017training,louizos2017learning} and large neural networks can be compressed by training a sparse sub-network without sacrificing performance  \citep{han2015deep,frankle2018lottery,hoefler2021sparsity}.} The Lipschitz continuity on $\mathbf g_{\bm\zeta}$ is often enforced by Lipschitz network training  \citep{gouk2021regularisation} or induced by implicit bias of the training algorithm  \citep{soudry2018implicit,bartlett2020benign}.

Now, using $\calS_{\rm ReLU}$, we design our score network architecture by first rearranging terms in \eqref{equ: score decomposition -- heterogeneous} as
\begin{align}
    \nabla \log p_t(\mathbf r) & = \bm\Lambda_t^{-1} \bm\beta \frac{\int \alpha_t \mathbf f \cdot \phi(\bm\Gamma_t\bm\beta^\top \bm\Lambda_t^{-1} \mathbf r ;\alpha_t \mathbf f, \bm\Gamma_t) p_{\textrm{fac}}(\mathbf f) \, \dd \mathbf f}{\int \phi(\bm\Gamma_t \bm\beta^\top \bm\Lambda_t^{-1} \mathbf r ;\alpha_t \mathbf f, \bm\Gamma_t) p_{\textrm{fac}}(\mathbf f) \, \dd \mathbf f} - \bm\Lambda_t^{-1} \mathbf r \nonumber \\
    & = \alpha_t\bm\Lambda_t^{-1} \bm\beta \cdot \bm\xi(\bm\beta^\top \bm\Lambda_t^{-1} \mathbf r, t) - \bm\Lambda_t^{-1} \mathbf r, \label{equ: score function rearranged}
\end{align}
where $\bm\xi: \bb R^{k} \times [0, T] \to \bb R^{k}$ is defined as 
\begin{equation}
\label{equ: definition of xi}
    \bm\xi(\mathbf z, t) := \frac{\int \mathbf f \cdot \phi(\bm\Gamma_t \mathbf z ;\alpha_t \mathbf f, \bm\Gamma_t) p_{\textrm{fac}}(\mathbf f) \, \dd \mathbf f}{\int \phi(\bm\Gamma_t \mathbf z ;\alpha_t \mathbf f, \bm\Gamma_t) p_{\textrm{fac}}(\mathbf f) \, \dd \mathbf f} \quad \text{for} \quad \mathbf z \in \mathbb{R}^k.
\end{equation}
The $i$-th element of $\bm\xi(\mathbf z, t)$ is denoted as $\xi_i(\mathbf z, t)$.
Note that the coefficient $\alpha_t$ forces the first term to decay exponentially. Therefore, for sufficiently large $t$, the score function $\nabla \log p_t(\mathbf{r})$ is approximately a linear function, corresponding to the second term in \eqref{equ: score function rearranged}. 

When choosing the score network architecture, we aim to reproduce the functional form in \eqref{equ: score function rearranged}. Accordingly, we define a class of neural networks built upon $\calS_{\rm ReLU}$ as
\begin{equation}
\label{equ: score network}
    \begin{aligned}
        &\calS_{\textrm{NN}}(L, M, J, K, \kappa, \gamma_1, \gamma_2,\sigma_{\max}) \\
        =& \Big\{ \mathbf s_{\mathbf \theta} (\mathbf r, t) = \alpha_t \mathbf D_t \mathbf V \cdot \mathbf g_{\bm\zeta}( \mathbf V^\top \mathbf D_t \mathbf r, t ) - \mathbf D_t \mathbf r \text{ with } \bm \theta:= \{\mathbf c, \mathbf V, \bm\zeta\}: \\
        &\qquad\qquad \mathbf c := [c_1, c_2, \dots, c_d]^\top \in [0, \sigma_{\max}]^d, \quad \mathbf V \in \bb R^{d \times k} \textrm{with orthogonal columns}, \\
        &\qquad\qquad \mathbf D_t := \operatorname{diag}\left\{ 1/(h_t + \alpha_t^2 c_1), \dots, 1/(h_t + \alpha_t^2 c_d) \right\} \textrm{induced by $\mathbf c$}, \\
        &\qquad\qquad \mathbf g_{\bm\zeta}\in \mathcal{S}_{\text{ReLU}}(L, M, J, K, \kappa, \gamma_1, \gamma_2)
        \Big\}.
    \end{aligned}
\end{equation}
In \eqref{equ: score network}, $\mathbf V$ represents the unknown factor loading $\bm\beta$ and $\mathbf D_t$ represents $\mathbf \Lambda_t^{-1}$. The ReLU network $\mathbf g_{\bm\zeta}$ is responsible for implementing $\bm\xi$. We remark that $\mathbf V^\top$ and $\mathbf V$ serve as the linear encoder and decoder, respectively, and  $-\mathbf D_t \mathbf r$ is incorporated as a shortcut connection within the U-Net framework  \citep{ronneberger2015u}. When there is no confusion, we drop the hyper-parameters and denote the network classes in \eqref{equ: ReLU network} and \eqref{equ: score network} as $\mathcal{S}_{\textrm{ReLU}}$ and $\calS_{\textrm{NN}}$, respectively.

\section{Score Approximation and Estimation}
\label{sec: score approximation and estimation}

Given the score decomposition and score network architecture $\calS_{\rm NN}$, this section establishes two intriguing properties: 1) with appropriate hyper-parameters, $\calS_{\rm NN}$ can well approximate any score function in the form \eqref{equ: score decomposition -- heterogeneous}, and 2) learning the score function within $\calS_{\rm NN}$ leads to an efficient sample complexity. Specifically, we establish an approximation theory to the score function in Section \ref{subsec: score approximation}. Building on the approximation guarantee, Section \ref{subsec: score estimation} derives bounds on the statistical error, providing finite-sample guarantees for score estimation, where the sample complexity bounds depend primarily on the number of factors $k$ rather than ambient dimension $d$.

\subsection{Theory of Score Approximation}
\label{subsec: score approximation}

The following assumptions on the factor distribution and score function are needed to establish our score approximation guarantee.

\begin{assumption}[Factor distribution]
\label{assumption: subgaussian}
    The density function for the factors, $p_{\rm{fac}}(\cdot)$, is non-negative and twice continuously differentiable. In addition $p_{\rm{fac}}(\cdot)$ has sub-Gaussian tail, namely, there exist constants $B, C_1$, and $C_2$ such that 
\begin{equation}
    p_{\rm{fac}}(\mathbf f) \leq (2\pi)^{-\frac{k}{2}}C_1 \exp(-C_2\|\mathbf f\|_2^2 /2) \textit{ when } \|\mathbf f\|_2 \ge B. \label{eqn:subgaussian}
\end{equation}
\end{assumption}
Assumption~\ref{assumption: subgaussian} is commonly adopted both in the literature of high-dimensional statistics \citep{vershynin2018high,wainwright2019high} and in recent work on diffusion-model theory \citep{de2021diffusion,chen2023score,oko2023diffusion,cole2024score}. In finance, Assumption~\ref{assumption: subgaussian} is standard in the factor/econometrics literature  \citep[e.g.,][]{bai2002determining, bai2023approximate} for modeling low-frequency returns, which are well known to exhibit the aggregated Gaussianity property \citep{fan2003nonlinear}. We also need the following regularity assumption on the score function.
\begin{assumption}
\label{assumption: Lipschitz}
    The subspace score function $\mathbf s_{\textrm{sub}}(\mathbf z, t)$ is $L_s$-Lipschitz in $\mathbf z$ for any~$t \in [0, T]$.
\end{assumption}

The Lipschitz assumption on the score function is a standard assumption in the diffusion model literature  \citep{lee2022convergence, chen2022sampling, han2024neural}. Note that Assumption \ref{assumption: Lipschitz} only requires the Lipschitz continuity for the subspace score. But it implies that $\nabla \log p_t$ is Lipschitz with coefficient $\left(L_s \cdot \frac{h_t + \sigma_1^2 \alpha_t^2}{h_t + \sigma_d^2 \alpha_t^2} + \frac{1}{h_t + \sigma_d^2 \alpha_t^2}\right)$, which is in a  similar  spirit to the condition proposed in \citet{lee2022convergence}. As a concrete example, a Gaussian distribution with a nondegenerate covariance satisfies Assumption \ref{assumption: Lipschitz}. 

\begin{example}[Gaussian factors]
\label{example: gaussian}
    Assume the factor $\mathbf F$ follows a nondegenerate Gaussian distribution, i.e., 
    \begin{equation}
    \label{equ: distribution of F in example gaussian}
    \mathbf F \sim \mathcal{N}(\mathbf 0, \bm\Sigma) \ \textrm{ with } \ \bm\Sigma = \operatorname{diag}\{\varsigma_{1}, \dots, \varsigma_{k}\} \succ \bm0.
    \end{equation}
    Then, an explicit calculation gives rise to
    $$
        \nabla \log p_t(\mathbf r) = ( -\bm\Lambda_t^{-1} \bm\beta \bm\Gamma_t (\bm\Gamma_t + \alpha_t^2 \bm\Sigma)^{-1} ) \bm\beta^{\top} \bm\Lambda_t^{\frac{1}{2}} \mathbf T_t \cdot \bm\Lambda_t^{-\frac{1}{2}} \mathbf r - \bm \Lambda_t^{-\frac{1}{2}} (\mathbf I - \mathbf T_t) \cdot \bm \Lambda_t^{-\frac{1}{2}} \mathbf r.
    $$
Correspondingly, the subspace score $\mathbf s_{\textrm{sub}}$ is written as
    \begin{equation*}
    \begin{aligned}
        \mathbf s_{\textrm{sub}}(\mathbf z, t) &= ( -\bm\Lambda_t^{-1} \bm\beta \bm\Gamma_t (\bm\Gamma_t + \alpha_t^2 \bm\Sigma)^{-1} ) \mathbf z,
    \end{aligned}
    \end{equation*}
    which is Lipschitz in $\mathbf z$.
\end{example}

We state our theory of score approximation as follows.
\begin{theorem}
\label{theorem: score approximation} 
    Suppose Assumptions \ref{assumption: factor}-\ref{assumption: Lipschitz} hold. Given an approximation error $\epsilon > 0$, there exists a network $\bar{\mathbf s}_{\bm\theta} \in \calS_{\rm NN}$ such that for any $t \in [0, T]$, it presents an upper bound
    \begin{equation}
        \label{equ: score approximation error}
        {\mathbb{E}_{\mathbf R_t \sim P_t}\|{\overline{\mathbf s}_{\bm\theta}(\mathbf R_t, t) - \nabla \log p_t(\mathbf R_t)}\|_2}
        \leq \frac{(\sqrt{k} + 1) \epsilon}{\min\{\sigma_d^2, 1\}}.
        \end{equation}
    The configuration of the network architecture $\calS_{\rm NN}$ satisfies
    \begin{equation}
    \begin{aligned}
    \label{equ: order of parameters in Theorem 1}
        & M=\order\bigg(T \tau (1+L_s)^k (1+\sigma_{\max}^k) \epsilon^{-(k+1)} \big(\log \frac{1}{\epsilon} + k\big)^{\frac{k}{2}}\bigg),~ \gamma_1 = 20k(1+L_s)(1+\sigma_{\max}^4), \\
        & L=\order\bigg(\log \frac{1}{\epsilon}+ k \bigg),~ J=\order\bigg(T \tau (1+L_s)^k (1+\sigma_{\max}^{k}) \epsilon^{-(k+1)} \big(\log \frac{1}{\epsilon} + k \big)^{\frac{k+2}{2}}\bigg),~ \gamma_{2}=10 \tau, \\
        & K = \order\bigg((1+L_s) (1+\sigma_{\max}^{4}) \left(\log \frac{1}{\epsilon} + k\right)^{\frac{1}{2}}\bigg),~ \kappa=\max \big\{(1+L_s) (1+\sigma_{\max}^4)\big(\log \frac{1}{\epsilon} + k\big)^{\frac{1}{2}}, T \tau\big\},
    \end{aligned}
    \end{equation}
    where $$\tau = \sup_{t \in [0, T], \| \mathbf z \|_\infty \leq \sqrt{(1+\sigma_{\max}^2)(k + (\log 1/\epsilon))}} \left\|\frac{\partial}{\partial t}\bm\xi(\mathbf z, t)\right\|_2\quad \text{with $\bm\xi$ defined in \eqref{equ: definition of xi}}.$$
\end{theorem}

\smallskip

As shown in \eqref{equ: score approximation error}, the approximation error has a benign dependence on the dimension. It primarily depends on $\min\{\sigma_d^2, 1\}$ and $k$, rather than $d$. The proof is deferred to Appendix \ref{subsec: proof of theorem -- score approximation}. Below, we provide key insights offered by Theorem \ref{theorem: score approximation}, along with a proof sketch and a discussion of the main technical challenges.
\smallskip

\paragraph{Discussion on Network Architecture.} In contrast to conventional neural network designs for universal approximation, such as those in \cite{yarotsky2017error}, our network employs only Lipschitz functions $\mathbf g_{\bm\zeta}$ rather than a broad family of unrestricted functions. As illustrated in \eqref{equ: score network}, we incorporate time $t$ as an additional input, and the network size is determined solely by the $k$-dimensional space due to the encoder-decoder architecture. Our results indicate that the error bound is determined by $k$ and remains free of the Lipschitz parameters $\gamma_1$ and $\gamma_2$.


\paragraph{Technical Challenges and Proof Overview.} One key challenge lies in approximating the score function under the factor model \eqref{equ: factor model} when data presents high-dimensional noise $\bm\varepsilon$. To address this challenge, we utilize the score function decomposition in \eqref{equ: score function rearranged} to separately approximate the low-dimensional term $\bm\xi(\mathbf z, t)$ and the noise-related term $\bm\Lambda_t ^{-1/2}\mathbf r$. With the designed network architecture in \eqref{equ: score network}, the noise-related term can be perfectly captured by setting $\mathbf D_t = \bm\Lambda_t^{-1}$. For the low-dimensional term,  we provide an approximation based on a partition of $\bb R^k$ into a compact subset $\calC = \{\mathbf z \in \bb{R}^k : \|\mathbf z\|_2 \le S\}$ with a radius $S = \order(\sqrt{(1+\sigma_{\max}^2)(k+\log (1/\epsilon))})$ and its complement. Specifically, we construct a network $\overline{\mathbf g}_{\bm\zeta}$ to achieve an $L^{\infty}$ approximation guarantee within the set $\calC \times [0, T]$, and take $\overline{\mathbf g}_{\bm\zeta} = 0$ in the complement of $\calC \times [0, T]$.

To construct $\overline{\mathbf g}_{\bm\zeta}$  as an approximation to $\bm\xi(\mathbf z, t)$ over the domain $\mathcal{C}\times[0,T]$, we begin by forming a uniform grid of hypercubes covering $\mathcal{C}\times[0,T]$ and build local approximations within each hypercube. For the $i$-th component $\xi_i$ of $\bm\xi$, we use a Taylor polynomial $\Bar{g}_i$ to obtain a local approximation satisfying $\|\Bar{g}_i - \xi_i\|_{\infty} = \order(\epsilon)$ on each hypercube. Since ReLU networks can approximate polynomials to arbitrary accuracy in the $L^\infty$ norm, we construct a network $\Bar{g}_{\bm\zeta, i}$ that approximates $\Bar{g}_i$ within error $\epsilon/2$. By combining these approximations across all hypercubes, we obtain a network $\overline{\mathbf g}_{\bm \zeta}$ that achieves an $L^\infty$ approximation of $\bm \xi$ on $\mathcal{C} \times [0, T]$.

Finally, the proof of Theorem~\ref{theorem: score approximation} is completed by showing that the $L^2$ approximation error on the complement of $\calC \times [0, T]$ can be well controlled due to the sub-Gaussian tail property assumed in Assumption~\ref{assumption: subgaussian}. Note that the designed network architecture takes the form $\overline{\mathbf s}_{\bm\theta}(\mathbf r, t) = \alpha_t \bm\Lambda_t^{-1} \bm\beta \, \overline{\mathbf g}_{\bm\zeta}(\bm\beta^\top \bm\Lambda_t^{-1} \mathbf r, t) - \bm\Lambda_t^{-1} \mathbf r$. See the details in Appendix \ref{subsec: proof of theorem -- score approximation}.

\subsection{Theory of Score Estimation}
\label{subsec: score estimation}
We now turn to the estimation of score functions using a finite number of samples. With the score function parameterized by $\calS_{\textrm{NN}}$ in \eqref{equ: score network}, we can express the score matching objective as
\begin{equation}
\label{equ: definition of hat s}
    \widehat{\mathbf s}_{\bm\theta} = \underset{\mathbf s_{\bm\theta} \in \calS_{\textrm{NN}}}{\arg\min}\ \widehat{\calL}(\mathbf s_{\bm\theta}),
\end{equation}
where recall $\widehat{\calL}$ is defined in \eqref{equ: empirical loss function}.
Given $n$ i.i.d. samples, we provide an $L^2$ error bound for the neural score estimator $\widehat{\mathbf s}_{\bm\theta}$. The result is presented in the following theorem.

\begin{theorem}
\label{theorem: score estimation}
Suppose Assumptions \ref{assumption: factor}-\ref{assumption: Lipschitz} hold. We choose $\calS_{\rm NN}$ in Theorem~\ref{theorem: score approximation} with $\epsilon = n^{-\frac{1-\delta(n)}{k+5}}$ for $\delta(n)=\frac{(k+10)\log(\log n)}{2\log n}$. Given $n$ i.i.d. samples from $P_{\mathrm{data}}$, with probability $1-\frac{1}{n}$, it holds that
$$
\begin{aligned}
   \frac{1}{T-t_0} \int_{t_0}^{T} \E_{\mathbf R_t \sim P_t} \left[ \left\|\widehat{\mathbf s}_{\bm\theta}(\mathbf R_t, t)-\nabla \log p_{t}(\mathbf R_t)\right\|_2^2 \right] \dd t  = \Tilde{\order}\bigg(\frac{1}{t_0}(1 + \sigma_{\max}^{2k}) d^{\frac{5}{2}} k^{\frac{k+10}{2}} n^{-\frac{2-2 \delta(n)}{k+5}} \log^4 n\bigg),
\end{aligned}
$$
where $\Tilde{\order}(\cdot)$ omits factors associated with $L_s$ and polynomial factors on $\log t_0$, $\log d$, and $\log k$. 
\end{theorem}

\paragraph{Discussion on the Convergence Rate.} The convergence rate in Theorem~\ref{theorem: score estimation} depends not only on the intrinsic factor dimension $k$ but also mildly on the asset return dimension $d$, which appears only in a non-leading polynomial term. This polynomial dependency arises because the noise term $\bm\varepsilon$ spans the entire $\mathbb{R}^d$ space, introducing a truncation error component that scales with $d$. Fortunately, this dependency does not appear in the leading term $n^{-\frac{2-2 \delta(n)}{k+5}}$, where $\delta(n) = \frac{(k+10) \log \log n}{2 \log n}$. This suggests that the convergence rate is primarily dominated by the sample size $n$ and the latent factor dimensionality $k$, rather than the ambient dimensionality $d$. When $n$ is sufficiently large, $\delta(n)$ becomes negligible, indicating the squared $L^2$ estimation error to converge at the rate of $\Tilde{\order}\left( \frac{1}{t_0}(1 + \sigma_{\max}^{2k}) d^{\frac{5}{2}} k^{\frac{k+10}{2}} n^{-\frac{2}{k+5}} \log^4 n \right)$.

\paragraph{Proof Sketch.} The full proof is deferred to Appendix \ref{subsec: proof of theorem -- score estimation}; here,  we present a sketch of the main argument. The proof relies on a decomposition of the population loss $\calL(\widehat{\mathbf s}_{\bm\theta})$. Specifically, for any $a \in (0,1)$, it holds that
\begin{equation*}
\calL\left(\widehat{\mathbf s}_{\bm\theta}\right) \leq \underbrace{\calL^{\textrm{trunc}}\left(\widehat{\mathbf s}_{\bm\theta}\right) - (1+a) \widehat{\calL}^{\textrm{trunc}}\left(\widehat{\mathbf s}_{\bm\theta}\right)}_{(A)}+\underbrace{\calL\left(\widehat{\mathbf s}_{\bm\theta}\right)-\calL^{\textrm{trunc}}\left(\widehat{\mathbf s}_{\bm\theta}\right)}_{(B)} + (1+a) \underbrace{\inf _{\mathbf s_{\bm\theta} \in \calS_{\textrm{NN}}} \widehat{\calL}\left(\mathbf s_{\bm\theta}\right)}_{(C)},
\end{equation*}
where $\calL^{\textrm{trunc}}$ is defined as
\begin{equation*}
\calL^{\textrm{trunc}}\left(\mathbf s_{\bm\theta}\right) := \int \ell^{\textrm{trunc}}(\mathbf r ; \mathbf s_{\bm\theta}) p_t(\mathbf r) \dd \mathbf r \quad \textrm{ with } \quad \ell^{\textrm{trunc}}(\mathbf r ; \mathbf s_{\bm\theta}) := \ell(\mathbf r ; \mathbf s_{\bm\theta}) \indicator\left\{\|\mathbf r\|_{2} \leq \rho \right\},
\end{equation*}
and a truncation radius $\rho$ to be determined. Here, the term $(A)$ captures the statistical error due to finite (training) samples, while terms $(B)$ and $(C)$ represent sources of \emph{bias} in the estimation of the score function. Specifically, $(B)$ captures the domain truncation error, while $(C)$ accounts for the approximation error of $\calS_{\textrm{NN}}$. We bound terms $(A), (B)$, and $(C)$  separately. For term $(A)$, we utilize a Bernstein-type concentration inequality on a compact domain. In addition, we show that the term $(B)$ is non-leading for sufficiently large radius $\rho$, thanks to the sub-Gaussian tail conditions. Then, we show that term $(C)$ is bounded by the network approximation error~\eqref{equ: score approximation error} in Theorem~\ref{theorem: score approximation}. To balance these three terms, we choose $\rho = \order(\sqrt{d + \log n})$, $a = n^{-\frac{1-\delta(n)}{k+5}}$, and set $\calS_{\textrm{NN}}$ in Theorem~\ref{theorem: score approximation} with $\epsilon = n^{-\frac{1 - \delta(n)}{k+5}}$ to obtain the desired result.

\section{Theory of Distribution Estimation}
\label{sec: distribution estimation}
This section establishes statistical guarantees for the estimation of high-dimensional return distribution.  Given the neural score estimator $\widehat{\mathbf s}_{\bm\theta}$ in Theorem \ref{theorem: score estimation}, we define the learned distribution $\widehat{P}_{t_0}$ as the marginal distribution of $\widehat{\mathbf R}_{T-t_0}^{\leftarrow}$ in \eqref{equ: learned backward SDE}, starting from $\mathcal{N}(\mathbf 0, \mathbf I_{d})$. To assess the quality of $\widehat{P}_{t_0}$,  we examine two key aspects: the estimation error relative to the ground-truth distribution $P_{\textrm{data}}$ and the accuracy of reconstructing the latent factor space.

\subsection{Main Results}
\label{sec:main_results_dis}
We estimate the latent subspace using generated samples as described in Algorithm \ref{algo: sampling}.
\begin{algorithm}
\caption{Sampling and Singular Value Decomposition (SVD)}
\label{algo: sampling}
\begin{algorithmic}[1]
\Require Score network $\widehat{\mathbf s}_{\bm\theta}$ in Theorem \ref{theorem: score estimation}, number of generated data $m$, and time $t_0$ and $T$.
\State Generate $m$ random samples $\{\mathbf R_{1}, \dots, \mathbf R_{m}\}$ at early stopping time $t_0$ via the backward process \eqref{equ: learned backward SDE}.\footnote{For practical implementation, we can use denoising diffusion probabilistic models (DDPM) discretization   \citep{ho2020denoising}. For $i=1, 2, \dots, m,$
$$
\mathbf R_{i, t_{j-1}} = \frac{1}{\sqrt{\alpha_{t_{j}}}}(\mathbf R_{i, t_{j}}+(1-\alpha_t)\widehat{\mathbf s}_{\bm\theta}(\mathbf R_{i, t_{j}}, t_{j})) + \frac{1-\alpha_{t_{j}}}{\alpha_{t_{j}}} \mathbf z_{t_{j}},   \textrm{ with } \mathbf R_{i, T} \sim \mathcal{N}(\mathbf 0, \mathbf I_d),
$$
where $t_0 < t_1 < t_2 \dots < t_\ell = T$ and $\{\mathbf z_t\}_{t=t_0}^{T}$ are i.i.d. following $\mathcal{N}(\mathbf 0, \mathbf I_d)$. In our analysis, the discretization error is a minor term that does not affect the main theoretical results. More specifically, previous work on score-based diffusion sampling (e.g., \citet{chen2022sampling,li2024sharp,tang2024score}) has shown that for a time step size $\Delta t$, the discretization error scales as $\mathcal{O}(\sqrt{\Delta t})$ and hence constitutes a negligible $o(1)$ when $\Delta t$ is sufficiently small.}
\State Perform SVD on sample covariance matrix:
    \begin{equation}
    \label{equ: cov and estimated cov}
    \widehat{\bm\Sigma}_0 := \frac{1}{m-1} \sum_{i=1}^{m}(\mathbf R_i - \bar{\mathbf R}) (\mathbf R_i - \bar{\mathbf R})^\top \quad \text{with} \quad \bar{\mathbf R} = \frac{1}{m} \sum_{i=1}^m \mathbf R_i.
    \end{equation}
\State Obtain the largest $k$ eigenvalues $\{\widehat{\lambda}_1, \dots, \widehat{\lambda}_k\}$ and the corresponding $k$-dimensional eigenspace $\widehat{\mathbf U} \in \bb R^{d \times k}$.
\State \Return $\{\mathbf R_{1}, \dots, \mathbf R_{m}\}$, $\widehat{\bm\Sigma}_0$, $\{\widehat{\lambda}_1, \dots, \widehat{\lambda}_k\}$, and $\widehat{\mathbf U}$.
\end{algorithmic}
\end{algorithm}

The following theorem shows that the simulated distribution and the recovered latent factor subspace are accurate with high probability.
\begin{theorem}
\label{theorem: distribution estimation}
Given the neural score estimator $\widehat{\mathbf s}_{\bm\theta}$ in Theorem \ref{theorem: score estimation}, we choose
$T = \frac{(4\gamma_1+2)(1-\delta(n))}{k+5} \log n$ and $t_0 = n^{-\frac{1-\delta(n)}{k+5}}$,
where $\gamma_1$ is the Lipschitz parameter in Theorem \ref{theorem: score approximation}. Denote {\tt Eigen-gap}$(k)= \lambda_k(\bm\Sigma_0) - \lambda_{k+1}(\bm\Sigma_0)$ with the covariance matrix of returns $\bm\Sigma_0$ defined in \eqref{equ: ground-truth covariance matrix}. Further denote $\mathbf U \in \mathbb{R}^{d \times k}$ as the $k$-dimensional leading eigenspace of $\bm \Sigma_0$. Then, the following two results hold.
\begin{enumerate}
    \item{\textbf{Estimation of return distribution.}} With probability $1 - 1/n$, the total variation distance between $\widehat{P}_{t_0}$ and $P_{\textrm{data}}$ satisfies
    $$
    \operatorname{TV}\left(P_{\textrm{data}}, \widehat{P}_{t_0}\right)
    = \Tilde{\order} \left((1+\sigma_{\max}^k) d^{\frac{5}{4}} k^{\frac{k+10}{4}} n^{-\frac{1-\delta(n)}{2(k+5)}} \log^{\frac{5}{2}} n \right).
    $$
    \item{\textbf{Latent subspace recovery.}} 
    Set $ m = \Tilde{\order}\left(\lambda_{\max}^{-2}(\bm\Sigma_{0}) d n^{\frac{2(1-\delta(n))}{k+5}}\log n \right).$
    For any $1 \leq i \leq k$, with probability $1 - 1/n$, it holds that 
    $$
        \left| \frac{\lambda_i (\widehat{\bm\Sigma}_{0})}{\lambda_i ( \bm\Sigma_0 )} - 1\right| = \Tilde{\order}\Bigg(\frac{\lambda_{\max}(\bm\Sigma_0) (1+\sigma_{\max}^k) d^{\frac{5}{4}} k^{\frac{k+10}{4}}}{\lambda_i(\bm\Sigma_0)} \cdot n^{-\frac{1-\delta(n)}{k+5}} \log^{\frac{5}{2}} n \Bigg).
    $$
    Meanwhile, the corresponding $k$-dimensional eigenspace can be recovered with
    $$
        \|\widehat{\mathbf U} \widehat{\mathbf U}^\top - \mathbf U \mathbf U^\top\|_{\rm{F}} = \Tilde{\order}\Bigg(\frac{\lambda_{\max}(\bm\Sigma_0) (1+\sigma_{\max}^k) d^{\frac{5}{4}} k^{\frac{k+12}{4}} }{\textrm{\tt Eigen-gap}(k)} \cdot n^{-\frac{1-\delta(n)}{k+5}} \log^{\frac{5}{2}} n\Bigg),
    $$
    where recall that $\widehat{\mathbf U}$ is the $k$-dimensional leading eigenspaces of $\widehat{\bm\Sigma}_{0}$.
\end{enumerate}
\end{theorem}

A few explanations are in line.

\paragraph{Trade-off on Early Stopping.} The distribution estimation in Theorem \ref{theorem: distribution estimation} highlights a trade-off associated with $t_0$. 
Specifically, we can upper bound $\operatorname{TV}(P_{\textrm{data}}, \widehat{P}_{t_0})$ by three terms
\begin{equation}
\label{equ: TV distance decomposition in Theorem 3}
\operatorname{TV}(P_{\textrm{data}}, \widehat{P}_{t_0}) \leq \operatorname{TV}(P_{\textrm{data}}, P_{t_0}) + \operatorname{TV}(P_{t_0}, \widetilde{P}_{t_0}) + \operatorname{TV}(\widetilde{P}_{t_0}, \widehat{P}_{t_0}),
\end{equation}
where $\widetilde{P}_{t_0}$ is the  distribution of $\widehat{\mathbf R}_{T - t_0}^{\leftarrow}$, defined in \eqref{equ: learned backward SDE}, initialized with $\widehat{\mathbf R}_{0}^{\leftarrow} \sim P_T$.
As shown in \eqref{equ: TV distance decomposition in Theorem 3}, the latent distribution error $\operatorname{TV}(P_{\textrm{data}}, \widehat{P}_{t_0})$ arises from early stopping, score network estimation, and the mixing of forward process \eqref{equ: diffusion forward SDE}. As $t_0$ increases, the score estimation error decreases according to Theorem~\ref{theorem: score estimation}. As a result, the error term $\operatorname{TV}(P_{t_0}, \widetilde{P}_{t_0})$ decreases. However, the early stopping error $\operatorname{TV}(P_{\textrm{data}}, P_{t_0})$ increases due to the heavier injected Gaussian noise.  Under a training horizon of $T=\Tilde{\order}(\log n)$, the choice of $t_0 = n^{-\frac{1-\delta(n)}{k+5}}$ optimally balances the early stopping error and the score estimation error.
    
\paragraph{Eigenspace Estimation using Generated Samples.} The latent subspace estimation in Theorem \ref{theorem: distribution estimation} shows that the subspace can be accurately recovered with high probability. Specifically, generating $ \Tilde{\order}\big(d n^{\frac{2(1-\delta(n))}{k+5}} \log n\big) $ samples from the trained diffusion model ensures that the eigenvalues and eigenspace of the sample covariance matrix $ \widehat{\bm\Sigma}_{0} $ closely approximate those of $ \bm\Sigma_0 $, with the error proportional to the score estimation error. Moreover, if {\tt Eigen-gap}$(k)$ increases—indicating an improvement in the factor model identification, then the estimation error of the $ k $-dimensional eigenspace decreases.

\paragraph{Further Discussion on Dimension Dependence.} Our sample complexity bounds in Theorem~\ref{theorem: distribution estimation} circumvent the curse of ambient dimensionality $d$ under very mild assumptions, namely, score function being Lipschitz and the distribution of factors being sub-Gaussian. In the meantime, we emphasize that our focus is not on optimizing the non-leading term (e.g., $k^{\frac{k+10}{4}}$) to derive a sharp bound, but on the structural dependence on $k$ (versus $d$). As a result, these bounds characterize learning efficiency being adaptive to the subspace dimension $k$ even in the most challenging scenarios. In practical applications, however, data distributions often possess more favorable regularity properties---such as higher-order smoothness in the score function or the distribution of returns---which may lead to better learning efficiency compared to the theoretical bound. For example, if one assumes that the $k$-dimensional subspace score function is Hölder-$s$ continuous, the error bound can be improved to $n^{-\frac{s}{2s+k}}$. As $s \to \infty$, the rate is approximately $n^{-\frac{1}{2}}$, which has been shown to be optimal  \citep{tsy2009nonparametric}. While refining our bounds under such additional properties is beyond the scope of this paper, we present comprehensive numerical and empirical results in Sections~\ref{sec:synthetic} and \ref{sec: empirical} to illustrate the strong performance of diffusion factor models, particularly in the small-data regime.

\paragraph{Proof Sketch.} The proof is deferred to Appendix \ref{subsec: proof of theorem -- distribution estimation}; here, we highlight its main ideas. The outline has two parts: (I) the key steps in establishing the distribution estimation result in Theorem~\ref{theorem: distribution estimation}, and (II) the technical components for proving the latent subspace recovery results, emphasizing novel coupling and concentration arguments.

\noindent \underline{(I) Estimation of return distribution.} We bound each term in the decomposition \eqref{equ: TV distance decomposition in Theorem 3} separately.
\begin{enumerate}
    \item Term $\operatorname{TV}(P_{\textrm{data}}, P_{t_0})$ is the early-stopping error. By direct calculations using the Gaussian transition kernel, we show that it is bounded by $\order(d t_0)$.
    
    \item Term $\operatorname{TV}(P_{t_0}, \widetilde{P}_{t_0})$ captures the statistical estimation error. We apply Girsanov's Theorem (\citealp{karatzas1991brownian}, Theorem 5.1; \citealp{revuz2013continuous}, Theorem 1.4) to show that the KL divergence ${\rm KL}(P_{t_0}, \widetilde{P}_{t_0})$ is bounded by the $L^2$ score estimation error developed in Theorem~\ref{theorem: score estimation}. Further, by Pinsker’s inequality \cite[Lemma 2.5]{tsy2009nonparametric}, we convert the KL divergence bound into a total variation distance bound. 
    
    \item Term $\operatorname{TV}(\widetilde{P}_{t_0}, \widehat{P}_{t_0})$ reflects the mixing error of the forward process \eqref{equ: diffusion forward SDE}. Using the data processing inequality \cite[Theorem 2.8.1]{thomas2006elements}, we show that it is a non-leading error term of order $\tilde{\mathcal{O}}(\exp(-T))$.
\end{enumerate}

\noindent \underline{(II) Latent subspace recovery.} The crux is to bound the covariance estimation error $\| \widehat{\bm\Sigma}_0 - \bm\Sigma_0 \|_{\mathrm{op}}$ by the following lemma.
\begin{lemma}
\label{lem: eigenvalue estimation error}
Assume the same assumptions as in Theorem~\ref{theorem: distribution estimation} and take  $\widehat{\mathbf \Sigma}_0$ as the estimator in \eqref{equ: cov and estimated cov} with  $m$ samples from Algorithm~\ref{algo: sampling}. It holds that, with probability at least $1 - \delta$,
    \begin{equation}
    \begin{aligned}
    \label{equ: upper bound of opnerator norm error}
        \| \widehat{\mathbf \Sigma}_0 - \bm\Sigma_0\|_{\rm{op}} = \order \bigg(\lambda_{\max}(\bm\Sigma_0) (1+\sigma_{\max}^k) d^{\frac{5}{4}} k^{\frac{k+10}{4}} n^{-\frac{1-\delta(n)}{k+5}} \log^{\frac{5}{2}} n \bigg).
    \end{aligned}
    \end{equation}
Here, $m$ satisfies
   \begin{equation}
   \label{equ: order of m in Lemma 3}
        m = \order\left( \lambda_{\max}^{-2}(\bm\Sigma_{0}) d n^{\frac{2(1-\delta(n))}{k+5}}\log n \right).
   \end{equation}
\end{lemma}
The complete proof of Lemma~\ref{lem: eigenvalue estimation error} is deferred to Appendix~\ref{pf: eigenvalue estimation error}. Using Lemma~\ref{lem: eigenvalue estimation error} in combination with Weyl’s theorem and Davis-Kahan theorem  \citep{davis1970rotation}, we derive the desired results for latent subspace recovery.

Proving Lemma~\ref{lem: eigenvalue estimation error} is similar to that for the estimation of return distribution. We upper bound $\| \widehat{\mathbf \Sigma}_0 - \bm\Sigma_0 \|_{\rm{op}}$ by
\begin{align*}
\big\| \widehat{\bm\Sigma}_{0} - \bm\Sigma_0 \big\|_{\textrm{op}} &\leq \underbrace{\big\| \bm\Sigma_0 - \bm\Sigma_{t_0} \big\|_{\textrm{op}}}_{(A)} + \underbrace{\big\| \bm\Sigma_{t_0} - \widetilde{\bm\Sigma}_{t_0} \big\|_{\textrm{op}}}_{(B)} + \underbrace{\big\|\widetilde{\bm\Sigma}_{t_0} - \widecheck{\bm\Sigma}_{t_0} \big\|_{\textrm{op}}}_{(C)} + \underbrace{\big\| \widehat{\bm\Sigma}_{0} - \widecheck{\bm\Sigma}_{t_0} \big\|_{\textrm{op}}}_{(D)},
\end{align*}
where $\bm\Sigma_{t_0}$, $\widetilde{\bm\Sigma}_{t_0}$, $\widecheck{\bm\Sigma}_{t_0}$ are the covariance of $P_{t_0}$, $\widetilde{P}_{t_0}$ and $\widehat{P}_{t_0}$, respectively. Analogous to the upper bound of the total variation distance in \eqref{equ: TV distance decomposition in Theorem 3}, term $(A)$ corresponds to the early-stopping error; term $(B)$ captures the statistical estimation error; and term $(C)$ reflects the mixing error. The additional term $(D)$ represents a finite-sample concentration error arising from the use of $m$ samples in Algorithm~\ref{algo: sampling}.

We bound each term separately. Term $(A)$ can be bounded by direct calculations using the Gaussian transition kernel; term $(D)$ is bounded using matrix concentration inequalities  \citep[Theorems 3.1.1 and 4.6.1]{vershynin2018high}. However, bounding terms $(B)$ and $(C)$ requires a novel analysis, as small total variation distances $\operatorname{TV}(P_{t_0} ,\widetilde{P}_{t_0})$ and $\operatorname{TV}(\widetilde{P}_{t_0}, \widehat{P}_{t_0})$   do not immediately imply small error bounds on the covariance matrix. In fact, we show the following $L^2$ bound based on a coupling between two backward SDEs, which converts to bounds on $(B)$ and $(C)$.

\begin{lemma}
\label{lem: backward SDEs L2 error}
    Assume the same assumptions as in Theorem~\ref{theorem: distribution estimation}. Consider the following coupled SDEs:
    \begin{equation}
    \label{equ: coupling SDEs system}
    \left\{
    \begin{array}{rcl}
        \dd \bRb &=& \Big( \frac{1}{2}\bRb + \nabla \log p_{T-t}(\bRb) \Big) \dd t + \dd \Bar{\mathbf W}_t, \ \ \mathrm{ with } \ \ \mathbf R_0^{\leftarrow} \sim P_T, \\
        \dd \hatbRb &=& \Big( \frac{1}{2}\hatbRb + \widehat{\mathbf s}_{\bm\theta}( \hatbRb, T-t ) \Big) \dd t + \dd \Bar{\mathbf W}_t, \ \ \mathrm{ with } \ \ \widehat{\mathbf R}_0^{\leftarrow} \sim \calN(\mathbf 0, \mathbf I_d) \ \mathrm{ or } \ P_T,
    \end{array}
    \right.
    \end{equation}
    where $P_T$ is the terminal distribution of the forward SDE \eqref{equ: diffusion forward SDE}. It holds that
    \begin{equation}
    \label{equ: score estimation bound in Lemma 2 of Theorem 3}
        \E\| \mathbf R_{T-t_0}^{\leftarrow} - \widehat{\mathbf R}_{T-t_0}^{\leftarrow} \|_2^2 = \order \left( (1+\sigma_{\max}^k) d^{\frac{5}{4}} k^{\frac{k+10}{4}} n^{-\frac{1-\delta(n)}{k+5}} \log^{\frac{5}{2}} n \right).
    \end{equation}
\end{lemma}
The proof of Lemma~\ref{lem: backward SDEs L2 error} is deferred to Appendix \ref{pf: backward SDEs L2 error}. By the Cauchy-Schwarz inequality and Lemma~\ref{lem: backward SDEs L2 error}, we bound $\|\bm\Sigma_{t_0} - \widetilde{\bm\Sigma}_{t_0}\|_{\textrm{op}}$ as well as $\|\widetilde{\bm\Sigma}_{t_0} - \widecheck{\bm\Sigma}_{t_0}\|_{\textrm{op}}$ by $\order \big( \sqrt{ \E \| \mathbf R_{T-t_0}^{\leftarrow} - \widehat{\mathbf R}_{T-t_0}^{\leftarrow} \|_2^2 } \cdot \big( \sqrt{ \E \| \mathbf R_{T-t_0}^{\leftarrow} \|_2^2 } + \sqrt{ \E \big\| \widehat{\mathbf R}_{T-t_0}^{\leftarrow} \big\|_2^2 } \big) \big)$, where the second moments of $\mathbf R_{T-t_0}^{\leftarrow}$ and $\hat{\mathbf R}_{T-t_0}^{\leftarrow}$ are clearly finite. Putting together all the error terms, we complete the proof of Lemma~\ref{lem: eigenvalue estimation error}.

\subsection{Highlights of Technical Novelties}
\label{sec:technical_highlights}
With the statements of our theoretical results in place, we now summarize the main differences between our setting and that of \cite{chen2023score}, and highlight the corresponding technical novelties.

\paragraph{Difference in Data Structure.} Our setting differs fundamentally from \citet{chen2023score} because our data follow a factor-model structure with heterogeneous idiosyncratic noise that spans the full high-dimensional space, rather than lying on a noise-free low-dimensional subspace, which is fundamental and crucial for financial applications. The perturbed data $\mathbf R_t$ therefore contain two sources of noise---a homogeneous diffusion noise with variance $h_t$ and a heterogeneous residual noise $\alpha_t^2\sigma_i^2$---so the $i$-th coordinate is perturbed by Gaussian noise with a total variance of $h_t + \alpha_t^2\sigma_i^2$. This heterogeneity introduces substantial technical challenges, particularly in controlling how this noise propagates into the final approximation error.

\paragraph{Technical Differences/Novelties.} To address the challenges posed by high-dimensional idiosyncratic noise, we develop a time-varying score decomposition alongside a tailored neural network architecture. Specifically, we derive a time-varying orthogonal decomposition of the score into subspace and complementary components, which we operationalize through a factor-aware, time-varying encoder-decoder with skip connections. We further introduce a time-dependent projection matrix that allows the model to effectively accommodate heterogeneous, high-dimensional noise. In contrast, in the noise-free linear setting of \cite{chen2023score}, a fixed projection matrix suffices for both score decomposition and network design.

From a technical standpoint, Lemmas \ref{lem: eigenvalue estimation error} and \ref{lem: backward SDEs L2 error} are both novel and essential to our analysis, and they highlight a key distinction between our work and \cite{chen2023score}:
    \begin{itemize}
    \item Lemma~\ref{lem: eigenvalue estimation error} establishes an error bound for covariance estimation and latent subspace recovery using samples generated by Algorithm~\ref{algo: sampling}, with only mild dependence on $d$ (non-leading order) despite the presence of heterogeneous high-dimensional noise. Because the idiosyncratic noise is full-rank, the problem cannot be reduced to a $k$-dimensional setting, making the analysis substantially more challenging. Consequently, our theoretical guarantees are derived in the full $d$-dimensional space, in contrast to \citet{chen2023score}, who work entirely within the $k$-dimensional latent subspace.
    
    \item Lemma~\ref{lem: backward SDEs L2 error} quantifies the discrepancy between the reverse processes driven by the true score $\nabla \log p_t$ and the learned score $\widehat{\mathbf s}_{\bm\theta}$, thereby converting score estimation error into distributional error (e.g., total variation). Unlike \citet{chen2023score}, which analyzes a time-reverse process projected onto a time-invariant subspace, our setting involves a time-varying latent subspace induced by the projections $\bm\Lambda_t^{-1/2}\bm\beta$. This introduces significant analytical challenges in characterizing the relationship between the corresponding time-varying projected backward processes. To overcome these difficulties, we develop a coupling argument that enables a general comparison between the true and estimated reverse dynamics $\mathbf R_{T-t_0}^{\leftarrow}$ and $\widehat{\mathbf R}_{T-t_0}^{\leftarrow}$ in \eqref{equ: diffusion backward SDE} and \eqref{equ: learned backward SDE}.
    
    \end{itemize}

\section{Numerical Study with Synthetic Data}\label{sec:synthetic}
In this section, we use our diffusion factor model to learn high-dimensional asset returns under a synthetic factor model setup. We numerically evaluate its effectiveness in terms of recovering both the latent subspace and the return distribution (as in Theorem~\ref{theorem: distribution estimation}).

To simulate a practically challenging scenario, we follow the widely used practice in the econometrics literature \citep{bai2002determining,bai2023approximate} to simulate a latent factor model with $d=2^{11}=2048$ assets and $k=16$ factors, which satisfies Assumptions~\ref{assumption: factor}--\ref{assumption: Lipschitz} in our framework. For diffusion models, we use a U-Net  \citep{ronneberger2015u} as a practical implementation of our theoretical neural network $\mathcal S_{\textrm{NN}}$ in \eqref{equ: score network}, and set its bottleneck width as the factor dimension $k$. Appendix \ref{app:exp_setup} provides more details on how we construct simulated returns and train diffusion models.

We denote $\bm\mu$ and $\bm{\Sigma}$ as the ground truth mean and covariance matrix of returns. Similarly, we denote $\bm\mu_{\textrm{Diff}}$ and $\bm{\Sigma}_{\textrm{Diff}}$ as the mean and covariance matrix estimated using diffusion-generated data, and $\bm\mu_{\textrm{Emp}}$ and $\bm{\Sigma}_{\textrm{Emp}}$ as the empirical mean and covariance matrix estimated using training data.

\paragraph{Latent Subspace Recovery.}
We compare the following two methods to recover the latent subspace:
\begin{enumerate}
    \item  {\tt Diff Method}: Our proposed diffusion factor model---we first estimate the return distribution using our diffusion factor model trained on the training dataset, then generate a large set of new data, and finally apply principal component analysis (PCA) on the generated data to estimate the eigenvalues and eigenspaces.
    \item {\tt Emp Method}: A na\"{i}ve PCA method---we directly perform PCA on the training data and extract the leading eigenvalues and eigenspaces.
\end{enumerate}
It is worth noting that the comparison between two methods is fair because the above two methods have access to exactly the same training data.

We denote $\{\lambda_i\}_{1 \leq i \leq k}$ as the top-$k$ eigenvalues and $\mathbf U \mathbf U^{\top}$  as the leading $k$-dimensional principal components of the ground-truth $\bm{\Sigma}$.
We perform SVD on $\bm{\Sigma}_{\textrm{Diff}}$ (resp. $\bm{\Sigma}_{\textrm{Emp}}$) to extract the top-$k$ eigenvalues $\{\lambda^{\textrm{Diff}}_i\}_{1 \leq i \leq k}$ (resp. $\{\lambda^{\textrm{Emp}}_i\}_{1 \leq i \leq k}$) and the leading $k$-dimensional principal components $\left(\mathbf U \mathbf U^\top\right)_{\textrm{Diff}}$ (resp. $\left(\mathbf U \mathbf U^\top\right)_{\textrm{Emp}}$).

To assess the accuracy of the eigenvalue estimation, we compute the $\ell^1$ relative error for {\tt Diff Method} and {\tt Emp Method} as
\begin{equation}
\label{equ: simulation--eigenvalues}
    \textrm{Diff RE}_1 = \frac{1}{k} \sum_{i=1}^{k} \left| \frac{\lambda_i^{\textrm{Diff}}}{\lambda_i} - 1 \right| \quad \text{and} \quad \textrm{Emp RE}_1 = \frac{1}{k} \sum_{i=1}^{k} \left| \frac{\lambda_i^{\textrm{Emp}}}{\lambda_i} - 1 \right|.
\end{equation}
To evaluate the recovery of the principal components, we compute the relative Frobenius norm errors for the two methods as
\begin{equation}
\label{equ: simulation--principal component}
    \textrm{Diff RE}_2 = \frac{\left\| \left( \mathbf U \mathbf U^\top \right)_{\textrm{Diff}} - \mathbf U \mathbf U^\top \right\|_{\textrm{F}}}{\left\| \mathbf U \mathbf U^\top \right\|_{\textrm{F}}} \quad \text{and} \quad
    \textrm{Emp RE}_2 = \frac{\left\| \left( \mathbf U \mathbf U^\top \right)_{\textrm{Emp}} - \mathbf U \mathbf U^\top \right\|_{\textrm{F}}}{\left\| \mathbf U \mathbf U^\top \right\|_{\textrm{F}}}.
\end{equation}

Table~\ref{tab: simulation--latent subspace recovery} reports the errors in estimating the top-$k$ eigenvalues \eqref{equ: simulation--eigenvalues} in Panel~A and the errors in recovering the $k$-dimensional principal components \eqref{equ: simulation--principal component} in Panel~B for both the {\tt Diff Method} and {\tt Emp Method}, for a variety of training sample sizes $N = 2^{9}, 2^{10}, \dots, 2^{13}$.

\begingroup
\renewcommand{\arraystretch}{0.9}
\setlength{\extrarowheight}{0pt}
\begin{table}[htbp]
\centering
\small
\caption{
Relative error of the estimated top-$k$ eigenvalues \eqref{equ: simulation--eigenvalues} and $k$-dimensional principal components \eqref{equ: simulation--principal component} for varying sample sizes (standard deviations in parentheses).
\label{tab: simulation--latent subspace recovery}}
{
\begin{tabular*}{\textwidth}{@{\extracolsep{\fill}}l *{3}{c}@{}}
\toprule
\multicolumn{4}{c}{\textbf{Panel A: Eigenvalues}} \\
\midrule
$N$ & \textbf{Diff RE$_1$} & \textbf{Emp RE$_1$} & \textbf{Diff RE$_1$/Emp RE$_1$} \\
\midrule
$2^{9}=512$  & 0.144 ($\pm$ 0.011) & 0.160 & 0.898 ($\pm$ 0.069) \\
$2^{10}=1024$ & 0.130 ($\pm$ 0.008) & 0.141 & 0.919 ($\pm$ 0.056) \\
$2^{11}=2048$ & 0.116 ($\pm$ 0.005) & 0.121 & 0.957 ($\pm$ 0.041) \\
$2^{12}=4096$ & 0.081 ($\pm$ 0.004) & 0.081 & 1.003 ($\pm$ 0.047) \\
$2^{13}=8192$ & 0.069 ($\pm$ 0.003)  & 0.067 & 1.024 ($\pm$ 0.045) \\
\midrule
\multicolumn{4}{c}{\textbf{Panel B: Principal Components}} \\
\midrule
$N$ & \textbf{Diff RE$_2$} & \textbf{Emp RE$_2$} & \textbf{Diff RE$_2$/Emp RE$_2$} \\
\midrule
$2^{9}=512$  & 0.247 ($\pm$ 0.012) & 0.274 & 0.901 ($\pm$ 0.044) \\
$2^{10}=1024$ & 0.202 ($\pm$ 0.006) & 0.218 & 0.926 ($\pm$ 0.028) \\
$2^{11}=2048$ & 0.153 ($\pm$ 0.005) & 0.159 & 0.960 ($\pm$ 0.033) \\
$2^{12}=4096$ & 0.110 ($\pm$ 0.004) & 0.109 & 1.009 ($\pm$ 0.036) \\
$2^{13}=8192$ & 0.085 ($\pm$ 0.004) & 0.084 & 1.012 ($\pm$ 0.047) \\
\bottomrule
\end{tabular*}
}
{}
\end{table}
\endgroup

Table~\ref{tab: simulation--latent subspace recovery} reveals the advantage of our method in small-data regimes ($N \leq d$), which is particularly important for practical applications.
In particular, when $N\leq 2048$, {\tt Diff Method} consistently outperforms {\tt Emp Method}, as shown by error ratios being statistically below~1. When there is enough sample ($N \geq 4096$), simply using empirical estimates suffices to yield good subspace recovery. It is worth highlighting that $N=2048$ corresponds to approximately 8 years of daily return observations or 39 years of weekly return observations. It is rarely the case that one enjoys the luxury of having that much data to estimate a factor model, because return distributions do not remain stable over such a long period of time.

\paragraph{Generated Return Distribution.}  In  Figure~\ref{fig:example_simulation_maxmin_variance_and_mean_stock}, we visualize the (empirical) return distribution generated by our diffusion factor model (trained on $2^{11}$ samples) for a few selected assets, which is compared with direct sampling from the ground truth. With the same number of $2^{11}$ samples, {\tt Diff Method} produces a {\it smoother} empirical distribution that more closely approximates the ground truth. This suggests that our diffusion factor model may be more effective at capturing patterns and regularities of the underlying distribution than direct sampling.

\begin{figure}[htbp]
    \centering
    \caption{Examples of asset return distribution (the blue is constructed using output samples from the diffusion model and the green is based on samples from the ground truth.)}
    \label{fig:example_simulation_maxmin_variance_and_mean_stock}
    
    \subfigure[Asset with the largest variance.\label{fig:example_simulation_max_variance_stock}]{
        \includegraphics[width=0.48\linewidth]{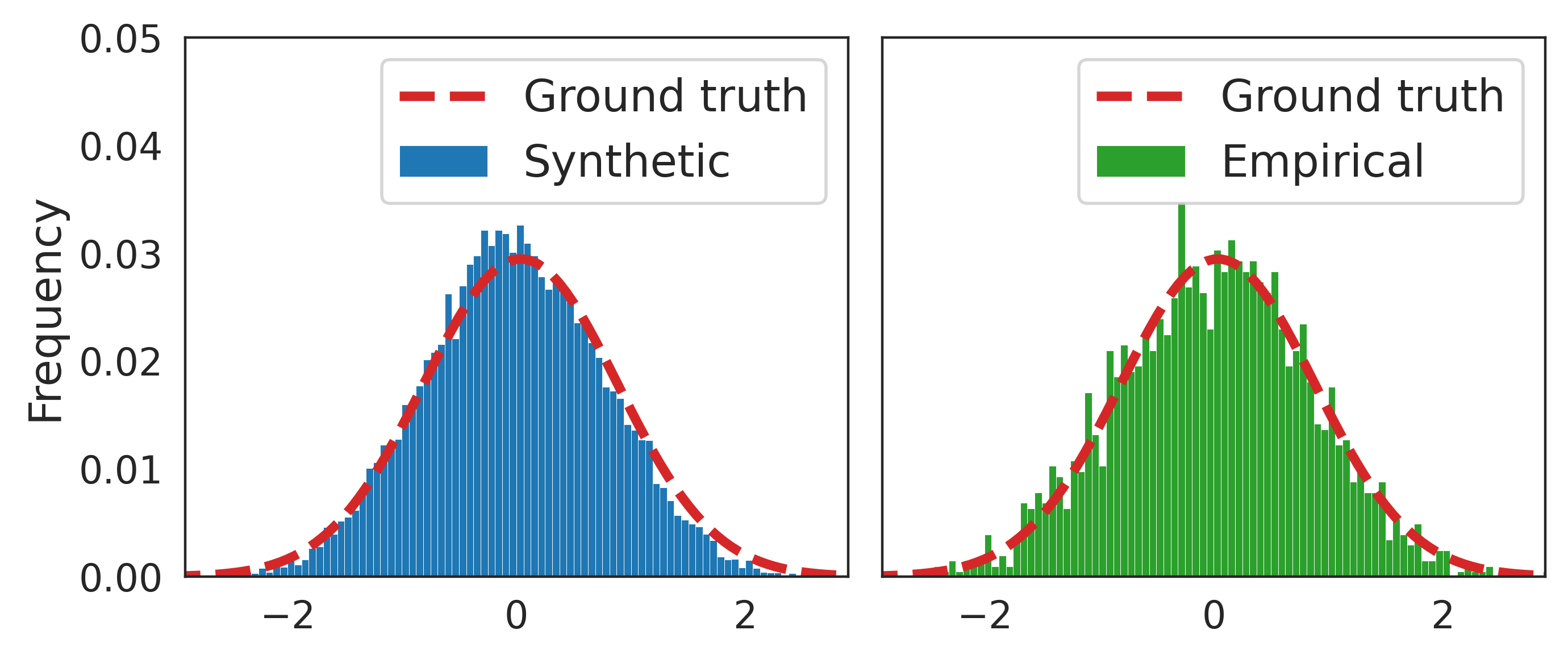}
    }
    \hfill
    \subfigure[Asset with the smallest variance.\label{fig:example_simulation_min_variance_stock}]{
        \includegraphics[width=0.48\linewidth]{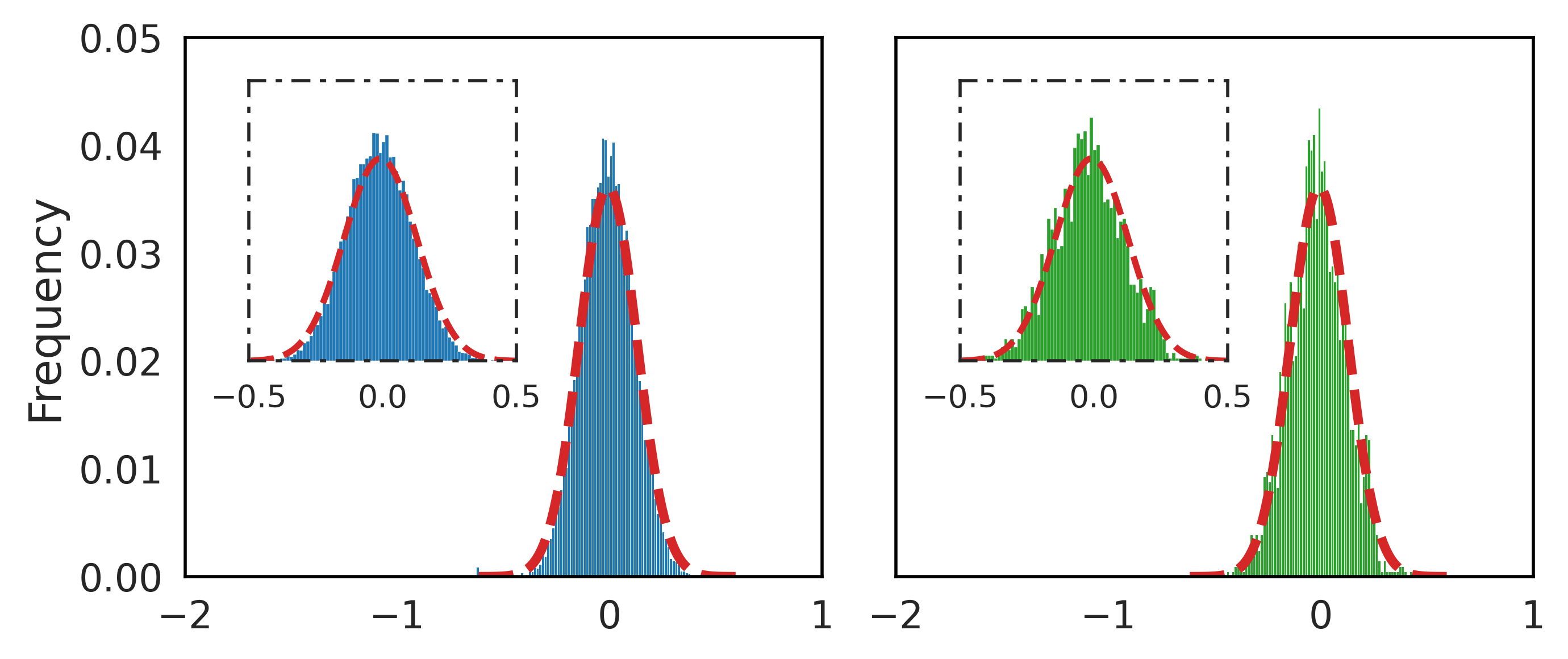}
    }
    
    \subfigure[Asset with the largest mean.]{
        \includegraphics[width=0.48\linewidth]{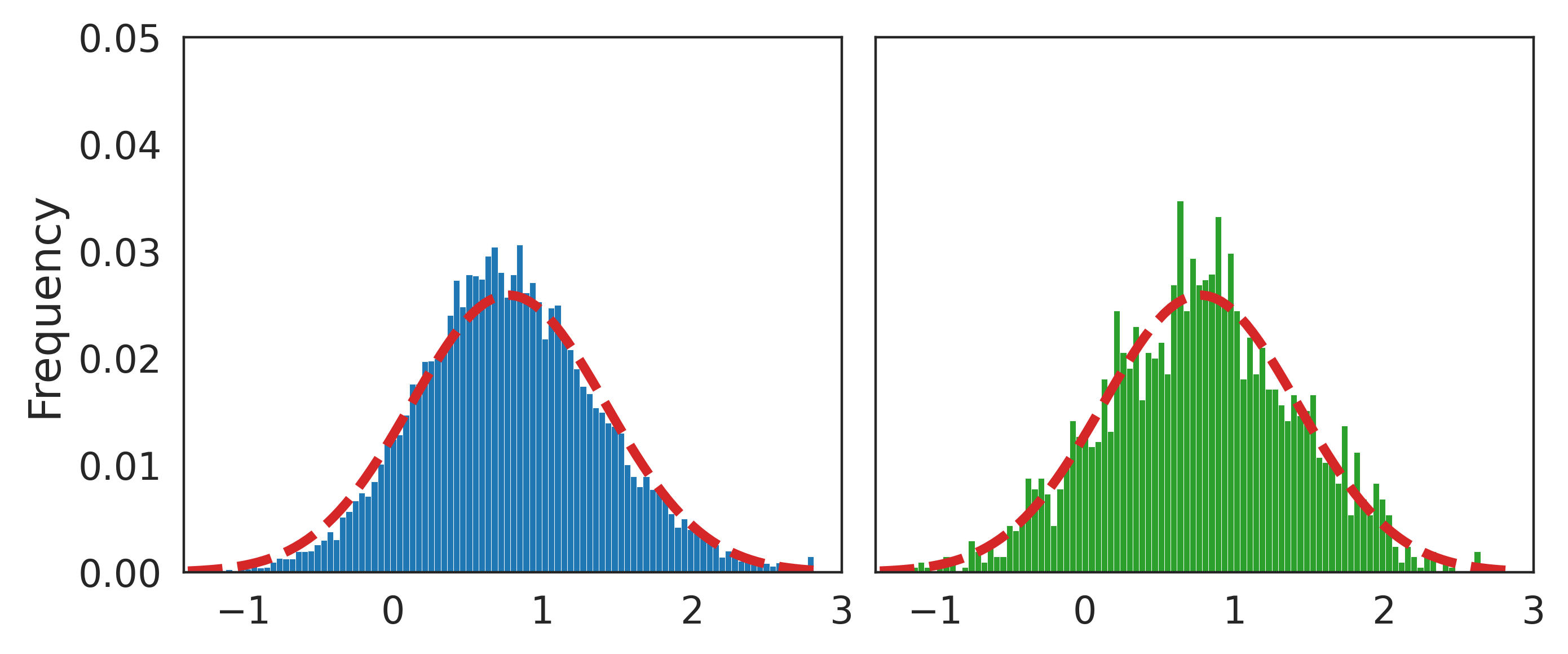}\label{fig:example_simulation_max_mean_stock}
    }
    \hfill
    \subfigure[Asset with the smallest mean.]{
        \includegraphics[width=0.48\linewidth]{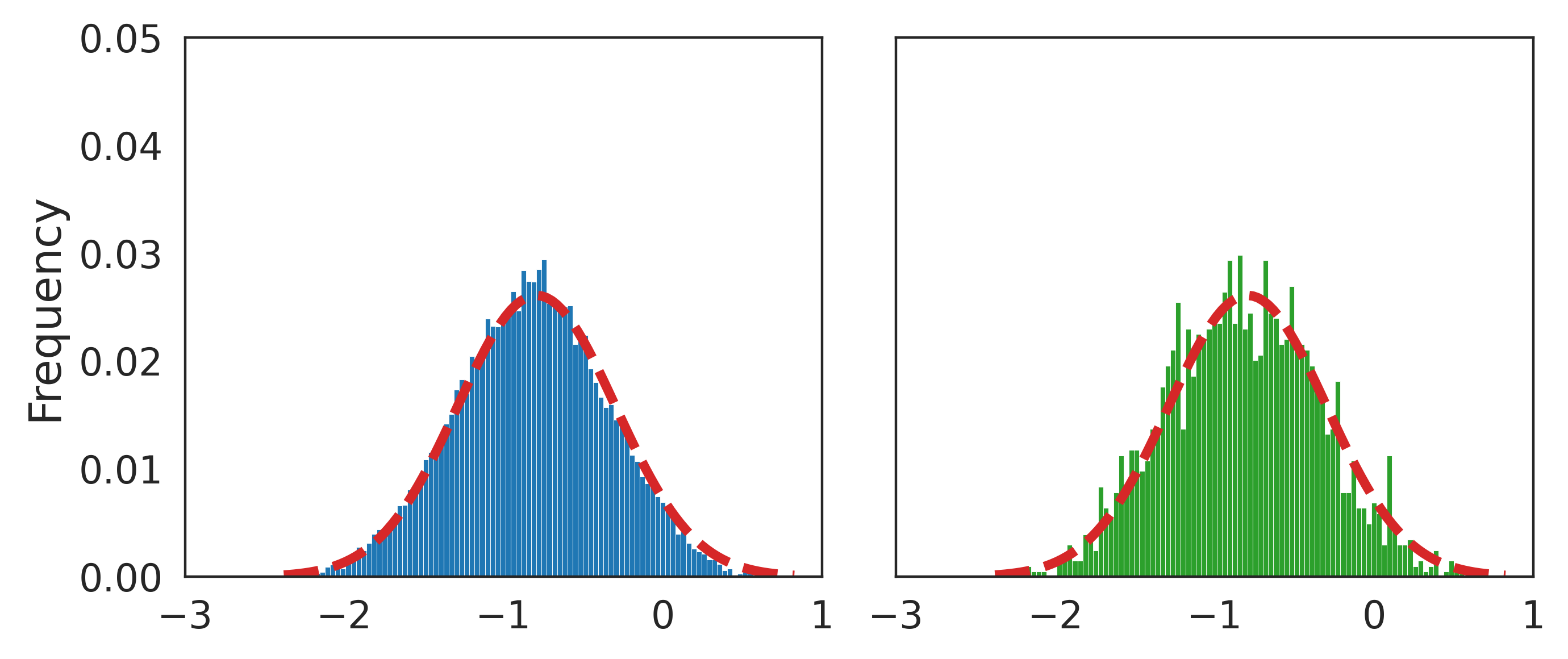}\label{fig:example_simulation_min_mean_stock}
    }
\end{figure}

Moreover, to evaluate the accuracy of estimating the first two moments, Table~\ref{tab: simulation--mean cov estimation} reports the relative estimation errors for the mean and the covariance matrix, across a range of sample sizes $N = 2^{9}, 2^{10}, \dots, 2^{13}$. Specifically, we compute the $\ell^2$ relative error of the mean and the relative Frobenius-norm error of the covariance matrix for {\tt Diff Method} and {\tt Emp Method} as follows:
\begin{align}
    &\textrm{Diff RE}_3 = \frac{\left\| \bm \mu_{\textrm{Diff}} - \bm \mu \right\|_{2}}{\left\| \bm \mu \right\|_{2}}, \quad
    \textrm{Emp RE}_3 = \frac{\left\| \bm \mu_{\textrm{Emp}} - \bm \mu \right\|_{2}}{\left\| \bm \mu \right\|_{2}}, \label{equ: simulation--mean} \\
    &\textrm{Diff RE}_4 = \frac{\left\| \bm \Sigma_{\textrm{Diff}} - \bm \Sigma \right\|_{\textrm{F}}}{\left\| \bm \Sigma \right\|_{\textrm{F}}}, \quad \text{and} \quad
    \textrm{Emp RE}_4 = \frac{\left\| \bm \Sigma_{\textrm{Emp}} - \bm \Sigma \right\|_{\textrm{F}}}{\left\| \bm \Sigma \right\|_{\textrm{F}}}. \label{equ: simulation--covariance matrix}
\end{align}

\begingroup
\renewcommand{\arraystretch}{0.9}
\setlength{\extrarowheight}{0pt}
\begin{table}[htbp]
\centering
\small
\caption{
Relative error of the estimated mean \eqref{equ: simulation--mean} and covariance matrix \eqref{equ: simulation--covariance matrix} for varying sample sizes (standard deviations in parentheses).
\label{tab: simulation--mean cov estimation}}
{
\begin{tabular*}{\textwidth}{@{\extracolsep{\fill}}l *{3}{c}@{}}
\toprule
\multicolumn{4}{c}{\textbf{Panel A: Mean}} \\
\midrule
$N$ & \textbf{Diff RE$_3$} & \textbf{Emp RE$_3$} & \textbf{Diff RE$_3$/Emp RE$_3$} \\
\midrule
$2^{9}=512$ & 0.089	($\pm$ 0.007) & 0.097 & 0.916 ($\pm$ 0.070) \\
$2^{10}=1024$ & 0.057 ($\pm$ 0.002) & 0.060 & 0.954 ($\pm$ 0.033) \\
$2^{11}=2048$ & 0.047 ($\pm$ 0.001) & 0.048 & 0.966 ($\pm$ 0.023) \\
$2^{12}=4096$ & 0.026 ($\pm$ 0.001) & 0.026 & 1.001 ($\pm$ 0.021) \\
$2^{13}=8192$ & 0.022 ($\pm$ 0.001) & 0.022 & 1.009 ($\pm$ 0.020) \\
\midrule
\multicolumn{4}{c}{\textbf{Panel B: Covariance Matrix}} \\
\midrule
$N$ & \textbf{Diff RE$_4$} & \textbf{Emp RE$_4$} & \textbf{Diff RE$_4$/Emp RE$_4$} \\
\midrule
$2^{9}=512$ & 0.257	($\pm$ 0.016) & 0.290 & 0.885 ($\pm$ 0.055) \\
$2^{10}=1024$ & 0.209 ($\pm$ 0.011) & 0.230 & 0.911 ($\pm$ 0.046) \\
$2^{11}=2048$ & 0.157 ($\pm$ 0.005) & 0.167 & 0.941 ($\pm$ 0.027) \\
$2^{12}=4096$ & 0.115 ($\pm$ 0.002) & 0.115 & 0.997 ($\pm$ 0.018) \\
$2^{13}=8192$ & 0.089 ($\pm$ 0.001) & 0.088 & 1.007 ($\pm$ 0.009) \\
\bottomrule
\end{tabular*}
}
{}
\end{table}
\endgroup

Table~\ref{tab: simulation--mean cov estimation} further demonstrates that, in small-sample regimes ($N \leq 2048$), the {\tt Diff Method} consistently attains lower estimation errors for both the mean and covariance than the {\tt Emp Method}, as evidenced by error ratios below~1. For $N \geq 4096$, the two methods become indistinguishable.

\paragraph{Statistical Interpretation.}
From a statistical perspective, our superior performance can be explained by the bias--variance trade-off. The {\tt Emp Method} has small bias but often suffers from high variance in small-data regimes. In contrast, the {\tt Diff Method} first fits a diffusion model to the available data and then uses the trained model to generate a large number of samples for downstream estimation tasks.  Once the diffusion model is well trained
(using observed data) to represent the underlying unknown distribution, our method can generate arbitrarily many additional samples.
The {\tt Diff Method} can therefore be viewed as a data-dependent regularization of the {\tt Emp Method}, increasing the effective sample size to reduce estimation variance at the cost of a small modeling bias \citep{li2025tabulardiffusion}. In small-data regimes, estimation variance typically dominates the error \citep{jorion1986bayes,ledoit2003improved,ledoit2004well}, and the resulting variance reduction leads the {\tt Diff Method} to outperform the {\tt Emp Method}.  This insight is further developed in the next section through empirical analysis.

\section{Empirical Analysis}
\label{sec: empirical}
In this section, we apply our diffusion factor model to real-world data and evaluate its economic relevance in constructing both mean-variance optimal portfolios  \citep{zhou2000continuous,demiguel2009generalized} and factor portfolios  \citep{giglio2022factor, feng2023deep}. Although our theoretical analysis relies on distributional assumptions of sub-Gaussian tails, we challenge our framework with real asset returns, which are well documented to be heavy-tailed while still well captured by factor models  \citep{chamberlain1983arbitrage,cont2001empirical,fan2013large}. Specifically, Section~\ref{subsec: mvo portfolio} compares mean-variance optimal portfolios derived from diffusion-generated data with those based on other robust portfolio rules in the literature. Section~\ref{subsec: factor and tangency portfolio} assesses the performance of factor portfolios estimated
using diffusion-generated data and benchmarks them against other prominent factor models in the literature.

We use daily excess return data for U.S. stocks from May 1, 2001, to April 30, 2024.\footnote{The U.S. Securities and Exchange Commission (SEC) mandated the conversion to decimal pricing for all U.S. stock markets by April 9, 2001.} The dataset is obtained from the Center for Research in Security Prices (CRSP), available through Wharton Research Data Services. We use a five-year rolling window to estimate the diffusion model, and we update the estimation quarterly. Specifically, at the beginning of each rebalancing quarter $T$ (i.e., on February 1, May 1, August 1, and November 1), we update model parameters using training data from the preceding five years ($T-20$ to $T$ offset quarters). We test the model on data in the next quarter $T+1$ to evaluate out-of-sample performance. Appendix~\ref{app:empirical_data_training_eval} provides more details on data preprocessing and training of diffusion models. Appendix~\ref{app:empirical_update_freq} provides a robustness check with annually updated diffusion models.\footnote{With a five-year rolling window, each quarterly update replaces roughly $5\%$ of the training data and reuses the remaining $95\%$, which balances information refresh against the computational cost of retraining the diffusion model. As a robustness check, an annually update frequency also leads to similar results; see Appendix~\ref{app:empirical_update_freq} for details.
}

\subsection{Mean-Variance Optimal Portfolio}
\label{subsec: mvo portfolio}

We follow the literature to consider the mean-variance optimization problem with a norm constraint \citep{jagannathan2003risk, bertsimas2004robust,demiguel2009generalized,gotoh2011role} to yield a fully invested and reasonably diversified portfolio:
\begin{equation}
\label{equ: empirical--mean-variance optimization problem}
\max_{\bm\omega}~ \bm\omega^\top \bm\mu - \frac{\eta}{2} \bm\omega^\top \bm\Sigma \bm\omega, \quad \textrm{ subject to } \bm\omega^\top \mathbf 1 = 1 \textrm{ and } \|\bm\omega\|_{\infty} \leq 0.05,
\end{equation}
where $\bm{\omega}$ denotes the portfolio weights, $\bm{\mu}$ is the expected return in excess of the risk-free rate, $\bm{\Sigma}$ is the covariance matrix, and $\eta > 0$ is the risk aversion parameter. As a robustness check, we also consider an $\ell_1$-norm constraint when solving the portfolio weights \citep{bertsimas2004robust,gotoh2011role}, and report the corresponding results in Appendix~\ref{app:empirical_norm}.
\footnote{In the norm-constrained portfolio literature, various norms have been studied, including $\ell_1$-, $\ell_2$-, $\ell_\infty$-, and $A$-norm constraints  \citep{jagannathan2003risk,bertsimas2004robust,demiguel2009generalized,gotoh2011role}. It has been shown that these constraints have certain mathematical equivalence in terms of obtaining the portfolio weights \citep{bertsimas2004robust,gotoh2011role}. As a robustness check, we construct mean–variance portfolios under $\ell_1$-norm constraints and evaluate their performance; see details in Appendix~\ref{app:empirical_norm}.
}

\paragraph{Methods of Portfolio Construction.} 
We evaluate a series of portfolio construction methods that differ in their data source (real observed data or diffusion-generated data) and in their estimation techniques for the mean and covariance matrix. We first describe classical approaches that rely solely on observed data.

\begin{enumerate}
    \item {\tt EW}: A simple strategy with equal weights on all risky assets. \cite{demiguel2009optimal} have documented its surprisingly efficient and robust performance.

    \item {\tt VW}: A value-weighted strategy that assigns each asset a weight proportional to its market capitalization relative to the total market capitalization in the dataset.
    
    \item {\tt Real Emp+Real Emp}: A baseline that directly uses the sample mean $\bm\mu_{\textrm{Emp}}$ and sample covariance matrix $\bm\Sigma_{\textrm{Emp}}$ as inputs to \eqref{equ: empirical--mean-variance optimization problem} to solve the optimal portfolio weights.

\end{enumerate}
As the empirical estimator is suboptimal and may be unstable in small-data regimes, we also include shrinkage estimators, which are well documented to improve upon empirical estimators in such settings \citep{jorion1986bayes,ledoit2003improved,ledoit2004well}.
\begin{enumerate}[resume]
    \item {\tt Real BS+Real Emp}: A robust portfolio proposed by \citet{jorion1986bayes} that utilizes a Bayes-Stein shrinkage mean $\bm\mu_{\textrm{BS}}$:
    \begin{equation*}
        \bm\mu_{\textrm{BS}} = (1 - \gamma_{\textrm{BS}}) \cdot \bm\mu_{\textrm{Emp}}  + \gamma_{\textrm{BS}} \cdot \mu_{\textrm{gmv}} \mathbf 1_d
    \end{equation*}
    to solve \eqref{equ: empirical--mean-variance optimization problem}, where $\mu_{\textrm{gmv}} = \mathbf 1_{d}^{\top} \bm\Sigma_{\textrm{Emp}}^{-1} \bm\mu_{d} / \mathbf 1_{d}^{\top} \bm\Sigma_{\textrm{Emp}}^{-1} \mathbf 1_{d}$ denotes the average excess return on the sample global minimum-variance portfolio, and $\gamma_{\textrm{BS}}$ is the shrinkage weight estimated by \citet[Equation (17)]{jorion1986bayes}. The covariance estimator is still the sample covariance $\bm\Sigma_{\textrm{Emp}}$.

    \item {\tt Real OLSE+Real Emp}: A robust portfolio proposed by \citet{bodnar2019optimal} that uses an Optimal Linear Shrinkage Estimator (OLSE) for high-dimensional mean $\bm\mu_{\textrm{OLSE}}$:
    \begin{equation*}
        \bm\mu_{\textrm{OLSE}} = \alpha_{\textrm{OLSE}} \cdot \bm\mu_{\textrm{Emp}}  + \beta_{\textrm{OLSE}} \cdot \mathbf 1_{d}
    \end{equation*}
    to solve \eqref{equ: empirical--mean-variance optimization problem}, where $\alpha_{\textrm{OLSE}}$ and $\beta_{\textrm{OLSE}}$ are the shrinkage weights estimated by \citet[Equations (6) and (7)]{jorion1986bayes}. The covariance estimator is still the sample covariance $\bm\Sigma_{\textrm{Emp}}$.

    \item {\tt Real Emp+Real LW}: A robust portfolio proposed by \cite{ledoit2003improved,ledoit2004well} that uses a shrinkage covariance matrix $\bm\Sigma_{\textrm{LW}}$:
    \begin{equation*}
        \bm\Sigma_{\textrm{LW}} = (1 - \gamma_{\textrm{LW}}) \cdot \bm\Sigma_{\textrm{Emp}}  + \gamma_{\textrm{LW}} \cdot u \mathbf I_d
    \end{equation*}
    to solve \eqref{equ: empirical--mean-variance optimization problem}, where $u = \operatorname{tr}(\bm\Sigma_{\textrm{Emp}}) / d$, and $\gamma_{\textrm{LW}}$ is the shrinkage parameter estimated by \citet[Equation (2.14)]{ledoit2022power}.
    The mean estimator is still the sample mean $\bm\mu_{\textrm{Emp}}$.

    \item {\tt Real BS+Real LW}: A robust portfolio that combines Bayes–Stein shrinkage mean $\bm\mu_{\textrm{BS}}$ with the shrinkage covariance $\bm\Sigma_{\textrm{LW}}$ to solve \eqref{equ: empirical--mean-variance optimization problem}.

    \item {\tt Real OLSE+Real LW}: A robust portfolio that combines OLSE mean $\bm\mu_{\textrm{OLSE}}$ with the shrinkage covariance $\bm\Sigma_{\textrm{LW}}$ to solve \eqref{equ: empirical--mean-variance optimization problem}.
\end{enumerate}
We further consider six methods that rely on our diffusion-generated data.
\begin{enumerate}[resume]
    \item {\tt Diff Emp+Diff Emp}: It extends {\tt Real Emp+Real Emp} by replacing the empirical mean and covariance estimates with those obtained from diffusion-generated data.

    \item {\tt Diff BS+Diff Emp}: It extends {\tt Real BS+Real Emp} by replacing the Bayes-Stein mean and empirical covariance estimates with those obtained from diffusion-generated data.

    \item {\tt Diff OLSE+Diff Emp}: It extends {\tt Real OLSE+Real Emp} by replacing the OLSE mean and empirical covariance estimates with those obtained from diffusion-generated data.

    \item {\tt Diff Emp+Diff LW}: It extends {\tt Real Emp+Real LW} by replacing the empirical mean and Ledoit-Wolf covariance estimates with those obtained from diffusion-generated data.

    \item {\tt Diff BS+Diff LW}: It extends {\tt Real BS+Real LW} by replacing the Bayes-Stein mean and Ledoit-Wolf covariance estimates with those obtained from diffusion-generated data.

    \item {\tt Diff OLSE+Diff LW}: It extends {\tt Real OLSE+Real LW} by replacing the OLSE mean and Ledoit-Wolf covariance estimates with those obtained from diffusion-generated data.
\end{enumerate}

\noindent
Finally, we consider two additional hybrid methods.
\begin{enumerate}[resume]
    \item {\tt Real Emp+Diff Emp}: It uses the empirical mean estimated from real data and the empirical covariance matrix estimated from diffusion-generated data to solve \eqref{equ: empirical--mean-variance optimization problem}.

    \item {\tt Diff Emp+Real Emp}: It uses the empirical mean estimated from diffusion-generated data and the empirical covariance matrix estimated from real data to solve \eqref{equ: empirical--mean-variance optimization problem}.
\end{enumerate}

Methods 9--14 serve as diffusion-based counterparts to Methods 3--8 to evaluate the benefits of using diffusion-generated data in both mean and covariance estimation. Methods 15 and 16 are hybrid approaches designed to understand the contribution of diffusion-generated data in mean and covariance estimation, respectively.

\paragraph{Main Results.}
Target weights are updated quarterly and rebalanced daily. Following \citet{kan2007optimal}, we set $\eta = 3$ and assume a transaction cost of 20 basis points. We also examine other values of $\eta$ and the scenario without transaction costs, and find similar results; see Appendix~\ref{app:empirical_cost_eta} for details. Table~\ref{tab: portfolio_performance_metrics_eta3_quarter} reports out-of-sample portfolio performance under scenarios with transaction costs, including the average return (Mean), standard deviation (Std), Sharpe ratio (SR), certainty equivalent return (CER, i.e., the objective value in \eqref{equ: empirical--mean-variance optimization problem}), maximum drawdown (MDD), and turnover (TO).\footnote{We rebalance daily to track the target weights, since daily return fluctuations cause realized weights to drift away from the targets and require small but frequent adjustments to maintain the target allocation. The reported turnover is annualized from daily trades, hence these adjustments can accumulate into a relatively high annual turnover, especially when the absolute values of the
target weights are more extreme.} Figure \ref{fig: example_for_cumulative_return_eta3_quarter} further shows the cumulative returns of different portfolios in log scale with transaction costs.

\begingroup
\renewcommand{\arraystretch}{0.9}
\setlength{\extrarowheight}{0pt}
\begin{table}[htbp]
\centering
\small
\caption{
Performance of different portfolios with transaction costs for $\eta = 3$ (model updated quarterly)}. 
\label{tab: portfolio_performance_metrics_eta3_quarter}
{
\begin{tabular*}{\textwidth}{@{\extracolsep{\fill}}l *{6}{c}@{}}
\toprule
Method & Mean & Std & SR & CER & MDD (\%) & TO \\
\midrule
\multicolumn{7}{c}{Methods based on real observed data} \\
\midrule
{\tt EW} & 0.100 & 0.206 & 0.486 & 0.037 & 53.128 & {\bf 3.031}\\
{\tt VW} & 0.096 & 0.220 & 0.437 & 0.024 & 58.086 & 3.464 \\
{\tt Real Emp+Real Emp}   & -0.017 & 0.128 & -0.129 & -0.041 & 45.011 & 46.722 \\
{\tt Real BS+Real Emp}    & -0.021 & 0.126 & -0.168 & -0.045 & 45.864 & 45.612 \\
{\tt Real OLSE+Real Emp}  & -0.039 & 0.127 & -0.305 & -0.063 & 55.596 & 45.952 \\
{\tt Real Emp+Real LW}    & -0.003 & 0.123 & -0.023 & -0.025 & 38.852 & 38.827 \\
{\tt Real BS+Real LW}     & -0.007 & {\bf 0.121} & -0.059 & -0.029 & 39.369 & 37.900 \\
{\tt Real OLSE+Real LW}   & -0.024 & 0.122 & -0.200 & -0.047 & 45.864 & 38.543 \\
\midrule
\multicolumn{7}{c}{Methods based on diffusion-generated data} \\
\midrule
{\tt Diff Emp+Diff Emp}    & {\bf 0.216} & 0.159 & {\bf 1.361} & {\bf 0.178} & 32.603 & 28.751 \\
{\tt Diff BS+Diff Emp}     & 0.213 & 0.157 & 1.357 & 0.176 & 32.610 & 27.978 \\
{\tt Diff OLSE+Diff Emp}   & 0.212 & 0.157 & 1.356 & 0.176 & 32.615 & 27.876 \\
{\tt Diff Emp+Diff LW}     & 0.180 & 0.152 & 1.186 & 0.145 & 32.781 & 26.353 \\
{\tt Diff BS+Diff LW}      & 0.178 & 0.150 & 1.184 & 0.144 & 32.862 & 25.773 \\
{\tt Diff OLSE+Diff LW}    & 0.178 & 0.150 & 1.184 & 0.144 & 32.882 & 25.697 \\
\midrule
\multicolumn{7}{c}{Methods based on both real observed data and diffusion-generated data} \\
\midrule
{\tt Diff Emp+Real Emp}    & 0.037 & 0.134 & 0.275 & 0.010 & 37.423 & 29.323 \\
{\tt Real Emp+Diff Emp}    & 0.163 & 0.150 & 1.090 & 0.130 & {\bf 30.760} & 23.313 \\
\bottomrule
\end{tabular*}
}
{}
\end{table}
\endgroup

\begin{figure}[htbp]
    \centering
    \caption{Cumulative returns of different portfolios in log scale with transaction cost for $\eta=3$.\label{fig: example_for_cumulative_return_eta3_quarter}}
    \includegraphics[width=0.5\linewidth]{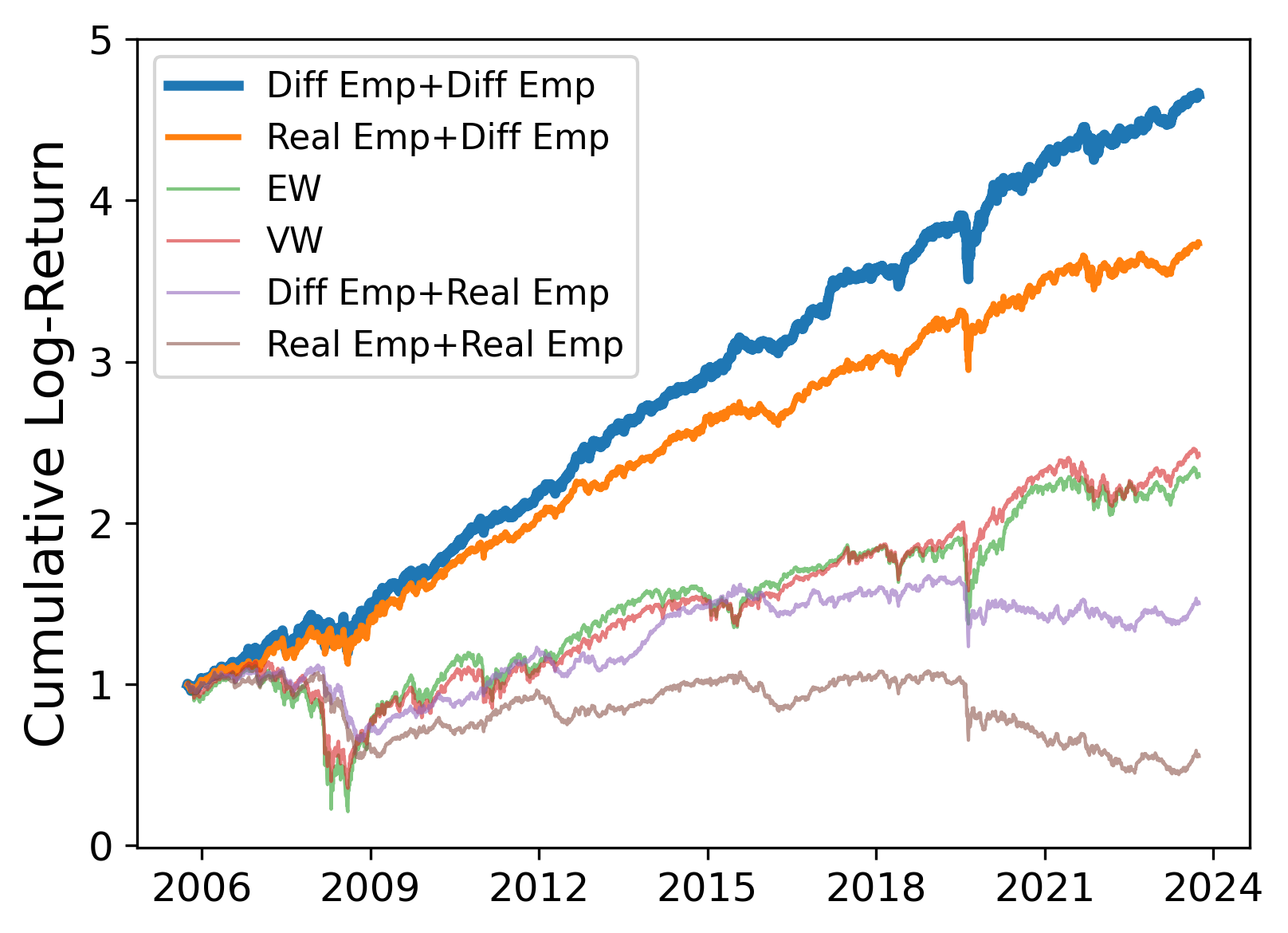}
\end{figure}

First, using diffusion-generated data, {\tt Diff Emp+Diff Emp} consistently outperforms by a large margin all alternative methods in Mean, SR, and CER, with transaction costs. Other diffusion-based methods also outperform their counterparts. In particular, {\tt Diff Emp+Diff Emp} outperforms {\tt EW} by a large margin, achieving approximately twice the Sharpe ratio. This is a highly nontrivial benchmark to beat, as shown by \citet{demiguel2009optimal}, because all other methods without diffusion-generated data fail to beat {\tt EW} in terms of risk-adjusted returns.

Second, {\tt Diff Emp+Diff Emp} outperforms both hybrid methods in risk-adjusted returns. Between them, {\tt Real Emp+Diff Emp} beats {\tt Diff Emp+Real Emp}, and both significantly outperform sample-based methods including {\tt Real Emp+Real Emp} and other classical shrinkage methods. This result reflects improvements in both the mean and covariance estimation from diffusion-generated data, but most of the improvements come from the improved covariance estimation, which is not surprising given the very design of our diffusion factor model. From a statistical perspective, again, the performance gains of {\tt Diff Emp+Diff Emp} can be interpreted through a bias--variance trade-off: our diffusion factor model acts as a form of data-dependent regularization, introducing a modest modeling bias while substantially reducing finite-sample estimation variance, thereby yielding more robust moment estimators \citep{kotelnikov2023tabddpm,li2025tabulardiffusion}.

Finally, with diffusion-generated data, shrinkage estimates of both the mean and covariance matrix are no longer necessary, as shown by the superior performance of {\tt Diff Emp+Diff Emp} compared with other diffusion-based shrinkage methods.
Although the shrinkage estimates have historically played an impactful role for robust portfolios, as shown by reviews in the literature  \citep{avramov_zhou:2010,bodnar2022optimal,ledoit2022power}, our results show that modern generative modeling techniques such as diffusion models may provide a simple yet effective and robust way to deal with data scarcity.

\subsection{Factor Portfolio}
\label{subsec: factor and tangency portfolio}
To further demonstrate the benefits of diffusion-generated data, we apply existing statistical methods on top of our diffusion-generated data to obtain factors and evaluate the performance of the corresponding tangency portfolios.

\paragraph{Methods of Factor Estimation.} 
We compare seven methods to estimate factors, where a projection matrix is first estimated from either observed data or diffusion-generated data, and then applied to test data to extract factors. Existing approaches that rely solely on observed data include:
\begin{enumerate}
    \item {\tt FF Method}: Firm characteristics-based factors that includes the \citet{fama_french:2015} five factors: market (Mkt-RF), size (SMB), value (HML), profitability (RMW), and investment (CMA), the momentum factor (MOM) of \cite{carhart:1997}, and the short-term and long-term reversal factors (ST-Rev and LT-Rev).\footnote{These two factors are obtained from French’s data library \url{https://mba.tuck.dartmouth.edu/pages/faculty/ken.french/data_library.html}.}
    \item {\tt PCA Method}: Perform PCA on observed training data to obtain a projection matrix $\mathbf W_{\textrm{PCA}}$.
    \item {\tt POET Method}: Principal Orthogonal complEment Thresholding (POET) proposed by \citet{fan2013large}, in which one computes a robust POET covariance estimator $\hat{\bm\Sigma}_{\textrm{POET}}$ and then apply SVD to obtain the projection matrix $\mathbf W_{\textrm{POET}}$.
    \item {\tt RPPCA Method}: Risk-premia PCA (RP-PCA) proposed by \citet{lettau2020factors}, in which one performs PCA on $\frac{1}{n}\sum_{i=1}^{n}\mathbf r_i \mathbf r_i^{\top} + \gamma_{\textrm{RPPCA}} \bar{\mathbf r}\bar{\mathbf r}^{\top}$ to obtain a projection matrix $\mathbf W_{\textrm{RPPCA}}$, where $\{\mathbf r_i\}_{i=1}^{n}$ denotes samples of asset returns, $\bar{\mathbf r}=\frac{1}{n}\sum_{i=1}^{n} \mathbf r_i$ is the sample mean, and $\gamma_{\textrm{RPPCA}}$ is a tuning parameter.
\end{enumerate}
Methods based on our diffusion factor model are implemented by applying the same factor estimation procedures to diffusion-generated data, rather than to the observed training data:
\begin{enumerate}[resume]
    \item {\tt Diff+PCA Method}: It extends {\tt PCA Method} by using diffusion-generated data. 
    \item {\tt Diff+POET Method}: It extends {\tt POET Method} by using diffusion-generated data. 
    \item {\tt Diff+RPPCA Method}: It extends {\tt RPPCA Method} by using diffusion-generated data. 
\end{enumerate}

\paragraph{Main Results.} 
Next, we construct tangency portfolios that maximize the Sharpe ratio using the extracted factors by solving the following optimization problem:
\begin{equation}
\max_{\bm\omega} \frac{\bm\omega^\top \bm\mu_{\textrm{fac}}}{\sqrt{\bm\omega^\top \bm\Sigma_{\textrm{fac}} \bm\omega}},\quad \textrm{ subject to } \bm\omega^{\top} \mathbf 1 = 1 ,
\end{equation}
where $\bm\omega$ denotes the portfolio weights, and $\bm\mu_{\textrm{fac}}$ and $\bm\Sigma_{\textrm{fac}}$ are the mean and covariance matrix of the factors, respectively. Table~\ref{tab: out_of_sample_sharpe_ratios_quarter} reports the Sharpe ratios of the tangency portfolios constructed across varying numbers of factors.\footnote{As noted by \citet{kelly2019characteristics}, factor tangency portfolios may not be directly implementable, but they serve as important theoretical benchmarks for evaluating mean–variance efficiency. Compared to mean–variance portfolios, their generally higher Sharpe ratios may stem from two main sources. First, the relatively low dimensionality of the factor space compared to individual assets improves the stability of estimated means and covariances. Second, the exclusion of transaction costs can further enhance performance. Similar observations have been made in the literature; see \citet{gu_etal:2020, gu2021autoencoder}.}

\begingroup
\renewcommand{\arraystretch}{0.9}
\setlength{\extrarowheight}{0pt}
\begin{table}[htbp]
\centering
\small
\caption
{Out-of-sample Sharpe ratios of factor tangency portfolios (model updated quarterly)}. The number of factors is set to be $3$, $5$, $6$, and $8$, respectively.\label{tab: out_of_sample_sharpe_ratios_quarter}
{
\begin{tabular*}{\textwidth}{@{\extracolsep{\fill}}l *{7}{c}@{}}
\toprule
\# Factors & {\tt Diff+PCA} & {\tt Diff+POET} & {\tt Diff+RPPCA} & {\tt FF} & {\tt PCA} & {\tt POET} & {\tt RPPCA} \\
\midrule
$ 3 $ & 1.204 & 1.167 & {\bf 1.474} & 0.431 & 0.480 & 0.856 & 0.857 \\
$ 5 $ & 2.280 & 2.264 & {\bf 2.555} & 0.474 & 0.519 & 0.963 & 0.928 \\
$ 6 $ & 2.698 & 2.701 & {\bf 3.070} & 0.552 & 0.630 & 1.375 & 1.433 \\
$ 8 $ & 3.536 & 3.471 & {\bf 3.668} & 0.810 & 0.735 & 1.811 & 1.861 \\
\bottomrule
\end{tabular*}
}
{}
\end{table}
\endgroup

Methods based on our diffusion factor model consistently outperform {\tt FF Method} and their corresponding PCA counterparts. In particular, {\tt Diff+PCA Method} exceeds both {\tt FF Method} and {\tt PCA Method} by a wide margin, achieving approximately three and five times their Sharpe ratios, respectively. Furthermore, applying the robust methods of factor estimation proposed by \citet{fan2013large, lettau2020factors} to diffusion-generated data yields additional improvements in portfolio performance. These results highlight the effectiveness of diffusion-generated factors in capturing systematic risk.

Finally, we assess whether diffusion-generated factors capture interpretable economic characteristics by analyzing their correlations with firm characteristics-based factors. For each method based on diffusion-generated data, Figure~\ref{fig:correlation_factors_quarter} reports the correlations between top eight factors estimated using diffusion-based methods and traditional factors in {\tt FF Method}. Diffusion-generated factors exhibit notable correlations with traditional factors, with Mkt-RF, LT-REV, and MOM being the three leading factors for all three methods.

\begin{figure}[htbp]
    \centering
    \caption{Correlation between the top 8 factors obtained using diffusion-based methods and those from the {\tt FF Method}.\label{fig:correlation_factors_quarter}}
    \subfigure[{\tt Diff+PCA Method}]{
        \includegraphics[width=0.31\textwidth]{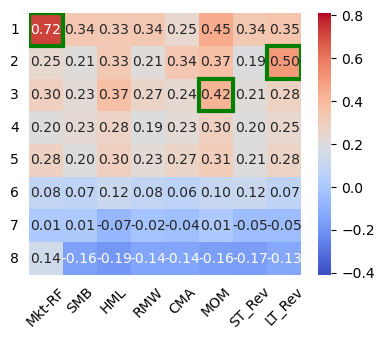}
    }
    \hfill
    \subfigure[{\tt Diff+POET Method}]{
        \includegraphics[width=0.31\textwidth]{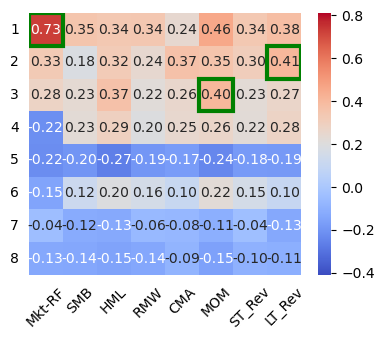}
    }  
    \subfigure[{\tt Diff+RPPCA Method}]{
        \includegraphics[width=0.31\textwidth]{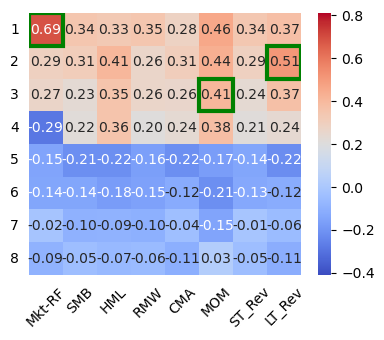}
    }
\end{figure}

\section{Conclusion and Future Work}
We propose a diffusion factor model that embeds the latent factor structure into generative diffusion processes. To exploit the low-dimensional nature of asset returns, we introduce a time-varying score decomposition via orthogonal projections and design a score network with an encoder-decoder architecture. These modeling choices lead to a concise and structure-aware representation of the score function.

On the theoretical front, we provide statistical guarantees for score approximation, score estimation, and distribution recovery. Our analysis introduces new techniques to address heterogeneous residual noise and time-varying subspaces, yielding error bounds that depend primarily on the intrinsic factor dimension $k$, with only mild dependence on the ambient dimension $d$. These results demonstrate that our framework effectively mitigates the curse of dimensionality in high-dimensional settings.

Simulation studies confirm that the proposed method achieves more accurate subspace recovery and smoother distribution estimation than classical baselines, particularly when the sample size is smaller than the asset dimension. The generated data reliably capture the true mean and covariance structure.

Finally, our empirical experiments on real data show that diffusion-generated data improves mean and covariance estimation, leading to superior mean-variance optimal portfolios. Our approach consistently outperforms traditional methods, achieving higher Sharpe ratios. Additionally, factors estimated from generated data exhibit interpretable economic characteristics, enabling tangency portfolios that better capture systematic risk. These findings highlight the substantial potential of our diffusion factor model for real-world financial applications.

With the widely used factor-model structure, our results open a path toward integrating financial data structures into generative AI architectures that admit statistical guarantees. Looking ahead, we plan to investigate several interesting and more sophisticated settings. These include understanding how the statistical guarantees change when the noise distribution is heavy-tailed; using diffusion models for missing-data imputation; establishing theoretical guarantees for how diffusion models improve decision-making; and developing methods for generating dynamic time-series data.

\theendnotes
\makeatletter

\singlespacing
\bibliography{references}

\begin{thebibliography}{183}
\expandafter\ifx\csname natexlab\endcsname\relax\def\natexlab#1{#1}\fi

\bibitem[Acciaio, Eckstein, and Hou(2024)]{acciaio2024time}
Acciaio, B., S.~Eckstein, and S.~Hou, 2024, Time-causal vae: Robust financial
  time series generator, {\em arXiv preprint arXiv:2411.02947\/} .

\bibitem[Acharya et~al.(2023)Acharya, Berner, Engle, Jung, Stroebel, Zeng, and
  Zhao]{acharya2023climate}
Acharya, V.~V., R.~Berner, R.~Engle, H.~Jung, J.~Stroebel, X.~Zeng, and
  Y.~Zhao, 2023, Climate stress testing, {\em Annual Review of Financial
  Economics\/} 15, 291--326.

\bibitem[Adrian, Etula, and Muir(2014)]{adrian_etal:2014}
Adrian, T., E.~Etula, and T.~Muir, 2014, Financial intermediaries and the
  cross-section of asset returns, {\em The Journal of Finance\/} 69,
  2557--2596.

\bibitem[A{\"\i}t-Sahalia and Xiu(2019)]{ait2019principal}
A{\"\i}t-Sahalia, Y., and D.~Xiu, 2019, Principal component analysis of
  high-frequency data, {\em Journal of the American Statistical Association\/}
  114, 287--303.

\bibitem[Albergo, Boffi, and Vanden-Eijnden(2023)]{albergo2023stochastic}
Albergo, M.~S., N.~M. Boffi, and E.~Vanden-Eijnden, 2023, Stochastic
  interpolants: A unifying framework for flows and diffusions, {\em arXiv
  preprint arXiv:2303.08797\/} .

\bibitem[Alexander(2005)]{alexander2005present}
Alexander, C., 2005, The present and future of financial risk management, {\em
  Journal of Financial Econometrics\/} 3, 3--25.

\bibitem[Anderson(1982)]{anderson1982reverse}
Anderson, B.~D., 1982, Reverse-time diffusion equation models, {\em Stochastic
  Processes and their Applications\/} 12, 313--326.

\bibitem[Anthony and Bartlett(2009)]{anthony2009neural}
Anthony, M., and P.~L. Bartlett, 2009, {\em Neural network learning:
  Theoretical foundations\/} (Cambridge University Press).

\bibitem[Avramov and Zhou(2010)]{avramov_zhou:2010}
Avramov, D., and G.~Zhou, 2010, Bayesian portfolio analysis, {\em Annual Review
  of Financial Economics\/} 2, 25--47.

\bibitem[Azangulov, Deligiannidis, and
  Rousseau(2024)]{azangulov2024convergence}
Azangulov, I., G.~Deligiannidis, and J.~Rousseau, 2024, Convergence of
  diffusion models under the manifold hypothesis in high-dimensions, {\em arXiv
  preprint arXiv:2409.18804\/} .

\bibitem[Bagnara(2024)]{bagnara2024asset}
Bagnara, M., 2024, Asset pricing and machine learning: a critical review, {\em
  Journal of Economic Surveys\/} 38, 27--56.

\bibitem[Bai and Ng(2002)]{bai2002determining}
Bai, J., and S.~Ng, 2002, Determining the number of factors in approximate
  factor models, {\em Econometrica\/} 70, 191--221.

\bibitem[Bai and Ng(2023)]{bai2023approximate}
Bai, J., and S.~Ng, 2023, Approximate factor models with weaker loadings, {\em
  Journal of Econometrics\/} 235, 1893--1916.

\bibitem[Bakry et~al.(2014)Bakry, Gentil, Ledoux, et~al.]{bakry2014analysis}
Bakry, D., I.~Gentil, M.~Ledoux, et~al., 2014, {\em Analysis and geometry of
  Markov diffusion operators\/}, volume 103 (Springer).

\bibitem[Baldi and Hornik(1989)]{baldi1989pca}
Baldi, P., and K.~Hornik, 1989, Neural networks and principal component
  analysis: Learning from examples without local minima, {\em Neural
  Networks\/} 2, 53--58.

\bibitem[Barancikova, Huang, and Salvi(2025)]{barancikova2024sigdiffusions}
Barancikova, B., Z.~Huang, and C.~Salvi, 2025, Sigdiffusions: Score-based
  diffusion models for long time series via log-signature embeddings, in {\em
  International Conference on Learning Representations\/}.

\bibitem[Bartlett, Foster, and Telgarsky(2017)]{bartlett2017spectrally}
Bartlett, P.~L., D.~J. Foster, and M.~J. Telgarsky, 2017, Spectrally-normalized
  margin bounds for neural networks, {\em Advances in Neural Information
  Processing Systems\/} 30.

\bibitem[Bartlett et~al.(2020)Bartlett, Long, Lugosi, and
  Tsigler]{bartlett2020benign}
Bartlett, P.~L., P.~M. Long, G.~Lugosi, and A.~Tsigler, 2020, Benign
  overfitting in linear regression, {\em National Academy of Sciences\/} 117,
  30063--30070.

\bibitem[Behn, Haselmann, and Vig(2022)]{behn2022limits}
Behn, M., R.~Haselmann, and V.~Vig, 2022, The limits of model-based regulation,
  {\em The Journal of Finance\/} 77, 1635--1684.

\bibitem[Benton et~al.(2024)Benton, De~Bortoli, Doucet, and
  Deligiannidis]{benton2024nearly}
Benton, J., V.~De~Bortoli, A.~Doucet, and G.~Deligiannidis, 2024, Nearly
  {$d$}-linear convergence bounds for diffusion models via stochastic
  localization, in {\em International Conference on Learning
  Representations\/}.

\bibitem[Bertsimas, Pachamanova, and Sim(2004)]{bertsimas2004robust}
Bertsimas, D., D.~Pachamanova, and M.~Sim, 2004, Robust linear optimization
  under general norms, {\em Operations Research Letters\/} 32, 510--516.

\bibitem[Bickel and Levina(2008)]{bickel2008regularized}
Bickel, P.~J., and E.~Levina, 2008, {Regularized estimation of large covariance
  matrices}, {\em The Annals of Statistics\/} 36, 199 -- 227.

\bibitem[Bisias et~al.(2012)Bisias, Flood, Lo, and Valavanis]{bisias2012survey}
Bisias, D., M.~Flood, A.~W. Lo, and S.~Valavanis, 2012, A survey of systemic
  risk analytics, {\em Annual Review of Finance Economics\/} 4, 255--296.

\bibitem[Bodnar, Okhrin, and Parolya(2019)]{bodnar2019optimal}
Bodnar, T., O.~Okhrin, and N.~Parolya, 2019, Optimal shrinkage estimator for
  high-dimensional mean vector, {\em Journal of Multivariate Analysis\/} 170,
  63--79.

\bibitem[Bodnar, Okhrin, and Parolya(2022)]{bodnar2022optimal}
Bodnar, T., Y.~Okhrin, and N.~Parolya, 2022, Optimal shrinkage-based portfolio
  selection in high dimensions, {\em Journal of Business \& Economic
  Statistics\/} 41, 140--156.

\bibitem[Borji(2019)]{borji2019pros}
Borji, A., 2019, Pros and cons of {GAN} evaluation measures, {\em Computer
  Vision and Image Understanding\/} 179, 41--65.

\bibitem[Brophy et~al.(2023)Brophy, Wang, She, and Ward]{brophy2023generative}
Brophy, E., Z.~Wang, Q.~She, and T.~Ward, 2023, Generative adversarial networks
  in time series: A systematic literature review, {\em ACM Computing Surveys\/}
  55, 1--31.

\bibitem[Bryzgalova et~al.(2023)Bryzgalova, DeMiguel, Li, and
  Pelger]{bryzgalova2023asset}
Bryzgalova, S., V.~DeMiguel, S.~Li, and M.~Pelger, 2023, Asset-pricing factors
  with economic targets, {\em SSRN Electronic Journal\/} Working paper.

\bibitem[B{\"u}chner and Kelly(2022)]{buchner2022factor}
B{\"u}chner, M., and B.~Kelly, 2022, A factor model for option returns, {\em
  Journal of Financial Economics\/} 143, 1140--1161.

\bibitem[Cao et~al.(2024)Cao, Tan, Gao, Xu, Chen, Heng, and Li]{cao2024survey}
Cao, H., C.~Tan, Z.~Gao, Y.~Xu, G.~Chen, P.-A. Heng, and S.~Z. Li, 2024, A
  survey on generative diffusion models, {\em IEEE Transactions on Knowledge
  and Data Engineering\/} .

\bibitem[Carhart(1997)]{carhart:1997}
Carhart, M.~M., 1997, On persistence in mutual fund performance, {\em The
  Journal of Finance\/} 52, 57--82.

\bibitem[Cetingoz and Lehalle(2025)]{cetingoz2025synthetic}
Cetingoz, A.~R., and C.-A. Lehalle, 2025, Synthetic data for portfolios: A
  throw of the dice will never abolish chance, {\em arXiv preprint
  arXiv:2501.03993\/} .

\bibitem[Chamberlain and Rothschild(1983)]{chamberlain1983arbitrage}
Chamberlain, G., and M.~Rothschild, 1983, Arbitrage, factor structure, and
  mean–variance analysis on large asset markets, {\em Econometrica\/} 51,
  1281--1304.

\bibitem[Chazottes, Collet, and Redig(2021)]{chazottes2019evolution}
Chazottes, J.-R., P.~Collet, and F.~Redig, 2021, Evolution of gaussian
  concentration bounds under diffusions, {\em Markov Processes and Related
  Fields\/} 27, 707--754.

\bibitem[Chen, Pelger, and Zhu(2024)]{chen2024deep}
Chen, L., M.~Pelger, and J.~Zhu, 2024, Deep learning in asset pricing, {\em
  Management Science\/} 70, 714--750.

\bibitem[Chen et~al.(2012)Chen, Xu, Weinberger, and Sha]{chen2012msda}
Chen, M., Z.~Xu, K.~Q. Weinberger, and F.~Sha, 2012, Marginalized denoising
  autoencoders for domain adaptation, in {\em International Conference on
  Machine Learning\/}.

\bibitem[Chen et~al.(2023{\natexlab{a}})Chen, Huang, Zhao, and
  Wang]{chen2023score}
Chen, M., K.~Huang, T.~Zhao, and M.~Wang, 2023{\natexlab{a}}, Score
  approximation, estimation and distribution recovery of diffusion models on
  low-dimensional data, in {\em International Conference on Machine
  Learning\/},  4672--4712, PMLR.

\bibitem[Chen et~al.(2022)Chen, Jiang, Liao, and Zhao]{chen2022nonparametric}
Chen, M., H.~Jiang, W.~Liao, and T.~Zhao, 2022, Nonparametric regression on
  low-dimensional manifolds using deep relu networks: Function approximation
  and statistical recovery, {\em Information and Inference: A Journal of the
  IMA\/} 11, 1203--1253.

\bibitem[Chen, Li, and Zhao(2020)]{chen2019generalization}
Chen, M., X.~Li, and T.~Zhao, 2020, On generalization bounds of a family of
  recurrent neural networks, in {\em International Conference on Artificial
  Intelligence and Statistics\/}, volume 108,  1233--1243 (PMLR).

\bibitem[Chen et~al.(2020)Chen, Liao, Zha, and Zhao]{chen2020statistical}
Chen, M., W.~Liao, H.~Zha, and T.~Zhao, 2020, Statistical guarantees of
  generative adversarial networks for distribution estimation, {\em arXiv
  preprint arXiv:2002.03938\/} 9.

\bibitem[Chen et~al.(2024)Chen, Mei, Fan, and Wang]{chen2024opportunities}
Chen, M., S.~Mei, J.~Fan, and M.~Wang, 2024, Opportunities and challenges of
  diffusion models for generative ai, {\em National Science Review\/} 11,
  nwae348.

\bibitem[Chen, Roll, and Ross(1986)]{chen_etal:1986}
Chen, N.-F., R.~Roll, and S.~A. Ross, 1986, Economic forces and the stock
  market, {\em Journal of Business\/}  383--403.

\bibitem[Chen et~al.(2023{\natexlab{b}})Chen, Chewi, Li, Li, Salim, and
  Zhang]{chen2022sampling}
Chen, S., S.~Chewi, J.~Li, Y.~Li, A.~Salim, and A.~Zhang, 2023{\natexlab{b}},
  Sampling is as easy as learning the score: theory for diffusion models with
  minimal data assumptions, in {\em International Conference on Learning
  Representations\/}.

\bibitem[Chen, Daras, and Dimakis(2023)]{chen2023restoration}
Chen, S., G.~Daras, and A.~Dimakis, 2023, Restoration-degradation beyond linear
  diffusions: A non-asymptotic analysis for ddim-type samplers, in {\em
  International Conference on Machine Learning\/},  4462--4484, PMLR.

\bibitem[Cole and Lu(2024)]{cole2024score}
Cole, F., and Y.~Lu, 2024, Score-based generative models break the curse of
  dimensionality in learning a family of sub-gaussian probability
  distributions, in {\em International Conference on Learning
  Representations\/}.

\bibitem[Coletta et~al.(2024)Coletta, Gopalakrishnan, Borrajo, and
  Vyetrenko]{coletta2024constrained}
Coletta, A., S.~Gopalakrishnan, D.~Borrajo, and S.~Vyetrenko, 2024, On the
  constrained time-series generation problem, {\em Advances in Neural
  Information Processing Systems\/} 36.

\bibitem[Coletta et~al.(2023)Coletta, Jerome, Savani, and
  Vyetrenko]{coletta2023conditional}
Coletta, A., J.~Jerome, R.~Savani, and S.~Vyetrenko, 2023, Conditional
  generators for limit order book environments: Explainability, challenges, and
  robustness, in {\em ACM International Conference on AI in Finance\/},
  27--35.

\bibitem[Connor, Hagmann, and Linton(2012)]{connor_etal:2012}
Connor, G., M.~Hagmann, and O.~Linton, 2012, Efficient semiparametric
  estimation of the fama--french model and extensions, {\em Econometrica\/} 80,
  713--754.

\bibitem[Cont(2001)]{cont2001empirical}
Cont, R., 2001, Empirical properties of asset returns: Stylized facts and
  statistical issues, {\em Quantitative Finance\/} 1, 223--236.

\bibitem[Cont et~al.(2023)Cont, Cucuringu, Kochems, and Prenzel]{cont2023limit}
Cont, R., M.~Cucuringu, J.~Kochems, and F.~Prenzel, 2023, Limit order book
  simulation with generative adversarial networks, {\em SSRN Electronic
  Journal\/} Working paper.

\bibitem[Cont et~al.(2022)Cont, Cucuringu, Xu, and Zhang]{cont2022tail}
Cont, R., M.~Cucuringu, R.~Xu, and C.~Zhang, 2022, Tail-{GAN}: Learning to
  simulate tail risk scenarios, {\em arXiv preprint arXiv:2203.01664\/} .

\bibitem[Creswell et~al.(2018)Creswell, White, Dumoulin, Arulkumaran, Sengupta,
  and Bharath]{creswell2018generative}
Creswell, A., T.~White, V.~Dumoulin, K.~Arulkumaran, B.~Sengupta, and A.~A.
  Bharath, 2018, Generative adversarial networks: An overview, {\em IEEE Signal
  Processing Magazine\/} 35, 53--65.

\bibitem[Dai, Wang, and Wipf(2020)]{dai2020usual}
Dai, B., Z.~Wang, and D.~Wipf, 2020, The usual suspects? reassessing blame for
  vae posterior collapse, in {\em International Conference on Machine
  Learning\/},  2313--2322 (PMLR).

\bibitem[Davis and Kahan(1970)]{davis1970rotation}
Davis, C., and W.~M. Kahan, 1970, The rotation of eigenvectors by a
  perturbation. iii, {\em SIAM Journal on Numerical Analysis\/} 7, 1--46.

\bibitem[De~Bortoli(2022)]{de2022convergence}
De~Bortoli, V., 2022, Convergence of denoising diffusion models under the
  manifold hypothesis, {\em arXiv preprint arXiv:2208.05314\/} .

\bibitem[De~Bortoli et~al.(2021)De~Bortoli, Thornton, Heng, and
  Doucet]{de2021diffusion}
De~Bortoli, V., J.~Thornton, J.~Heng, and A.~Doucet, 2021, Diffusion
  schr{\"o}dinger bridge with applications to score-based generative modeling,
  {\em Advances in Neural Information Processing Systems\/} 34, 17695--17709.

\bibitem[{DeMiguel} et~al.(2009){DeMiguel}, Garlappi, Nogales, and
  Uppal]{demiguel2009generalized}
{DeMiguel}, V., L.~Garlappi, F.~J. Nogales, and R.~Uppal, 2009, A generalized
  approach to portfolio optimization: Improving performance by constraining
  portfolio norms, {\em Management Science\/} 55, 798--812.

\bibitem[DeMiguel, Garlappi, and Uppal(2009)]{demiguel2009optimal}
DeMiguel, V., L.~Garlappi, and R.~Uppal, 2009, Optimal versus naive
  diversification: How inefficient is the 1/n portfolio strategy?, {\em The
  Review of Financial Studies\/} 22, 1915--1953.

\bibitem[Dhariwal and Nichol(2021)]{dhariwal2021diffusion}
Dhariwal, P., and A.~Nichol, 2021, Diffusion models beat gans on image
  synthesis, {\em Advances in Neural Information Processing Systems\/} 34,
  8780--8794.

\bibitem[Dou et~al.(2024)Dou, Kotekal, Xu, and Zhou]{dou2024optimal}
Dou, Z., S.~Kotekal, Z.~Xu, and H.~H. Zhou, 2024, From optimal score matching
  to optimal sampling, {\em arXiv preprint arXiv:2409.07032\/} .

\bibitem[Eckerli and Osterrieder(2021)]{eckerli2021generative}
Eckerli, F., and J.~Osterrieder, 2021, Generative adversarial networks in
  finance: an overview, {\em arXiv preprint arXiv:2106.06364\/} .

\bibitem[Elkamhi, Jo, and Nozawa(2024)]{elkamhi2024one}
Elkamhi, R., C.~Jo, and Y.~Nozawa, 2024, A one-factor model of corporate bond
  premia, {\em Management Science\/} 70, 1875--1900.

\bibitem[Fabozzi, Huang, and Zhou(2010)]{fabozzi2010robust}
Fabozzi, F.~J., D.~Huang, and G.~Zhou, 2010, Robust portfolios: contributions
  from operations research and finance, {\em Annals of Operations Research\/}
  176, 191--220.

\bibitem[Fama and French(1993)]{fama_french:1993}
Fama, E.~F., and K.~R. French, 1993, Common risk factors in the returns on
  stocks and bonds, {\em Journal of Financial Economics\/} 33, 3--56.

\bibitem[Fama and French(2004)]{fama2004capital}
Fama, E.~F., and K.~R. French, 2004, The capital asset pricing model: Theory
  and evidence, {\em Journal of Economic Perspectives\/} 18, 25--46.

\bibitem[Fama and French(2015{\natexlab{a}})]{fama_french:2015}
Fama, E.~F., and K.~R. French, 2015{\natexlab{a}}, A five-factor asset pricing
  model, {\em Journal of Financial Economics\/} 116, 1--22.

\bibitem[Fama and French(2015{\natexlab{b}})]{fama2015five}
Fama, E.~F., and K.~R. French, 2015{\natexlab{b}}, A five-factor asset pricing
  model, {\em Journal of Financial Economics\/} 116, 1--22.

\bibitem[Fan, Guo, and Zheng(2022)]{fan2022estimating}
Fan, J., J.~Guo, and S.~Zheng, 2022, Estimating number of factors by adjusted
  eigenvalues thresholding, {\em Journal of the American Statistical
  Association\/} 117, 852--861.

\bibitem[Fan, Liao, and Liu(2016)]{fan2016overview}
Fan, J., Y.~Liao, and H.~Liu, 2016, An overview of the estimation of large
  covariance and precision matrices, {\em The Econometrics Journal\/} 19,
  C1--C32.

\bibitem[Fan, Liao, and Mincheva(2013)]{fan2013large}
Fan, J., Y.~Liao, and M.~Mincheva, 2013, Large covariance estimation by
  thresholding principal orthogonal complements, {\em Journal of the Royal
  Statistical Society Series B: Statistical Methodology\/} 75, 603--680.

\bibitem[Fan, Liao, and Wang(2016)]{fan2016projected}
Fan, J., Y.~Liao, and W.~Wang, 2016, Projected principal component analysis in
  factor models, {\em Annals of Statistics\/} 44, 219.

\bibitem[Fan and Yao(2003)]{fan2003nonlinear}
Fan, J., and Q.~Yao, 2003, {\em Nonlinear time series: nonparametric and
  parametric methods\/} (Springer).

\bibitem[{Federal Reserve Board}(2023)]{federal}
{Federal Reserve Board}, 2023, 2023 stress test scenarios,
  \url{https://www.federalreserve.gov/publications/2023-stress-test-scenarios.htm}.

\bibitem[Feng, Giglio, and Xiu(2020)]{feng_etal:2020}
Feng, G., S.~Giglio, and D.~Xiu, 2020, Taming the factor zoo: A test of new
  factors, {\em The Journal of Finance\/} 75, 1327--1370.

\bibitem[Feng et~al.(2024)Feng, He, Polson, and Xu]{feng2024deep}
Feng, G., J.~He, N.~G. Polson, and J.~Xu, 2024, Deep learning in
  characteristics-sorted factor models, {\em Journal of Financial and
  Quantitative Analysis\/} 59, 3001--3036.

\bibitem[Feng et~al.(2023)Feng, Jiang, Li, and Song]{feng2023deep}
Feng, G., L.~Jiang, J.~Li, and Y.~Song, 2023, Deep tangency portfolio, {\em
  Available at SSRN 3971274\/} .

\bibitem[Ferson and Harvey(1991)]{ferson_harvey:1991}
Ferson, W.~E., and C.~R. Harvey, 1991, The variation of economic risk premiums,
  {\em Journal of Political Economy\/} 99, 385--415.

\bibitem[Frankle and Carbin(2019)]{frankle2018lottery}
Frankle, J., and M.~Carbin, 2019, The lottery ticket hypothesis: Finding
  sparse, trainable neural networks, in {\em International Conference on
  Learning Representations\/}.

\bibitem[Fu et~al.(2024)Fu, Yang, Wang, and Chen]{fu2024unveil}
Fu, H., Z.~Yang, M.~Wang, and M.~Chen, 2024, Unveil conditional diffusion
  models with classifier-free guidance: A sharp statistical theory.

\bibitem[Gao, Nguyen, and Zhu(2025)]{gao2023wasserstein}
Gao, X., H.~M. Nguyen, and L.~Zhu, 2025, Wasserstein convergence guarantees for
  a general class of score-based generative models, {\em Journal of Machine
  Learning Research\/} 26, 1--54.

\bibitem[Gao, Zha, and Zhou(2024)]{gao2024reward}
Gao, X., J.~Zha, and X.~Y. Zhou, 2024, Reward-directed score-based diffusion
  models via q-learning, {\em arXiv preprint arXiv:2409.04832\/} .

\bibitem[Giglio, Kelly, and Xiu(2022)]{giglio2022factor}
Giglio, S., B.~Kelly, and D.~Xiu, 2022, Factor models, machine learning, and
  asset pricing, {\em Annual Review of Financial Economics\/} 14, 337--368.

\bibitem[Giglio, Liao, and Xiu(2021)]{giglio2021thousands}
Giglio, S., Y.~Liao, and D.~Xiu, 2021, Thousands of alpha tests, {\em The
  Review of Financial Studies\/} 34, 3456--3496.

\bibitem[Giglio and Xiu(2021)]{giglio2021asset}
Giglio, S., and D.~Xiu, 2021, Asset pricing with omitted factors, {\em Journal
  of Political Economy\/} 129, 1947--1990.

\bibitem[Giglio, Xiu, and Zhang(2025)]{giglio2025test}
Giglio, S., D.~Xiu, and D.~Zhang, 2025, Test assets and weak factors, {\em The
  Journal of Finance\/} 80, 259--319.

\bibitem[Gotoh and Takeda(2011)]{gotoh2011role}
Gotoh, J.-y., and A.~Takeda, 2011, On the role of norm constraints in portfolio
  selection, {\em Computational Management Science\/} 8, 323--353.

\bibitem[Gouk et~al.(2021)Gouk, Frank, Pfahringer, and
  Cree]{gouk2021regularisation}
Gouk, H., E.~Frank, B.~Pfahringer, and M.~J. Cree, 2021, Regularisation of
  neural networks by enforcing lipschitz continuity, {\em Machine Learning\/}
  110, 393--416.

\bibitem[Gu, Kelly, and Xiu(2020)]{gu_etal:2020}
Gu, S., B.~Kelly, and D.~Xiu, 2020, Empirical asset pricing via machine
  learning, {\em The Review of Financial Studies\/} 33, 2223--2273.

\bibitem[Gu, Kelly, and Xiu(2021)]{gu2021autoencoder}
Gu, S., B.~Kelly, and D.~Xiu, 2021, Autoencoder asset pricing models, {\em
  Journal of Econometrics\/} 222, 429--450.

\bibitem[Gui et~al.(2021)Gui, Sun, Wen, Tao, and Ye]{gui2021review}
Gui, J., Z.~Sun, Y.~Wen, D.~Tao, and J.~Ye, 2021, A review on generative
  adversarial networks: Algorithms, theory, and applications, {\em IEEE
  Transactions on Knowledge and Data Engineering\/} 35, 3313--3332.

\bibitem[Guo et~al.(2024)Guo, Liu, Wang, Chen, Wang, Xu, and
  Cheng]{guo2024diffusion}
Guo, Z., J.~Liu, Y.~Wang, M.~Chen, D.~Wang, D.~Xu, and J.~Cheng, 2024,
  Diffusion models in bioinformatics and computational biology, {\em Nature
  Reviews Bioengineering\/} 2, 136--154.

\bibitem[Hambly, Xu, and Yang(2023)]{hambly2023recent}
Hambly, B., R.~Xu, and H.~Yang, 2023, Recent advances in reinforcement learning
  in finance, {\em Mathematical Finance\/} 33, 437--503.

\bibitem[Han, Mao, and Dally(2016)]{han2015deep}
Han, S., H.~Mao, and W.~J. Dally, 2016, Deep compression: Compressing deep
  neural networks with pruning, trained quantization and huffman coding, in
  {\em International Conference on Learning Representations\/}.

\bibitem[Han, Razaviyayn, and Xu(2024)]{han2024neural}
Han, Y., M.~Razaviyayn, and R.~Xu, 2024, Neural network-based score estimation
  in diffusion models: Optimization and generalization, in {\em International
  Conference on Learning Representations\/}.

\bibitem[Harvey, Liu, and Zhu(2016)]{harvey2016and}
Harvey, C.~R., Y.~Liu, and H.~Zhu, 2016, … and the cross-section of expected
  returns, {\em The Review of Financial Studies\/} 29, 5--68.

\bibitem[Haussmann and Pardoux(1986)]{haussmann1986time}
Haussmann, U.~G., and E.~Pardoux, 1986, Time reversal of diffusions, {\em The
  Annals of Probability\/}  1188--1205.

\bibitem[He et~al.(2019)He, Spokoyny, Neubig, and
  Berg-Kirkpatrick]{he2019lagging}
He, J., D.~Spokoyny, G.~Neubig, and T.~Berg-Kirkpatrick, 2019, Lagging
  inference networks and posterior collapse in variational autoencoders, in
  {\em International Conference on Learning Representations\/}.

\bibitem[He, Kou, and Peng(2022)]{he2022risk}
He, X.~D., S.~Kou, and X.~Peng, 2022, Risk measures: robustness, elicitability,
  and backtesting, {\em Annual Review of Statistics and Its Application\/} 9,
  141--166.

\bibitem[He, Kelly, and Manela(2017)]{he_etal:2017}
He, Z., B.~Kelly, and A.~Manela, 2017, Intermediary asset pricing: New evidence
  from many asset classes, {\em Journal of Financial Economics\/} 126, 1--35.

\bibitem[Ho, Jain, and Abbeel(2020)]{ho2020denoising}
Ho, J., A.~Jain, and P.~Abbeel, 2020, Denoising diffusion probabilistic models,
  {\em Advances in Neural Information Processing Systems\/} 33, 6840--6851.

\bibitem[Hoefler et~al.(2021)Hoefler, Alistarh, Ben-Nun, Dryden, and
  Peste]{hoefler2021sparsity}
Hoefler, T., D.~Alistarh, T.~Ben-Nun, N.~Dryden, and A.~Peste, 2021, Sparsity
  in deep learning: Pruning and growth for efficient inference and training in
  neural networks, {\em Journal of Machine Learning Research\/} 22, 1--124.

\bibitem[Hoffman and Johnson(2016)]{hoffman2016elbo}
Hoffman, M.~D., and M.~J. Johnson, 2016, Elbo surgery: yet another way to carve
  up the variational evidence lower bound, in {\em NIPS 2016 Workshop in
  Advances in Approximate Bayesian Inference\/}, volume~1.

\bibitem[Hou, Xue, and Zhang(2015)]{hou_etal:2015}
Hou, K., C.~Xue, and L.~Zhang, 2015, Digesting anomalies: An investment
  approach, {\em The Review of Financial Studies\/} 28, 650--705.

\bibitem[Huang, Huang, and Lin(2025)]{huang2024convergence}
Huang, D.~Z., J.~Huang, and Z.~Lin, 2025, Convergence analysis of probability
  flow {ODE} for score-based generative models, {\em IEEE Transactions on
  Information Theory\/} 71, 4581--4601.

\bibitem[Hultin et~al.(2023)Hultin, Hult, Proutiere, Samama, and
  Tarighati]{hultin2023generative}
Hultin, H., H.~Hult, A.~Proutiere, S.~Samama, and A.~Tarighati, 2023, A
  generative model of a limit order book using recurrent neural networks, {\em
  Quantitative Finance\/} 23, 931--958.

\bibitem[Hyv{\"a}rinen and Dayan(2005)]{hyvarinen2005estimation}
Hyv{\"a}rinen, A., and P.~Dayan, 2005, Estimation of non-normalized statistical
  models by score matching., {\em Journal of Machine Learning Research\/} 6.

\bibitem[Jacquier and Polson(2011)]{jacquier2011bayesian}
Jacquier, E., and N.~Polson, 2011, Bayesian methods in finance, in {\em The
  Oxford Handbook of Bayesian Econometrics\/}, chapter~9,  439--512 (Oxford
  University Press).

\bibitem[Jagannathan and Ma(2003)]{jagannathan2003risk}
Jagannathan, R., and T.~Ma, 2003, Risk reduction in large portfolios: Why
  imposing the wrong constraints helps, {\em The Journal of Finance\/} 58,
  1651--1683.

\bibitem[Jagannathan and Wang(1996)]{jagannathan_wang:1996}
Jagannathan, R., and Z.~Wang, 1996, The conditional capm and the cross-section
  of expected returns, {\em The Journal of Finance\/} 51, 3--53.

\bibitem[Jegadeesh and Titman(1993)]{jegadeesh_titman:1993}
Jegadeesh, N., and S.~Titman, 1993, Returns to buying winners and selling
  losers: Implications for stock market efficiency, {\em The Journal of
  Finance\/} 48, 65--91.

\bibitem[Jorion(1986)]{jorion1986bayes}
Jorion, P., 1986, Bayes-stein estimation for portfolio analysis, {\em Journal
  of Financial and Quantitative analysis\/} 21, 279--292.

\bibitem[Kan and Zhou(2007)]{kan2007optimal}
Kan, R., and G.~Zhou, 2007, Optimal portfolio choice with parameter
  uncertainty, {\em Journal of Financial and Quantitative Analysis\/} 42,
  621--656.

\bibitem[Karatzas and Shreve(1991)]{karatzas1991brownian}
Karatzas, I., and S.~Shreve, 1991, {\em Brownian motion and stochastic
  calculus\/}, volume 113 (Springer Science \& Business Media).

\bibitem[Karras et~al.(2022)Karras, Aittala, Aila, and
  Laine]{karras2022elucidating}
Karras, T., M.~Aittala, T.~Aila, and S.~Laine, 2022, Elucidating the design
  space of diffusion-based generative models, {\em Advances in Neural
  Information Processing Systems\/} 35, 26565--26577.

\bibitem[Kelly, Malamud, and Pedersen(2023)]{kelly2023principal}
Kelly, B., S.~Malamud, and L.~H. Pedersen, 2023, Principal portfolios, {\em The
  Journal of Finance\/} 78, 347--387.

\bibitem[Kelly, Palhares, and Pruitt(2023)]{kelly2023modeling}
Kelly, B., D.~Palhares, and S.~Pruitt, 2023, Modeling corporate bond returns,
  {\em The Journal of Finance\/} 78, 1967--2008.

\bibitem[Kelly, Xiu et~al.(2023)Kelly, Xiu, et~al.]{kelly2023financial}
Kelly, B., D.~Xiu, et~al., 2023, Financial machine learning, {\em Foundations
  and Trends{\textregistered} in Finance\/} 13, 205--363.

\bibitem[Kelly, Pruitt, and Su(2019)]{kelly2019characteristics}
Kelly, B.~T., S.~Pruitt, and Y.~Su, 2019, Characteristics are covariances: A
  unified model of risk and return, {\em Journal of Financial Economics\/} 134,
  501--524.

\bibitem[Koehler, Heckett, and Risteski(2023)]{koehler2022statistical}
Koehler, F., A.~Heckett, and A.~Risteski, 2023, Statistical efficiency of score
  matching: The view from isoperimetry, in {\em International Conference on
  Learning Representations\/}.

\bibitem[Kotelnikov et~al.(2023)Kotelnikov, Baranchuk, Rubachev, and
  Babenko]{kotelnikov2023tabddpm}
Kotelnikov, A., D.~Baranchuk, I.~Rubachev, and A.~Babenko, 2023, Tabddpm:
  Modelling tabular data with diffusion models, in {\em International
  Conference on Machine Learning\/},  17564--17579, PMLR.

\bibitem[Ledoit and Wolf(2003)]{ledoit2003improved}
Ledoit, O., and M.~Wolf, 2003, Improved estimation of the covariance matrix of
  stock returns with an application to portfolio selection, {\em Journal of
  Empirical Finance\/} 10, 603--621.

\bibitem[Ledoit and Wolf(2004)]{ledoit2004well}
Ledoit, O., and M.~Wolf, 2004, A well-conditioned estimator for
  large-dimensional covariance matrices, {\em Journal of Multivariate
  Analysis\/} 88, 365--411.

\bibitem[Ledoit and Wolf(2022)]{ledoit2022power}
Ledoit, O., and M.~Wolf, 2022, The power of (non-) linear shrinking: A review
  and guide to covariance matrix estimation, {\em Journal of Financial
  Econometrics\/} 20, 187--218.

\bibitem[Lee, Lu, and Tan(2022)]{lee2022convergence}
Lee, H., J.~Lu, and Y.~Tan, 2022, Convergence for score-based generative
  modeling with polynomial complexity, {\em Advances in Neural Information
  Processing Systems\/} 35, 22870--22882.

\bibitem[Lee, Lu, and Tan(2023)]{lee2023convergence}
Lee, H., J.~Lu, and Y.~Tan, 2023, Convergence of score-based generative
  modeling for general data distributions, in {\em International Conference on
  Algorithmic Learning Theory\/},  946--985, PMLR.

\bibitem[Lettau and Ludvigson(2001)]{lettau_ludvigson:2001}
Lettau, M., and S.~Ludvigson, 2001, Consumption, aggregate wealth, and expected
  stock returns, {\em The Journal of Finance\/} 56, 815--849.

\bibitem[Lettau and Pelger(2020{\natexlab{a}})]{lettau2020estimating}
Lettau, M., and M.~Pelger, 2020{\natexlab{a}}, Estimating latent asset-pricing
  factors, {\em Journal of Econometrics\/} 218, 1--31.

\bibitem[Lettau and Pelger(2020{\natexlab{b}})]{lettau2020factors}
Lettau, M., and M.~Pelger, 2020{\natexlab{b}}, Factors that fit the time series
  and cross-section of stock returns, {\em The Review of Financial Studies\/}
  33, 2274--2325.

\bibitem[Li et~al.(2024{\natexlab{a}})Li, Wei, Chen, and Chi]{li2024towards}
Li, G., Y.~Wei, Y.~Chen, and Y.~Chi, 2024{\natexlab{a}}, Towards non-asymptotic
  convergence for diffusion-based generative models, in {\em International
  Conference on Learning Representations\/}.

\bibitem[Li et~al.(2024{\natexlab{b}})Li, Wei, Chi, and Chen]{li2024sharp}
Li, G., Y.~Wei, Y.~Chi, and Y.~Chen, 2024{\natexlab{b}}, A sharp convergence
  theory for the probability flow odes of diffusion models, {\em arXiv preprint
  arXiv:2408.02320\/} .

\bibitem[Li, Di, and Gu(2025)]{li2024unified}
Li, R., Q.~Di, and Q.~Gu, 2025, Unified convergence analysis for score-based
  diffusion models with deterministic samplers, in {\em International
  Conference on Learning Representations\/}.

\bibitem[Li, Dai, and Qu(2024)]{li2024understanding}
Li, X., Y.~Dai, and Q.~Qu, 2024, Understanding generalizability of diffusion
  models requires rethinking the hidden gaussian structure, {\em Advances in
  Neural Information Processing Systems\/} 37, 57499--57538.

\bibitem[Li et~al.(2025)Li, Huang, Yang, and van
  Leeuwen]{li2025tabulardiffusion}
Li, Z., Q.~Huang, L.~Yang, and M.~van Leeuwen, 2025, Diffusion models for
  tabular data: Challenges, current progress, and future directions, {\em arXiv
  preprint arXiv:2502.17119\/} .

\bibitem[Liao et~al.(2024)Liao, Ni, Sabate-Vidales, Szpruch, Wiese, and
  Xiao]{liao2024sig}
Liao, S., H.~Ni, M.~Sabate-Vidales, L.~Szpruch, M.~Wiese, and B.~Xiao, 2024,
  Sig-wasserstein gans for conditional time series generation, {\em
  Mathematical Finance\/} 34, 622--670.

\bibitem[Liu et~al.(2024)Liu, Zhu, Jia, He, and Zheng]{liu2024learning}
Liu, H., T.~Zhu, N.~Jia, J.~He, and Z.~Zheng, 2024, Learning to simulate from
  heavy-tailed distribution via diffusion model, {\em Available at SSRN
  4975931\/} .

\bibitem[Liu, Tsyvinski, and Wu(2022)]{liu2022common}
Liu, Y., A.~Tsyvinski, and X.~Wu, 2022, Common risk factors in cryptocurrency,
  {\em The Journal of Finance\/} 77, 1133--1177.

\bibitem[Locatello et~al.(2019)Locatello, Bauer, Lucic, R{\"a}tsch, Gelly,
  Sch{\"o}lkopf, and Bachem]{locatello2019challenging}
Locatello, F., S.~Bauer, M.~Lucic, G.~R{\"a}tsch, S.~Gelly, B.~Sch{\"o}lkopf,
  and O.~Bachem, 2019, Challenging common assumptions in the unsupervised
  learning of disentangled representations, in {\em International Conference on
  Machine Learning\/}, volume~97,  4114--4124.

\bibitem[Louizos, Welling, and Kingma(2018)]{louizos2017learning}
Louizos, C., M.~Welling, and D.~P. Kingma, 2018, Learning sparse neural
  networks through $l_0$ regularization, in {\em International Conference on
  Learning Representations\/}.

\bibitem[Lyu et~al.(2022)Lyu, Xu, Yang, Lin, and Dai]{lyu2022accelerating}
Lyu, Z., X.~Xu, C.~Yang, D.~Lin, and B.~Dai, 2022, Accelerating diffusion
  models via early stop of the diffusion process, {\em arXiv preprint
  arXiv:2205.12524\/} .

\bibitem[Nagel(2013)]{nagel2013empirical}
Nagel, S., 2013, Empirical cross-sectional asset pricing, {\em Annual Review of
  Finance Economics\/} 5, 167--199.

\bibitem[Nichol and Dhariwal(2021)]{nichol2021improved}
Nichol, A.~Q., and P.~Dhariwal, 2021, Improved denoising diffusion
  probabilistic models, in {\em International Conference on Machine
  Learning\/},  8162--8171, PMLR.

\bibitem[Oko, Akiyama, and Suzuki(2023)]{oko2023diffusion}
Oko, K., S.~Akiyama, and T.~Suzuki, 2023, Diffusion models are minimax optimal
  distribution estimators, in {\em International Conference on Machine
  Learning\/}, volume 202,  26517--26582 (PMLR).

\bibitem[Onatski(2010)]{onatski2010determining}
Onatski, A., 2010, Determining the number of factors from empirical
  distribution of eigenvalues, {\em The Review of Economics and Statistics\/}
  92, 1004--1016.

\bibitem[P{\'a}stor and Stambaugh(2003)]{pastor_stambaugh:2003}
P{\'a}stor, L., and R.~F. Stambaugh, 2003, Liquidity risk and expected stock
  returns, {\em Journal of Political Economy\/} 111, 642--685.

\bibitem[Raponi, Robotti, and Zaffaroni(2020)]{raponi_etal:2020}
Raponi, V., C.~Robotti, and P.~Zaffaroni, 2020, Testing beta-pricing models
  using large cross-sections, {\em The Review of Financial Studies\/} 33,
  2796--2842.

\bibitem[Reppen and Soner(2023)]{reppen2023deep}
Reppen, A.~M., and H.~M. Soner, 2023, Deep empirical risk minimization in
  finance: Looking into the future, {\em Mathematical Finance\/} 33, 116--145.

\bibitem[Revuz and Yor(2013)]{revuz2013continuous}
Revuz, D., and M.~Yor, 2013, {\em Continuous martingales and Brownian
  motion\/}, volume 293 (Springer Science \& Business Media).

\bibitem[Ronneberger, Fischer, and Brox(2015)]{ronneberger2015u}
Ronneberger, O., P.~Fischer, and T.~Brox, 2015, U-net: Convolutional networks
  for biomedical image segmentation, in {\em Medical Image Computing and
  Computer-Assisted Intervention--MICCAI 2015: 18th International Conference,
  Munich, Germany, October 5-9, 2015, proceedings, part III 18\/},  234--241,
  Springer.

\bibitem[Ross(2013)]{ross2013arbitrage}
Ross, S.~A., 2013, The arbitrage theory of capital asset pricing, in {\em
  Handbook of the fundamentals of financial decision making: Part I\/},  11--30
  (World Scientific).

\bibitem[Saatci and Wilson(2017)]{saatci2017bayesian}
Saatci, Y., and A.~G. Wilson, 2017, Bayesian gan, {\em Advances in Neural
  Information Processing Systems\/} 30.

\bibitem[Schneider, Strahan, and Yang(2023)]{schneider2023bank}
Schneider, T., P.~E. Strahan, and J.~Yang, 2023, Bank stress testing: Public
  interest or regulatory capture?, {\em Review of Finance\/} 27, 423--467.

\bibitem[Shapiro and Zeng(2024)]{shapiro2024stress}
Shapiro, J., and J.~Zeng, 2024, Stress testing and bank lending, {\em The
  Review of Financial Studies\/} 37, 1265--1314.

\bibitem[Song and Ermon(2019)]{song2019generative}
Song, Y., and S.~Ermon, 2019, Generative modeling by estimating gradients of
  the data distribution, {\em Advances in Neural Information Processing
  Systems\/} 32.

\bibitem[Song and Ermon(2020)]{song2020improved}
Song, Y., and S.~Ermon, 2020, Improved techniques for training score-based
  generative models, {\em Advances in Neural Information Processing Systems\/}
  33, 12438--12448.

\bibitem[Song et~al.(2020)Song, Garg, Shi, and Ermon]{song2020sliced}
Song, Y., S.~Garg, J.~Shi, and S.~Ermon, 2020, Sliced score matching: A
  scalable approach to density and score estimation, in {\em Uncertainty in
  Artificial Intelligence\/},  574--584, PMLR.

\bibitem[Song et~al.(2021)Song, Sohl-Dickstein, Kingma, Kumar, Ermon, and
  Poole]{song2020score}
Song, Y., J.~Sohl-Dickstein, D.~P. Kingma, A.~Kumar, S.~Ermon, and B.~Poole,
  2021, Score-based generative modeling through stochastic differential
  equations, {\em International Conference on Learning Representations\/} .

\bibitem[Soudry et~al.(2018)Soudry, Hoffer, Nacson, Gunasekar, and
  Srebro]{soudry2018implicit}
Soudry, D., E.~Hoffer, M.~S. Nacson, S.~Gunasekar, and N.~Srebro, 2018, The
  implicit bias of gradient descent on separable data, {\em Journal of Machine
  Learning Research\/} 19, 1--57.

\bibitem[Srinivas, Subramanya, and Babu(2017)]{srinivas2017training}
Srinivas, S., A.~Subramanya, and R.~V. Babu, 2017, Training sparse neural
  networks, in {\em Conference on Computer Vision and Pattern Recognition
  Workshops\/},  455--462, IEEE.

\bibitem[Tang and Yang(2024)]{tang2024adaptivity}
Tang, R., and Y.~Yang, 2024, Adaptivity of diffusion models to manifold
  structures, in {\em International Conference on Artificial Intelligence and
  Statistics\/},  1648--1656, PMLR.

\bibitem[Tang and Zhao(2024)]{tang2024contractive}
Tang, W., and H.~Zhao, 2024, Contractive diffusion probabilistic models, {\em
  arXiv preprint arXiv:2401.13115\/} .

\bibitem[Tang and Zhao(2025)]{tang2024score}
Tang, W., and H.~Zhao, 2025, Score-based diffusion models via stochastic
  differential equations, {\em Statistics Surveys\/} 19, 28--64.

\bibitem[Tashiro et~al.(2021)Tashiro, Song, Song, and Ermon]{tashiro2021csdi}
Tashiro, Y., J.~Song, Y.~Song, and S.~Ermon, 2021, Csdi: Conditional
  score-based diffusion models for probabilistic time series imputation, {\em
  Advances in Neural Information Processing Systems\/} 34, 24804--24816.

\bibitem[Thomas and Joy(2006)]{thomas2006elements}
Thomas, M., and A.~T. Joy, 2006, {\em Elements of information theory\/}
  (Wiley-Interscience).

\bibitem[Tsybakov(2009)]{tsy2009nonparametric}
Tsybakov, A.~B., 2009, {\em Introduction to Nonparametric Estimation\/}, first
  edition (Springer).

\bibitem[Tu and Zhou(2010)]{tu_zhou:2010}
Tu, J., and G.~Zhou, 2010, Incorporating economic objectives into bayesian
  priors: Portfolio choice under parameter uncertainty, {\em Journal of
  Financial and Quantitative Analysis\/} 45, 959--986.

\bibitem[Tukey(1962)]{tukey1962future}
Tukey, J.~W., 1962, The future of data analysis, in {\em Breakthroughs in
  Statistics: Methodology and Distribution\/},  408--452 (Springer).

\bibitem[Vershynin(2018)]{vershynin2018high}
Vershynin, R., 2018, {\em High-dimensional probability: An introduction with
  applications in data science\/}, volume~47 (Cambridge university press).

\bibitem[Vincent(2011)]{vincent2011connection}
Vincent, P., 2011, A connection between score matching and denoising
  autoencoders, {\em Neural Computation\/} 23, 1661--1674.

\bibitem[Vuleti{\'c} and Cont(2025)]{vuletic2025volgan}
Vuleti{\'c}, M., and R.~Cont, 2025, {VOLGAN}: A generative model for
  arbitrage-free implied volatility surfaces, {\em Applied Mathematical
  Finance\/}  1--36.

\bibitem[Vuleti{\'c}, Prenzel, and Cucuringu(2024)]{vuletic2024fin}
Vuleti{\'c}, M., F.~Prenzel, and M.~Cucuringu, 2024, Fin-{GAN}: Forecasting and
  classifying financial time series via generative adversarial networks, {\em
  Quantitative Finance\/} 24, 175--199.

\bibitem[Wainwright(2019)]{wainwright2019high}
Wainwright, M.~J., 2019, {\em High-dimensional statistics: A non-asymptotic
  viewpoint\/}, volume~48 (Cambridge university press).

\bibitem[Wang et~al.(2024)Wang, Zhang, Zhang, Chen, Ma, and
  Qu]{wang2024diffusion}
Wang, P., H.~Zhang, Z.~Zhang, S.~Chen, Y.~Ma, and Q.~Qu, 2024, Diffusion models
  learn low-dimensional distributions via subspace clustering, {\em arXiv
  preprint arXiv:2409.02426\/} .

\bibitem[Weitzner et~al.(2025)Weitzner, Delbracio, Milanfar, and
  Giryes]{weitzner2024linear}
Weitzner, D., M.~Delbracio, P.~Milanfar, and R.~Giryes, 2025, The diffusion
  process as a correlation machine: Linear denoising insights, {\em
  Transactions on Machine Learning Research\/} .

\bibitem[Wibisono, Wu, and Yang(2024)]{wibisono2024optimal}
Wibisono, A., Y.~Wu, and K.~Y. Yang, 2024, Optimal score estimation via
  empirical bayes smoothing, in {\em Conference on Learning Theory\/}, volume
  247 (PMLR).

\bibitem[Xiao, Kreis, and Vahdat(2022)]{trilemma2022}
Xiao, Z., K.~Kreis, and A.~Vahdat, 2022, Tackling the generative learning
  trilemma with denoising diffusion gans, in {\em International Conference on
  Learning Representations\/}.

\bibitem[Yakovlev and Puchkin(2025)]{yakovlev2025generalization}
Yakovlev, K., and N.~Puchkin, 2025, Generalization error bound for denoising
  score matching under relaxed manifold assumption, in {\em Conference on
  Learning Theory\/}, volume 291,  5824--5891 (PMLR).

\bibitem[Yang and Wibisono(2022)]{yang2022convergence}
Yang, K.~Y., and A.~Wibisono, 2022, Convergence in kl and r{\'e}nyi divergence
  of the unadjusted langevin algorithm using estimated score, in {\em NeurIPS
  2022 Workshop on Score-Based Methods\/}.

\bibitem[Yang et~al.(2023)Yang, Zhang, Song, Hong, Xu, Zhao, Zhang, Cui, and
  Yang]{yang2023diffusion}
Yang, L., Z.~Zhang, Y.~Song, S.~Hong, R.~Xu, Y.~Zhao, W.~Zhang, B.~Cui, and
  M.-H. Yang, 2023, Diffusion models: A comprehensive survey of methods and
  applications, {\em ACM Computing Surveys\/} 56, 1--39.

\bibitem[Yarotsky(2017)]{yarotsky2017error}
Yarotsky, D., 2017, Error bounds for approximations with deep relu networks,
  {\em Neural Networks\/} 94, 103--114.

\bibitem[Yogo(2006)]{yogo:2006}
Yogo, M., 2006, A consumption-based explanation of expected stock returns, {\em
  The Journal of Finance\/} 61, 539--580.

\bibitem[Yoon, Jarrett, and Van~der Schaar(2019)]{yoon2019time}
Yoon, J., D.~Jarrett, and M.~Van~der Schaar, 2019, Time-series generative
  adversarial networks, {\em Advances in Neural Information Processing
  Systems\/} 32.

\bibitem[Zhang et~al.(2024)Zhang, Yin, Liang, and Liu]{zhang2024minimax}
Zhang, K., H.~Yin, F.~Liang, and J.~Liu, 2024, Minimax optimality of
  score-based diffusion models: Beyond the density lower bound assumptions, in
  {\em International Conference on Machine Learning\/}, volume 235,
  60134--60178 (PMLR).

\bibitem[Zhou and Li(2000)]{zhou2000continuous}
Zhou, X.~Y., and D.~Li, 2000, Continuous-time mean-variance portfolio
  selection: A stochastic {LQ} framework, {\em Applied Mathematics and
  Optimization\/} 42, 19--33.

\end{thebibliography}

\clearpage

\newpage

\setlength{\baselineskip}{1.5\baselineskip}
\onehalfspacing
\counterwithin{figure}{section}
\makeatletter
\renewcommand\p@subfigure{\thefigure}
\makeatother
\counterwithin{equation}{section}
\counterwithin{table}{section}
\counterwithin{theorem}{section}
\counterwithin{proposition}{section}
\counterwithin{definition}{section}
\counterwithin{example}{section}
\counterwithin{lemma}{section}
\counterwithin{remark}{section}
\counterwithin{assumption}{section}

\appendix

\section{Omitted Proof in Section \ref{sec: model}}
\label{sec: proof of lemma -- decomposition}
In this section, we provide the formal proof of Lemma~\ref{lem: score decomposition -- heterogeneous}.

\begin{proof}{Proof of Lemma~\ref{lem: score decomposition -- heterogeneous}.}
By definition, the marginal distribution of $\mathbf R_t$ is
\begin{align*}
p_t(\mathbf r) 
& = \int \underbrace{\phi(\mathbf r; \alpha_t \mathbf r_0, h_t \mathbf I_d)}_{\textrm{Gaussian transition kernel}} p_{\textrm{data}}(\mathbf r_0) \dd \mathbf r_0 \\
& \stackrel{(i)}{=} \int \phi(\mathbf r; \alpha_t (\bm\beta \mathbf f + \bm\varepsilon), h_t \mathbf I_d) {p_{\textrm{fac}}(\mathbf f) \phi\big(\bm\varepsilon;\mathbf 0, \operatorname{diag}\{\sigma_1^2, \sigma_2^2, \dots, \sigma_d^2\}\big) \dd \mathbf f \dd \bm\varepsilon.}
\end{align*}
Here, equality $(i)$ invokes the factor model \eqref{equ: factor model} to represent $\mathbf r_0$ and the independence between factor and noise.

Since $\bm\varepsilon$ is Gaussian with uncorrelated entries, we can simplify $p_t$ as
\begin{align}\label{equ: p_t}
p_t(\mathbf r) 
& = \int \frac{1}{(2\pi h_t)^{d/2}} \exp\left(-\frac{\| \mathbf r - \alpha_t (\bm\beta \mathbf f + \bm\varepsilon)\|_2^2}{2 h_t}\right) \prod_{i=1}^d \frac{1}{\sqrt{2\pi} \sigma_i} \exp\left(-\frac{\varepsilon_{i}^2}{2\sigma_{i}^2}\right) p_{\textrm{fac}}(\mathbf f) \dd \varepsilon_{i} \dd \mathbf f \nonumber \\
& \overset{(i)}{=} \int \prod_{i=1}^d \frac{1}{\sqrt{2\pi(h_t + \sigma_i^2 \alpha_t^2)}} \exp\left(-\frac{([\mathbf r - \alpha_t \bm\beta \mathbf f]_i)^2}{2 (h_t + \sigma_i^2 \alpha_t^2)}\right) p_{\textrm{fac}}(\mathbf f) \dd \mathbf f \nonumber \\
& = \int \frac{1}{\sqrt{(2\pi)^d \det(\bm\Lambda_t)}} \exp\left(-\frac{\|\bm\Lambda_t^{-\frac{1}{2}} \mathbf r - \alpha_t \bm\Lambda_t^{-\frac{1}{2}} \bm\beta \mathbf f \|_2^2}{2}\right) p_{\textrm{fac}}(\mathbf f) \dd \mathbf f,
\end{align}
where $(i)$ holds by completing the squares and integrating with respect to $\varepsilon_{i}$, and the last equality holds by applying the formula of $\bm\Lambda_{t}$ in \eqref{equ: Lambda_t}.

Now we define orthogonal decomposition of the rescaled returns $\bm\Lambda_t^{-\frac{1}{2}} \mathbf r$ into the subspace spanned by $\bm\Lambda_t^{-\frac{1}{2}} \bm\beta$ and its complement:
$$
\bm\Lambda_t^{-\frac{1}{2}} \mathbf r = (\mathbf I - \mathbf T_t)\bm\Lambda_t^{-\frac{1}{2}} \mathbf r + \mathbf T_t \bm\Lambda_t^{-\frac{1}{2}} \mathbf r,
$$
where $\bm\Gamma_t$ and $\mathbf T_t$ are defined in \eqref{equ: Projection_t}, respectively. Along with the fact that $\mathbf T_t(\mathbf I - \mathbf T_t) = \mathbf 0$, we can rewrite $p_t(\mathbf r)$ in \eqref{equ: p_t} as
\begin{equation}
\label{equ: density of r}
\begin{aligned}
    p_t(\mathbf r) = \frac{(2\pi)^{-\frac{d}{2}}}{\sqrt{\det(\bm\Lambda_t)}} \exp \bigg( - \frac{\| (\mathbf I - \mathbf T_t) \bm\Lambda_t^{-\frac{1}{2}} \mathbf r \|_2^2}{2} \bigg) \int \exp \bigg( - \frac{\| \mathbf T_t \bm\Lambda_t^{-\frac{1}{2}} \mathbf r - \alpha_t \bm\Lambda_t^{-\frac{1}{2}}\bm\beta \mathbf f \|_2^2}{2} \bigg) p_{\textrm{fac}}(\mathbf f) \dd \mathbf f.
\end{aligned}
\end{equation}
Take the take gradient of $\log p_t$ with respect to $\mathbf r$ using expression in \eqref{equ: density of r}, we obtain:
\begin{align*}
    \nabla \log p_t(\mathbf r) &= -\bm\Lambda_t^{-\frac{1}{2}} \left(\mathbf I - \mathbf T_t\right)^2 \bm \Lambda_t^{-\frac{1}{2}} \mathbf r \\
    &\qquad - \frac{\int \bm\Lambda_t^{-\frac{1}{2}} \mathbf T_t ( \mathbf T_t \bm\Lambda_t^{-\frac{1}{2}} \mathbf r - \alpha_t \bm\Lambda_t^{-\frac{1}{2}}\bm\beta \mathbf f ) \exp \left( - \frac{\| \mathbf T_t \bm\Lambda_t^{-\frac{1}{2}} \mathbf r - \alpha_t \bm\Lambda_t^{-\frac{1}{2}}\bm\beta \mathbf f \|_2^2}{2} \right) p_{\textrm{fac}}(\mathbf f) \dd \mathbf f}{\int \exp \left( - \frac{\| \mathbf T_t \bm\Lambda_t^{-\frac{1}{2}} \mathbf r - \alpha_t \bm\Lambda_t^{-\frac{1}{2}}\bm\beta \mathbf f \|_2^2}{2} \right) p_{\textrm{fac}}(\mathbf f) \dd \mathbf f} \\ 
    &\stackrel{(i)}{=} -\bm\Lambda_t^{-\frac{1}{2}} \left(\mathbf I - \mathbf T_t\right) \bm \Lambda_t^{-\frac{1}{2}} \mathbf r \\
    &\qquad - \frac{\int \bm\Lambda_t^{-1} \bm\beta ( \bm\beta^\top \bm\Lambda_t^{\frac{1}{2}} \mathbf T_t \bm\Lambda_t^{-\frac{1}{2}} \mathbf r - \alpha_t \mathbf f ) \exp \left( - \frac{\| \bm\Lambda_t^{-\frac{1}{2}} \bm\beta ( \bm\beta^\top \bm\Lambda_t^{\frac{1}{2}} \mathbf T_t \bm\Lambda_t^{-\frac{1}{2}} \mathbf r - \alpha_t \mathbf f ) \|_2^2}{2} \right) p_{\textrm{fac}}(\mathbf f) \dd \mathbf f}{\int \exp \left( - \frac{\| \bm\Lambda_t^{-\frac{1}{2}} \bm\beta ( \bm\beta^\top \bm\Lambda_t^{\frac{1}{2}} \mathbf T_t \bm\Lambda_t^{-\frac{1}{2}} \mathbf r - \alpha_t \mathbf f ) \|_2^2}{2} \right) p_{\textrm{fac}}(\mathbf f) \dd \mathbf f} \\
    &\stackrel{(ii)}{=} - \underbrace{\bm \Lambda_t^{-\frac{1}{2}} (\mathbf I - \mathbf T_t) \cdot \bm \Lambda_t^{-\frac{1}{2}} \mathbf r}_{\mathbf s_{\textrm{comp}}(\mathbf r, t):\textrm{ Complement score }} + \underbrace{\mathbf T_t \bm\Lambda_t^{\frac{1}{2}} \bm\beta \cdot \nabla \log p_t^{\rm{fac}}( \bm\beta^\top \bm\Lambda_t^{\frac{1}{2}} \mathbf T_t \cdot \bm\Lambda_t^{-\frac{1}{2}} \mathbf r)}_{\mathbf s_{\textrm{sub}}( \bm\beta^\top \bm\Lambda_t^{\frac{1}{2}} \mathbf T_t \cdot \bm\Lambda_t^{-\frac{1}{2}} \mathbf r, t):\textrm{ Subspace score }},
\end{align*}
where $(i)$ holds due to the fact that $(\mathbf I -\mathbf T_t)^2 = \mathbf I -\mathbf T_t$ and the following straightforward calculation by invoking the formula $\mathbf T_t$ in \eqref{equ: Projection_t} and $\bm\beta^{\top} \bm\beta = \mathbf I_k$:
\begin{align*}
    \mathbf T_t \bm\Lambda_t^{-\frac{1}{2}} \mathbf r - \alpha_t \bm\Lambda_t^{-\frac{1}{2}}\bm\beta \mathbf f = \bm\Lambda_t^{-\frac{1}{2}} \bm\beta \bm\Gamma_t \bm\beta^\top \bm\Lambda_t^{-\frac{1}{2}} \cdot \bm\Lambda_t^{-\frac{1}{2}} \mathbf r - \alpha_t \bm\Lambda_t^{-\frac{1}{2}}\bm\beta \mathbf f = \bm\Lambda_t^{-\frac{1}{2}} \bm\beta \left( \bm\beta^\top \bm\Lambda_t^{\frac{1}{2}} \mathbf T_t \bm\Lambda_t^{-\frac{1}{2}} \mathbf r - \alpha_t \mathbf f \right).
\end{align*}
In addition, $(ii)$ follows the definition of $p_t^{\rm{fac}}$. \hfill\Halmos
\end{proof}

\section{Omitted Proofs in Section \ref{sec: score approximation and estimation}}
\label{sec:proofs}
In this section, we provide proofs of Theorem~\ref{theorem: score approximation}, Theorem~\ref{theorem: score estimation} and the lemmas used in the proof.

\subsection{Proof of Theorem~\ref{theorem: score approximation}}
\label{subsec: proof of theorem -- score approximation}
\iftrue
\begin{proof}{Proof.}
Given the neural network architecture defined in \eqref{equ: ReLU network},
our goal is to construct a diagonal matrix $\mathbf D_t = \operatorname{diag}\left\{ 1/(h_t + \alpha_t^2 c_1), \dots, 1/(h_t + \alpha_t^2 c_d)\right\} \in \bb R^d$ induced by a vector $\mathbf c = (c_1, \dots, c_d) \in \bb R^d$, a matrix $\mathbf V \in \mathbb{R}^{d \times k}$ with orthonormal columns, and a ReLU network $\mathbf g_{\bm\zeta}(\mathbf V^\top \mathbf D_t \mathbf r, t) \in \calS_{\textrm{ReLU}}$ so that $\mathbf s_{\bm\theta}(\mathbf r, t)$ serves as a good approximator to $\nabla \log p_t(\mathbf r)$. Thanks to the score decomposition in \eqref{equ: score function rearranged}, we choose $\mathbf D_t(\sigma_1, \dots, \sigma_d) = \operatorname{diag}\{ 1/(h_t + \sigma_1^2 \alpha_t^2), \dots, 1/(h_t + \sigma_d^2 \alpha_t^2) \}$ and $\mathbf V = \bm\beta$. It remains to choose neural network hyper-parameters to guarantee the desired approximation power.

\noindent \underline{\bf Step 1: Approximation on $\calC \times \left[0, T\right]$.} 
Define $\calC = \left\{\mathbf z\in \mathbb{R}^k |\left\|\mathbf z\right\|_{2} \leq S \right\}$ as a $k$-dimensional  ball of radius $S > 0$, with the choice of 
\begin{equation}
\label{equ: choice of S}
S = \order(\sqrt{(1+\sigma_{\max}^2) (k + \log (1/\epsilon))} ).
\end{equation}
On $\calC \times\left[0, T\right]$, we approximate the coordinate $\xi_i$ separately for each $i = 1, \dots, d$. 
To ease the analysis, we define the following linear transformation:
\begin{equation}
\label{equ: definition of (xi, y, t) prime}
    \bm\xi^{\prime}(\mathbf y^{\prime}, t^{\prime}) := \bm\xi(\mathbf z, t) \textrm{ with } \mathbf y^{\prime} := (\mathbf z + S \mathbf{1}) / (2S) \textrm{ and } t^{\prime} := t / T
\end{equation}
such that the domain of $\calC \times [0, T]$ is transformed to be contained within $[0,1]^{k} \times [0, 1]$. Therefore, we can equivalently approximate $\xi_i^{\prime}$ for each $i = 1, \dots, d$ on the new domain $[0,1]^{k} \times [0, 1]$.

Recall that the subspace score $\mathbf{s}_{\textrm{sub}}\left(\mathbf z, t\right)$ is $L_s$-Lipschitz in $\mathbf z$ by Assumption \ref{assumption: Lipschitz}. Then, by the definition of $\mathbf s_{\textrm{sub}}$ and $\bm\xi$ in \eqref{equ: s_parallel} and \eqref{equ: definition of xi}, we derive that $\bm\xi(\mathbf z, t)$ is $2(1+L_s)(1+\sigma_{\max}^4)$-Lipschitz in $\mathbf z$. Immediately, we obtain that $\bm\xi^{\prime}(\mathbf y^{\prime}, t^{\prime})$ is $4S (1+L_s)(1+\sigma_{\max}^4)$-Lipschitz in $\mathbf{y}^{\prime}$; so is each coordinate~$\xi_i$. Denote $L_z = 2(1+L_s)(1+\sigma_{\max}^4)$.

Next, define
\begin{equation}
\label{equ: definition of tauS}
\tau(S) := \sup_{t \in [0, T]} \sup_{\mathbf z \in \calC}\left\|\frac{\partial}{\partial t} \bm\xi\left(\mathbf z, t\right)\right\|_2.
\end{equation}
Then for any $\mathbf{y}^{\prime} \in [0,1]^k$, the Lipschitz constant of $\bm\xi^{\prime}(\mathbf{y}^{\prime}, t^{\prime})$ with respect to $t^{\prime}$ is bounded by~$T \tau(S)$. Substituting the order of $S$ in \eqref{equ: choice of S} into the upper bound of $\tau(S)$ in \eqref{equ: upper bound of tauS} in Lemma \ref{lem: lipschitz constant bound in Theorem 1}, we have
\begin{equation}
\label{equ: upper bound of tauS in Theorem 1}
    \tau(S) = \order(L_s (1+\sigma_{\max}^7) \operatorname{poly} k^{3/2} \log^{3/2} (1/\epsilon)).
\end{equation}
For notation simplicity, we abbreviate $\tau(S)$ as $\tau$ when there is no confusion. 

Now we construct a partition of the product space $[0, 1]^k \times [0, 1]$. For the hypercube $[0, 1]^k$, we partition it uniformly into smaller, non-overlapping hypercubes, each with an edge length of $e_1$. Similarly, we partition the interval $[0, 1]$ into non-overlapping subintervals, each of length $e_2$. Here, we take
$$
e_{1} = \calO\left(\frac{\epsilon}{S L_z}\right) \quad \text{and} \quad e_{2} = \calO\left(\frac{\epsilon}{T \tau}\right).
$$
In addition, we denote $N_1=\left\lceil\frac{1}{e_{1}}\right\rceil, N_2=\left\lceil\frac{1}{e_{2}}\right\rceil$.

Let $\mathbf m=\left[m_1, \ldots, m_k \right]^{\top} \in\left\{0, \ldots, N_1 - 1\right\}^{k}$ be a multi-index. We define a function $\Bar{g}^{\prime} : \bb R^{k+1} \mapsto \bb R^k$, with the $i$-th component $g_i^{\prime}$ being
\begin{equation}\label{eq:gprimei}
    \Bar{g}_i^{\prime} \left(\mathbf y^{\prime}, t^{\prime}\right)= \sum_{\mathbf m, j=0, \ldots, N_{2}-1} \xi_i^{\prime} \left(\frac{\mathbf m}{N_1}, \frac{j}{N_{2}}\right) \Psi_{\mathbf{m}, j}\left(\mathbf y^{\prime}, t^{\prime}\right).
\end{equation}
Here $\Psi_{\mathbf{m}, j}\left(\mathbf{y}^{\prime}, t^{\prime}\right)$ is a partition of unity function. Specifically, we choose $\Psi_{\mathbf m, j}$ as a product of coordinate-wise trapezoid functions:
$$
\Psi_{\mathbf{m}, j}\left(\mathbf{y}^{\prime}, t^{\prime}\right)=\psi\left(3 N_{2}\left(t^{\prime}-\frac{j}{N_{2}}\right)\right) \prod_{i=1}^{d} \psi\left(3 N_{1}\left(y_{i}^{\prime}-\frac{m_{i}}{N_{1}}\right)\right),
$$
where $\psi$ is a one-demensional trapezoid function with the specific formula:
$$
\psi(a)= \begin{cases}1, & |a|<1 \\ 2-|a|, & |a| \in[1,2] \\ 0, & |a|>2\end{cases}.
$$
For any $1 \leq i \leq k$, we claim that 
\begin{enumerate}
    \item $\Bar{g}_i^{\prime}$ defined in \eqref{eq:gprimei} can approximate $\xi_i$ arbitrarily well as long as $N_1$ and $N_2$ are sufficiently large;
    \item $\Bar{g}_i^{\prime}$ can be well approximated by a ReLU neural network $\Bar{g}_{\bm\zeta, i}^{\prime}$ with a controllable error.
\end{enumerate}
The above two claims can be verified using Lemma 10 in \citet{chen2020statistical}, in which we substitute the Lipschitz coefficients $4S(1+L_s)(1+\sigma_{\max}^4)$ and $T \tau$ of $\bm\xi^{\prime}$ into the error analysis. Specifically, for any $1 \leq i \leq k$, we consider the ReLU neural network $\Bar{g}_{\bm\zeta, i}^{\prime}$ that satisfies the following Lipschitz property:
\begin{align*}
    &\left\|\Bar{g}_{\bm\zeta, i}^{\prime}\left(\mathbf y_1^{\prime}, t^{\prime}\right)-\Bar{g}_{\bm\zeta, i}^{\prime}\left(\mathbf y_2^{\prime}, t^{\prime}\right)\right\|_{\infty} \leq 10 k S L_z\left\|\mathbf y_1^{\prime} - \mathbf y_1^{\prime}\right\|_{2},\quad \forall~ \mathbf y_1^{\prime}, \mathbf y_1^{\prime} \in [0, 1]^k, t^{\prime} \in [0, 1],~ \textrm{and} \\
    &\left\|\Bar{g}_{\bm\zeta, i}^{\prime}\left(\mathbf y^{\prime}, t_1^{\prime}\right)-\Bar{g}_{\bm\zeta, i}^{\prime}\left(\mathbf y^{\prime}, t_2^{\prime}\right)\right\|_{\infty} \leq 10 T \tau\left\|t_{1}-t_{2}\right\|_{2}, \forall~ t_1^{\prime}, t_2^{\prime} \in [0, 1], \mathbf y^{\prime} \in [0, 1]^k.
\end{align*} 
By concatenating $\Bar{g}_{\bm\zeta, i}^{\prime}$'s together, we construct $\overline{\mathbf g}_{\bm\zeta} = \left[\Bar{g}_{\bm\zeta, 1}, \dots, \Bar{g}_{\bm\zeta, k}\right]^{\top}$. For a given error level $\epsilon > 0$, with a neural network configuration
\begin{gather*}
M=\order\left(T \tau (L_s+1)^k (1+\sigma_{\max}^k) \epsilon^{-(k+1)} \left(\log \frac{1}{\epsilon} + k\right)^{\frac{k}{2}}\right), \gamma_1 = 20k(1+L_s)(1+\sigma_{\max}^4), \\
L=\order\left(\log \frac{1}{\epsilon}+ k \right),\ J=\order\left(T \tau (1+L_s)^k (1+\sigma_{\max}^{k}) \epsilon^{-(k+1)} \left(\log \frac{1}{\epsilon} + k \right)^{\frac{k+2}{2}}\right), \gamma_{2}=10 \tau, \\
K = \order\bigg((1+L_s) (1+\sigma_{\max}^{4}) \left(\log \frac{1}{\epsilon} + k\right)^{\frac{1}{2}}\bigg), \kappa=\max \bigg\{(1+L_s) (1+\sigma_{\max}^4)\bigg(\log \frac{1}{\epsilon} + k\bigg)^{\frac{1}{2}}, T \tau\bigg\},
\end{gather*}
we have
$$
\sup_{(\mathbf y^{\prime}, t^{\prime}) \in [0,1]^k \times [0, 1]}\left\|\overline{\mathbf g}_{\bm\zeta}^{\prime}(\mathbf y^{\prime}, t^{\prime})-\bm\xi^{\prime}(\mathbf y^{\prime}, t^{\prime})\right\|_{\infty} \leq \epsilon.
$$
{To transform the function $\overline{\mathbf g}_{\bm\zeta}^{\prime}$ back to domain $\calC \times (0,T]$, we define 
\begin{equation}
\label{equ: definition of overline g_zeta}
\overline{\mathbf g}_{\bm\zeta}(\mathbf z, t) := \overline{\mathbf g}_{\bm\zeta}^{\prime}(\mathbf y^{\prime}, t^{\prime}) \indicator\{\| \mathbf z \|_2 \leq S\}.
\end{equation}
} By the definition of $\bm\xi^{\prime}$ in \eqref{equ: definition of (xi, y, t) prime}, we deduce that
\begin{equation}
\label{equ: L_infty bound of g_zeta minus xi}
\sup_{(\mathbf z, t) \in \calC \times [0, T]}\left\|\overline{\mathbf g}_{\bm\zeta}(\mathbf z, t)-\bm\xi(\mathbf z, t)\right\|_{\infty} \leq \epsilon.
\end{equation}
Also by the variable transformation in \eqref{equ: definition of (xi, y, t) prime}, we obtain that  $\overline{\mathbf g}_{\bm\zeta}$ is Lipschitz continuous in $\mathbf z$ and $t$. Specifically, for any $\mathbf z_1, \mathbf z_2 \in \calC$ and $t \in [0, T]$, it holds that
$$
\left\|\overline{\mathbf g}_{\bm\zeta}\left(\mathbf z_1, t\right)-\overline{\mathbf g}_{\bm\zeta}\left(\mathbf z_2, t\right)\right\|_{\infty} \leq 10kL_z\left\|\mathbf z_1 - \mathbf z_2\right\|_{2}.
$$
In addition, for any $t_{1}, t_{2} \in [0, T]$ and $\mathbf z \in \calC$, it holds that
$$
\left\|\overline{\mathbf g}_{\bm\zeta}\left(\mathbf z, t_{1}\right)-\overline{\mathbf g}_{\bm\zeta}\left(\mathbf z, t_{2}\right)\right\|_{\infty} \leq 10 \tau|t_{1}-t_{2}|.
$$
By definition of $\overline{\mathbf g}_{\bm\zeta}$ in \eqref{equ: definition of overline g_zeta}, we have $\overline{\mathbf g}_{\bm\zeta}\left(\mathbf z, t\right)=\mathbf{0}$ for $\| \mathbf z \|_2 > S$. Therefore, the Lipschitz continuity property in $\mathbf z$ can be extended to  $\bb R^k$.

\vspace{5pt}
\noindent \underline{\bf Step 2: Bounding $L^2$ Approximation Error.} {Denote $\mathbf Z =\bm\beta^\top \bm\Lambda_t^{-1} \mathbf R_t$ with the distribution $P_{t}^{\textrm{fac}}$. The $L^2$ approximation error of $\overline{\mathbf g}_{\bm\zeta}$ can be decomposed into two terms
\begin{equation}
\begin{aligned}
\label{equ: L2 approximation error decomposition}
    \left\| \bm\xi\left(\mathbf Z, t\right)-\overline{\mathbf g}_{\bm\zeta}\left(\mathbf Z, t\right)\right\|_{L^2\left(P_{t}^{\rm fac}\right)} &= \left\|\left(\bm\xi\left(\mathbf Z, t\right)-\overline{\mathbf g}_{\bm\zeta}\left(\mathbf Z, t\right)\right) \indicator\{\left\| \mathbf Z \right\|_2 < S \}\right\|_{L^2(P_{t}^{\rm fac})} \\
    &\quad + \left\| \bm\xi\left(\mathbf Z, t\right) \indicator \{\|\mathbf Z\|_2 > S\} \right\|_{L^2 \left(P_{t}^{\rm fac}\right)}.
\end{aligned}
\end{equation}
By applying the $L^{\infty}$ approximation error bound in \eqref{equ: L_infty bound of g_zeta minus xi},} the first term in \eqref{equ: L2 approximation error decomposition} is bounded by
\begin{equation}\label{eq:first_term}
\left\|\left(\bm\xi\left(\mathbf Z, t\right)-\overline{\mathbf g}_{\bm\zeta}\left(\mathbf Z, t\right)\right) \indicator\{\left\| \mathbf Z \right\|_2 < S \}\right\|_{L^2(P_{t}^{\rm fac})} \leq \sqrt{k} \underset{(\mathbf z, t) \in \calC \times [0, T]}{\sup}\left\|\left(\bm\xi\left(\mathbf z, t\right)-\overline{\mathbf g}_{\bm\zeta}\left(\mathbf z, t\right)\right) \right\|_{\infty} \leq \sqrt{k} \epsilon.
\end{equation}
The second term on the right-hand side of \eqref{equ: L2 approximation error decomposition} is controlled by the upper bound \eqref{equ: truncation error -- heterogeneous} in Lemma \ref{lem: truncation error -- heterogeneous}. Specifically, by choosing $S = \order
(\sqrt{(1+\sigma_{\max}^2) (k + \log (1/\epsilon))} )$, we have
\begin{equation}\label{eq:second_term}
    \left\| \bm\xi\left(\mathbf Z, t\right) \indicator \{\|\mathbf Z\|_2 > S\} \right\|_{L^2 \left(P_{t}^{\rm fac}\right)} \leq \epsilon.
\end{equation}
Combining \eqref{eq:first_term} and \eqref{eq:second_term}, we deduce that
\begin{equation}
\label{eq:inter_error}
\left\|\bm\xi\left(\mathbf Z, t\right)-\overline{\mathbf g}_{\bm\zeta}\left(\mathbf Z, t\right)\right\|_{L^2\left(P_{t}^{\rm fac}\right)} \leq(\sqrt{k}+1) \epsilon.    
\end{equation}
Furthermore, by involving $\bar{\mathbf g}_{\bm\zeta}$, we construct the following approximator $\overline{\mathbf s}_{\bm\theta}$ for $\nabla \log p_{t}(\mathbf r)$ 
\begin{equation}
\label{equ: Bar s theta}
\overline{\mathbf s}_{\bm\theta}(\mathbf r, t) := \alpha_t\bm\Lambda_t^{-1} \bm\beta \overline{\mathbf g}_{\bm\zeta}(\bm\beta^\top \bm\Lambda_t^{-1} \mathbf r, t) - \bm\Lambda_t^{-1} \mathbf r.
\end{equation}
Then, by applying the formula of $\nabla \log p_t$ and $\overline{\mathbf s}_{\bm\theta}$ in \eqref{equ: score function rearranged} and \eqref{equ: Bar s theta} respectively, we obtain that
\begin{align*}
    \|\nabla \log p_{t}(\cdot)-\overline{\mathbf s}_{\bm\zeta}(\cdot, t)\|_{L^2\left(P_{t}\right)} &= \left\|\alpha_t \bm\Lambda_t^{-1/2} \bm\beta (\bm\xi(\mathbf Z, t) - \overline{\mathbf g}_{\bm\zeta}(\mathbf Z, t)) \right\|_{L^2\left(P_{t}^{\rm fac}\right)} \leq \frac{(\sqrt{k}+1)\epsilon}{\min\{\sigma_d^2, 1\}},
\end{align*}
where the inequality invokes $\| \bm\Lambda_t^{-1} \bm\beta \|_{\mathrm{op}} \leq 1 /(h_t + \sigma_d^2 \alpha_t^2) \leq 1/\min\{\sigma_d^2, 1\}$ and the error bound \eqref{eq:inter_error}. \hfill\Halmos
\end{proof}
\fi

\subsection{Proof of Theorem~\ref{theorem: score estimation}}
\label{subsec: proof of theorem -- score estimation}
\iftrue
\begin{proof}{Proof.}

\vspace{4pt}
\noindent \underline{\bf Step 1: Error Decomposition.} 
The proof is based on the following bias-variance decomposition on $\calL\left(\widehat{\mathbf s}_{\bm\theta}\right)$.
For any $a \in(0,1)$, we decompose $\calL\left(\widehat{\mathbf s}_{\bm\theta}\right)$ as
\begin{align*}
\calL\left(\widehat{\mathbf s}_{\bm\theta}\right) &= \calL\left(\widehat{\mathbf s}_{\bm\theta}\right)-(1+a) \widehat{\calL}\left(\widehat{\mathbf s}_{\bm\theta}\right)+(1+a) \widehat{\calL}\left(\widehat{\mathbf s}_{\bm\theta}\right) \\
& \leq \underbrace{\calL^{\textrm{trunc}}\left(\widehat{\mathbf s}_{\bm\theta}\right) - (1+a) \widehat{\calL}^{\textrm{trunc}}\left(\widehat{\mathbf s}_{\bm\theta}\right)}_{(A)}+\underbrace{\calL\left(\widehat{\mathbf s}_{\bm\theta}\right)-\calL^{\textrm{trunc}}\left(\widehat{\mathbf s}_{\bm\theta}\right)}_{(B)} + (1+a) \underbrace{\inf _{\mathbf s_{\bm\theta} \in \calS_{\textrm{NN}}} \widehat{\calL}\left(\mathbf s_{\bm\theta}\right)}_{(C)},
\end{align*}
where $\calL^{\textrm{trunc}}$ is defined as
\begin{equation}
\label{equ: definition of truncated loss function}
\calL^{\textrm{trunc}}\left(\mathbf s_{\bm\theta}\right) := \int \ell^{\textrm{trunc}}(\mathbf r ; \mathbf s_{\bm\theta}) p_t(\mathbf r) \dd \mathbf r \quad \textrm{ with } \quad \ell^{\textrm{trunc}}(\mathbf r ; \mathbf s_{\bm\theta}) := \ell(\mathbf r ; \mathbf s_{\bm\theta}) \indicator\left\{\|\mathbf r\|_{2} \leq \rho \right\}
\end{equation}
subject to some truncation radius $\rho$ to be determined. In the sequel, we bound $(A)$ -- $(C)$ separately. The term $(A)$ is the statistical error due to finite samples, term $(B)$ is the truncation error, term $(C)$ reflects the approximation error of $\calS_{\textrm{NN}}$.

Note that the introduction of the hyper-parameter $a>0$ (to be determined) is to handle the bias by applying Lemma 15 of \citet{chen2023score}, which is a standard Bernstein-type concentration inequality for empirical processes over function classes (see, for example, Theorem 3.27, Theorem 4.10, and Chapter 5 in \citet{wainwright2019high}). Conversely, setting $a=0$ results in a convergence rate at $\order(n^{-1/2})$, as derived using only Hoeffding's concentration inequality.


\vspace{5pt}
\noindent \underline{\bf Step 2: Bounding Term $(A)$.} We denote $\mathcal{G}:=\left\{\ell^{\textrm{trunc}}\left(\cdot ; \mathbf s_{\bm\theta}\right): \mathbf s_{\bm\theta} \in \calS_{\textrm{NN}}\right\}$ as the class of loss functions induced by the score network $\calS_{\textrm{NN}}$. We first determine an upper bound on all functions in $\mathcal{G}$ by  bounding $\sup_{\mathbf s_{\bm\theta} \in \calS_{\textrm{NN}}} \sup_{\mathbf r \in \bb{R}^{d}} 
| \ell^{\textrm{trunc}}(\mathbf r; \mathbf s_{\bm\theta}) |$.

To start, we consider 
\begin{align}
\bigg\|\mathbf s_{\bm\theta}(\mathbf r^{\prime}, t)+\frac{(\mathbf r^{\prime}-\alpha_t \mathbf r)}{h_t}\bigg\|_2 &\leq \|\mathbf s_{\bm\theta}(\mathbf r^{\prime}, t) + \mathbf D_t \mathbf r^{\prime} \|_2 + \bigg\| \frac{(\mathbf I - h_t \mathbf D_t) \mathbf r^{\prime}}{h_t} \bigg\|_2 + \bigg\|\frac{\alpha_t \mathbf r}{h_t} \bigg\|_2 \nonumber \\
&\stackrel{(i)}{=} \alpha_t \|\mathbf D_t \mathbf V \mathbf g_{\bm\theta}(\mathbf V^{\top} \mathbf D_t \mathbf r^{\prime}, t)\|_2 + \frac{\| \mathbf I - h_t \mathbf D_t \|_2 \| \mathbf r^{\prime} \|_2}{h_t} + \frac{\alpha_t \|\mathbf r \|_2}{h_t} \nonumber \\
&= \order \bigg(\frac{K + \| \mathbf r^{\prime} \|_2 + \| \mathbf r \|_2}{h_t}\bigg), \label{equ: first bound for ell in (A) of Theorem 2}
\end{align}
where $(i)$ holds by applying the formula of $\mathbf s_{\bm\theta}$ in \eqref{equ: score network} and $(ii)$ follows from the facts that $\alpha_t^2 \leq 1$, $\| \mathbf D_t \|_{\textrm{op}} \leq 1/h_t$, $\| \mathbf V \|_{\textrm{op}} = 1$, and $\| \mathbf g_{\bm\theta} \|_2 \leq K$.

By the definition of $\ell^{\textrm{trunc}}$ in \eqref{equ: definition of truncated loss function}, for any $\mathbf s_{\bm\theta} \in \calS_{\textrm{NN}}$ we have  $\ell^{\textrm{trunc}}(\mathbf r ; \mathbf s_{\bm\theta}) = 0$  if $\| \mathbf r \|_2 > \rho$. For any $\| \mathbf r \|_2 \leq \rho$, we have
\begin{align*}
\ell^{\textrm{trunc}}(\mathbf r ; \mathbf s_{\bm\theta}) & = \frac{1}{T - t_0} \int_{t_0}^T \E_{\mathbf R_t|\mathbf R_0 = \mathbf r} \bigg\|\mathbf s_{\bm\theta}(\mathbf R_t, t) + \frac{\mathbf R_t - \alpha_t \mathbf r}{h_t} \bigg\|_2^2 \cdot \indicator\left\{\|\mathbf r\|_{2} \leq \rho\right\} \dd t \\
& \stackrel{(i)}{=} \order \bigg( \frac{1}{T - t_0} \int_{t_0}^{T} \E_{\mathbf R_t|\mathbf R_0 = \mathbf r} \bigg(\frac{K^2 + \| \mathbf R_t \|_2^2 + \| \mathbf r \|_2^2 }{h_t^2}\bigg) \cdot \indicator\left\{\|\mathbf r\|_{2} \leq \rho\right\} \dd t \bigg) \\
& \stackrel{(ii)}{=} \order \bigg( \frac{1}{T - t_0} \int_{t_0}^{T} \bigg( \frac{2 \rho^2 + K^2}{h_t^2} + \frac{d}{h_t} \bigg) \dd t \bigg) \\
& = \order \bigg( \frac{\rho^2 + K^2}{t_0\left(T-t_0\right)} + \frac{d}{T - t_0}\log \frac{T}{t_0} \bigg),
\end{align*}
where $(i)$ holds by the uniform upper bound \eqref{equ: first bound for ell in (A) of Theorem 2}; $(ii)$ holds by applying the facts that $(\mathbf R_t | \mathbf R_0 = \mathbf r) \sim \mathcal{N}(\alpha_t \mathbf r, h_t \mathbf I_d)$ and $\| \mathbf r \|_2 \leq \rho$ and $\alpha_t^2 \leq 1$.

To bound term $(A)$, it is essential to consider the covering number of $\calS_{\textrm{NN}}$,  as it measures the approximation power of the neural network class. Take $\mathbf s_{\bm\theta_1}$ and $\mathbf s_{\bm\theta_2}$ such that 
$$
\sup_{\|\mathbf r^{\prime}\|_{2} \leq 3\rho + \sqrt{d \log d}, t \in\left[t_0, T\right]}\left\|\mathbf s_{\bm\theta_1}(\mathbf r^{\prime}, t)-\mathbf s_{\bm\theta_2}(\mathbf r^{\prime}, t)\right\|_{2} \leq \iota,
$$
we then have
\begin{align*}
&\quad \left\|\ell^{\textrm{trunc}}\left(\cdot ; \mathbf s_{\bm\theta_1}\right)-\ell^{\textrm{trunc}}\left(\cdot ; \mathbf s_{\bm\theta_2}\right)\right\|_{\infty} \\
&= \sup _{\|\mathbf r\|_{2} \leq \rho} \frac{1}{T - t_0} \int_{t_0}^{T}  \E_{\mathbf R_t|\mathbf R_0 = \mathbf r} \bigg[\left\|\mathbf s_{\bm\theta_1}\left(\mathbf R_t, t\right)-\mathbf s_{\bm\theta_2}\left(\mathbf R_t, t\right)\right\|_{2} \cdot \left\|\mathbf s_{\bm\theta_1}\left(\mathbf R_t, t\right) + \mathbf s_{\bm\theta_2}\left(\mathbf R_t, t\right) + \frac{2(\mathbf R_t - \alpha_t \mathbf r)}{h_t}\right\|_{2}\bigg] \dd t \\
& \stackrel{(i)}{\leq} \sup _{\|\mathbf r\|_{2} \leq \rho} \frac{1}{T-t_0} \int_{t_0}^{T} \E_{\mathbf R_t | \mathbf R_0 = \mathbf r} \left[ \frac{2}{h_t} \left(K + \left\| \mathbf R_t \right\|_2 + \left\| \mathbf r \right\|_2 \right) \left\|\mathbf s_{\bm\theta_1}\left(\mathbf R_t, t\right)-\mathbf s_{\bm\theta_2}\left(\mathbf R_t, t\right)\right\|_{2} \right. \\
&\qquad \qquad \left. \cdot \indicator\left\{\left\|\mathbf R_t\right\|_{2} \leq 3\rho + \sqrt{d \log d}\right\}\right] \dd t \\
&\quad + \sup _{\|\mathbf r\|_{2} \leq \rho} \frac{1}{T-t_0} \int_{t_0}^{T} \E_{\mathbf R_t | \mathbf R_0 = \mathbf r} \left[ \frac{2}{h_t} \left(K + \left\| \mathbf R_t \right\|_2 + \left\| \mathbf r \right\|_2 \right) \left\|\mathbf s_{\bm\theta_1}\left(\mathbf R_t, t\right)-\mathbf s_{\bm\theta_2}\left(\mathbf R_t, t\right)\right\|_{2} \right. \\
&\qquad \qquad \left. \cdot \indicator\left\{\left\|\mathbf R_t\right\|_{2} > 3\rho + \sqrt{d \log d}\right\}\right] \dd t \\
& \stackrel{(ii)}{\leq} \sup_{\|\mathbf r\|_{2} \leq \rho} \frac{2 \iota}{T-t_0} \int_{t_0}^{T} \E_{\mathbf R_t | \mathbf R_0 = \mathbf r} \left[ \frac{1}{h_t} \left(K + \left\| \mathbf R_t \right\|_2 + \left\| \mathbf r \right\|_2 \right) \cdot \indicator\left\{\left\|\mathbf R_t\right\|_{2} \leq 3\rho + \sqrt{d \log d}\right\}\right] \dd t \\
&\quad + \sup _{\|\mathbf r\|_{2} \leq \rho} \frac{1}{T-t_0} \int_{t_0}^{T} \E_{\mathbf R_t | \mathbf R_0 = \mathbf r} \left[ \frac{2}{h_t} \left(K + \left\| \mathbf R_t \right\|_2 + \left\| \mathbf r \right\|_2 \right) \left\|\mathbf s_{\bm\theta_1}\left(\mathbf R_t, t\right)-\mathbf s_{\bm\theta_2}\left(\mathbf R_t, t\right)\right\|_{2} \right. \\
&\qquad \qquad \left. \cdot \indicator\left\{\left\|\mathbf R_t\right\|_{2} > 3\rho + \sqrt{d \log d}\right\}\right] \dd t,
\end{align*}
where $(i)$ follows from applying the upper bound in \eqref{equ: first bound for ell in (A) of Theorem 2} and decomposing the error into two parts: within the compact domain of radius $3\rho + \sqrt{d \log d}$, and outside this domain; $(ii)$ holds since $\left\|\mathbf s_{\bm\theta_1}(\mathbf r^{\prime}, t)-\mathbf s_{\bm\theta_2}(\mathbf r^{\prime}, t)\right\|_{2} \leq \iota$ in the compact domain $\| \mathbf r^{\prime} \| \leq 3\rho + \sqrt{d \log d}$. Then, we deduce that
\begin{align*}
& \quad \|\ell^{\textrm{trunc}}(\cdot ; \mathbf s_{\bm\theta_1})-\ell^{\textrm{trunc}}(\cdot ; \mathbf s_{\bm\theta_2})\|_{\infty} \\
& \stackrel{(i)}{=} \order \bigg( \frac{2 \iota}{T-t_0} \int_{t_0}^{T} \frac{K + \sqrt{h_t d} + 2\rho}{h_t} \dd t + \frac{2}{T-t_0} \int_{t_0}^{T} \frac{1}{h_t} \bigg(\rho K^2 h_t^{-\frac{d+4}{2}} \left(\frac{\rho}{d}\right)^{d} \exp \bigg(-\frac{\rho^2}{h_t} \bigg) \bigg) \dd t \bigg) \\
& \stackrel{(ii)}{=} \order\bigg(\iota \cdot \frac{(\rho + K)\log (T/t_0) + \sqrt{d}(\sqrt{T}-\sqrt{t_0})}{T-t_0} + \rho K^2 \left( \frac{\rho}{d} \right)^{\frac{d}{2}} \exp \left(- \frac{\rho^2}{2 h_T}\right) \bigg),
\end{align*}
where $(i)$ holds by applying $ (\mathbf R_t | \mathbf R_0 = \mathbf r) \sim \mathcal{N}(\alpha_t \mathbf r, h_t \mathbf I_d)$, $\| \mathbf r \|_2 \leq \rho$ and the upper bound \eqref{equ: SNN truncation error} in Lemma \ref{lem: SNN truncation error -- heterogeneous}; $(ii)$ follows from the facts that $h_t = \order(t)$ as $t \to 0$ and the second term in $(i)$ has a dominating exponential decay rate $\exp (-\rho^2 / h_t)$. 
For notational simplicity, we denote
\begin{equation}
\label{equ: definition of eta}
\eta := \rho K^2 ( \rho / d )^{d/2} \exp (- \rho^2 / (2 h_T) ).
\end{equation}
Denote the $\tau$-covering number of a class of functions $\mathcal{H}$ under a metric $\Psi(\cdot)$ by
\begin{equation}
\label{equ: definition of covering number}
\mathfrak{N}(\tau, \mathcal{H}, \Psi(\cdot)) = \inf\{|\mathcal{H}_1|: \mathcal{H}_1 \subseteq \mathcal{H},\ \forall h \in \mathcal{H},\ \exists h_1 \in \mathcal{H},\ \mathrm{s.t.} \ \Psi(h, h_1) \leq \tau \},
\end{equation}
where $|\mathcal{H}_1|$ represents the number of functions in the class $\mathcal{H}_1$.
Immediately, we can deduce that an $\iota$-covering of $\calS_{\textrm{NN}}$ induces a covering of $\mathcal{G}$ with an accuracy $\iota \cdot \frac{(\rho + K)\log (T/t_0) + \sqrt{d}(\sqrt{T}-\sqrt{t_0}))}{T-t_0}+ \eta$. To apply Bernstein-type concentration inequality \cite[Lemma 15]{chen2023score} for $\ell^{\textrm{trunc}}\left(\cdot ; \mathbf s_{\bm\theta}\right)$, let us take $B = \order(\frac{\rho^2 + K^2 + t_0 d \log (T/t_0)}{t_0(T-t_0)})$, $\tau = \iota$ and the corresponding covering number of $\calS_{\textrm{NN}}$ as
\begin{equation*}
\mathfrak{N}\left( \frac{(\iota - \eta)(T-t_0)}{(\rho + K)\log (T/t_0) + \sqrt{d}(\sqrt{T}-\sqrt{t_0}))}, \calS_{\textrm{NN}},\|\cdot\|_{2}\right)
\end{equation*}
Then by Lemma 15 of \cite{chen2023score}, with probability $1-\delta$, it holds that
\begin{equation}
\label{equ: error bound of term (A) in Theorem 2}
    (A) = \order\left(\frac{(1+\frac{3}{a}) \big(\frac{\rho^2 + K^2 + t_0 d \log (\frac{T}{t_0})}{t_0(T-t_0)}\big)}{n} \log \frac{\mathfrak{N}\Big( \frac{(\iota - \eta)(T-t_0)}{(\rho + K)\log (\frac{T}{t_0}) + \sqrt{d}(\sqrt{T}-\sqrt{t_0})}, \calS_{\textrm{NN}},\|\cdot\|_{2}\Big)}{\delta}+(2+a) \iota \right).
\end{equation}

\vspace{5pt}
\noindent \underline{\bf Step 3: Bounding Term $(B)$.} By applying the formulas of $\ell$ and $\ell^{\textrm{trunc}}$ in \eqref{equ: empirical loss function} and \eqref{equ: definition of truncated loss function}, we have
\begin{align*}
(B) &= \frac{1}{T - t_0} \int_{t_0}^{T} \E_{\mathbf R_0 \sim P_{\textrm{data}}} \E_{\mathbf R_t | \mathbf R_0}\bigg[\bigg\|\widehat{\mathbf s}_{\bm\theta}(\mathbf R_t, t)+\frac{(\mathbf R_t-\alpha_t \mathbf R_0)}{h_t}\bigg\|_2^2 \indicator\left\{\|\mathbf R_0\|_{2} > \rho\right\}\bigg] \dd t \\
& \leq \frac{2}{T - t_0} \int_{t_0}^{T} \E_{\mathbf R_0 \sim P_{\textrm{data}}} \bigg[ \frac{h_t d + K^2 + 2\|\mathbf R_0\|^2_2}{h_t^2} \indicator\{ \| \mathbf R_0 \|_2 > \rho \} \bigg] \dd t,
\end{align*}
where the inequality follows from applying the upper bound \eqref{equ: first bound for ell in (A) of Theorem 2} and $(\mathbf R_t | \mathbf R_0 = \mathbf r) \sim \mathcal{N}(\alpha_t \mathbf r, h_t \mathbf I_d)$.
Notice that the density of $\mathbf R_0 = \bm\beta \mathbf F + \bm\varepsilon$ can be bounded by
\begin{align}
    p_{\textrm{data}}(\mathbf r) &= \prod_{i=1}^d \frac{1}{\sqrt{2\pi} \sigma_i} \exp\left(-\frac{\varepsilon_{i}^2}{2\sigma_{i}^2}\right) p_{\textrm{fac}}(\mathbf f) \nonumber \\
    & \stackrel{(i)}{\leq} \frac{(2 \pi)^{-(d+k)/2} C_1}{\prod_{i=1}^d \sigma_i} \exp\left(- \frac{\sigma_{\max}^{-2} \| \bm\epsilon \|_2^2 + C_2 \| \bm\beta \mathbf f \|_2^2}{2}\right) \nonumber \\
    & \leq \frac{C_1 (2 \pi)^{-(d+k)/2}}{\prod_{i=1}^d \sigma_i} \exp\left(- \frac{\| \mathbf r \|_2^2 }{2 (\sigma_{\max}^2 + 1 / C_2)}\right), \label{equ: R0 sub-Gaussian tail -- heterogeneous}
\end{align}
where $(i)$ follows from the sub-Gaussian tail \eqref{eqn:subgaussian} in Assumption \ref{assumption: subgaussian} and the fact that $\bm\beta$ is a norm-preserving transformation satisfying $\bm\beta^\top \bm\beta = \mathbf I$ in Assumption~\ref{assumption: factor}. Therefore, by applying the upper bound of $p_{\textrm{data}}$ in \eqref{equ: R0 sub-Gaussian tail -- heterogeneous}, we obtain
\begin{align}
(B) & \leq \frac{2}{T - t_0} \int_{t_0}^{T} \bigg[ \bigg( \frac{h_t d + K^2 + 2\|\mathbf r\|^2}{h_t^2} \bigg) \frac{C_1 (2 \pi)^{-\frac{d+k}{2}}}{\prod_{i=1}^d \sigma_i} \exp\left(- \frac{\| \mathbf r \|_2^2 }{2 (\sigma_{\max}^2 + 1 / C_2)}\right) \indicator\{ \| \mathbf r \|_2 > \rho \} \bigg] \dd t \nonumber \\
& \stackrel{(i)}{\leq} \frac{C_1 d (\sigma_{\max}^2 + 1 / C_2) 2^{-(d+k)/2}}{(\prod_{i=1}^{d} \sigma_i) (T - t_0)\Gamma(d/2 + 1)} \exp \bigg(- \frac{\rho^2}{2 (\sigma_{\max}^2 + 1 / C_2)} \bigg) \displaystyle\int_{t_0}^{T} \frac{\rho^{d-1} (h_t d + K^2 + \rho^2)}{h_t^2} \dd t \nonumber \\
& =\order\left( \frac{d \rho^{d-1} 2^{-(d+k)/2} (\sigma_{\max}^2 + 1 / C_2)}{(\prod_{i=1}^{d} \sigma_i)(T - t_0)\Gamma(d/2 + 1)} \left(\frac{\rho^2 + K^2}{t_0} + d\log \frac{T}{t_0} \right) \exp \left(- \frac{\rho^2}{2 (\sigma_{\max}^2 + 1 / C_2)} \right) \right), \label{equ: error bound of term (B) in Theorem 2}
\end{align}
where $(i)$ follows from the tail estimation in Proposition 2.6.6 of \citet{vershynin2018high}.

\vspace{5pt}
\noindent \underline{\bf Step 4: Bounding Term $(C)$.} Recall that $\overline{\mathbf s}_{\bm\theta}$ is the constructed network approximator in Theorem \ref{theorem: score approximation}. For any $\epsilon>0$, we have
$$
(C) \leq \underbrace{\widehat{\calL}(\overline{\mathbf s}_{\bm\theta})-(1+a) \calL^{\textrm{trunc}}(\overline{\mathbf s}_{\bm\theta})}_{(\spadesuit)}+(1+a) \underbrace{\calL^{\textrm{trunc}}(\overline{\mathbf s}_{\bm\theta})}_{(\clubsuit)},
$$
where $(\spadesuit)$ is the statistical error and $(\clubsuit)$ is the approximation error.

First, we can bound $(\spadesuit)$ with high probability using the fact that $\mathbf R_0$ has a sub-Gaussian tail. Specifically, applying Proposition 2.6.6 of \citet{vershynin2018high} to $\mathbf R_0$ with density bound \eqref{equ: R0 sub-Gaussian tail -- heterogeneous}, we obtain
\begin{equation}
\label{equ: definition of probability q}
    \mathbb{P}\left(\|\mathbf R_0\|_{2} > \rho\right) \leq \frac{C_1 d (\sigma_{\max}^2 + 1 / C_2) 2^{-(d+k)/2}}{(T - t_0)\Gamma(d/2 + 1) \left(\prod_{i=1}^d \sigma_i\right)} \rho^{d-1} \exp \left(- \frac{\rho^2}{2 (\sigma_{\max}^2 + 1/C_2)} \right) := q.
\end{equation}
Applying union bound for $n$ i.i.d. data samples $\{\mathbf R_0^i\}_{i=1}^{n}$ from $P_{\textrm{data}}$ leads to
\begin{align*}
\mathbb{P}\left(\left\|\mathbf R_{0}^{i}\right\|_{2} \leq \rho \textrm{ for all } i=1, \ldots, n\right) \geq 1 - nq.
\end{align*}
Immediately, we obtain that, with probability $1 - nq$, it holds
$$
(\spadesuit) = \widehat{\calL}^{\textrm{trunc}}(\overline{\mathbf s}_{\bm\theta})-(1+a) \calL^{\textrm{trunc}}(\overline{\mathbf s}_{\bm\theta}).
$$
Meanwhile, Lemma 15 of \cite{chen2023score} implies that with probability $1 - \delta$, it holds that
\begin{equation}
\label{equ: error bound of (Box)}
(\spadesuit) = \order\bigg(\frac{(1+6 / a)}{n}\bigg( \frac{\rho^2 + K^2}{t_0(T-t_0)} + \frac{d}{T - t_0}\log \frac{T}{t_0} \bigg) \log \frac{1}{\delta}\bigg).
\end{equation}
Here, we take $\delta$ defined in \eqref{equ: error bound of term (A) in Theorem 2}. For $(\clubsuit)$, we have 
\begin{align}
(\clubsuit) \leq \calL(\overline{\mathbf s}_{\bm\theta}) &= \frac{1}{T-t_0} \int_{t_0}^{T}\E_{\mathbf R_t \sim P_t}\left\|\overline{\mathbf s}_{\bm\theta}(\mathbf R_t, t)-\nabla \log p_{t}(\mathbf R_t)\right\|_{2}^{2} \dd t \nonumber \\
&\qquad + \underbrace{\calL\left(\overline{\mathbf s}_{\bm\theta}\right)-\frac{1}{T-t_0} \int_{t_0}^{T}\E_{\mathbf R_t \sim P_t}\left\|\overline{\mathbf s}_{\bm\theta}(\mathbf R_t, t)-\nabla \log p_{t}(\mathbf R_t)\right\|_{2}^{2} \dd t}_{(E)}. \label{equ: definition of (E)}
\end{align}
Then, by Theorem~\ref{theorem: score approximation}, we have
\begin{equation}
\label{equ: error bound of (Diamond)}
(\clubsuit) = \order\bigg(\frac{k \epsilon^2}{t_0 (T-t_0)}\bigg) + (E).
\end{equation}
Note that the two terms in $(E)$ are equivalent to the score-matching objectives \eqref{equ: score-matching loss function} and \eqref{equ: denoising score-matching loss function}, hence $(E)$ is on a constant order.

Combining two error bounds for $(\spadesuit)$ and $(\clubsuit)$ in \eqref{equ: error bound of (Box)} and \eqref{equ: error bound of (Diamond)}, we deduce that, with probability $(1 - nq)(1-\delta)$, it holds
\begin{equation}
\label{equ: error bound of term (C) in Theorem 2}
    (C) = \order\bigg(\frac{(1+6 / a)}{n}\left( \frac{\rho^2 + K^2}{t_0\left(T-t_0\right)} + \frac{d}{T - t_0}\log \frac{T}{t_0} \right) \log \frac{1}{\delta} + \frac{(1+a) k \epsilon^2}{t_0 (T-t_0)} \bigg) + (1+a) \cdot (E).
\end{equation}

\vspace{5pt}
\noindent \underline{\bf Step 5: Choosing $\rho$ and putting together $(A)$, $(B)$ and $(C)$.} 
Under a fixed $\delta > 0$ in \eqref{equ: error bound of term (A) in Theorem 2}, we choose $\rho$ and $\tau$ as the following to balance terms $(A)$, $(B)$, and $(C)$,
\begin{equation}
\label{equ: choice of rho and tau}
    \rho = \order \left( \sqrt{\sigma_1^2 \left(d + \log K + \log (n / \delta) \right) }\right) \textrm{ and } \tau = \order \bigg( \frac{1}{n t_0 (T-t_0)} \bigg).
\end{equation}
By direct calculation, our choice of $\rho$ implies that $q \leq \delta / n$, where $q$ is defined in \eqref{equ: definition of probability q}. Next, we derive the error bound for terms (A)-(C) under our choice of the hyper-parameters.
\begin{enumerate}
    \item For term $(A)$, we first give an upper bound for $\eta$ defined in \eqref{equ: definition of eta}. Substituting the order of $\rho$ in \eqref{equ: choice of rho and tau} into \eqref{equ: definition of eta}, we deduce that
    \begin{equation}
    \label{equ: order of eta}
        \eta = \order\bigg( \frac{1}{nt_0(T-t_0)} \bigg).
    \end{equation}
    Then, substituting the order of $K$ in \eqref{equ: order of parameters in Theorem 1} in Theorem \ref{theorem: score approximation} and the hyperparameters in \eqref{equ: choice of rho and tau} and \eqref{equ: order of eta} into \eqref{equ: error bound of term (A) in Theorem 2}, we obtain that with probability $1 - \delta$, it holds that
    \begin{equation}
    \begin{aligned}
    \label{equ: final error bound of term (A) in Theorem 2}
        (A) &= \order\Bigg(\frac{(1+\sigma_{\max}^8) (1 + 3/a) \left((1+L_s)^2 \left(\log \frac{1}{\epsilon} + k\right)^2 + d + \log \frac{n}{\delta} \right)}{n t_0 (T - t_0)} \\
        &\qquad \qquad \qquad \cdot \log \frac{\mathfrak{N}\left( \frac{1}{n t_0 \left( \rho + K + \sqrt{d}(\sqrt{T} - \sqrt{t_0}) \right)}, \calS_{\textrm{NN}},\|\cdot\|_{2}\right)}{\delta} + \frac{1}{n t_0 (T - t_0)} \Bigg) \\
        & \stackrel{(i)}{=} \order\Bigg(\frac{(1+\sigma_{\max}^8) (1 + 3/a) \left((L_s+1)^2 (\log \frac{1}{\epsilon} + k)^2 + d + \log \frac{n}{\delta} \right)}{n t_0 (T - t_0)} \\
        &\qquad \cdot \bigg( dk + T \tau(1+L_s)^k (1+\sigma_{\max}^k) \epsilon^{-(k+1)} \big(\log \frac{1}{\epsilon}+ k \big)^{\frac{k+4}{2}} \bigg) \log \frac{T \tau d k}{t_0 \iota \epsilon} + \frac{1}{n t_0 (T - t_0)} \Bigg),
    \end{aligned}
    \end{equation}
    where $(i)$ follows from applying the upper bound \eqref{equ: covering number of SNN} for the covering number of $\calS_{\textrm{NN}}$ in Lemma~\ref{lem: SNN covering number}.
    \item For term $(B)$, by plugging the order of $\rho$ and $K$, defined in \eqref{equ: choice of rho and tau} and \eqref{equ: order of parameters in Theorem 1}, into \eqref{equ: error bound of term (B) in Theorem 2} and by straightforward calculations, we have
    \begin{equation}
    \label{equ: final error bound of term (B) in Theorem 2}
        (B) = \order \bigg( \frac{1}{nt_0(T-t_0)} \bigg).
    \end{equation}
    \item For term $(C)$, applying the fact that $q \leq \frac{\delta}{n}$ and the order of $\rho$ and $K$ in \eqref{equ: choice of rho and tau} and \eqref{equ: order of parameters in Theorem 1} to \eqref{equ: error bound of term (C) in Theorem 2}, with probability $1 - 2\delta$, it holds
    \begin{equation}
    \begin{aligned}
    \label{equ: final error bound of term (C) in Theorem 2}
        (C) &= \order\left(\frac{(1+\sigma_{\max}^8) (1 + 6/a) \left((1+L_s)^2 \left(\log \frac{1}{\epsilon} + k\right)^2 + d + \log \frac{n}{\delta} \right)}{n t_0 (T - t_0)} \log \frac{1}{\delta} \right. \\
        &\qquad \qquad \qquad \left. + \frac{1}{n t_0 (T - t_0)} + \frac{k \epsilon^2}{\min \{ \sigma_d^4, 1 \}}\right) + (1+a) \cdot (E).
    \end{aligned}
    \end{equation}
\end{enumerate}
Summing up the error terms in \eqref{equ: final error bound of term (A) in Theorem 2}-\eqref{equ: final error bound of term (C) in Theorem 2}, we derive that with probability $1 - 3\delta$, it holds that
\begin{align*}
\calL\left(\widehat{\mathbf s}_{\bm\theta}\right) & \leq  (A) + (B) + (1+a) \cdot(C) \\
& =  \order\Bigg(\frac{(1+\sigma_{\max}^8) (1 + 6/a) \left((1+L_s)^2 (\log \frac{1}{\epsilon} + k\right)^2 + d + \log \frac{n}{\delta} )}{n t_0 (T - t_0)} \\
& \qquad \cdot \bigg( dk + T \tau(1+L_s)^k (1+\sigma_{\max}^k) \epsilon^{-(k+1)} \bigg(\log \frac{1}{\epsilon}+ k \bigg)^{\frac{k+4}{2}} \bigg) \log \frac{T \tau d k}{t_0 \iota \epsilon} \Bigg) \\
& \qquad \qquad + \order \Bigg( \frac{1}{n t_0 (T - t_0)} + \frac{k \epsilon^2}{t_0 (T - t_0)} \Bigg) + (1+a)^2 \cdot (E).
\end{align*}
By the definition of $(E)$ in \eqref{equ: definition of (E)} and setting $a=\epsilon^2$, with probability $1 - 3\delta$, it holds that
\begin{equation*}
\begin{aligned}
    &\quad \frac{1}{T-t_0} \int_{t_0}^{T}\E_{\mathbf R_t \sim P_t}\left\|\overline{\mathbf s}_{\bm\theta}(\mathbf R_t, t)-\nabla \log p_{t}(\mathbf R_t)\right\|_{2}^{2} \dd t \\
    & =  \order\Bigg(\frac{(1+\sigma_{\max}^8) \left((1+L_s)^2 (\log \frac{1}{\epsilon} + k\right)^2 + d + \log \frac{n}{\delta} )}{\epsilon^2 n t_0 (T - t_0)} \\
    & \qquad \cdot \bigg( dk + T \tau(1+L_s)^k (1+\sigma_{\max}^k) \epsilon^{-(k+1)} \bigg(\log \frac{1}{\epsilon}+ k \bigg)^{\frac{k+4}{2}} \bigg) \log \frac{T \tau d k}{t_0 \iota \epsilon} \Bigg) \\
    & \qquad \qquad + \order \Bigg( \frac{1}{n t_0 (T - t_0)} + \frac{k \epsilon^2}{t_0 (T - t_0)} \Bigg) \\
    & =: (E_1) + (E_2).
\end{aligned}
\end{equation*}

\vspace{5pt}
\noindent \underline{\bf Step 6: Balancing Error Terms.} 
Take $\delta = 1/(3n)$ such that with probability $1-1/n$, it holds that
\begin{equation}
\begin{aligned}
    (E_1) &= \order\bigg( \frac{(1+\sigma_{\max}^{k+8}) (1+L_s)^{k} (d^2\log d) (k^{\frac{k+7}{2}} \log k) (\tau \log \tau) T \epsilon^{-(k+3)} \log^{\frac{k+7}{2}} (\frac{1}{\epsilon}) \log^3 n}{n t_0} \bigg) \nonumber \\
    & \stackrel{(i)}{=} \order\bigg( \frac{(1+\sigma_{\max}^{2k}) (1+L_s)^{k} (d^{\frac{7}{2}}\log d) (k^{\frac{k+10}{2}} \log^{\frac{5}{2}} k) \epsilon^{-(k+3)} \log^{\frac{k+10}{2}} (\frac{1}{\epsilon}) \log^{\frac{9}{2}} n}{n t_0} \bigg) \\
    & \stackrel{(ii)}{=} \Tilde{\order} \bigg( \frac{1}{n}\epsilon^{-(k+3)}\log^{\frac{k+10}{2}} (\frac{1}{\epsilon}) \bigg)
\end{aligned}
\end{equation}
and
\begin{equation}
    (E_2) = \Tilde{\order}\bigg( \frac{1}{n} + \epsilon^2 \bigg).
\end{equation}
Here $(i)$ follows from invoking the upper bound of $\tau(S)$ in \eqref{equ: upper bound of tauS in Theorem 1} and $\Tilde{\order}(\cdot)$ in $(i)$ holds by keeping terms only on the sample size $n$ and the error term $\epsilon$.

To balance two error terms $(E_1)$ and $(E_2)$, we choose $\epsilon$ as the following
\begin{equation}
\label{equ: choice of epsilon and delta(n)}
\epsilon = n^{-\frac{1-\delta(n)}{k+5}} \textrm{ with } \delta(n) = \frac{(k+10)\log \log n}{2 \log n}.
\end{equation}
Consequently, we obtain
\begin{equation*}
\begin{aligned}
    \frac{1}{n} \epsilon^{-(k+3)} \log^{\frac{k+10}{2}} (1/\epsilon) &= n^{-1+\frac{(k+3)(1 - \delta(n))}{k+5}} (1/\epsilon)^{\frac{(k+10)\log \log (1/\epsilon)}{2\log (1/\epsilon)}} \\
    &= n^{-\frac{2-2\delta(n)}{k+5}} \cdot n^{-1 + \big(1 + \frac{(k+10)\log \log (1/\epsilon)}{2(k+5)\log (1/\epsilon)}\big) \big(1 - \delta(n)\big)} \\
    &\stackrel{(i)}{=} \order \bigg( n^{-\frac{2-2\delta(n)}{k+5}} \bigg) = \epsilon^2,
\end{aligned}
\end{equation*}
where $(i)$ holds by the formula of $\delta(n)$ in \eqref{equ: choice of epsilon and delta(n)}. By straightforward calculations, we deduce that, with probability $1-\frac{1}{n}$, it holds
\begin{equation*}
\begin{aligned}
    &\quad \frac{1}{T-t_0} \int_{t_0}^{T}\E_{\mathbf R_t \sim P_t}\left\|\overline{\mathbf s}_{\bm\theta}(\mathbf R_t, t)-\nabla \log p_{t}(\mathbf R_t)\right\|_{2}^{2} \dd t \\
    & = \order\bigg(\frac{1}{t_0}(1+\sigma_{\max}^{2k})(1+L_s)^{k} d^2 k^{\frac{k+10}{2}}\left(\sqrt{d} n^{-\frac{2-2 \delta(n)}{k+5}}+ n^{-\frac{k+3+2 \delta(n)}{k+5}}\right) \log d \log^4 n\bigg) \\
    & = \Tilde{\order}\bigg( \frac{1}{t_0}(1 + \sigma_{\max}^{2k}) d^{\frac{5}{2}} k^{\frac{k+10}{2}} n^{-\frac{2-2 \delta(n)}{k+5}} \log^4 n \bigg),
\end{aligned}
\end{equation*}
where the last equality follows from omitting terms associated with $L_s$ and polynomial terms in $\log t_0$, $\log d$, and $\log k$. \hfill\Halmos
\end{proof}
\fi

\subsection{Supporting Lemmas and Proofs}

\begin{lemma}
\label{lem: lipschitz constant bound in Theorem 1}

Under the same assumptions as in Theorem \ref{theorem: score approximation}, it holds that
\begin{equation}
\label{equ: upper bound of tauS}
    \tau(S) = \order\bigg(L_s \operatorname{poly} (1+\sigma_{\max}^2) \operatorname{poly} (\sqrt{k} S) \bigg),
\end{equation}
where $\tau(S)$ is defined in \eqref{equ: definition of tauS} and $\operatorname{poly}(\cdot)$ represents a cubic polynomial.
\end{lemma}

\begin{proof}{Proof of Lemma~\ref{lem: lipschitz constant bound in Theorem 1}.}

Recall that $\tau(S)$ is associated with $\bm\xi(\mathbf z, t)$, which is defined in \eqref{equ: definition of xi}.
By direct calculation, we have
\begin{align}
    \frac{\partial \bm\xi}{\partial t} &= - \frac{1}{2} \frac{\int \mathbf f \frac{\partial \| \bm\Gamma_t^{-\frac{1}{2}} (\mathbf z - \alpha_t \mathbf f) \|_2^2}{\partial t} \phi(\mathbf z ; \alpha_t \mathbf f, \bm\Gamma_t) p_{\textrm{fac}}(\mathbf f) \, \dd \mathbf f}{\int \phi(\mathbf z ; \alpha_t \mathbf f, \bm\Gamma_t) p_{\textrm{fac}}(\mathbf f) \, \dd \mathbf f}
    + \frac{1}{2} \bm\xi \frac{\int \frac{\partial \| \bm\Gamma_t^{-\frac{1}{2}} (\mathbf z - \alpha_t \mathbf f) \|_2^2}{\partial t} \phi(\mathbf z ; \alpha_t \mathbf f, \bm\Gamma_t) p_{\textrm{fac}}(\mathbf f) \, \dd \mathbf f}{\int \phi(\mathbf z ; \alpha_t \mathbf f, \bm\Gamma_t) p_{\textrm{fac}}(\mathbf f) \, \dd \mathbf f} \nonumber \\
    &\stackrel{(i)}{=} \frac{\alpha_t^2}{2} \E[\mathbf F \mathbf F^\top \bm\beta^\top \bm\Lambda_t^{-2} \bm\beta \mathbf F | \mathbf Z = \mathbf z] + \frac{\alpha_t}{2} \operatorname{Cov}[\mathbf F | \mathbf Z = \mathbf z] \mathbf C_t \mathbf z + \frac{\alpha_t^2}{2} \E[\mathbf F| \mathbf Z = \mathbf z]\E[\mathbf F^\top \bm\beta^\top \bm\Lambda_t^{-2} \bm\beta \mathbf F | \mathbf Z = \mathbf z], \label{equ: partial xi partial t}
\end{align}
where $(i)$ follows from plugging in 
\begin{align*}
\frac{\partial \| \bm\Gamma_t^{-\frac{1}{2}} (\mathbf z - \alpha_t \mathbf f) \|_2^2}{\partial t} &= - \alpha_t^2 \mathbf f^\top \bm\beta^\top \bm\Lambda_t^{-2} \bm\beta \mathbf f + \alpha_t \mathbf f^\top \mathbf C_t \mathbf z + \mathbf z^\top \bm\beta^\top (\bm\Lambda_t^{-1} - \bm\Lambda_t^{-2}) \bm\beta \mathbf z
\end{align*}
with $\mathbf C_t = \bm\beta^\top (2\bm\Lambda_t^{-2} - \bm\Lambda_t^{-1}) \bm\beta$ and re-arranging terms. To bound $\|\partial \bm\xi/\partial t \|_2$, we provide the following two upper bounds.

\paragraph{Conditional Third Moment Bound.} By Cauchy-Schwarz inequality, we have
\begin{align}
    \left\| \E[\mathbf F \mathbf F^\top \bm\beta^\top\bm\Lambda_t^{-2}\bm\beta \mathbf F |\mathbf Z = \mathbf z] \right\|_2 &\leq \sqrt{\E[\| \mathbf F \|_2^2 |\mathbf Z = \mathbf z] \cdot \E [\| \mathbf F^\top \bm\beta^\top\bm\Lambda_t^{-2}\bm\beta \mathbf F \|_2^2 |\mathbf Z = \mathbf z]} \nonumber \\
    &\leq \frac{1}{h_t + \sigma_d^2 \alpha_t^2} \sqrt{\E[\| \mathbf F \|_2^2 |\mathbf Z = \mathbf z] \cdot \E [\| \mathbf F \|_2^4 |\mathbf Z = \mathbf z]}, \label{equ: conditional third moment bounds}
\end{align}
where the second inequality holds due to $\bm\beta^\top \bm\beta = \mathbf I_k$ and $\| \bm\Lambda_t^{-2} \|_{\rm{op}} \leq 1 / (h_t + \sigma_d^2 \alpha_t^2)$.

\paragraph{Conditional Covariance Bound.}

Recall $\mathbf s_{\textrm{sub}}$ defined in \eqref{equ: s_parallel}. Taking the derivative of $\mathbf s_{\textrm{sub}}$ with respect to $\mathbf z$, we have
\begin{align}
    \frac{\partial \mathbf{s}_{\textrm{sub}}\left(\mathbf z, t\right)}{\partial \mathbf z} & = -\bm\Lambda_t^{-1} \bm\beta + \alpha_t^2 \bm\Lambda_t^{-1} \bm\beta \frac{\int \mathbf f \mathbf f^\top \bm\Gamma_t^{-1} \phi(\mathbf z ; \alpha_t \mathbf f, \bm\Gamma_t) p_{\textrm{fac}}(\mathbf f) \dd \mathbf f}{\int \phi(\mathbf z ; \alpha_t \mathbf f, \bm\Gamma_t) p_{\textrm{fac}}(\mathbf f) \dd \mathbf f} \nonumber \\
    &\quad - \alpha_t^2 \bm\Lambda_t^{-1} \bm\beta \frac{\int \mathbf f \phi(\mathbf z ; \alpha_t \mathbf f, \bm\Gamma_t) p_{\textrm{fac}}(\mathbf f) \dd \mathbf f}{\int \phi(\mathbf z ; \alpha_t \mathbf f, \bm\Gamma_t) p_{\textrm{fac}}(\mathbf f) \dd \mathbf f} \frac{\int \mathbf f^\top \bm\Gamma_t^{-1} \phi(\mathbf z ; \alpha_t \mathbf f, \bm\Gamma_t) p_{\textrm{fac}}(\mathbf f) \dd \mathbf f}{\int \phi(\mathbf z ; \alpha_t \mathbf f, \bm\Gamma_t) p_{\textrm{fac}}(\mathbf f) \dd \mathbf f} \nonumber \\
    & = \alpha_t^2 \bm\Lambda_t^{-1} \bm\beta \left[\operatorname{Cov}\left(\mathbf F |\mathbf Z = \mathbf z \right)\bm\Gamma_t^{-1} - \frac{1}{\alpha_t^2} \mathbf I_k\right]. \label{equ: equation s_parallel and Cov(F|Z=z)}
\end{align}
Since $\mathbf s_{\textrm{sub}}$ is $L_s$-Lipschitz by Assumption \ref{assumption: Lipschitz}, we deduce from \eqref{equ: equation s_parallel and Cov(F|Z=z)} that for any $t \in (0, T]$, it holds
$$
\left\|\operatorname{Cov}\left(\mathbf F | \mathbf Z = \mathbf z \right)\right\|_{\textrm{op}} \leq \frac{(h_t+\sigma_{\max}^2\alpha_t^2)(1+L_s(h_t+\sigma_{\max}^2\alpha_t^2))}{\alpha_t^2} \leq (1+\sigma_{\max}^4) (1+L_s),
$$
where the second inequality follows from taking $t = 0$.

Furthermore, as 
$$
\|\mathbf C_t\|_{\rm{op}} = \|\bm\beta^\top (2\bm\Lambda_t^{-2} - \bm\Lambda_t^{-1}) \bm\beta\|_{\rm{op}} \leq \| 2\bm\Lambda_t^{-2} - \bm\Lambda_t^{-1} \|_{\rm{op}} \leq \frac{3}{(h_t + \sigma_d^2 \alpha_t^2)^2},
$$ it holds that
\begin{align}
\label{equ: conditional covariance bounds}
    \left\| \operatorname{Cov}[\mathbf F |\mathbf Z = \mathbf z] \mathbf C_t \mathbf z \right\|_{2} &\leq \frac{3}{(h_t + \sigma_d^2 \alpha_t^2)^2} \left\| \operatorname{Cov}[\mathbf F |\mathbf Z = \mathbf z] \right\|_{\rm{op}} \| \mathbf z \|_2 \\
    &\leq \frac{3(1+\sigma_{\max}^4) (1+L_s)}{(h_t + \sigma_d^2 \alpha_t^2)^2} \| \mathbf z \|_2.
\end{align}

By substituting the conditional third moment bound in \eqref{equ: conditional third moment bounds} and covariance bound in \eqref{equ: conditional covariance bounds} into \eqref{equ: partial xi partial t}, and using the fact that $\mathbf F, \mathbf Z$ have the sub-Gaussian tails in the compact domain $\calS$, we conclude that
\begin{align*}
    \tau(S) = \order\bigg(L_s (1+\sigma_{\max}^4) \operatorname{poly} (\sqrt{k} S) \bigg).
\end{align*}
where $\operatorname{poly}(\cdot)$ represents a cubic polynomial. \hfill\Halmos

\end{proof}

\begin{lemma}
\label{lem: truncation error -- heterogeneous}
Suppose Assumption \ref{assumption: subgaussian} holds. Let $\bm\xi$ be defined in \eqref{equ: definition of xi} and $\mathbf Z = \bm\beta^\top \bm\Lambda_t^{-1} \mathbf R$ with distribution $P_t^{\rm fac}$. Given $\epsilon > 0$, with $S = c\left(\sqrt{(1+\sigma_{\max}^2) (k + \log (1/\epsilon))} \right)$
for some constant $c > 0$, it holds that
\begin{equation}
\label{equ: truncation error -- heterogeneous}
\left\| \bm\xi\left(\mathbf Z, t\right) \indicator \{\|\mathbf Z\|_2 > S\} \right\|_{L^2 \left(P_{t}^{\rm fac}\right)} \leq \epsilon, \quad \forall t \in (0, T].
\end{equation}
\end{lemma}

\begin{proof}{Proof of Lemma~\ref{lem: truncation error -- heterogeneous}.}

Plugging in the expression of $\bm\xi$ in \eqref{equ: definition of xi}, we obtain that
\begin{align*}
&\quad \int \left\|\int \frac{\mathbf f \phi (\bm\Gamma_t \bm\beta^\top \bm\Lambda_t^{-1} \mathbf r ; \alpha_t \mathbf f, \bm\Gamma_t )  p_{\textrm{fac}}(\mathbf f) \dd \mathbf f}{\int \phi (\bm\Gamma_t \bm\beta^\top \bm\Lambda_t^{-1} \mathbf r ; \alpha_t \mathbf f, \bm\Gamma_t ) p_{\textrm{fac}}(\mathbf f) \dd \mathbf f}\right\|_2^2 \indicator \{\|\bm\beta^\top \bm\Lambda_t^{-1} \mathbf r\|_2 > S\} p_t(\mathbf r) \dd \mathbf r \\
& \stackrel{(i)}{\leq} \int_{\|\bm\beta^\top \bm\Lambda_t^{-1} \mathbf r\|_2 > S} \|\mathbf f\|_2^2 \frac{\phi (\bm\Gamma_t \bm\beta^\top \bm\Lambda_t^{-1} \mathbf r ; \alpha_t \mathbf f, \bm\Gamma_t )  p_{\textrm{fac}}(\mathbf f) \dd \mathbf f}{\int \phi (\bm\Gamma_t \bm\beta^\top \bm\Lambda_t^{-1} \mathbf r ; \alpha_t \mathbf f, \bm\Gamma_t )  p_{\textrm{fac}}(\mathbf f) \dd \mathbf f} p_t(\mathbf r) \dd \mathbf r \\
& \stackrel{(ii)}{\leq} \int \int_{\|\bm\beta^\top \bm\Lambda_t^{-1} \mathbf r\|_2 > S} \|\mathbf f\|_2^2 \phi (\mathbf T_t \bm\Lambda_t^{-\frac{1}{2}} \mathbf r ; \alpha_t \bm\Lambda_t^{-\frac{1}{2}} \bm\beta \mathbf f, \mathbf I )  \phi ((\mathbf I - \mathbf T_t)\bm\Lambda_t^{-\frac{1}{2}} \mathbf r ; \mathbf 0, \mathbf I )  p_{\textrm{fac}}(\mathbf f) \dd \mathbf r \dd \mathbf f \\
& = \underbrace{\int_{\|\bm\beta^\top \bm\Lambda_t^{-1} \mathbf r\|_2 > S} \int_{\|\bm\Lambda_t^{-\frac{1}{2}} \bm\beta \mathbf f\|_2 \leq \frac{1}{2}\|\mathbf T_t \bm\Lambda_t^{-\frac{1}{2}} \mathbf r\|_2} \|\mathbf f\|_2^2 \phi (\mathbf T_t \bm\Lambda_t^{-\frac{1}{2}} \mathbf r ; \alpha_t \bm\Lambda_t^{-\frac{1}{2}} \bm\beta \mathbf f, \mathbf I ) p_{\textrm{fac}}(\mathbf f) \dd \mathbf f \dd (\mathbf T_t \bm\Lambda_t^{-\frac{1}{2}} \mathbf r)}_{(A)} \\
&\quad + \underbrace{\int_{\|\bm\beta^\top \bm\Lambda_t^{-1} \mathbf r\|_2 > S} \int_{\|\bm\Lambda_t^{-\frac{1}{2}} \bm\beta \mathbf f\|_2 > \frac{1}{2}\|\mathbf T_t \bm\Lambda_t^{-\frac{1}{2}} \mathbf r\|_2} \|\mathbf f\|_2^2 \phi(\mathbf T_t \bm\Lambda_t^{-\frac{1}{2}} \mathbf r ; \alpha_t \bm\Lambda_t^{-\frac{1}{2}} \bm\beta \mathbf f, \mathbf I ) p_{\textrm{fac}}(\mathbf f) \dd \mathbf f \dd (\mathbf T_t \bm\Lambda_t^{-\frac{1}{2}} \mathbf r)}_{(B)},
\end{align*}
where $(i)$ holds due to the Cauchy-Schwarz inequality, $(ii)$ invokes the expression of $p_t(\mathbf r)$ in \eqref{equ: density of r}  and re-arranging terms, and the last equality holds by straightforward calculations.

\paragraph{Bounding Term $(A)$.} We define the change of variable $\mathbf X := \mathbf T_t \bm\Lambda_t^{-\frac{1}{2}} \mathbf R_t$ and denote by $\mathbf x$ a realization of $\mathbf X$.  By the Cauchy-Schwarz inequality and $\|\bm\Lambda_t^{-\frac{1}{2}} \bm\beta \mathbf f\|_2 \leq \frac{1}{2}\|\mathbf T_t \bm\Lambda_t^{-\frac{1}{2}} \mathbf r\|_2$, we have
$$
\|\mathbf x - \alpha_t \bm\Lambda_t^{-\frac{1}{2}} \bm\beta \mathbf f\|_2^2 \geq \frac{1}{2}\|\mathbf x\|_2^2 - \alpha^2_t \|\bm\Lambda_t^{-\frac{1}{2}} \bm\beta \mathbf f\|_2^2 \geq \frac{1}{4}\|\mathbf x\|_2^2.
$$
As a result, we can deduce that
\begin{align}
    (A) & \leq \int_{\|\bm\beta^\top \bm\Lambda_t^{-\frac{1}{2}} \mathbf x\|_2 > S} \int_{\|\bm\Lambda_t^{-\frac{1}{2}} \bm\beta \mathbf f\|_2 \leq \frac{1}{2}\|\mathbf x\|_2} \| \mathbf f\|_2^2 (2 \pi)^{-\frac{k}{2}} \exp \left(-\frac{\| \mathbf x\|_2^2}{8}\right) p_{\textrm{fac}}(\mathbf f) \dd \mathbf f \dd \mathbf x \nonumber \\
    & \leq \E\left[\|\mathbf f\|_2^2\right] \cdot \int_{\|\bm\beta^\top \bm\Lambda_t^{-\frac{1}{2}} \mathbf x\|_2 > S} (2 \pi)^{-\frac{k}{2}} \exp \left(-\frac{\| \mathbf x\|_2^2}{8}\right) \dd \mathbf x \nonumber \\
    & \stackrel{(i)}{\leq} \E\left[\|\mathbf f\|_2^2\right] \cdot \int_{\left\|\mathbf x\right\|_2 > (h_t + \sigma_d^2 \alpha_t^2)^{\frac{1}{2}} S} (2 \pi)^{-\frac{k}{2}} \exp \left(-\frac{\| \mathbf x\|_2^2}{8}\right) \dd \mathbf x \nonumber \\
    & \stackrel{(ii)}{\leq} \E\left[\|\mathbf f\|_2^2\right] \cdot \frac{2^{-\frac{k}{2} + 2} k S^{k-2} (h_t + \sigma_d^2 \alpha_t^2)^{(k-2)/2}}{(\frac{1}{2}-\eta) \Gamma(\frac{k}{2} + 1) } \exp \left(-\frac{(h_t + \sigma_d^2 \alpha_t^2) S^2}{8}\right). \label{equ: bound (A) in truncation error}
\end{align}
where $(i)$ holds due to $\| \bm\beta^\top \bm\Lambda_t^{-\frac{1}{2}} \|_{\textrm{{op}}} \leq (h_t + \sigma_d^2 \alpha_t^2)^{-\frac{1}{2}}$, and $(ii)$ follows from the sub-Gaussian tail in Proposition 2.6.6 of \citet{vershynin2018high}.

\paragraph{Bounding Term $(B)$.} We define the change of variable $\mathbf Y := \bm\beta^\top \bm\Lambda_t^{-1} \mathbf R_t$ and denote by $\mathbf y$ the realization of $\mathbf Y$. Given $S > \min\{B/2, 1\}$, applying the sub-Gaussian tail of $p_{\textrm{fac}}(\mathbf f)$ in \eqref{eqn:subgaussian}, we obtain that
\begin{align}
(B) & \leq \int_{\left\|\mathbf y\right\|_2 > S} \int_{\|\bm\Lambda_t^{-\frac{1}{2}} \bm\beta \mathbf f\|_2 > \frac{1}{2}\|\bm\Lambda_t^{-\frac{1}{2}} \bm\beta \bm\Gamma_t \mathbf y\|_2} \phi(\bm\Lambda_t^{-\frac{1}{2}} \bm\beta \bm\Gamma_t \mathbf y; \alpha_t \bm\Lambda_t^{-\frac{1}{2}} \bm\beta \mathbf f, \mathbf I) \cdot \frac{C_1}{(2 \pi)^{\frac{k}{2}}} \|\mathbf f\|_2^2 \exp \left(-\frac{C_{2}\|\mathbf f\|_2^2}{2}\right) \dd \mathbf f \dd \mathbf y \nonumber \\
& \stackrel{(i)}{\leq} \frac{C_1}{(2 \pi)^{k}} \int_{\left\|\mathbf y\right\|_2 > S} \int_{\|\bm\Lambda_t^{-\frac{1}{2}} \bm\beta \mathbf f\|_2 > \frac{1}{2}\|\bm\Lambda_t^{-\frac{1}{2}} \bm\beta \bm\Gamma_t \mathbf y\|_2} \exp \Bigg( - \frac{C_{2}\big\| (\alpha_t^2 \mathbf I_k + C_{2} \bm\Gamma_t)^{-\frac{1}{2}} \bm\Gamma_t \mathbf y \big\|_2^2}{2} \Bigg) \nonumber \\
&\qquad \cdot \|\mathbf f\|_2^2 \exp \Bigg( - \frac{\big\| (\alpha_t^2 \bm\Gamma_t^{-1} + C_{2} \mathbf I_k)^{\frac{1}{2}} \left( \mathbf f - \alpha_t (\alpha_t^2 \bm\Gamma_t^{-1} + C_{2} \mathbf I_k)^{-1} \mathbf y \right) \big\|_2^2}{2} \Bigg) \dd \mathbf f \dd \mathbf y \nonumber \\
& \stackrel{(ii)}{\leq} \frac{C_1}{(2 \pi)^{k}} \int_{\left\|\mathbf y\right\|_2 > S} \int_{\|\bm\Lambda_t^{-\frac{1}{2}} \bm\beta \mathbf f\|_2 > \frac{1}{2}\|\bm\Lambda_t^{-\frac{1}{2}} \bm\beta \bm\Gamma_t \mathbf y\|_2} \exp \Bigg( - \frac{C_{2}\big\| (\alpha_t^2 \mathbf I_k + C_{2} \bm\Gamma_t)^{-\frac{1}{2}} \bm\Gamma_t \mathbf y \big\|_2^2}{2} \Bigg) \nonumber \\
&\qquad \cdot \|\mathbf f\|_2^2 \exp \left( - \frac{C_{2}\left\| \mathbf f - \alpha_t (\alpha_t^2 \bm\Gamma_t^{-1} + C_{2} \mathbf I_k)^{-1} \mathbf y \right\|_2^2}{2} \right) \dd \mathbf f \dd \mathbf y, \label{equ: temp_bound (B) in truncation error}
\end{align}
where $(i)$ invokes the formula of $\phi(\mathbf y; \alpha_t \mathbf f, \bm\Gamma_t)$ in \eqref{equ: score decomposition -- heterogeneous} and completing the square for $\mathbf f$, $(ii)$ follows from $\|\alpha_t^2 \bm\Gamma_t^{-1} + C_{2} \mathbf I_k \|_{\textrm{op}} \geq C_2$.

Furthermore, applying $\E[\|\mathbf f\|_2^2] \leq \alpha_t^2 \| (\alpha_t^2 \bm\Gamma_t^{-1} + C_{2} \mathbf I_k)^{-\frac{1}{2}} \mathbf y \|_2^2 + k$ to \eqref{equ: temp_bound (B) in truncation error}, we deduce that
\begin{align}
(B) & \leq \frac{C_1}{C_2^{\frac{k}{2}} (2 \pi)^{k} } \int_{\left\|\mathbf y\right\|_2 > S}\left[ \alpha_t^2 \| (\alpha_t^2 \bm\Gamma_t^{-1} + C_{2} \mathbf I_k)^{-1} \mathbf y \|_2^2 + k \right] \cdot \exp \Bigg( - \frac{C_{2}\big\| (\alpha_t^2 \mathbf I_k + C_{2} \bm\Gamma_t)^{-\frac{1}{2}} \bm\Gamma_t \mathbf y \big\|_2^2}{2} \Bigg) \dd \mathbf y \nonumber \\
& \leq \frac{C_1 2^{-\frac{k}{2} + 2} k S^k}{C_2 \Gamma(\frac{k}{2} + 1) (\alpha_t^2 + C_2(h_t + \sigma_{\max}^2 \alpha_t^2))} \exp \left(-\frac{ (\alpha_t^2 + C_2(h_t + \sigma_{\max}^2 \alpha_t^2)) C_2 S^2}{2} \right), \label{equ: bound (B) in truncation error}
\end{align}
where the last inequality is due to $\| \bm\Gamma_t \|_{\textrm{op}} \geq h_t + \sigma_d^2 \alpha_t^2$, $\|(\alpha_t^2 \mathbf I_k + C_2 \bm\Gamma_t)^{-\frac{1}{2}} \|_{\textrm{op}} \geq \sqrt{\alpha_t^2 + C_2(h_t + \sigma_{\max}^2 \alpha_t^2)}$ and the sub-Gaussian tail in Proposition 2.6.6 of \citet{vershynin2018high} and similar operator norm bounds in $(ii)$.

Combining the error bounds \eqref{equ: bound (A) in truncation error} and \eqref{equ: bound (B) in truncation error} for $(A)$ and $(B)$, we conclude that
\begin{align}
\left\| \bm\xi\left(\mathbf Z, t\right) \indicator \{\|\mathbf Z\|_2 > S\} \right\|_{L^2 \left(P_{t}^{\rm fac}\right)} \leq & c^{\prime} \frac{2^{-\frac{k}{2} + 3} k (h_t + \sigma_d^2 \alpha_t^2)^{\frac{k}{2}} S^k}{\Gamma(\frac{k}{2} + 1) (\alpha_t^2 + C_2(h_t + \sigma_{\max}^2 \alpha_t^2))} \exp \left(-\frac{(h_t + \sigma_d^2 \alpha_t^2) S^2}{8} \right) \label{equ: final truncation error}
\end{align}
for some constant $c^{\prime} > 0$. Given any $\epsilon > 0$, by the upper bound of truncation error in \eqref{equ: final truncation error}, we can choose
$$
S = c\left(\sqrt{\bigg(1+\sigma_{\max}^2\bigg) \bigg(k + \log \frac{1}{\epsilon}\bigg)} \right),
$$
such that $\left\| \bm\xi\left(\mathbf Z, t\right) \indicator \{\|\mathbf Z\|_2 > S\} \right\|_{L^2 \left(P_{t}^{\rm fac}\right)} \leq \epsilon$. Here, $c$ is an absolute constant. \hfill\Halmos
\end{proof}

\begin{lemma}
\label{lem: SNN truncation error -- heterogeneous}
Suppose Assumption \ref{assumption: subgaussian} holds. For any $\mathbf s_{\bm\theta_1}\left(\cdot, t\right)$ and $\mathbf s_{\bm\theta_2}\left(\cdot, t\right)$, when $\rho$ is sufficiently large, it holds that
\begin{align}
    &\quad \sup_{\|\mathbf r\|_{2} \leq \rho} \E_{\mathbf R_t | \mathbf R_0 = \mathbf r} \bigg[ \left(K + \left\| \mathbf R_t \right\|_2 + \left\| \mathbf r \right\|_2 \right) \left\|\mathbf s_{\bm\theta_1}\left(\mathbf R_t, t\right)-\mathbf s_{\bm\theta_2}\left(\mathbf R_t, t\right)\right\|_{2} \cdot \indicator\left\{\left\|\mathbf R_t\right\|_{2} > 3\rho + \sqrt{d \log d}\right\}\bigg] \nonumber \\
    & = \order \bigg(\rho K^2 h_t^{-2-\frac{d}{2}} \left(\frac{\rho}{d}\right)^{d} \exp \left(-\frac{\rho^2}{h_t} \right) \bigg).
    \label{equ: SNN truncation error}
\end{align}
\end{lemma}

\begin{proof}{Proof of Lemma~\ref{lem: SNN truncation error -- heterogeneous}.}
Denote $\mathbf D_{ti} = \operatorname{diag}\{ 1/(h_t + c_{i1} \alpha_t^2), \dots, 1/(h_t + c_{i1} \alpha_t^2) \}$ for $i =1, 2$. Applying the formula of $\mathbf s_{\bm\theta_1}$ and $\mathbf s_{\bm\theta_2}$ in \eqref{equ: score network}, we calculate
\begin{align}
&\quad \E_{\mathbf R_t | \mathbf R_0 = \mathbf r} \bigg[ \left(K + \left\| \mathbf R_t \right\|_2 + \left\| \mathbf r \right\|_2 \right) \left\|\mathbf s_{\bm\theta_1}\left(\mathbf R_t, t\right)-\mathbf s_{\bm\theta_2}\left(\mathbf R_t, t\right)\right\|_{2} \cdot \indicator\left\{\left\|\mathbf R_t\right\|_{2} > 3\rho + \sqrt{d \log d}\right\}\bigg] \nonumber \\
&\stackrel{(i)}{\leq} \displaystyle\int \bigg(K + \left\| \mathbf r^{\prime} \right\|_2 + \left\| \mathbf r \right\|_2 \bigg) \bigg(\| (\mathbf D_{t1} - \mathbf D_{t2}) \mathbf r^{\prime} \|_2 + \|\alpha_t(\mathbf D_{t1} \mathbf V_1 - \mathbf D_{t2} \mathbf V_2) \mathbf g_{\bm\zeta_1}(\mathbf V_1^\top \mathbf D_{t1} \mathbf r^{\prime}, t) \|_2 \nonumber \\
&\quad + \|\alpha_t\mathbf D_{t2} \mathbf V_2 (\mathbf g_{\bm\zeta_1}(\mathbf V_1^\top \mathbf D_{t1} \mathbf r^{\prime}, t) - \mathbf g_{\bm\zeta_2}(\mathbf V_2^\top \mathbf D_{t2} \mathbf r^{\prime}, t))\|_2 \bigg) \cdot \indicator\left\{\left\|\mathbf r^{\prime}\right\|_{2}>3\rho + \sqrt{d \log d}\right\} \phi(\mathbf r^{\prime} ; \alpha_t \mathbf r, h_t \mathbf I) \dd \mathbf r^{\prime} \nonumber \\
& \stackrel{(ii)}{=} \order \bigg( \displaystyle\int_{\| \mathbf r^{\prime} \|_2 > 3\rho + \sqrt{d \log d}} \frac{(K + \left\| \mathbf r^{\prime} \right\|_2 + \left\| \mathbf r \right\|_2)(K + \left\| \mathbf r^{\prime} \right\|_2)}{h_t^2 (2 \pi h_t)^{\frac{d}{2}}} \exp \left(-\frac{1}{2 h_t}\left(\frac{1}{2}\left\|\mathbf r^{\prime}\right\|_2^2-\|\mathbf r\|_2^2\right)\right) \dd \mathbf r^{\prime} \bigg), \label{equ: upper bound of score tail estimation in Theorem 2}
\end{align}
where $(i)$ is due to Cauchy-Schwarz inequality; $(ii)$ follows from the upper bounds \eqref{equ: ReLU network} and \eqref{equ: score network} of $\{\mathbf g_{\bm\theta_i}, \mathbf V_{i}, \mathbf D_{ti}\}_{i=1,2}$, $\alpha_t^2 \leq 1$ and
$$
\phi(\mathbf r^{\prime} ; \alpha_t \mathbf r, h_t \mathbf I) \leq (2 \pi h_t)^{-\frac{d}{2}} \exp \left(-\frac{1}{2 h_t}\left(\frac{1}{2}\left\|\mathbf r^{\prime}\right\|_2^2-\|\mathbf r\|_2^2\right)\right).
$$
Then, substituting the upper bound for the tail estimation in Proposition 2.6.6 of \citet{vershynin2018high} into \eqref{equ: upper bound of score tail estimation in Theorem 2}, we deduce that
\begin{align}
\eqref{equ: upper bound of score tail estimation in Theorem 2} & = \order \bigg( \frac{(K^2 + K \|\mathbf r\|_2) (2 h_t)^{-2-\frac{d}{2}} (3\rho + \sqrt{d \log d})^{d-2}}{\Gamma(\frac{d}{2} + 1)} \exp \left( - \frac{(3\rho + \sqrt{d \log d})^2}{4 h_t} + \frac{\| \mathbf r \|_2^2}{2 h_t} \right) \bigg) \nonumber \\
&\qquad + \frac{(2K + \|\mathbf r\|_2)(2 h_t)^{-2-\frac{d}{2}} (3\rho + \sqrt{d \log d})^{d-1}}{\Gamma(\frac{d}{2} + 1)} \exp \left( - \frac{(3\rho + \sqrt{d \log d})^2}{4 h_t} + \frac{\| \mathbf r \|_2^2}{2 h_t} \right) \nonumber \\
&\qquad \qquad + \frac{(2 h_t)^{-2-\frac{d}{2}} (3\rho + \sqrt{d \log d})^{d}}{\Gamma(\frac{d}{2} + 1)} \exp \left( - \frac{(3\rho + \sqrt{d \log d})^2}{4 h_t} + \frac{\| \mathbf r \|_2^2}{2 h_t} \right) \nonumber \\
& = \order \left( K^2 \|\mathbf r\|_2 h_t^{-2-\frac{d}{2}} \left(\frac{\rho}{d}\right)^{d} \exp \left(-\frac{9 \rho^2 - 2\| \mathbf r \|_2^2}{4 h_t} \right) \right). \label{equ: upper bound 1 of score tail estimation in Theorem 2}
\end{align}
Here, the last inequality holds since $\Gamma(\frac{d}{2} + 1) = \order (\prod_{j=1}^{\frac{d}{2}} j)$ and
$$
\frac{2^{-\frac{d}{2}}(3 \rho + \sqrt{d \log d})^{d}\exp\left(-\frac{6\rho\sqrt{d \log d} + d \log d}{4 h_t}\right)}{\Gamma(\frac{d}{2} + 1)} = \order \bigg( \left(\frac{\rho}{d}\right)^{d} \bigg)
$$
for a sufficiently large $\rho > \max\{ B, d \}$.

Immediately, substituting $\|\mathbf r\|_2 \leq \rho$ into \eqref{equ: upper bound 1 of score tail estimation in Theorem 2}, we obtain the desired result.
\hfill\Halmos
\end{proof}

\begin{lemma}
\label{lem: SNN covering number}
For any given $\epsilon > 0$, $\delta > 0$, and $\rho = \order \left( \sqrt{\sigma_{\max}^2 \left(d + \log K + \log (n/\delta) \right) }\right)$ defined in \eqref{equ: choice of rho and tau}, the $\nu$-covering number of $\calS_{\rm NN}$ in \eqref{equ: score network} is
\begin{equation}
\begin{aligned}
\label{equ: covering number of SNN}
    \log \mathfrak{N}(\nu, \calS_{\textrm{NN}},\|\cdot\|_2) &= \bigg( \bigg( dk + T \tau(1+L_s)^k (1+\sigma_{\max}^k) \epsilon^{-(k+1)} \bigg(\log \frac{1}{\epsilon}+ k \bigg)^{\frac{k+4}{2}} \bigg) \log \frac{T \tau d k}{t_0 \nu \epsilon} \bigg).
\end{aligned}
\end{equation}
\end{lemma}

\begin{proof}{Proof of Lemma~\ref{lem: SNN covering number}.}
$\calS_{\textrm{NN}}$ consists of three components:
\begin{enumerate}
    \item A vector $\mathbf c = (c_1, c_2, \dots, c_d) \in [0, \sigma_{\max}]^d$ and its induced matrix 
    $$\mathbf D_t = \operatorname{diag}\{ 1/(h_t + \alpha_t^2 c_1), 1/(h_t + \alpha_t^2 c_2), \dots, 1/(h_t + \alpha_t^2 c_d) \}.$$
    \item A matrix $\mathbf V$ with orthonormal columns.
    \item A ReLU network $\mathbf g_{\bm\zeta}$.
\end{enumerate}
Denote $\mathbf D_{ti} = \operatorname{diag}\{ 1/(h_t + \alpha_t^2 c_{i1}), 1/(h_t + \alpha_t^2 c_{i2}), \dots, 1/(h_t + \alpha_t^2 c_{id}) \}$ for $i=1,2$. Directly incorporating the sub-additive property of $L^2$ norm and $\alpha_t^2 \leq 1$, we have
\begin{align}
& \quad \|\mathbf s_{\bm\theta_1}(\mathbf r, t) - \mathbf s_{\bm\theta_2}(\mathbf r, t)\|_2 \nonumber \\
& \leq \|(\mathbf D_{t1} \mathbf V_1 - \mathbf D_{t2} \mathbf V_1) \mathbf g_{\bm\zeta_1}(\mathbf V_1^\top \mathbf D_{t1} \mathbf r, t) \|_2 + \|(\mathbf D_{t2} \mathbf V_1 - \mathbf D_{t2} \mathbf V_2) \mathbf g_{\bm\zeta_1}(\mathbf V_1^\top \mathbf D_{t1} \mathbf r, t) \|_2 \nonumber \\
&\quad + \|\mathbf D_{t2} \mathbf V_2 (\mathbf g_{\bm\zeta_1}(\mathbf V_1^\top \mathbf D_{t1} \mathbf r, t) - \mathbf g_{\bm\zeta_2}(\mathbf V_1^\top \mathbf D_{t1} \mathbf r, t))\|_2 + \|\mathbf D_{t2} \mathbf V_2 (\mathbf g_{\bm\zeta_2}(\mathbf V_1^\top \mathbf D_{t1} \mathbf r, t) - \mathbf g_{\bm\zeta_2}(\mathbf V_2^\top \mathbf D_{t1} \mathbf r, t))\|_2 \nonumber \\
&\quad \quad + \|\mathbf D_{t2} \mathbf V_2 (\mathbf g_{\bm\zeta_2}(\mathbf V_2^\top \mathbf D_{t1} \mathbf r, t) - \mathbf g_{\bm\zeta_2}(\mathbf V_2^\top \mathbf D_{t2} \mathbf r, t))\|_2 + \| (\mathbf D_{t1} - \mathbf D_{t2}) \mathbf r \|_2 \nonumber \\
& \leq \|\mathbf D_{t1} - \mathbf D_{t2} \|_{\textrm{op}} \| \mathbf g_{\bm\zeta_1}(\mathbf V_1^\top \mathbf D_{t1} \mathbf r, t) \|_2 + \|\mathbf D_{t2}\|_{\textrm{op}} \| \mathbf V_1 - \mathbf V_2\|_{\textrm{op}} \| \mathbf g_{\bm\zeta_1}(\mathbf V_1^\top \mathbf D_{t1} \mathbf r, t) \|_2 \nonumber \\
&\quad + \|\mathbf D_{t2} \|_{\textrm{op}} \| \mathbf g_{\bm\zeta_1}(\mathbf V_1^\top \mathbf D_{t1} \mathbf r, t) - \mathbf g_{\bm\zeta_2}(\mathbf V_1^\top \mathbf D_{t1} \mathbf r, t) \|_2 + \gamma\|\mathbf D_{t2} \|_{\textrm{op}} \| \mathbf V_1 - \mathbf V_2^\top \|_{\textrm{op}} \|\mathbf D_{t1} \|_{\textrm{op}} \| \| \mathbf r \|_2 \nonumber \\
&\quad \quad + \gamma\|\mathbf D_{t2} \|_{\textrm{op}} \| \mathbf D_{t1} - \mathbf D_{t2} \|_{\textrm{op}} \| \mathbf r \|_2 + \| \mathbf D_{t1} - \mathbf D_{t2} \|_{\textrm{op}} \| \mathbf r \|_2, \label{equ: diff bound of score}
\end{align}
where the last inequality follows from the fact that $\{\mathbf V_i\}_{i=1,2}$ are orthogonal and $\{\mathbf g_{\bm\zeta_i}\}_{i=1,2}$ is $\gamma$-Lipschitz.

To analyze the covering number of $\calS_{\textrm{NN}}$, we consider 
\begin{equation}
\label{equ: diff bound of c, V, f}
\| \mathbf c_1 - \mathbf c_2 \|_{\infty} \leq \delta_c, \|\mathbf V_1 - \mathbf V_2\|_{\textrm{op}} \leq \delta_V, \textrm{ and }
\sup_{\| \mathbf r \|_2 \leq 3\rho+ \sqrt{d \log d}, t \in [t_0, T]} \| \mathbf g_{\bm\zeta_1}(\mathbf r, t) - \mathbf g_{\bm\zeta_2}(\mathbf r, t)\|_2 \leq \delta_f.
\end{equation}
Immediately, we can deduce that
\begin{equation}
\label{equ: diff bound of D}
\sup_{t \in [t_0, T]} \| \mathbf D_{t1} - \mathbf D_{t2} \|_{\textrm{op}} \leq \frac{\delta_c}{t_0^2}.
\end{equation}

Then, on the domain with radius $\| \mathbf r \|_2 \leq 3\rho + \sqrt{d \log d}$ and $t \in [t_0, T]$, by substituting the upper bounds \eqref{equ: diff bound of c, V, f} and \eqref{equ: diff bound of D} into \eqref{equ: diff bound of score}, we obtain
\begin{align*}
& \quad \sup_{\|\mathbf r\|_2 \leq 3\rho + \sqrt{d \log d}, t \in [t_0, T]} \|\mathbf s_{\bm\theta_1}(\mathbf r, t) - \mathbf s_{\bm\theta_2}(\mathbf r, t)\|_2 \\
&\leq \frac{\delta_c K}{t_0^2} + \frac{\delta_V K}{t_0} + \frac{\delta_f}{t_0} + \frac{\gamma \delta_V (3\rho + \sqrt{d \log d})}{t_0^2} + \frac{\gamma \delta_c (3\rho + \sqrt{d \log d})}{t_0^3} + \frac{\delta_c (3\rho + \sqrt{d \log d})}{t_0^2} \\
&\stackrel{(i)}{=} \frac{\delta_c (\gamma (3\rho + \sqrt{d \log d}) + t_0 K + t_0 (3\rho + \sqrt{d \log d}))}{t_0^3} + \frac{\delta_V (\gamma (3\rho + \sqrt{d \log d}) + t_0 K)}{t_0^2} + \frac{\delta_f}{t_0} \\
&\stackrel{(ii)}{=} \order\left( \frac{\delta_c \gamma(3\rho + \sqrt{d \log d}) + t_0 \delta_V \gamma (3\rho + \sqrt{d \log d}) + t_0^2 \delta_f}{t_0^3} \right),
\end{align*}
where $(i)$ follows from rearranging terms, and $(ii)$ holds by omitting higher-order terms on $\delta_c$, $\delta_V$, and $\delta_f$. For a hypercube $[0, \sigma_{\max}]^d$, the $\delta_c$-covering number is bounded by $\left(\frac{\sigma_{\max}}{\delta_c}\right)^{d}$. For a set of matrices $\{\mathbf V \in \mathbb R^{d \times k}: \|\mathbf V\|_{\textrm{op}} \leq 1\}$, its $\delta_V$-covering number is bounded by $\left(1 + \frac{2\sqrt{k}}{\delta_V}\right)^{dk}$ (a standard volume-ratio bound for matrices with a bounded operator norm; see Lemma~5.7 and Example~5.8 of \citet{wainwright2019high} and Lemma~8 in \citet{chen2019generalization}). Following Lemma~5.3 in \citet{chen2022nonparametric}, which provides $\delta_f$-covering number bounds for ReLU network classes (see also Chapters~14 and 16 of \citet{anthony2009neural}), we take the upper bound $\left(\frac{2L^2 M (3\rho + \sqrt{d \log d}) \kappa^L M^{L+1}}{\delta_f}\right)^J$ for the $\delta_f$-covering number of the function class~\eqref{equ: ReLU network}. Therefore, with $\rho = \order \big( \sqrt{\sigma_{\max}^2 (d + \log K + \log (n/\delta)) }\big)$, we have
\begin{align*}
    \log \mathfrak{N}(\nu, \calS_{\textrm{NN}},\|\cdot\|_2) & \leq \order\left(d\log \left( \frac{\sigma_{\max} \gamma(3\rho+\sqrt{d \log d})}{t_0^3 \nu}\right) + dk\log\left(1+\frac{2\sqrt{k}\gamma(3\rho+\sqrt{d \log d})}{t_0^2 \nu}\right) \right. \\
    &\quad \quad \left. + J\log \left( \frac{2 L^2 M(3\rho+\sqrt{d \log d})\kappa^LM^{L+1}}{t_0 \nu} \right)\right) \\
    &\stackrel{(i)}{=} \order \bigg( \bigg( dk + T \tau(1+L_s)^k (1+\sigma_{\max}^k) \epsilon^{-(k+1)} \bigg(\log \frac{1}{\epsilon}+ k \bigg)^{\frac{k+4}{2}} \bigg) \log \frac{T \tau d k}{t_0 \nu \epsilon} \bigg),
\end{align*}
where $(i)$ follows from invoking the order of network parameters in \eqref{equ: order of parameters in Theorem 1} in Theorem \ref{theorem: score approximation} and omitting higher orders terms such as $\log \log d$ and $\log \log k$. \hfill\Halmos

\end{proof}

\section{Omitted Proofs in Section \ref{sec: distribution estimation}}
\label{sec: proof of theorem -- distribution estimation}
In this section, we provide the proof of Theorem \ref{theorem: distribution estimation} and the lemmas used in the proof.

\subsection{Proof of Theorem~\ref{theorem: distribution estimation}}
\label{subsec: proof of theorem -- distribution estimation}
\iftrue
\begin{proof}{Proof.}
The proof contains two parts: the distribution estimation and the latent subspace recovery. For notational simplicity, let us denote
\begin{equation}
\label{equ: definition of simplified notation epsilon}
\epsilon := \frac{1}{t_0}(1+\sigma_{\max}^{2k}) d^{\frac{5}{2}} k^{\frac{k+10}{2}} n^{-\frac{2-2 \delta(n)}{k+5}} \log^4 n.
\end{equation}

\noindent \underline{\bf Part 1: Return Distribution Estimation.} 
First, we can decompose $\operatorname{TV}(P_{\textrm{data}}, \widehat{P}_{t_0})$ into
\begin{equation*}
\operatorname{TV}(P_{\textrm{data}}, \widehat{P}_{t_0}) \leq \operatorname{TV}(P_{\textrm{data}}, P_{t_0}) + \operatorname{TV}(P_{t_0}, \widetilde{P}_{t_0}) + \operatorname{TV}(\widetilde{P}_{t_0}, \widehat{P}_{t_0}),
\end{equation*}
where $P_{\textrm{data}}$ is the initial distribution of $\mathbf R$ in \eqref{equ: factor model}, $\widehat{P}_{t_0}$ and $\widetilde{P}_{t_0}$ are the marginal distribution of the estimated backward process $\widehat{\mathbf R}_{T-t_0}^{\leftarrow}$ in \eqref{equ: learned backward SDE} initialized with $\widehat{\mathbf R}_{0}^{\leftarrow} \sim \calN(\mathbf 0, \mathbf I_d)$ and $\widehat{\mathbf R}_{0}^{\leftarrow} \sim P_T$, respectively. Here, $\operatorname{TV}(P_{\textrm{data}}, P_{t_0})$ is the early-stopping error, $\operatorname{TV}(P_{t_0}, \widetilde{P}_{t_0})$ captures the approximation error of the score estimation, and $\operatorname{TV}(\widetilde{P}_{t_0}, \widehat{P}_{t_0})$ reflects the mixing error. We bound each term in Lemma \ref{lem: total variation distance error} and the error bound~\eqref{equ: upper bound of TV distance} is given by
\begin{equation*}
\begin{aligned}
    \operatorname{TV}\left(P_{\textrm{data}}, \widehat{P}_{t_0}\right) &= \order \bigg( d t_0 L_s (1 + \sigma_{\max}^2) + \sqrt{\epsilon (T-t_0)}+\sqrt{\operatorname{KL}\left(P_{\textrm{data}} \| \mathcal{N}\left(\mathbf 0, \mathbf I_d\right)\right)} \exp (-T) \bigg) \\
    &= \Tilde{\order} \left((1+\sigma_{\max}^k) d^{\frac{5}{4}} k^{\frac{k+10}{4}} n^{-\frac{1-\delta(n)}{2(k+5)}} \log^{\frac{5}{2}} n \right)
\end{aligned}
\end{equation*}
where the last equality follows from invoking the order of $\epsilon$ in \eqref{equ: definition of simplified notation epsilon}, $t_0 = n^{-\frac{1-\delta(n)}{k+5}}$, and $T = \order(\log n)$ and omitting the lower-order terms in $d t_0 $. Hence the distribution estimation result in Theorem~\ref{theorem: distribution estimation} is completed.

\noindent \underline{\bf Part 2: Latent Subspace Recovery.} First, we generate $m = \order\left( \lambda_{\max}^{-2}(\bm\Sigma_{0}) d n^{\frac{2(1-\delta(n))}{k+5}}\log n \right)$ new samples via Algorithm~\ref{algo: sampling}. By the error bound~\eqref{equ: upper bound of opnerator norm error} in Lemma \ref{lem: eigenvalue estimation error}, we obtain that, with probability $1 - 1/n$, it holds
\begin{align*}
    \left\| \widehat{\bm\Sigma}_{0} - \bm\Sigma_0 \right\|_{\textrm{op}} =
    \Tilde{\order}\Bigg(\lambda_{\max}(\bm\Sigma_0) (1+\sigma_{\max}^k) d^{\frac{5}{4}} k^{\frac{k+10}{4}} n^{-\frac{1-\delta(n)}{k+5}} \log^{\frac{5}{2}} n\Bigg). 
\end{align*}
Therefore, applying Weyl's theorem to $\widehat{\bm\Sigma}_{0}$ and $\bm\Sigma_0$, we deduce that for any $i = 1, 2, \dots, d$, it holds that
\begin{align*}
    \left| \lambda_i (\widehat{\bm\Sigma}_{0}) - \lambda_i ( \bm\Sigma_0 )\right| = \Tilde{\order}\Bigg(\lambda_{\max}(\bm\Sigma_0) (1+\sigma_{\max}^k) d^{\frac{5}{4}} k^{\frac{k+10}{4}} n^{-\frac{1-\delta(n)}{k+5}} \log^{\frac{5}{2}} n\Bigg).
\end{align*}
In addition, for any $i = 1, 2, \dots, k$, it holds that
\begin{align*}
    \left| \frac{\lambda_i (\widehat{\bm\Sigma}_{0})}{\lambda_i ( \bm\Sigma_0 )} - 1\right| = \Tilde{\order}\Bigg(\frac{\lambda_{\max}(\bm\Sigma_0) (1+\sigma_{\max}^k) d^{\frac{5}{4}} k^{\frac{k+10}{4}} }{\lambda_i(\bm\Sigma_0)} n^{-\frac{1-\delta(n)}{k+5}} \log^{\frac{5}{2}} n\Bigg).
\end{align*}
Next, we analyze the SVD of $\widehat{\bm\Sigma}_{0}$. Recall that the top $k$-dimensional eigenspace of $\bm\Sigma_0$ and $\widehat{\bm\Sigma}_{0}$ are denoted as $\mathbf U$ and $\widehat{\mathbf U}$, respectively. For any $j = 1, 2, \dots, k$, define
\begin{equation*}
\begin{aligned}
\label{equ: definition of cosine between matrices}
&\cos \angle_j(\widehat{\mathbf U}, \mathbf U) := \max_{\widehat{\mathbf u} \in \operatorname{Col}(\widehat{\mathbf U}), \mathbf u \in \operatorname{Col}(\mathbf U)} \frac{| \widehat{\mathbf u}^\top \mathbf u |}{\|\widehat{\mathbf u}\|_2 \|\mathbf u\|_2} := |\widehat{\mathbf u}_j^\top \mathbf u_j|, \\
&\textrm{ subject to } \widehat{\mathbf u}_j^\top \widehat{\mathbf u}_{\ell} = 0 \textrm{ and } \mathbf u_j^\top \mathbf u_{\ell} = 0, \textrm{ for any } \ell = 1, 2, \dots, j-1,
\end{aligned}
\end{equation*}
where $\operatorname{Col}(\cdot)$ represents the column space, $\widehat{\mathbf u}_0 := \mathbf 0$, and $\mathbf u_0 := \mathbf 0$. Applying Davis-Kahan-sin$(\theta)$ Theorem of \citet{davis1970rotation} to $\widehat{\mathbf U}$ and $\mathbf U$, we have
\begin{equation}
\label{equ: upper bound of sine error}
    \| \sin \angle (\widehat{\mathbf U}, \mathbf U) \|_{\rm{F}} \leq \frac{\| \widehat{\bm\Sigma}_{0} - \bm\Sigma_0 \|_{\rm{F}}}{\lambda_{k}(\bm\Sigma_0) - \lambda_{k+1}(\bm\Sigma_0)} = \order \bigg( \frac{\lambda_{\max}(\bm\Sigma_0) (1+\sigma_{\max}^k) d^{\frac{5}{4}} k^{\frac{k+10}{4}}}{{\textrm{\tt Eigen-gap}(k)}} \cdot n^{-\frac{1-\delta(n)}{k+5}} \log^{\frac{5}{2}} n \bigg),
\end{equation}
where
\begin{equation}
\label{equ: definition of sine between matrices}
\sin \angle (\widehat{\mathbf U}, \mathbf U ) := \Big(1 - \cos^2 \angle_1(\widehat{\mathbf U}, \mathbf U), 1 - \cos^2 \angle_2 (\widehat{\mathbf U}, \mathbf U), \dots, 1 - \cos^2 \angle_k (\widehat{\mathbf U}, \mathbf U) \Big).
\end{equation}
The inequality in \eqref{equ: upper bound of sine error} follows from the fact that $\|\mathbf A\|_{\textrm{F}} \leq \sqrt{k} \|\mathbf A\|_{\textrm{op}}$ for any $\mathbf A \in \mathbb{R}^{d \times k}$. By the property of SVD, we can find two orthogonal matrices $\mathbf O_1, \mathbf O_2 \in \bb R^{k \times k}$ such that 
$$
\widehat{\mathbf U}^\top \mathbf U = \mathbf O_1^\top \operatorname{diag}\left\{\cos \angle_1 (\widehat{\mathbf U}, \mathbf U), \cos \angle_2 (\widehat{\mathbf U}, \mathbf U), \dots, \cos \angle_k (\widehat{\mathbf U}, \mathbf U)\right\} \mathbf O_2.
$$
Immediately, it holds that
\begin{equation}
\label{equ: QR decomposition of cosine}
\widehat{\mathbf U}^\top \mathbf U \mathbf U^\top \widehat{\mathbf U} = \mathbf O_1^\top \operatorname{diag}\left\{\cos^2(\angle_1), \dots, \cos^2(\angle_k)\right\} \mathbf O_1.
\end{equation}
Therefore, we deduce that
\begin{align}
    \|\widehat{\mathbf U} \widehat{\mathbf U}^\top - \mathbf U \mathbf U^\top\|_{\rm{F}}^2 &= \operatorname{tr}(\widehat{\mathbf U} \widehat{\mathbf U}^\top + \mathbf U \mathbf U^\top - 2\widehat{\mathbf U} \widehat{\mathbf U}^\top \mathbf U \mathbf U^\top) \nonumber \\
    &\stackrel{(i)}{=} \operatorname{tr}(\widehat{\mathbf U}^\top \widehat{\mathbf U}) + \operatorname{tr}(\mathbf U^\top \mathbf U) - 2 \operatorname{tr}(\mathbf O_1^\top \operatorname{diag}\left\{\cos^2(\angle_1), \dots, \cos^2(\angle_k)\right\} \mathbf O_1) \nonumber \\
    &\stackrel{(ii)}{=} 2k - \sum_{i=1}^k \cos^2(\angle_i) = 2\| \sin \angle (\widehat{\mathbf U}, \mathbf U) \|_{\rm{F}}^2,
    \label{equ: relationship between subspace and sin}
\end{align}
where $(i)$ invokes \eqref{equ: QR decomposition of cosine} and $(ii)$ holds due to the fact that $\widehat{\mathbf U}$, $\mathbf U$, and $\mathbf O_1$ have orthogonal columns, and the last equality holds by the definition of $\sin \angle$ defined in \eqref{equ: definition of sine between matrices}.
Therefore, substituting the error bound of $\| \sin \angle (\widehat{\mathbf U}, \mathbf U) \|_{\rm{F}}^2$ in \eqref{equ: upper bound of sine error} into \eqref{equ: relationship between subspace and sin}, we obtain
\begin{align*}
    \|\widehat{\mathbf U} \widehat{\mathbf U}^\top - \mathbf U \mathbf U^\top\|_{\rm{F}} = \Tilde{\order}\Bigg(\frac{\lambda_{\max}(\bm\Sigma_0) (1+\sigma_{\max}^k) d^{\frac{5}{4}} k^{\frac{k+12}{4}} }{{\textrm{\tt Eigen-gap}(k)}} n^{-\frac{1-\delta(n)}{k+5}} \log^{\frac{5}{2}} n\Bigg). \tag*{\Halmos}
\end{align*}

\end{proof}
\fi

\subsection{Supporting Lemmas for Theorem \ref{theorem: distribution estimation}}
\label{subsec: proof of lemmas -- distribution estimation}
Recall that Lemma~\ref{lem: eigenvalue estimation error} and Lemma Lemma~\ref{lem: backward SDEs L2 error} are key results of our development. We provide their proofs in Appendices \ref{pf: eigenvalue estimation error} and \ref{pf: backward SDEs L2 error} respectively. Some additional lemmas that support the proof of Theorem \ref{theorem: distribution estimation} are stated and proved in Appendix \ref{ap:thm3_others}.
 \subsubsection{Proof of Lemma~\ref{lem: eigenvalue estimation error}}\label{pf: eigenvalue estimation error}
\iftrue
\begin{proof}{Proof.}

For notational simplicity, let us define 
\begin{equation}
\begin{aligned}
\label{equ: definition of three covariance matrices}
    &\bm\Sigma_{t_0} := \E_{\mathbf R_{t_0} \sim P_{t_0}}\left[\mathbf R_{t_0} \mathbf R_{t_0}^\top \right] - \E_{\mathbf R_{t_0} \sim P_{t_0}}[\mathbf R_{t_0}]\E_{\mathbf R_{t_0} \sim P_{t_0}}[\mathbf R_{t_0}]^\top, \\
    &\widetilde{\bm\Sigma}_{t_0} := \E_{\mathbf R_{t_0} \sim \widetilde{P}_{t_0}}\left[\mathbf R_{t_0} \mathbf R_{t_0}^\top \right] - \E_{\mathbf R_{t_0} \sim \widetilde{P}_{t_0}}[\mathbf R_{t_0}]\E_{\mathbf R_{t_0} \sim \widetilde{P}_{t_0}}[\mathbf R_{t_0}]^\top, \textrm{ and } \\
    &\widecheck{\bm\Sigma}_{t_0} = \E_{\mathbf R_{t_0} \sim \widehat{P}_{t_0}}\left[\mathbf R_{t_0} \mathbf R_{t_0}^\top \right] - \E_{\mathbf R_{t_0} \sim \widehat{P}_{t_0}}[\mathbf R_{t_0}]\E_{\mathbf R_{t_0} \sim \widehat{P}_{t_0}}[\mathbf R_{t_0}]^\top.
\end{aligned}
\end{equation}
The proof is based on the following error decomposition.
\paragraph{Error Decomposition.}
We decompose the target operator norm as
\begin{equation}
\label{equ: oracle inequality in Theorem 3}
    \big\| \widehat{\bm\Sigma}_{0} - \bm\Sigma_0 \big\|_{\textrm{op}} \leq \underbrace{\big\| \bm\Sigma_0 - \bm\Sigma_{t_0} \big\|_{\textrm{op}}}_{(A)} + \underbrace{\big\| \bm\Sigma_{t_0} - \widetilde{\bm\Sigma}_{t_0} \big\|_{\textrm{op}}}_{(B)} + \underbrace{\big\|\widetilde{\bm\Sigma}_{t_0} - \widecheck{\bm\Sigma}_{t_0} \big\|_{\textrm{op}}}_{(C)} + \underbrace{\big\| \widehat{\bm\Sigma}_{0} - \widecheck{\bm\Sigma}_{t_0} \big\|_{\textrm{op}}}_{(D)},
\end{equation}
where term $(A)$ is the early-stopping error, term $(B)$ is the approximation error of $\calS_{\textrm{NN}}$ term $(C)$ is the mixing error of forward process \eqref{equ: diffusion forward SDE}, term $(D)$ is the finite-sample error.

\paragraph{Bounding Term $(A)$.}
Using the fact that $\mathbf R_{t_0} = e^{-t_0 / 2} \mathbf R_0 + \mathbf B_{1-e^{-t_0}}$, we have 
\begin{align*}
    \bm\Sigma_0 - \bm\Sigma_{t_0} = \bm\Sigma_0 - e^{-t_0} \bm\Sigma_0 - (1-e^{-t_0}) \mathbf I_d =  (1-e^{-t_0})(\bm\Sigma_0 - \mathbf I_d).
\end{align*}
Therefore, by the definition of $(A)$ in \eqref{equ: oracle inequality in Theorem 3} and $t_0 = n^{-\frac{1-\delta(n)}{k+5}}$, we obtain
\begin{align}
    (A) = \order \Big( \lambda_{\max}(\bm\Sigma_0) \cdot n^{-\frac{1-\delta(n)}{k+5}} \Big). \label{equ: bound term (A) in Theorem 3}
\end{align}

\paragraph{Bounding Term $(B)$.}
Under the coupled SDE system \eqref{equ: coupling SDEs system}, we have
\begin{align}
    (B) &\leq \bigg\|\E_{\mathbf R_{t_0} \sim P_{t_0}}\left[\mathbf R_{t_0} \mathbf R_{t_0}^\top \right] - \E_{\mathbf R_{t_0} \sim \widehat{P}_{t_0}}\left[\mathbf R_{t_0} \mathbf R_{t_0}^\top \right]\bigg\|_{\textrm{op}} \nonumber \\
    &\qquad + \bigg\| \E_{\mathbf R_{t_0} \sim P_{t_0}}[\mathbf R_{t_0}]\E_{\mathbf R_{t_0} \sim P_{t_0}}[\mathbf R_{t_0}]^\top - \E_{\mathbf R_{t_0} \sim \widehat{P}_{t_0}}[\mathbf R_{t_0}]\E_{\mathbf R_{t_0} \sim \widehat{P}_{t_0}}[\mathbf R_{t_0}]^\top \bigg\|_{\textrm{op}} \nonumber \\
    &= \bigg\| \E\left[ (\mathbf R_{T-t_0}^{\leftarrow}) (\mathbf R_{T-t_0}^{\leftarrow})^\top \right] - \E\left[ (\widehat{\mathbf R}_{T-t_0}^{\leftarrow}) (\widehat{\mathbf R}_{T-t_0}^{\leftarrow})^\top \right] \bigg\|_{\textrm{op}} \label{equ: temp1 of bounding term B in latent subspace in Theorem 3} \\
    &\qquad + \bigg\| \E\left[\mathbf R_{T-t_0}^{\leftarrow}\right]\E\left[\mathbf R_{T-t_0}^{\leftarrow}\right]^\top - \E\left[\widehat{\mathbf R}_{T-t_0}^{\leftarrow}\right]\E\left[\widehat{\mathbf R}_{T-t_0}^{\leftarrow}\right]^\top \bigg\|_{\textrm{op}}, \label{equ: temp2 of bounding term B in latent subspace in Theorem 3}
\end{align}
where the last equality invokes $\bRb$ and $\hatbRb$ defined in \eqref{equ: coupling SDEs system}. For term \eqref{equ: temp1 of bounding term B in latent subspace in Theorem 3}, we have
\begin{align}
    \eqref{equ: temp1 of bounding term B in latent subspace in Theorem 3} &\leq \left\| \E\left[ (\mathbf R_{T-t_0}^{\leftarrow} - \widehat{\mathbf R}_{T-t_0}^{\leftarrow}) (\mathbf R_{T-t_0}^{\leftarrow})^\top \right] \right\|_{\text{op}} + \left\| \E\left[ (\widehat{\mathbf R}_{T-t_0}^{\leftarrow}) (\mathbf R_{T-t_0}^{\leftarrow} - \widehat{\mathbf R}_{T-t_0}^{\leftarrow})^\top \right] \right\|_{\text{op}} \nonumber \\
    &\leq \sqrt{ \E \big\| \mathbf R_{T-t_0}^{\leftarrow} - \widehat{\mathbf R}_{T-t_0}^{\leftarrow} \big\|_2^2 } \cdot \left( \sqrt{ \E \big\| \mathbf R_{T-t_0}^{\leftarrow} \big\|_2^2 } + \sqrt{ \E \big\| \widehat{\mathbf R}_{T-t_0}^{\leftarrow} \big\|_2^2 } \right), \label{equ: bounding term1 in B}
\end{align}
where \eqref{equ: bounding term1 in B} follows from the Cauchy-Schwarz inequality and rearranging terms. Similarly, for term \eqref{equ: temp2 of bounding term B in latent subspace in Theorem 3}, using the Cauchy-Schwarz inequality, we have
\begin{align}
    \eqref{equ: temp2 of bounding term B in latent subspace in Theorem 3} &\leq \left\|(\E\left[\mathbf R_{T-t_0}^{\leftarrow}\right] - \E[\widehat{\mathbf R}_{T-t_0}^{\leftarrow}]) \E\left[\mathbf R_{T-t_0}^{\leftarrow}\right]^\top \right\|_{\textrm{op}} + \left\| \E[\widehat{\mathbf R}_{T-t_0}^{\leftarrow}] (\E\left[\mathbf R_{T-t_0}^{\leftarrow}\right]^\top - \E[\widehat{\mathbf R}_{T-t_0}^{\leftarrow}]^\top) \right\|_{\textrm{op}} \nonumber \\
    &\leq \sqrt{ \E \big\| \mathbf R_{T-t_0}^{\leftarrow} - \widehat{\mathbf R}_{T-t_0}^{\leftarrow} \big\|_2^2 } \cdot \left( \sqrt{ \E \big\| \mathbf R_{T-t_0}^{\leftarrow} \big\|_2^2 } + \sqrt{ \E \big\| \widehat{\mathbf R}_{T-t_0}^{\leftarrow} \big\|_2^2 } \right), \label{equ: bounding term2 in B}
\end{align}
Then, substituting \eqref{equ: bounding term1 in B} and \eqref{equ: bounding term2 in B} into \eqref{equ: temp1 of bounding term B in latent subspace in Theorem 3} and \eqref{equ: temp2 of bounding term B in latent subspace in Theorem 3}, we deduce that
\begin{align}
    (B) &\leq 2\sqrt{ \E \big\| \mathbf R_{T-t_0}^{\leftarrow} - \widehat{\mathbf R}_{T-t_0}^{\leftarrow} \big\|_2^2 } \cdot \left( \sqrt{ \E \big\| \mathbf R_{T-t_0}^{\leftarrow} \big\|_2^2 } + \sqrt{ \E \big\| \widehat{\mathbf R}_{T-t_0}^{\leftarrow} \big\|_2^2 } \right) \nonumber \\
    &= \order \bigg( (1+\sigma_{\max}^k) d^{\frac{5}{4}} k^{\frac{k+10}{4}} n^{-\frac{1-\delta(n)}{k+5}} \log^{\frac{5}{2}} n \bigg), \label{equ: bound term (B) in Theorem 3}
\end{align}
where the last equality follows from applying the upper bound \eqref{equ: score estimation bound in Lemma 2 of Theorem 3} in Lemma \ref{lem: backward SDEs L2 error} and using the fact that
\begin{align*}
    &\E \big\| \mathbf R_{T-t_0}^{\leftarrow} \big\|_2^2 = \E \big\| e^{-t_0 / 2} \mathbf R_0 + \mathbf B_{1-e^{-t_0}} \big\|_2^2 \leq e^{-t_0} \E \big\| \mathbf R_0 \big\|_2^2 + 1 - e^{-t_0} = \order(1) \ \textrm{ and } \\
    &\E \big\| \widehat{\mathbf R}_{T-t_0}^{\leftarrow} \big\|_2^2 \leq 2\E \big\| \mathbf R_{T-t_0}^{\leftarrow} - \widehat{\mathbf R}_{T-t_0}^{\leftarrow} \big\|_2^2 + 2\E \big\| \mathbf R_{T-t_0}^{\leftarrow} \big\|_2^2 = \order(1).
\end{align*}

\paragraph{Bounding Term $(C)$.}
Applying Lemma~\ref{lem: backward SDEs L2 error} to the estimated backward process starting from $P_T$ and $\calN(\mathbf 0, \mathbf I_d)$, respectively, we obtain
\begin{equation}
\label{equ: bound term (C) in Theorem 3}
    \big\| \widetilde{\bm\Sigma}_{t_0} - \widecheck{\bm\Sigma}_{t_0} \big \|_{\textrm{op}} = \order \big( 2\E\| \widehat{\mathbf R}_{T-t_0}^{\leftarrow} - \mathbf R_{T-t_0}^{\leftarrow} \|_2^2 \big) = \order \bigg( (1+\sigma_{\max}^k) d^{\frac{5}{4}} k^{\frac{k+10}{4}} n^{-\frac{1-\delta(n)}{k+5}} \log^{\frac{5}{2}} n \bigg).
\end{equation}

\paragraph{Bounding Term $(D)$.}
By introducing the estimation error between $\bar{\mathbf R}_{0}$ and $\E[\mathbf R_i]$, we have
\begin{align*}
    (D) &\leq \bigg\| \frac{1}{m-1} 
    \sum_{i=1}^m \left( (\mathbf R_i - \E[\mathbf R_i]) (\mathbf R_i - \E[\mathbf R_i])^\top \right) - \widecheck{\bm\Sigma}_{t_0} \bigg\|_{\textrm{op}} \\
    &\qquad + \bigg\| \bar{\mathbf R}_0 (\bar{\mathbf R}_0 - \E[\mathbf R_i])^\top \bigg\|_{\textrm{op}} + \bigg\| (\bar{\mathbf R}_0 - \E[\mathbf R_i])\E[\mathbf R_i]^\top \bigg\|_{\textrm{op}} \\
    &\stackrel{(i)}{\leq} \bigg\| \frac{1}{m-1} 
    \sum_{i=1}^m \left( (\mathbf R_i - \E[\mathbf R_i]) (\mathbf R_i - \E[\mathbf R_i])^\top \right) - \widecheck{\bm\Sigma}_{t_0} \bigg\|_{\textrm{op}} \\
    &\qquad + \| \bar{\mathbf R}_0 \|_2 \| \bar{\mathbf R}_0 - \E[\mathbf R_i] \|_2 + \| \bar{\mathbf R}_0 - \E[\mathbf R_i] \|_2 \| \E[\mathbf R_i] \|_2 \\
    &\stackrel{(ii)}{\leq} \bigg\| \frac{1}{m-1} 
    \sum_{i=1}^m \left( (\widecheck{\bm\Sigma}_{t_0})^{-1/2}(\mathbf R_i - \E[\mathbf R_i]) (\mathbf R_i - \E[\mathbf R_i])^\top (\widecheck{\bm\Sigma}_{t_0})^{-1/2} - \mathbf I_d \right) \bigg\|_{\textrm{op}} \bigg\| \widecheck{\bm\Sigma}_{t_0} \bigg\|_{\textrm{op}} \\
    &\qquad + \| (\widecheck{\bm\Sigma}_{t_0})^{-\frac{1}{2}} \bar{\mathbf R}_0 \|_2 \| (\widecheck{\bm\Sigma}_{t_0})^{-\frac{1}{2}} (\bar{\mathbf R}_0 - \E[\mathbf R_i]) \|_2 \| \widecheck{\bm\Sigma}_{t_0} \|_{\textrm{op}} \\
    &\qquad \qquad + \| (\widecheck{\bm\Sigma}_{t_0})^{-\frac{1}{2}} (\bar{\mathbf R}_0 - \E[\mathbf R_i]) \|_2 \| (\widecheck{\bm\Sigma}_{t_0})^{-\frac{1}{2}} \E[\mathbf R_i] \|_2 \| \widecheck{\bm\Sigma}_{t_0} \|_{\textrm{op}},
\end{align*}
where $(i)$ follows from the H$\ddot{\text{o}}$lder inequality and $(ii)$ holds due to the covariance normalization using~$(\widecheck{\bm\Sigma}_{t_0})^{-\frac{1}{2}}$. Applying Theorem 3.1.1 and Theorem 4.6.1 of \citet{vershynin2018high} to
$$
(\widecheck{\bm\Sigma}_{t_0})^{-1/2}(\bar{\mathbf R}_0 - \E[\mathbf R_i]) \ \textrm{  and  } \
\frac{1}{m-1} \sum_{i=1}^m (\widecheck{\bm\Sigma}_{t_0})^{-1/2}(\mathbf R_i - \E[\mathbf R_i]) (\mathbf R_i - \E[\mathbf R_i])^\top (\widecheck{\bm\Sigma}_{t_0})^{-1/2},
$$
respectively, we obtain that with probability $1 - \delta$, it holds
\begin{align}
    (D) &= \order \bigg( \max\bigg\{\frac{\sqrt{d} + \sqrt{\log (2/\delta)}}{\sqrt{m}}, \bigg( \frac{\sqrt{d} + \sqrt{\log (2/\delta)}}{\sqrt{m}} \bigg)^2 \bigg\} \cdot \big\|\widecheck{\bm\Sigma}_{t_0}\big\|_{\textrm{op}} \bigg) \nonumber \\
    &= \order \bigg( \lambda_{\max}(\bm\Sigma_0) (1+\sigma_{\max}^k) d^{\frac{5}{4}} k^{\frac{k+10}{4}} n^{-\frac{1-\delta(n)}{k+5}} \log^{\frac{5}{2}} n \bigg), \label{equ: bound term (D) in Theorem 3}
\end{align}
where the last inequality the last inequality invokes the order of $m$ in \eqref{equ: order of m in Lemma 3} and the fact that
$$
\big\|\widecheck{\bm\Sigma}_{t_0}\big\|_{\textrm{op}} \leq \big\|\widecheck{\bm\Sigma}_{t_0} - \widetilde{\bm\Sigma}_{t_0}\big \|_{\textrm{op}} + \big\| \widetilde{\bm\Sigma}_{t_0} - \bm\Sigma_{t_0} \big \|_{\textrm{op}} + \big\| \bm\Sigma_{t_0} - \bm\Sigma_{0} \big \|_{\textrm{op}} + \big\| \bm\Sigma_{0}\big\|_{\textrm{op}}.
$$

Summing up the upper bound of $(A)$--$(D)$ in \eqref{equ: bound term (A) in Theorem 3} and \eqref{equ: bound term (B) in Theorem 3}--\eqref{equ: bound term (D) in Theorem 3}, we obtain the desired result. \hfill\Halmos
\end{proof}
\fi

\iftrue
\subsubsection{Proof of Lemma~\ref{lem: backward SDEs L2 error}}\label{pf: backward SDEs L2 error}

\begin{proof}{Proof.}
    For notational simplicity, we denote $\widehat{\mathbf s}_{T-t}(\cdot) := \widehat{\mathbf s}_{\bm\theta}\left( \cdot, T-t \right)$ and let $\widehat{\sigma}_{\min}^2$ and $\widehat{\sigma}_{\max}^2$ be the minimal and maximal elements of $\mathbf c$ in $\widehat{\mathbf s}_{\bm\theta}$, respectively. First, by direct calculation, we obtain
    \begin{align*}
        \frac{\dd \E\| \bRb - \hatbRb \|_2^2}{\dd t} &= 2\E\left[\Big(\bRb - \hatbRb\Big)^\top \Big( \frac{1}{2}\bRb - \frac{1}{2} \hatbRb + \nabla \log p_{T-t}\left(\bRb\right) - \widehat{\mathbf s}_{T-t}(\hatbRb) \Big)\right] \\
        &= \E\| \bRb - \hatbRb \|_2^2 + 2 \underbrace{\E\left[(\bRb - \hatbRb)^\top (\nabla \log p_{T-t}\left(\bRb\right) - \widehat{\mathbf s}_{T-t}(\hatbRb))\right]}_{(*)}.
    \end{align*}
  Consider $\Tilde{\mathbf g}_{\bm\zeta} : \bb R^k \times [0, T] \rightarrow \bb R^k$, equivalent to $\mathbf g_{\bm\zeta}$ defined via transformation, defined as
    \begin{equation}
    \label{equ: anothor formula for f relu network}
        \Tilde{\mathbf g}_{\bm\zeta}(\mathbf z, t) := \mathbf g_{\bm\zeta}(\mathbf V^\top \mathbf D_t \mathbf V \mathbf z, t),
    \end{equation}
    where $\mathbf V$, $\mathbf D_t$, and $\mathbf g_{\bm\zeta}$ are components of $\widehat{\mathbf s}_{\bm\theta}$ defined in \eqref{equ: score network}. Note that the Lipschitz constant of the ReLU network $\Tilde{\mathbf g}_{\bm\zeta}$ with respect to $\mathbf z$ is also on the order of $\gamma$ defined in \eqref{equ: order of parameters in Theorem 1}.
    Then, for term $(*)$, we have
    \begin{align*}
        (*) &= \E\left[(\bRb - \hatbRb)^\top (\nabla \log p_{T-t}\left(\bRb\right) - \widehat{\mathbf s}_{T-t}(\bRb))\right] \\
        &\qquad + \E\left[(\bRb - \hatbRb)^\top (\widehat{\mathbf s}_{T-t}(\bRb) - \widehat{\mathbf s}_{T-t}(\hatbRb))\right] \\
        &\stackrel{(i)}{\leq} \frac{\E\| \bRb - \hatbRb \|_2^2}{4} + \E\left\|\widehat{\mathbf s}_{T-t}\left(\bRb\right) - \nabla \log p_{T-t}\left(\bRb\right)\right\|_2^2 \\
        &\quad + \E\left[ (\bRb - \hatbRb)^\top \mathbf D_{T-t}^{1/2}(\alpha_{T-t} \gamma_1 \mathbf D_{T-t}^{1/2} \mathbf V (\mathbf V^{\top} \mathbf D_{T-t} \mathbf V)^{-1} \mathbf V^\top \mathbf D_{T-t}^{1/2} - \mathbf I) \mathbf D_{T-t}^{1/2} (\bRb - \hatbRb) \right] \\
        &\stackrel{(ii)}{\leq} \bigg(\frac{1}{4} + \frac{(\alpha_{T-t} \gamma_1 - 1)\indicator\{\alpha_{T-t} \gamma_1 > 1\}}{h_{T-t} + \widehat{\sigma}_{\min}^2\alpha_{T-t}^2} + \frac{(\alpha_{T-t} \gamma_1 - 1)\indicator\{\alpha_{T-t} \gamma_1 \leq 1\}}{h_{T-t} + \widehat{\sigma}_{\max}^2\alpha_{T-t}^2} \bigg)\E\| \bRb - \hatbRb \|_2^2 \\
        &\quad + \E\left\|\widehat{\mathbf s}_{T-t}\left(\bRb\right) - \nabla \log p_{T-t}\left(\bRb\right)\right\|_2^2,
    \end{align*}
    where $(i)$ holds due to the Cauchy-Schwarz inequality and the fact that $\widehat{\mathbf s}_{T-t}(\cdot)$ is $\gamma_1$-Lipschitz; $(ii)$ follows from
    \begin{equation*}
    \begin{aligned}
        & \lambda_{\max}(\alpha_{T-t} \gamma_1 \mathbf V (\mathbf V^\top \mathbf D_{T-t} \mathbf V)^{-1} \mathbf V^\top \mathbf D_{T-t} - \mathbf I) = \alpha_{T-t} \gamma_1 - 1,
    \end{aligned}
    \end{equation*}
    and
    $$
        \frac{1}{h_{T-t} + \widehat{\sigma}_{\max}^2\alpha_{T-t}^2} \leq \|\widehat{\mathbf D}_{T-t}\|_{\textrm{op}} \leq \frac{1}{h_{T-t} + \widehat{\sigma}_{\min}^2\alpha_{T-t}^2}.
    $$
    
 Notice that $\alpha_{T-t} \gamma_1 \leq 1$ is equivalent to $t \leq T - 2 \log \gamma_1$. By Gr$\ddot{\text{o}}$nwall's inequality, we obtain
    \begin{align}
        \quad &\E\| \mathbf R_{T-t_0}^{\leftarrow} - \widehat{\mathbf R}_{T-t_0}^{\leftarrow} \|_2^2 \nonumber \\
        &\leq \left( \E\| \mathbf R_0^{\leftarrow} - \widehat{\mathbf R}_0^{\leftarrow} \|_2^2 + \int_{0}^{T-t_0} 2\E\left\|\widehat{\mathbf s}_{T-t}\left(\bRb\right) - \nabla \log p_{T-t}\left(\bRb\right)\right\|_2^2 \dd t \right) \nonumber \\
        &\qquad \cdot \exp\left(\int_{0}^{T - 2 \log \gamma_1} \left(\frac{3}{2} + \frac{2(\alpha_{T-t} \gamma_1 - 1)}{h_{T-t} + \widehat{\sigma}_{\max}^2\alpha_{T-t}^2} \right) \dd t + \int_{T - 2 \log \gamma_1}^{T-t_0} \left(\frac{3}{2} + \frac{2(\alpha_{T-t} \gamma_1 - 1)}{h_{T-t} + \widehat{\sigma}_{\min}^2\alpha_{T-t}^2} \right) \dd t \right) \nonumber \\
        &= \left( \E\| \mathbf R_0^{\leftarrow} - \widehat{\mathbf R}_0^{\leftarrow} \|_2^2 + \int_{0}^{T-t_0} 2\E\left\|\widehat{\mathbf s}_{T-t}\left(\bRb\right) - \nabla \log p_{T-t}\left(\bRb\right)\right\|_2^2 \dd t \right) \label{equ: dynamic bound of backward process line 1} \\
        &\qquad \cdot \exp\left(\frac{3}{2}(T-t_0) + \int_{t_0}^{2 \log \gamma_1} \frac{2(\alpha_{w} \gamma_1 - 1)}{h_{w} + \widehat{\sigma}_{\min}^2\alpha_{w}^2} \dd w + \int_{2 \log \gamma_1}^{T} \frac{2(\alpha_{w} \gamma_1 - 1)}{h_{w} + \widehat{\sigma}_{\max}^2\alpha_{w}^2} \dd w \right), \label{equ: dynamic bound of backward process line 2}
    \end{align}
    where the last equality follows from rearranging the terms and a change of variable $T - t = w$.

   Now we claim that  
    \begin{equation}
    \label{equ: integral bound in Theorem 3}
        \int_{t_0}^{2 \log \gamma_1} \frac{2(\alpha_{w} \gamma_1 - 1)}{h_{w} + \widehat{\sigma}_{\min}^2\alpha_{w}^2} \dd w + \int_{2 \log \gamma_1}^{T} \frac{2(\alpha_{w} \gamma_1 - 1)}{h_{w} + \widehat{\sigma}_{\max}^2\alpha_{w}^2} \dd w \leq 4\gamma_1 \left(1 - \log (\widehat{\sigma}_{\min}^2 + t_0)\right) - 2(T - t_0).
    \end{equation}
    To verify this, consider the integral
    \begin{numcases}{\int \frac{\alpha_w \gamma_1 - 1}{h_w + c\alpha_w^2} \dd w =}
    C - \dfrac{2 \gamma_1 \arctan(\sqrt{c - 1}e^{-w/2})}{\sqrt{c - 1}} - \log (e^w + c - 1), & $\forall \ c > 1$ \label{equ: integral case1} \\
    C - 2\gamma_1 e^{-w/2} - w, & $c = 1$ \label{equ: integral case2} \\
    C - \dfrac{\gamma_1 \log \left(\frac{1 + e^{-w/2}\sqrt{1 - c}}{1 - e^{-w/2}\sqrt{1 - c}}\right) }{\sqrt{1 - c}} - \log (e^w + c - 1), & $\forall \ 0 < c < 1$ \label{equ: integral case3}
    \end{numcases}

\begin{enumerate}
        \item For the case $ c > 1 $, note that
        \begin{equation}
        \begin{aligned}
        \label{equ: temp of inregral bound 1 in Theorem 3}
            & -\log\bigg(\frac{e^T + c - 1}{e^{t_0} + c - 1}\bigg) \leq -(T-t_0) + \log (1+(c-1)e^{-t_0})
            \leq \log(c-(c-1)t_0)-(T - t_0), \\
            & -\frac{\arctan(\sqrt{c - 1} e^{-T/2}) - \arctan(\sqrt{c - 1} e^{-t_0/2})}{\sqrt{c - 1}} \leq e^{-t_0/2} - e^{-T/2} \leq 1.
        \end{aligned}
        \end{equation}
        Then, by substituting \eqref{equ: temp of inregral bound 1 in Theorem 3} into the integral \eqref{equ: integral case1}, we obtain
        \begin{equation}
        \label{equ: integral bound 1 in Theorem 3}
            \displaystyle\int_{t_0}^{T} \frac{\alpha_w \gamma_1 - 1}{h_w + c\alpha_w^2} \dd w \leq 2\gamma_1 + \log(c-(c-1)t_0) - (T-t_0).
        \end{equation}
        \item For the case $ c = 1 $, applying $e^{-w/2} \leq 1$, we obtain that the integral in \eqref{equ: integral case2} satisfies
        \begin{equation}
        \label{equ: integral bound 2 in Theorem 3}
            \displaystyle\int_{t_0}^{T} \frac{\alpha_w \gamma_1 - 1}{h_w + c\alpha_w^2} \dd w \leq -2\gamma_1(e^{-T} - e^{-t_0})-(T-t_0) \leq 2\gamma_1 - (T-t_0).
        \end{equation}
        \item For the case $0 < c < 1$, due to the continuity of the integral \eqref{equ: integral case3} with respect to $c$ and the bound in \eqref{equ: integral bound 2 in Theorem 3}, we only need to focus on the case $c \ll 1$. Without loss of generality, we consider $c < 1/2$. By direct calculation, we have 
   \begin{align}
            &\quad -\frac{1}{\sqrt{1 - c}} \left(\log \bigg(\frac{1 + e^{-T/2}\sqrt{1 - c}}{1 - e^{-T/2}\sqrt{1 - c}}\bigg) - \log \bigg(\frac{1 + e^{-t_0/2}\sqrt{1 - c}}{1 - e^{-t_0/2}\sqrt{1 - c}}\bigg)\right) \nonumber \\
            &\leq \frac{1}{\sqrt{1 - c}} \log \bigg(\frac{1 + e^{-t_0/2}\sqrt{1 - c}}{1 - e^{-t_0/2}\sqrt{1 - c}}\bigg) \nonumber \\
            &\leq \sqrt{2}(\log 4 - \log (c + t_0)), \label{equ: temp2 of integral bound 3}
        \end{align}
        where the last inequality follows from the fact that $e^{-x} \leq 1/(1+x)$, $\sqrt{1-c} \leq 1-c/2$, and $\log((1+x) / (1-x))$ is increasing in $x$; and rearranging terms. Then, by substituting \eqref{equ: temp of inregral bound 1 in Theorem 3} and \eqref{equ: temp2 of integral bound 3} into \eqref{equ: integral case3}, we obtain
        \begin{equation}
        \label{equ: integral bound 3 in Theorem 3}
        \displaystyle\int_{t_0}^{T} \frac{\alpha_w \gamma_1 - 1}{h_w + c\alpha_w^2} \dd w \leq 2\gamma_1(1 - \log (c + t_0))- (T-t_0).
        \end{equation}
    \end{enumerate}
    Combining the results in \eqref{equ: integral bound 1 in Theorem 3}, \eqref{equ: integral bound 2 in Theorem 3}, and \eqref{equ: integral bound 3 in Theorem 3}, we verified the claim in \eqref{equ: integral bound in Theorem 3}. 
    Finally, applying the upper bound of score estimation in \eqref{equ: score estimation bound in Lemma 2 of Theorem 3} and substituting \eqref{equ: integral bound in Theorem 3} into \eqref{equ: dynamic bound of backward process line 1} and \eqref{equ: dynamic bound of backward process line 2}, we deduce that
    \begin{align*} 
        &\quad \E \| \mathbf R_{T-t_0}^{\leftarrow} - \widehat{\mathbf R}_{T-t_0}^{\leftarrow} \|_2^2 \\ &\leq \left( \E\| \mathbf R_0^{\leftarrow} - \widehat{\mathbf R}_0^{\leftarrow} \|_2^2 + 2\epsilon (T-t_0) \right) \cdot \exp\left( \frac{3}{2}(T-t_0) + 4\gamma_1 (1 - \log (\widehat{\sigma}_{\min}^2 + t_0)) - 2(T-t_0) \right) \\
        &= \order \bigg( (1+\sigma_{\max}^k) d^{\frac{5}{4}} k^{\frac{k+10}{4}} n^{-\frac{1-\delta(n)}{k+5}} \log^{\frac{5}{2}} n \bigg),
    \end{align*}
    where the last equality follows from invoking $\E\| \mathbf R_0^{\leftarrow} - \widehat{\mathbf R}_0^{\leftarrow} \|_2^2 = \order (e^{-T})$ and rearranging terms. \hfill\Halmos
\end{proof}
\fi

\subsubsection{Other Supporting Lemmas for Theorem \ref{theorem: distribution estimation}}
\label{ap:thm3_others}
\begin{lemma}
\label{lem: total variation distance error}
Suppose that $P_{\mathrm{data}}$ is sub-Gaussian, and both $\widehat{\mathbf s}_{\bm\theta}(\mathbf r, t)$ and $\nabla \log p_t(\mathbf r)$ are Lipschitz with respect to both $\mathbf r$ and $t$. Consider the score estimation error satisfying
$$
\int_{t_0}^{T} \E_{\mathbf R_t \sim P_t}\left\|\widehat{\mathbf s}_{\bm\theta}(\mathbf R_t, t) - \nabla \log p_t \left(\mathbf R_t\right)\right\|_{2}^{2} \dd t = \order (\epsilon (T - t_0)).
$$
Then, the total variation distance is bounded by
\begin{equation}
\label{equ: upper bound of TV distance}
\operatorname{TV}(P_{\mathrm{data}}, \widehat{P}_{t_0}) = \order \bigg( d t_0 L_s (1 + \sigma_{\max}^2) + \sqrt{\epsilon (T-t_0)}+\sqrt{\operatorname{KL}\left(P_{\mathrm{data}} \| \mathcal{N}\left(\mathbf 0, \mathbf I_d\right)\right)} \exp (-T) \bigg),
\end{equation}
where $P_{\mathrm{data}}$ is the initial distribution of $\mathbf R$ in \eqref{equ: factor model} and $\widehat{P}_{t_0}$ is the marginal distribution of the backward process $\widehat{\mathbf R}_{T-t_0}^{\leftarrow}$ in \eqref{equ: learned backward SDE} starting from $\mathcal{N}(\mathbf 0, \mathbf I_d)$.
\end{lemma}

\begin{proof}{Proof of Lemma~\ref{lem: total variation distance error}.}
To estimate $\operatorname{TV}(P_{\textrm{data}}, \widehat{P}_{t_0})$, we leverage the error decomposition in \eqref{equ: TV distance decomposition in Theorem 3}.
Recall that $\widetilde{P}_{t_0}$ is the marginal distribution of $\widehat{\mathbf R}_{T - t_0}^{\leftarrow}$ in \eqref{equ: learned backward SDE} initialized with $\widehat{\mathbf R}_{0}^{\leftarrow} \sim P_T$. In the decomposition \eqref{equ: TV distance decomposition in Theorem 3}, $\operatorname{TV}(P_{\textrm{data}}, P_{t_0})$ is the early stopping error, $\operatorname{TV}(P_{t_0}, \widetilde{P}_{t_0})$ is the statistical error arising from the score estimation, and $\operatorname{TV}(\widetilde{P}_{t_0}, \widehat{P}_{t_0})$ is the mixing error of the forward process \eqref{equ: diffusion forward SDE}. 
\begin{enumerate}
    \item For term $\operatorname{TV}(P_{\textrm{data}}, P_{t_0})$, applying the upper bound \eqref{equ: TV between 0 and t0} with $t = t_0$ in Lemma \ref{lem: TV between 0 and t0}, we obtain
    \begin{equation}
    \label{equ: bounding term1 in TV distance}
    \operatorname{TV}\left(P_{\textrm{data}}, P_{t_0} \right) = \order (d t_0).
    \end{equation}
    \item For term $\operatorname{TV}(P_{t_0}, \widetilde{P}_{t_0})$, by Pinsker's inequality \cite[Lemma 2.5]{tsy2009nonparametric} and the upper bound of KL-divergence  \eqref{equ: KL divergence between backward processes} in Lemma \ref{lem: KL divergence between backward processes}, we have
    \begin{equation}
    \label{equ: bounding term2 in TV distance}
    \operatorname{TV}(P_{t_0}, \widetilde{P}_{t_0}) \leq \operatorname{KL}(P_{t_0} || \widetilde{P}_{t_0}) = \order \big( \sqrt{\epsilon\left(T-t_0\right)} \big).
    \end{equation}
    \item For term $\operatorname{TV}(\widetilde{P}_{t_0}, \widehat{P}_{t_0})$, by Pinsker's inequality \cite[Lemma 2.5]{tsy2009nonparametric} and Data processing inequality \cite[Theorem 2.8.1]{thomas2006elements} , we deduce that
    \begin{equation}
    \label{equ: bounding term3 in TV distance}
    {\operatorname{TV}}(\widetilde{P}_{t_0}, \widehat{P}_{t_0}) \leq \sqrt{\operatorname{KL}(\widetilde{P}_{t_0} \| \widehat{P}_{t_0})} \leq \sqrt{\operatorname{KL}\left(P_T \| \mathcal{N}\left(0, \mathbf I_d\right)\right)} = \order\big( \sqrt{\operatorname{KL}\left(P_{\textrm{data}} \| \mathcal{N}\left(0, \mathbf I_d\right)\right)} \exp (-T)\big),
    \end{equation}
where in the last inequality, we use the exponential mixing property of the O-U process.
\end{enumerate}
Substituting the upper bounds \eqref{equ: bounding term1 in TV distance}, \eqref{equ: bounding term2 in TV distance}, and \eqref{equ: bounding term3 in TV distance} into \eqref{equ: TV distance decomposition in Theorem 3}, we obtain the desired result. \hfill\Halmos
\end{proof}

\begin{lemma}[Novikov's condition]
\label{lem: Novikov's condition}
Under the assumptions in Lemma \ref{lem: total variation distance error}, it holds
\begin{equation}
\label{equ: Novikov's condition}
\E_{(\bRb)_{t \in [0, T-t_0]}} \bigg[\exp \left(\frac{1}{2} \int_0^{T-t_0}\left\|\widehat{\mathbf s}_{\bm\theta}(\bRb, t) - \nabla \log p_{T-t}\left(\bRb\right)\right\|_2^2 \dd t\right) \bigg] < \infty,
\end{equation}
where the expectation is taken over the backward diffusion process $(\bRb)_{t \in [0, T-t_0]}$ in \eqref{equ: diffusion backward SDE}.
\end{lemma}

\begin{proof}{Proof of Lemma~\ref{lem: Novikov's condition}.}
The result follows from a straightforward calculation using the same techniques as in \cite[Lemma 11]{chen2023score}. \hfill\Halmos

\end{proof}

\begin{lemma}
\label{lem: KL divergence between backward processes}
Suppose that the assumptions in Lemma \ref{lem: total variation distance error} hold. When both the ground-truth and the learned backward processes start with $\mathbf R_0^{\leftarrow} \stackrel{\textrm{d}}{=} \widehat{\mathbf R}_{0}^{\leftarrow} \sim P_T$, the KL-divergence between the laws of the terminal distributions of the processes $\mathbf R_{T-t_0}^{\leftarrow}$ and $\widehat{\mathbf R}_{T-t_0}^{\leftarrow}$ can be bounded by
\begin{equation}
\label{equ: KL divergence between backward processes}
\operatorname{KL}(P_{t_0} || \widetilde{P}_{t_0}) \leq \E\left(\frac{1}{2} \int_0^{T-t_0}\left\|\widehat{\mathbf s}_{\bm\theta}\left(\bRb, T-t\right) - \nabla \log p_{T-t}\left(\bRb\right)\right\|_2^2 \dd t\right) = \order \Big( \epsilon(T-t_0) \Big).
\end{equation}
\end{lemma}

\begin{proof}{Proof of Lemma~\ref{lem: KL divergence between backward processes}.}
By Lemma \ref{lem: Novikov's condition}, the Novikov's condition holds. By the data processing inequality of the KL divergence \cite[Theorem 2.8.1]{thomas2006elements}, we have
\begin{align*}
    \operatorname{KL}\big(P_{t_0}\big\|\widetilde P_{t_0}\big) 
    &\leq
    \operatorname{KL}\big(({\bf R}_t^{\leftarrow})_{0 \leq t \leq T-t_0} || (\widehat{\bf R}_t^{\leftarrow})_{0 \leq t \leq T-t_0}\big) \\
    &=\frac{1}{2}\E\left[\int_{0}^{T-t_0}\big\|\widehat{\mathbf s}_{\boldsymbol{\theta}}(\mathbf R_t^{\leftarrow},T-t)-\nabla\log p_{T-t}(\mathbf R_t^{\leftarrow})\big\|_2^2\mathrm dt\right].
\end{align*}
Immediately, we obtain the results directly using Girsanov's Theorem \cite[Theorem 9]{chen2022sampling}. \hfill\Halmos
\end{proof}

\begin{lemma}
\label{lem: TV between 0 and t0}
Suppose that the assumptions in Lemma \ref{lem: total variation distance error} hold. Then, for any $t < 1/d$, we have
\begin{equation}
\label{equ: TV between 0 and t0}
    \operatorname{TV}(P_{\textrm{data}}, P_t) = \order (d t).
\end{equation}
\end{lemma}

\begin{proof}{Proof of Lemma~\ref{lem: TV between 0 and t0}.}
Given $\mathbf R_0$, $\mathbf R_t$ can be represented as
\begin{equation*}
\mathbf R_t = e^{-t/2} \mathbf R_0 + \int_0^{t} e^{-(t-s)/2} \dd \mathbf W_s,
\end{equation*}
where $\mathbf R_0$ and $\int_0^t e^{-(t-s)/2} \dd \mathbf W_s$ are independent. Then, the density of $\mathbf R_t$ is given by
    $$
        p_t(\mathbf r) = \int p_{\textrm{data}}(\mathbf y) \phi(\mathbf r; \alpha_t \mathbf y, h_t) \dd \mathbf y.
    $$
    Define
    \begin{equation}
    \label{equ: choice of S(d, t)}
    S(d, t) := \order(\sqrt{d + \log (1/t)})
    \end{equation}
    as a truncation radius and we have
    \begin{align}
        \operatorname{TV}\left( P_{\textrm{data}}, P_t \right) &= \frac{1}{2} \int \left| p_t(\mathbf r) - p_{\textrm{data}}(\mathbf r) \right| \dd \mathbf r \nonumber \\
        &\leq \frac{1}{2} \int_{\| \mathbf r \|_2 > S(d, t)} \Big(p_t(\mathbf r) + p_{\textrm{data}}(\mathbf r)\Big) \dd \mathbf r \label{equ: term1 in TV between 0 and t0} \\
        &\qquad + \frac{1}{2} \int_{\| \mathbf r \|_2 \leq S(d, t)} \bigg| \int ( p_{\textrm{data}}(\mathbf y) \phi(\mathbf r; \alpha_t \mathbf y, h_t) - p_{\textrm{data}}(\mathbf r) ) \dd \mathbf y \bigg| \dd \mathbf r. \label{equ: term2 in TV between 0 and t0}
    \end{align}
     By the density upper bound in \eqref{equ: R0 sub-Gaussian tail -- heterogeneous} and Theorem 3.1 of \citet{chazottes2019evolution}, it holds that $p_{\textrm{data}}$ and $p_t(\mathbf r)$ are sub-Gaussian and there exists a constant $A_1 > 0$ such that $(p_t(\mathbf r) + p_{\textrm{data}}(\mathbf r)) \leq \exp(-A_1 \|\mathbf r\|_2^2/2)$.

    For term \eqref{equ: term1 in TV between 0 and t0}, using the sub-Gaussian tail in Proposition 2.6.6 of \citet{vershynin2018high} and invoking the order of $S(d, t)$ in \eqref{equ: choice of S(d, t)}, we obtain that
    \begin{equation}
    \label{equ: bound of term1 in TV between 0 and t0}
    \eqref{equ: term1 in TV between 0 and t0} = \order \bigg( \frac{2^{-\frac{d}{2}} d S(d, t)^{d-2}}{A_1 \Gamma\left( \frac{d}{2} + 1 \right)} \exp\left(- \frac{A_1 S(d, t)^2}{2} \right) \bigg) = \order \big( t\exp\left( -A_1 d \right) \big).
    \end{equation}
    For term \eqref{equ: term2 in TV between 0 and t0}, by taking a change of variable $\mathbf z := (\mathbf r - \alpha_t \mathbf y)/\sqrt{h_t}$, we deduce that
    \begin{align}
        \eqref{equ: term2 in TV between 0 and t0} &= \int_{\| \mathbf r \|_2 \leq S(d, t)} \left| \int \left( p_{\textrm{data}}(\alpha_t^{-1} ( \mathbf r - \sqrt{h_t} \mathbf z )) \phi( \mathbf z; \mathbf 0, \mathbf I) - p_{\textrm{data}}(\mathbf r) \right) \dd \mathbf z \right| \dd \mathbf r \nonumber \\
        &\stackrel{(i)}{=} \order \bigg( \int_{\| \mathbf r \|_2 \leq S(d, t)} \bigg| \int \bigg(\nabla p_{\textrm{data}}(\mathbf r) \bigg( \frac{t \mathbf r}{2} - \sqrt{t} \mathbf z \bigg) \phi(\mathbf z; \mathbf 0, \mathbf I) \dd \mathbf z \bigg| \dd \mathbf r \label{equ: upper bound of TV in domain S(d, t) line1} \\
        &\qquad \qquad + \int_{\| \mathbf r \|_2 \leq S(d, t)} \bigg| \frac{1}{2} \bigg( \frac{t \mathbf r}{2} - \sqrt{t} \mathbf z \bigg)^\top \nabla^2 p_{\textrm{data}}(\mathbf r) \bigg( \frac{t \mathbf r}{2} - \sqrt{t} \mathbf z \bigg) \bigg) \phi(\mathbf z; \mathbf 0, \mathbf I) \dd \mathbf z \bigg| \dd \mathbf r \bigg), \label{equ: upper bound of TV in domain S(d, t) line2}
    \end{align}
    where $(i)$ involves the Taylor expansion $\alpha_t^{-1} = 1 + t/2 + \order(t^2)$, $h_t = t + \order(t^2)$ and the fact that $$
    p_{\textrm{data}}(e^{t/2} (\mathbf r - \sqrt{h_t} \mathbf z)) = p_{\textrm{data}}(\mathbf r) + \nabla p_{\textrm{data}}(\mathbf r) \bigg( \frac{t \mathbf r}{2} - \sqrt{t} \mathbf z \bigg) + \frac{1}{2} \bigg( \frac{t \mathbf r}{2} - \sqrt{t} \mathbf z \bigg)^\top \nabla^2 p_{\textrm{data}}(\mathbf r) \bigg( \frac{t \mathbf r}{2} - \sqrt{t} \mathbf z \bigg).
    $$
    The integrals associated with the kernel function $\phi(\mathbf z; \mathbf 0, \mathbf I)$ satisfy
    \begin{equation}
    \label{equ: integral about kernel function -- 0 moment}
        \int p_{\textrm{data}}(\mathbf r) \phi(\mathbf z; \mathbf 0, \mathbf I) \dd \mathbf z = p_{\textrm{data}}(\mathbf r)
    \end{equation}
    and
    \begin{align}
        \int \nabla p_{\textrm{data}}(\mathbf r) \bigg( \frac{t \mathbf r}{2} - \sqrt{t} \mathbf z \bigg) \phi(\mathbf z; \mathbf 0, \mathbf I) \dd \mathbf z = \frac{t}{2} \nabla \log p_{\textrm{data}}(\mathbf r) \mathbf r p_{\textrm{data}}(\mathbf r) = \order ( t \|\mathbf r\|_2 \| \nabla \log p_{\textrm{data}}(\mathbf r) \|_2 \cdot p_{\textrm{data}}(\mathbf r)). \label{equ: integral about kernel function -- 1 moment}
    \end{align}
    Moreover, since the Hessian matrix satisfies the following property
    \begin{equation}
    \begin{aligned}
    \label{equ: Hessian matrix representation}
        \nabla^2 p_{\textrm{data}}(\mathbf r) = (\nabla^2 \log p_{\textrm{data}}(\mathbf r) + \nabla \log p_{\textrm{data}}(\mathbf r) \nabla \log p_{\textrm{data}}(\mathbf r)^\top ) \cdot p_{\textrm{data}}(\mathbf r),
    \end{aligned}
    \end{equation}
    we deduce
    \begin{align}
    \label{equ: integral about kernel function -- 2 moment}
        &\quad \int \bigg(\frac{1}{2} \bigg( \frac{t \mathbf r}{2} - \sqrt{t} \mathbf z \bigg)^\top \nabla^2 p_{\textrm{data}}(\mathbf r) \bigg( \frac{t \mathbf r}{2} - \sqrt{t} \mathbf z \bigg) \bigg) \phi(\mathbf z; \mathbf 0, \mathbf I) \dd \mathbf z \nonumber \\
        &\stackrel{(i)}{=} \operatorname{tr}\bigg( \frac{1}{8} t^2 \nabla^2 p_{\textrm{data}}(\mathbf r) \mathbf r \mathbf r^\top + \frac{1}{2} t \nabla^2 p_{\textrm{data}}(\mathbf r) \bigg) \nonumber \\
        &\stackrel{(ii)}{=} \operatorname{tr}\bigg( \bigg( \frac{1}{8} t^2 \mathbf r \mathbf r^\top + \frac{1}{2} t \bigg) (\nabla^2 \log p_{\textrm{data}}(\mathbf r) + \nabla \log p_{\textrm{data}}(\mathbf r) \nabla \log p_{\textrm{data}}(\mathbf r)^\top ) \cdot p_{\textrm{data}}(\mathbf r) \bigg) \nonumber \\
        &= \order \bigg( (t^2 \|\mathbf r\|_2^2 + t) \operatorname{tr}\big( (\nabla^2 \log p_{\textrm{data}}(\mathbf r) + \nabla \log p_{\textrm{data}}(\mathbf r) \nabla \log p_{\textrm{data}}(\mathbf r)^\top ) \cdot p_{\textrm{data}}(\mathbf r) \big) \bigg).
    \end{align}
    where $(i)$ is follows from $\int \mathbf z \mathbf z^\top \phi(\mathbf z; \mathbf 0, \mathbf I) \dd \mathbf z = \mathbf I_d$ and rearranging terms, and $(ii)$ follows \eqref{equ: Hessian matrix representation}.
    
Therefore, by substituting \eqref{equ: integral about kernel function -- 0 moment}, \eqref{equ: integral about kernel function -- 1 moment} and \eqref{equ: integral about kernel function -- 2 moment} into \eqref{equ: upper bound of TV in domain S(d, t) line1} and \eqref{equ: upper bound of TV in domain S(d, t) line2}, we obtain that
    \begin{align}
        \eqref{equ: term2 in TV between 0 and t0}
        &= \order \bigg( \int_{\| \mathbf r \|_2 \leq S(d, t)} \bigg( t \|\mathbf r\|_2 \| \nabla \log p_{\textrm{data}}(\mathbf r) \|_2 \nonumber \\
        &\qquad + \operatorname{tr} \big( (t^2 \|\mathbf r\|_2^2 + t) (\nabla^2 \log p_{\textrm{data}}(\mathbf r) + \nabla \log p_{\textrm{data}}(\mathbf r) \nabla \log p_{\textrm{data}}(\mathbf r)^\top ) \big) \bigg) p_{\textrm{data}}(\mathbf r) \dd \mathbf r \bigg) \nonumber \\
        &\stackrel{(i)}{=} \order \bigg( t S(d, t) \sqrt{\E_{\mathbf R_0 \sim P_{\textrm{data}}}[\| \nabla \log p_{\textrm{data}}(\mathbf R_0) \|_2^2]} \nonumber \\
        &\qquad + (t^2 S^2(d, t)  + t) \operatorname{tr}\bigg(\int (\nabla^2 \log p_{\textrm{data}}(\mathbf r) + \nabla \log p_{\textrm{data}}(\mathbf r) \nabla \log p_{\textrm{data}}(\mathbf r)^\top) p_{\textrm{data}}(\mathbf r) \dd \mathbf r) \bigg) \bigg) \nonumber \\
        &\stackrel{(ii)}{=} \order \big( t \sqrt{d} S(d, t) + t^2 d S^2(d, t) \cdot L_s (\sigma_{\max}^2 + 1) \big) = \order \big( d t L_s (\sigma_{\max}^2 + 1) \big). \label{equ: bound of term2 in TV between 0 and t0}
    \end{align}
    where $(i)$ is due to the Cauchy-Schwarz inequality and $\| \mathbf r \|_2 \leq S(d, t)$, and $(ii)$ invokes the upper bound \eqref{equ: upper bound of log p_data L2 norm} in Lemma \ref{lem: bound for score function L2 norm}.

    Combining the upper bound of \eqref{equ: term1 in TV between 0 and t0} and \eqref{equ: term2 in TV between 0 and t0} in \eqref{equ: bound of term1 in TV between 0 and t0} and \eqref{equ: bound of term2 in TV between 0 and t0}, we obtain the desired result. \hfill\Halmos
\end{proof}

\begin{lemma}
\label{lem: bound for score function L2 norm}
Suppose Assumptions \ref{assumption: factor}-\ref{assumption: Lipschitz} holds. Then, it holds
\begin{align}
    \E_{\mathbf R_0 \sim P_{\mathrm{data}}}\left\|\nabla \log p_{\mathrm{data}}(\mathbf R_0)\right\|_2^2 = \order (d L_s (\sigma_{\max}^2 + 1)). \label{equ: upper bound of log p_data L2 norm}
\end{align}
\end{lemma}

\begin{proof}{Proof of Lemma~\ref{lem: bound for score function L2 norm}.}
Taking $t = 0$ in the formula of $\nabla \log p_t$ in \eqref{equ: score decomposition -- heterogeneous} of Lemma \ref{lem: score decomposition -- heterogeneous}, we have
\begin{align*}
    \nabla \log p_{\textrm{data}}(\mathbf r) &= \mathbf s_{\textrm{sub}}(\bm\Gamma_0\bm\beta^\top \bm\Lambda_0^{-1} \mathbf r, 0) - \bm \Lambda_0^{-\frac{1}{2}} (\mathbf I - \bm\Lambda_0^{-\frac{1}{2}} \bm\beta \bm\Gamma_0 \bm\beta^\top \bm\Lambda_0^{-\frac{1}{2}}) \bm \Lambda_0^{-\frac{1}{2}} \mathbf r.
\end{align*}
Under Assumption \ref{assumption: Lipschitz}, for any $\mathbf r_1, \mathbf r_2 \in \bb R^d$, it holds that
\begin{align}
    \|\nabla \log p_{\textrm{data}}(\mathbf r_1) - \nabla \log p_{\textrm{data}}(\mathbf r_2)\|_2 &\leq L_s \| \bm\Gamma_0\bm\beta^\top \bm\Lambda_0^{-1} \|_{\textrm{op}} \|\mathbf r_1 - \mathbf r_2\|_2 + \|\bm\Lambda_0^{-1} \|_{\textrm{op}} \|\mathbf r_1 - \mathbf r_2\|_2 \nonumber \\
    &\leq \frac{L_s (\sigma_{\max}^2 + 1)}{\sigma_d^2} \cdot \|\mathbf r_1 - \mathbf r_2\|_2, \label{equ: Lipschitz constant of log p_data}
\end{align}
where the last equality follows from $\| \bm\Gamma_0\|_{\textrm{op}} \leq \sigma_{\max}^2$ and $\| \bm\Lambda_0^{-1} \|_{\textrm{op}} \leq 1/\sigma_d^2$. This indicates that the Lipschitz constant of $\nabla \log p_{\textrm{data}}$ is bounded by $L_s (1+\sigma_{\max}^2)/\sigma_d^2$. 
Furthermore, we have
\begin{align*}
\E_{\mathbf R_0 \sim P_{\textrm{data}}}\bigg[ \| \nabla \log p_{\textrm{data}}(\mathbf R_0) \|_2^2 \bigg]
&= \operatorname{tr}\bigg( \int \nabla \log p_{\textrm{data}}(\mathbf r) \nabla \log p_{\textrm{data}}(\mathbf r)^{\top} p_{\textrm{data}}(\mathbf r) \dd \mathbf r \bigg) \\
&\stackrel{(i)}{=} \operatorname{tr}\bigg( -\int \nabla^2 \log p_{\textrm{data}}(\mathbf r)^{\top} p_{\textrm{data}}(\mathbf r) \dd \mathbf r\bigg) \\
&= \order (d L_s (\sigma_{\max}^2 + 1)),
\end{align*}
where $(i)$ is due to the integration by parts and the last inequality follows from invoking \eqref{equ: Lipschitz constant of log p_data}. \hfill\Halmos

\end{proof}

\section{Additional Details of the Numerical Study with Synthetic Data} \label{app:exp_setup}
Here we explain additional details of the numerical experiment setup for Section \ref{sec:synthetic}. Following the standard setup in the econometrics literature \citep{bai2002determining, bai2023approximate}, we construct the ground-truth environment of high-dimensional asset returns using a sub-Gaussian factor model. Specifically, the universe consists of $ d = 2048 $ assets, whose returns are driven by $ k = 16 $ latent factors. Here, the choice of $ d $ as a power of 2 enhances the computational efficiency. 

Denote $\bm{\mu}_F = (\mu_{F1}, \mu_{F2}, \dots, \mu_{Fk})$ as the expected return and $\bm{\Sigma}_F = \operatorname{diag}\{\sigma_{F1}^2, \sigma_{F2}^2, \dots, \sigma_{Fk}^2\}$ the covariance matrix of the latent factor. In addition, denote
$\bm{\Sigma}_{\varepsilon} = \operatorname{diag}\{\sigma_{\varepsilon 1}^2, \sigma_{\varepsilon 2}^2, \dots, \sigma_{\varepsilon d}^2\}$ as the covariance of the idiosyncratic noise of the asset. We then construct samples from the ground-truth environment as follows:
\begin{enumerate}
    \item \textbf{Latent Factor.}  
    The components of $\bm{\mu}_F$ are drawn i.i.d. from $\operatorname{Uniform}(0, 0.1)$ and we set $ \sigma_{Fi} = 1.5 \mu_{Fi} $ for $ i = 1, 2, \dots, k $ to ensure that the volatility scales proportionally to the corresponding mean.
    
    \item \textbf{Factor Loadings.}  
    We generate the factor loading matrix $\bm \beta \in \mathbb{R}^{d \times k}$, where each element is drawn i.i.d. from $\mathcal{N}(0, 1)$, ensuring that the loadings are symmetrically distributed with comparable magnitudes across assets and factors.

    \item \textbf{Idiosyncratic Risk.}  
    $\{\sigma_{\varepsilon i}\}_{i=1}^{d}$ are drawn i.i.d. from $ \operatorname{Uniform}(0, 0.4)$, ensuring uncorrelated idiosyncratic returns across assets.

    \item \textbf{Asset Return.} We generate a total of $ 2^{13}=8192 $ simulated samples. Asset returns are sampled i.i.d. according to the following procedure. First, the factor is drawn from a multivariate normal distribution $\mathbf F \sim \mathcal{N}(\bm{\mu}_F, \bm{\Sigma}_F)$. Then the asset-specific noise terms are drawn i.i.d. from $\bm{\varepsilon} \sim \mathcal{N}(\mathbf 0, \bm{\Sigma}_{\varepsilon})$. Finally, the asset return is constructed by $\mathbf R = \bm\beta \mathbf F + \bm{\varepsilon}$. We denote by $ \mu_{Ri} $ and $ \sigma_{Ri} $ the mean and standard deviation of the ground-truth return for asset $ i $, where $ i = 1, 2, \dots, d $.

\end{enumerate}

\paragraph{Summary Statistics of the Synthetic Data.} To show that our simulation setting is close to the realistic market scenario, we benchmark our simulation set-up against the S\&P 500 index. Specifically, denote by $ \mu_{\text{S\&P 500}, i} $ and $ \sigma_{\text{S\&P500}, i} $ the mean and standard deviation of historical returns for stock $ i $ in the S\&P 500 index over the period 2000–2020. Table \ref{tab:summary_statistics_for_simulation} reports the summary statistics of $ \{\mu_{Ri}\}_{i=1}^{d} $ and compares them with $ \{\mu_{\text{S\&P500}, i}\}_{i=1}^{500} $. The range of both the simulated mean and standard deviation of returns closely matches that of the empirical quantities of stocks in the S\&P 500 index. 

In addition, the variance of the factors accounts for 50.42\% of the total variance in our synthetic data, which corresponds to the population $R$-squared. 

\begingroup
\renewcommand{\arraystretch}{0.9}
\setlength{\extrarowheight}{0pt}
\begin{table}[H]
\centering
\small
\caption
{Summary statistics for simulation return data and comparison with S\&P 500 over the period 2000-2020.
\label{tab:summary_statistics_for_simulation}}
{
\begin{tabular*}{\textwidth}{@{\extracolsep{\fill}}l *{7}{c}@{}}
\toprule
& \textbf{Mean} & \textbf{Std} & \textbf{Min} & \textbf{25\%} & \textbf{50\%} & \textbf{75\%} & \textbf{Max} \\
\midrule
Synthetic $\{\mu_{Ri}\}$ & 0.000 & 0.235 & -0.809 & -0.154 & -0.007 & 0.155 & 0.751 \\
S\&P 500 $\{\mu_{\text{S\&P500}, i}\}$ & 0.070 & 0.234 & -0.817 & -0.057 & -0.124 & 0.182 & 0.929 \\
Synthetic $\{\sigma_{Ri}\}$ & 0.475 & 0.126 & 0.243 & 0.377 & 0.473 & 0.576 & 0.739 \\
S\&P 500 $\{\sigma_{\text{S\&P500}, i}\}$ & 0.380 & 0.142 & 0.203 & 0.273 & 0.345 & 0.450 & 0.725 \\
\bottomrule
\end{tabular*}
}
{}
\end{table}
\endgroup

\paragraph{Data Preprocessing.}  
We preprocess the data in the following steps.

\begin{enumerate}
    \item First, we sort the asset returns by their variance, prioritizing those with greater variability for subsequent analysis. 
    \item Next, we normalize the data by subtracting the mean return of each asset and reshape the data from a one-dimensional vector of length $2^{11}$ into a two-dimensional matrix of size $(2^5, 2^6)$. This reshaping step ensures compatibility with the 2D-Unet architecture and allows the model to effectively leverage spatial hierarchies in the data.
\end{enumerate}

\paragraph{Training.}
We train our diffusion factor model using a 2-dimensional U-Net architecture \citep{ronneberger2015u}, which is a convolutional encoder-decoder network with skip connections. The U-Net serves as a practical implementation of the encoder-decoder architecture analyzed in theory. Specifically, in our theoretical analysis, we consider a linear encoder-decoder architecture, which first projects the input data into a low-dimensional latent space and then reconstructs it back to the original space. This \emph{project-then-lift network architecture} not only reduces the number of trainable parameters by operating in a low-dimensional latent space, but also effectively captures data intrinsic structures such as the factor model in Assumption~\ref{assumption: factor}. In experiments, U-Net implements the same encoder-decoder architecture with nonlinear networks, which is further shown to accommodate complex real-world data, such as the U.S. stock returns in Section~\ref{sec: empirical} and Appendix~\ref{app:empirical_setup}. Nonetheless, whether the mappings are linear or nonlinear is \emph{secondary}, since linear transformations are only assumed to enable tractable analysis and are standard in the theoretical study of deep learning \citep{baldi1989pca,chen2012msda,chen2023score,weitzner2024linear}.

Besides, guided by the theoretical results in
Sections 3-5, we choose the latent space dimension---the bottleneck dimension of the encoder and decoder architecture---based on the underlying factor model. More specifically, in synthetic experiments, we set the bottleneck width equal to the true factor dimension, $k=16$. The U-Net has approximately one billion parameters and is trained to approximate the score function by minimizing the empirical loss defined in \eqref{equ: empirical loss function}. To assess performance under different data regimes, we set the number of training samples to be $N = 2^{9}, 2^{10}, \dots, 2^{13}$. For fairness, both {\tt Diff Method} and {\tt Emp Method} are trained on the \emph{same} $N$ simulated returns for each $N$. The {\tt Diff Method} then uses the trained model to generate $2^{13}$ new samples for latent subspace recovery and distribution estimation, while the {\tt Emp Method} provides empirical estimates solely on $N$ training samples. To ensure that the results are robust rather than due to chance, each experiment is repeated five times.

\section{Additional Details of the Empirical Analysis}
\label{app:empirical_setup}
In this section, we provide further details on the empirical setup in Section~\ref{sec: empirical} and report robustness analysis with respect to transaction costs, risk aversions, norm constraints, and model update frequency.

\subsection{Data Preprocessing, Training, and Evaluation}
\label{app:empirical_data_training_eval}
\paragraph{Data Selection and Preprocessing.} We select and preprocess the stock return data in the following steps:
\begin{enumerate}
    \item We first exclude stocks with more than 5\% missing values and then select the 512 stocks with the largest market capitalizations from the remaining universe.
    \item Rank the selected stocks by return volatility in descending order.
    \item Within each rolling window of the training data, we standardize the returns by subtracting the (empirical) mean and dividing by the (empirical) standard deviation for each stock.
    \item Winsorize returns for each stock at 2.5\% each side by resampling non-extreme values with the same sign, which preserves the empirical distribution while mitigating the influence of outliers  \citep{tukey1962future}.
\end{enumerate}

\paragraph{Training and Sampling.}
We employ a 2D-UNet architecture with approximately one billion parameters to train our diffusion factor model. For real data, the true factor dimension is unknown; following common practice, we set the bottleneck width to $8$ as an educated estimate consistent with factor dimensions commonly used in the literature \citep{fama_french:1993,bai2002determining,onatski2010determining,fama2015five,fan2022estimating,bai2023approximate}. Following a similar setup as \citet{lyu2022accelerating}, we set the total number of training steps to $T=200$ and apply early stopping at $T^{\prime}=180$ for the sampling of time-reversed process~\eqref{equ: learned backward SDE}. For the downstream evaluation on mean-variance portfolios and factor-tangency portfolios, we use the trained model to generate $2^{12}=4096$ new samples for each rolling window.

\paragraph{Performance Evaluations.}
Next, we specify the performance evaluation metrics used in Section~\ref{sec: empirical}.
\begin{enumerate}
    \item SR is defined as $\hat{\mu}/\hat{\sigma}$, where $\hat{\mu}$ and $\hat{\sigma}$ denote the sample mean and standard deviation, respectively, of excess portfolio returns over the testing periods.

    \item CER is defined as $\hat{\mu} - \frac{1}{2}\eta\hat{\sigma}^2$, where $\eta$ is the risk aversion parameter.

    \item MDD is defined as 
    \begin{equation*}
        \text{MDD} = \max_{t \in \mathcal{D}_t} \left( \frac{\max_{s \leq t} V_s - V_t}{\max_{s \leq t} V_s} \right),
    \end{equation*}
    where $\mathcal{D}_t$ contains all the dates of the test set and $V_t$ denotes the portfolio value on day $t$.

    \item TO on day $t$ is defined as
    \begin{align*}
        \textrm{TO}_t = \sum_{i \in \mathcal{A}_t} \left| w_{i, t} - \frac{w_{i, t-1} (1+r_{i,t-1})}{\sum_{i=1}^{N}w_{i, t-1} (1+r_{i,t-1})} \right|,
    \end{align*}
    where $\mathcal{A}_t$ contains all assets of the test set on day $t$, $w_{i, t}$ is the target weight of stock $i$ on day $t$, and $r_{i, t}$ denotes the return of stock $i$ on day $t$.
\end{enumerate} 

We visualize the return distribution generated by our diffusion factor model for selected assets in Figure~\ref{fig:example_empirical_maxmin_variance_and_mean_stock} (trained on data from May 1, 2009 to April 30, 2014), which is compared with the observed training data. The generated data distribution is smoother and closely approximates the empirical distribution.
\begin{figure}[htbp]
    \centering
    \caption{Examples of asset return distribution (the blue histogram is constructed using samples generated from the diffusion model and the green one uses actual data samples.)}
    \label{fig:example_empirical_maxmin_variance_and_mean_stock}
    
    \subfigure[The asset with the largest variance.\label{fig:example_empirical_max_variance_stock}]{
        \includegraphics[width=0.48\linewidth]{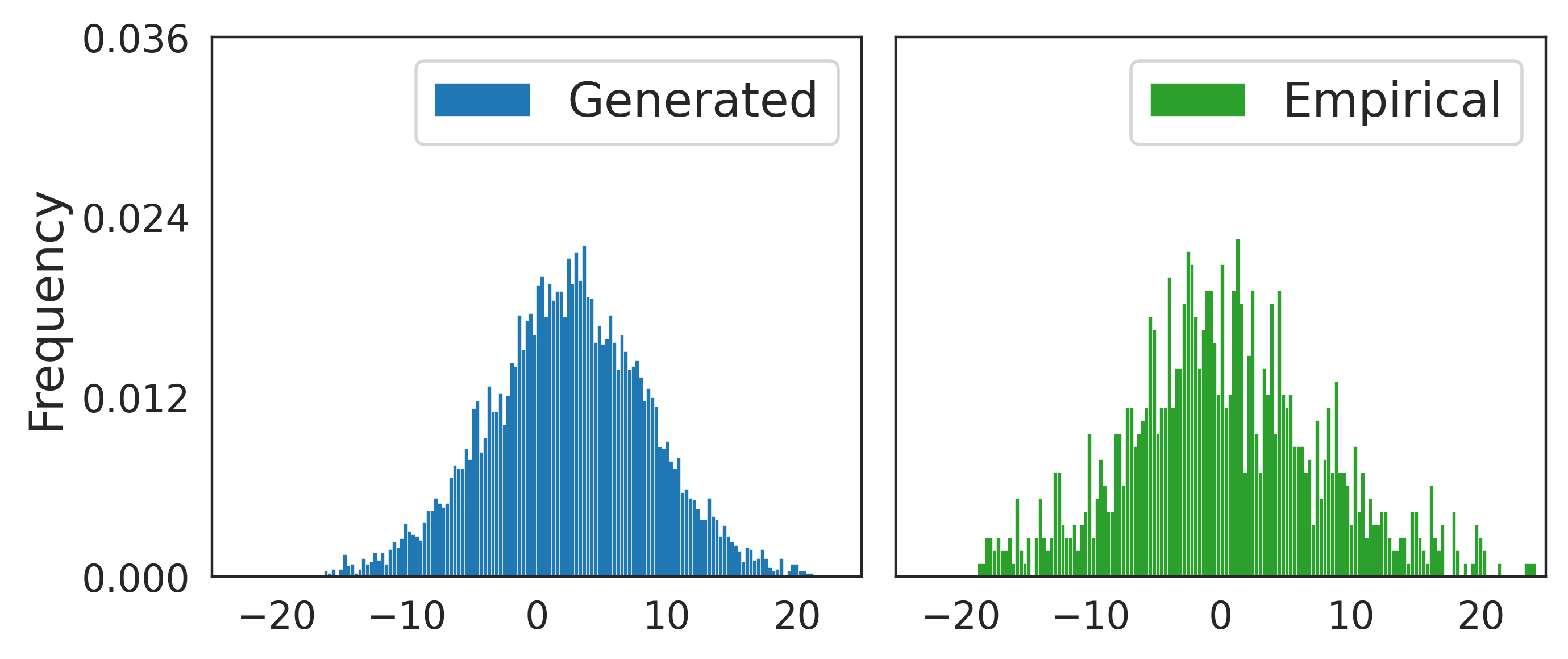}
    }
    \hfill
    \subfigure[The asset with the smallest variance.\label{fig:example_empirical_min_variance_stock}]{
        \includegraphics[width=0.48\linewidth]{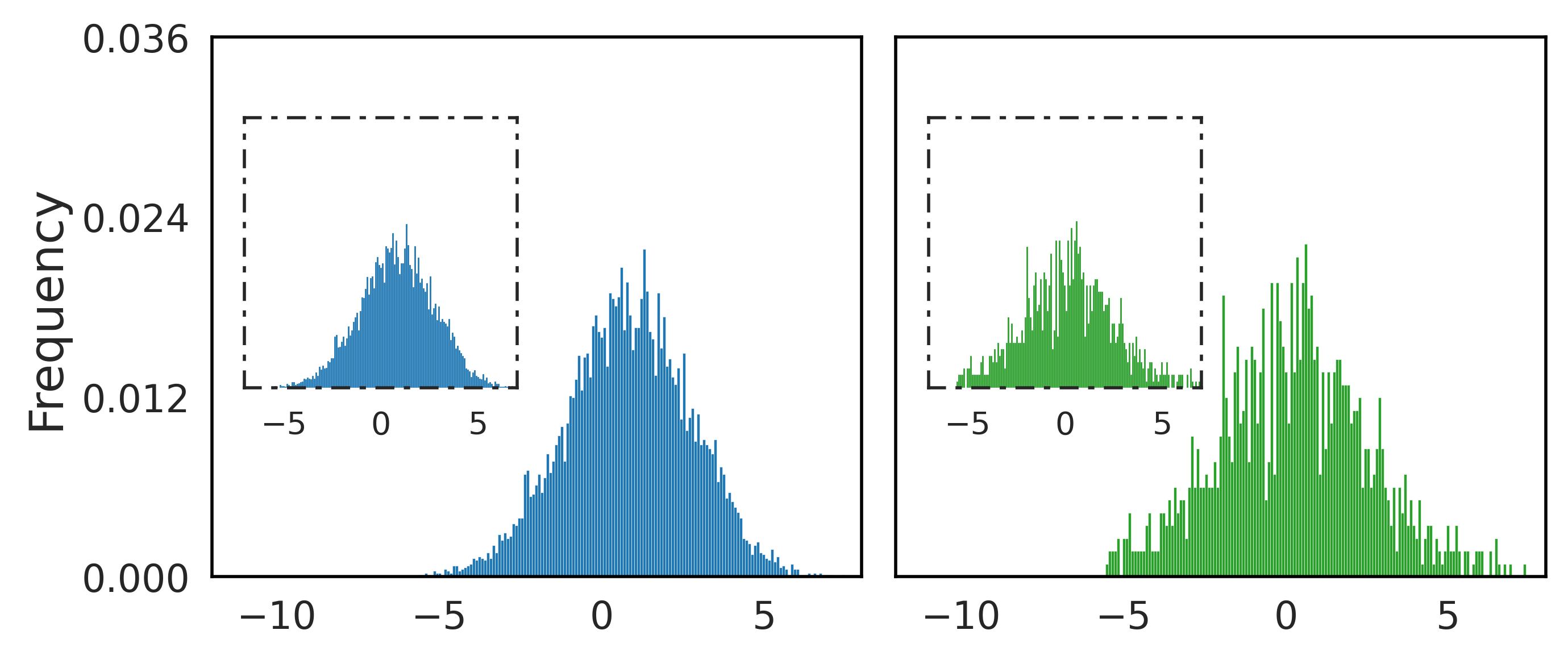}
    }
    
    \subfigure[The asset with the largest mean.\label{fig:example_empirical_max_mean_stock}]{
        \includegraphics[width=0.48\linewidth]{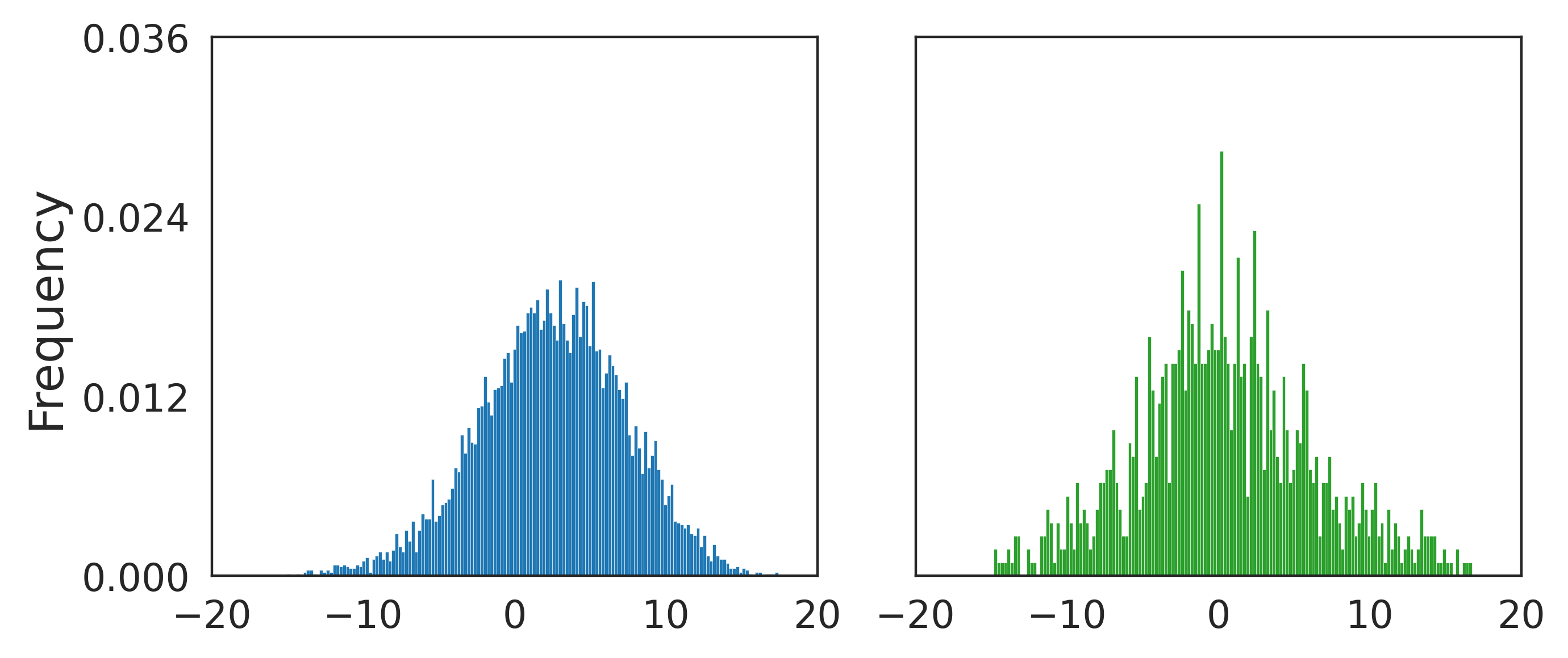}
    }
    \hfill
    \subfigure[The asset with the smallest mean.\label{fig:example_empirical_min_mean_stock}]{
        \includegraphics[width=0.48\linewidth]{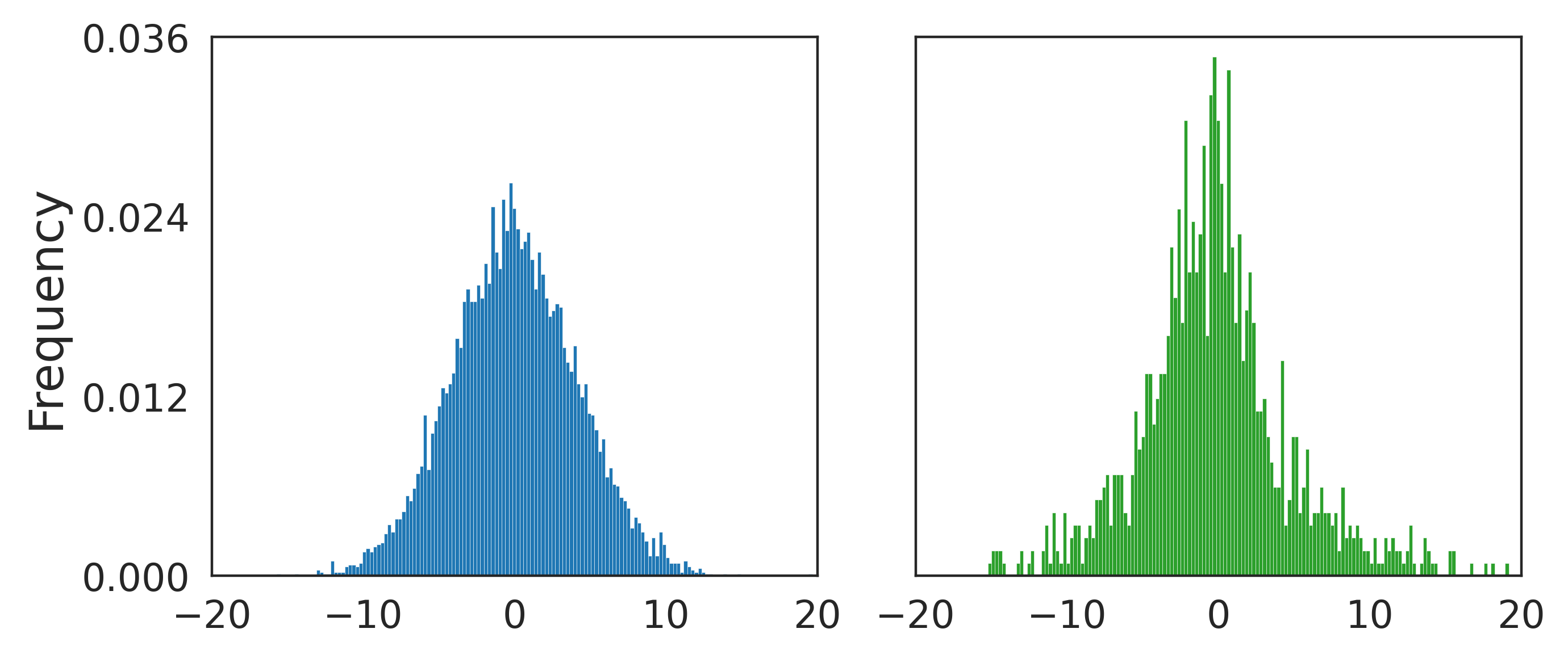}
    }
\end{figure}

\subsection{Robustness Analysis on Transaction Costs and Risk Aversion}
\label{app:empirical_cost_eta}
For mean-variance portfolios with $\eta=3$, we report out-of-sample portfolio performance under the scenario without transaction costs in Table~\ref{tab: portfolio_performance_metrics_eta3_nocost_quarter}. The {\tt Diff Emp+Diff Emp} outperforms all other methods, which are consistent with those observed in the scenario with transaction costs.

\begingroup
\renewcommand{\arraystretch}{0.9}
\setlength{\extrarowheight}{0pt}
\begin{table}[htbp]
\centering
\small
\caption{Performance of different portfolios without transaction costs for $\eta = 3$ (model updated quarterly).\label{tab: portfolio_performance_metrics_eta3_nocost_quarter}}
{
\begin{tabular*}{\textwidth}{@{\extracolsep{\fill}}l *{6}{c}@{}}
\toprule
Method & Mean & Std & SR & CER & MDD (\%) & TO \\
\midrule
\multicolumn{7}{c}{Methods based on real observed data} \\
\midrule
{\tt EW} & 0.106 & 0.206 & 0.516 & 0.043 & 52.437 & {\bf 3.031} \\
{\tt VW} & 0.103 & 0.220 & 0.468 & 0.031 & 57.322 & 3.464 \\
{\tt Real Emp+Real Emp}    & 0.077 & 0.126 & 0.608 & 0.053 & 33.642 & 46.722 \\
{\tt Real BS+Real Emp}     & 0.070 & 0.124 & 0.565 & 0.047 & 32.092 & 45.612 \\
{\tt Real OLSE+Real Emp}   & 0.053 & 0.125 & 0.427 & 0.030 & 33.188 & 45.952 \\
{\tt Real Emp+Real LW}     & 0.075 & 0.121 & 0.617 & 0.053 & 32.264 & 38.827 \\
{\tt Real BS+Real LW}      & 0.069 & {\bf 0.119} & 0.575 & 0.047 & 31.558 & 37.900 \\
{\tt Real OLSE+Real LW}    & 0.053 & 0.120 & 0.438 & 0.031 & 33.503 & 38.543 \\
\midrule
\multicolumn{7}{c}{Methods based on diffusion-generated data} \\
\midrule
{\tt Diff Emp+Diff Emp}    & {\bf 0.273} & 0.157 & {\bf 1.740} & {\bf 0.236} & 32.159 & 28.751 \\
{\tt Diff BS+Diff Emp}     & 0.269 & 0.155 & 1.729 & 0.232 & 32.113 & 27.978 \\
{\tt Diff OLSE+Diff Emp}   & 0.268 & 0.155 & 1.728 & 0.232 & 32.110 & 27.876 \\
{\tt Diff Emp+Diff LW}     & 0.233 & 0.150 & 1.547 & 0.199 & 32.048 & 26.353 \\
{\tt Diff BS+Diff LW}      & 0.230 & 0.149 & 1.539 & 0.196 & 31.995 & 25.773 \\
{\tt Diff OLSE+Diff LW}    & 0.229 & 0.149 & 1.539 & 0.196 & 31.991 & 25.697 \\
\midrule
\multicolumn{7}{c}{Methods based on both real observed data and diffusion-generated data} \\
\midrule
{\tt Real Emp+Diff Emp}    & 0.210 & 0.149 & 1.410 & 0.177 & {\bf 30.098} & 23.313 \\
{\tt Diff Emp+Real Emp}    & 0.095 & 0.133 & 0.720 & 0.069 & 33.777 & 29.323 \\
\bottomrule
\end{tabular*}
}
{}
\end{table}
\endgroup

For the case of $\eta=5$, we report out-of-sample portfolio performance under scenarios with and without transaction costs in Table~\ref{tab: portfolio_performance_metrics_eta5_quarter} and plot the cumulative returns with transaction cost in Figure~\ref{fig: example_for_cumulative_return_eta5_quarter}. The {\tt Diff Emp+Diff Emp} outperforms all other methods with the highest Mean, SR, and CER. These results are consistent with those observed in the case of $\eta=3$.

\begingroup
\renewcommand{\arraystretch}{0.9}
\setlength{\extrarowheight}{0pt}
\begin{table}[htbp]
\centering
\small
\caption{Performance of different portfolios with and without transaction costs for $\eta = 5$.\label{tab: portfolio_performance_metrics_eta5_quarter}}
{
\begin{tabular*}{\textwidth}{@{\extracolsep{\fill}}l *{6}{c}@{}}
\toprule
Method & Mean & Std & SR & CER & MDD (\%) & TO \\
\midrule
\multicolumn{7}{c}{{\bf Panel A: Without Transaction Costs}} \\
\midrule
\multicolumn{7}{c}{Methods based on real observed data} \\
\midrule
{\tt EW} & 0.106 & 0.206 & 0.516 & 0.000  & 52.437 & {\bf 3.031} \\
{\tt VW} & 0.103 & 0.220 & 0.468 & -0.018 & 57.322 & 3.464 \\
{\tt Real Emp+Real Emp}     & 0.071 & 0.124 & 0.573 & 0.033  & 32.477 & 45.764 \\
{\tt Real BS+Real Emp}      & 0.067 & 0.123 & 0.546 & 0.029  & 31.390 & 45.359 \\
{\tt Real OLSE+Real Emp}    & 0.057 & 0.124 & 0.462 & 0.019  & 32.432 & 45.405 \\
{\tt Real Emp+Real LW}      & 0.070 & 0.120 & 0.584 & 0.034  & 31.503 & 38.025 \\
{\tt Real BS+Real LW}       & 0.066 & {\bf 0.119} & 0.556 & 0.031  & 31.657 & 37.706 \\
{\tt Real OLSE+Real LW}     & 0.057 & 0.120 & 0.438 & 0.021  & 32.850 & 38.004 \\
\midrule
\multicolumn{7}{c}{Methods based on diffusion-generated data} \\
\midrule
{\tt Diff Emp+Diff Emp}     & {\bf 0.234} & 0.146 & {\bf 1.607} & {\bf 0.181} & 31.836 & 22.226 \\
{\tt Diff BS+Diff Emp}      & 0.232 & 0.145 & 1.597 & 0.179 & 31.804 & 21.881 \\
{\tt Diff OLSE+Diff Emp}    & 0.231 & 0.145 & 1.595 & 0.179 & 31.802 & 21.834 \\
{\tt Diff Emp+Diff LW}      & 0.207 & 0.142 & 1.452 & 0.156 & 32.409 & 21.410 \\
{\tt Diff BS+Diff LW}       & 0.205 & 0.142 & 1.445 & 0.155 & 32.442 & 21.176 \\
{\tt Diff OLSE+Diff LW}     & 0.205 & 0.142 & 1.444 & 0.155 & 32.456 & 21.144 \\
\midrule
\multicolumn{7}{c}{Methods based on both real observed data and diffusion-generated data} \\
\midrule
{\tt Real Emp+Diff Emp}     & 0.203 & 0.141 & 1.153 & 0.154 & {\bf 30.194} & 20.049 \\
{\tt Diff Emp+Real Emp}     & 0.082 & 0.127 & 0.643 & 0.041 & 32.531 & 25.921 \\
\midrule
\multicolumn{7}{c}{{\bf Panel B: With Transaction Costs}} \\
\midrule
\multicolumn{7}{c}{Methods based on real observed data} \\
\midrule
{\tt EW} & 0.100 & 0.206 & 0.486 & -0.006 & 53.128 & {\bf 3.031} \\
{\tt VW} & 0.096 & 0.220 & 0.437 & -0.025 & 58.086 & 3.464 \\
{\tt Real Emp+Real Emp}     & -0.020 & 0.126 & -0.160 & -0.060 & 45.671 & 45.764 \\
{\tt Real BS+Real Emp}      & -0.023 & 0.125 & -0.186 & -0.063 & 46.280 & 45.359 \\
{\tt Real OLSE+Real Emp}    & -0.034 & 0.126 & -0.267 & -0.073 & 51.562 & 45.405 \\
{\tt Real Emp+Real LW}      & -0.006 & 0.121 & -0.051 & -0.043 & 39.176 & 38.025 \\
{\tt Real BS+Real LW}       & -0.009 & {\bf 0.120} & -0.077 & -0.046 & 39.598 & 37.706 \\
{\tt Real OLSE+Real LW}     & -0.019 & 0.121 & -0.160 & -0.056 & 43.379 & 38.004 \\
\midrule
\multicolumn{7}{c}{Methods based on diffusion-generated data} \\
\midrule
{\tt Diff Emp+Diff Emp}     & {\bf 0.190} & 0.147 & {\bf 1.292} & {\bf 0.136} & 32.459 & 22.226 \\
{\tt Diff BS+Diff Emp}      & 0.188 & 0.146 & 1.285 & 0.134 & 32.532 & 21.881 \\
{\tt Diff OLSE+Diff Emp}    & 0.188 & 0.146 & 1.284 & 0.134 & 32.551 & 21.834 \\
{\tt Diff Emp+Diff LW}      & 0.164 & 0.144 & 1.143 & 0.113 & 33.410 & 21.410 \\
{\tt Diff BS+Diff LW}       & 0.163 & 0.143 & 1.138 & 0.112 & 33.435 & 21.176 \\
{\tt Diff OLSE+Diff LW}     & 0.163 & 0.143 & 1.138 & 0.112 & 33.448 & 21.144 \\
\midrule
\multicolumn{7}{c}{Methods based on both real observed data and diffusion-generated data} \\
\midrule
{\tt Real Emp+Diff Emp}     & 0.163 & 0.142 & 1.153 & 0.113 & {\bf 31.111} & 20.049 \\
{\tt Diff Emp+Real Emp}     & 0.030 & 0.128 & 0.234 & -0.011 & 35.536 & 25.921 \\
\bottomrule
\end{tabular*}
}
{}
\end{table}
\endgroup

\begin{figure}[htbp]
\centering
    \caption{Cumulative returns of different portfolios in log scale with transaction cost for $\eta=5$ (model updated quarterly).\label{fig: example_for_cumulative_return_eta5_quarter}}
    \includegraphics[width=0.5\linewidth]{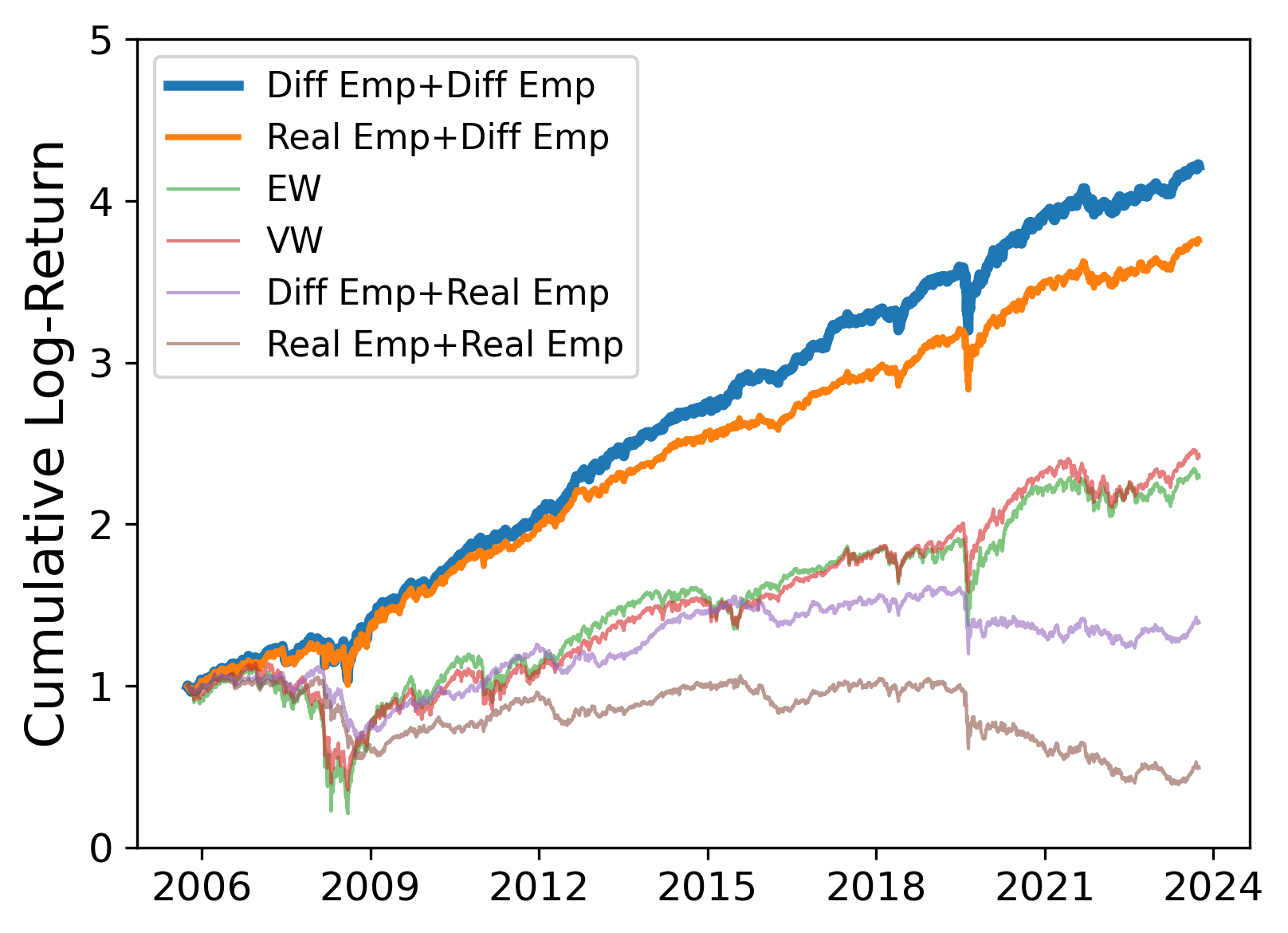}
\end{figure}

\subsection{Robustness Analysis on Norm Constraints}
\label{app:empirical_norm}
As a robustness check on the choice of norm constraints, for mean-variance portfolios, we solve the target weight by replacing the $\ell_{\infty}$-norm constraint in \eqref{equ: empirical--mean-variance optimization problem} with an $\ell_1$-norm constraint $\|\bm\omega\|_1=\sum_{i=1}^{d}|w_i|\le 3$.
We report the out-of-sample portfolio performance with and without transaction costs in Table~\ref{tab: portfolio_performance_metrics_eta3_l1norm} and plot the cumulative returns with transaction costs (in log scale) in Figure~\ref{fig: example_for_cumulative_return_eta3_l1norm}. The {\tt Diff Emp+Diff Emp} achieves the highest Mean, SR, and CER, which is similar to the results under the $\ell_\infty$-norm constraint.

\begin{figure}[htbp]
\centering
    \caption{Cumulative returns of different portfolios in log scale with transaction cost under $\ell_1$ norm constraints.\label{fig: example_for_cumulative_return_eta3_l1norm}}
    \includegraphics[width=0.5\linewidth]{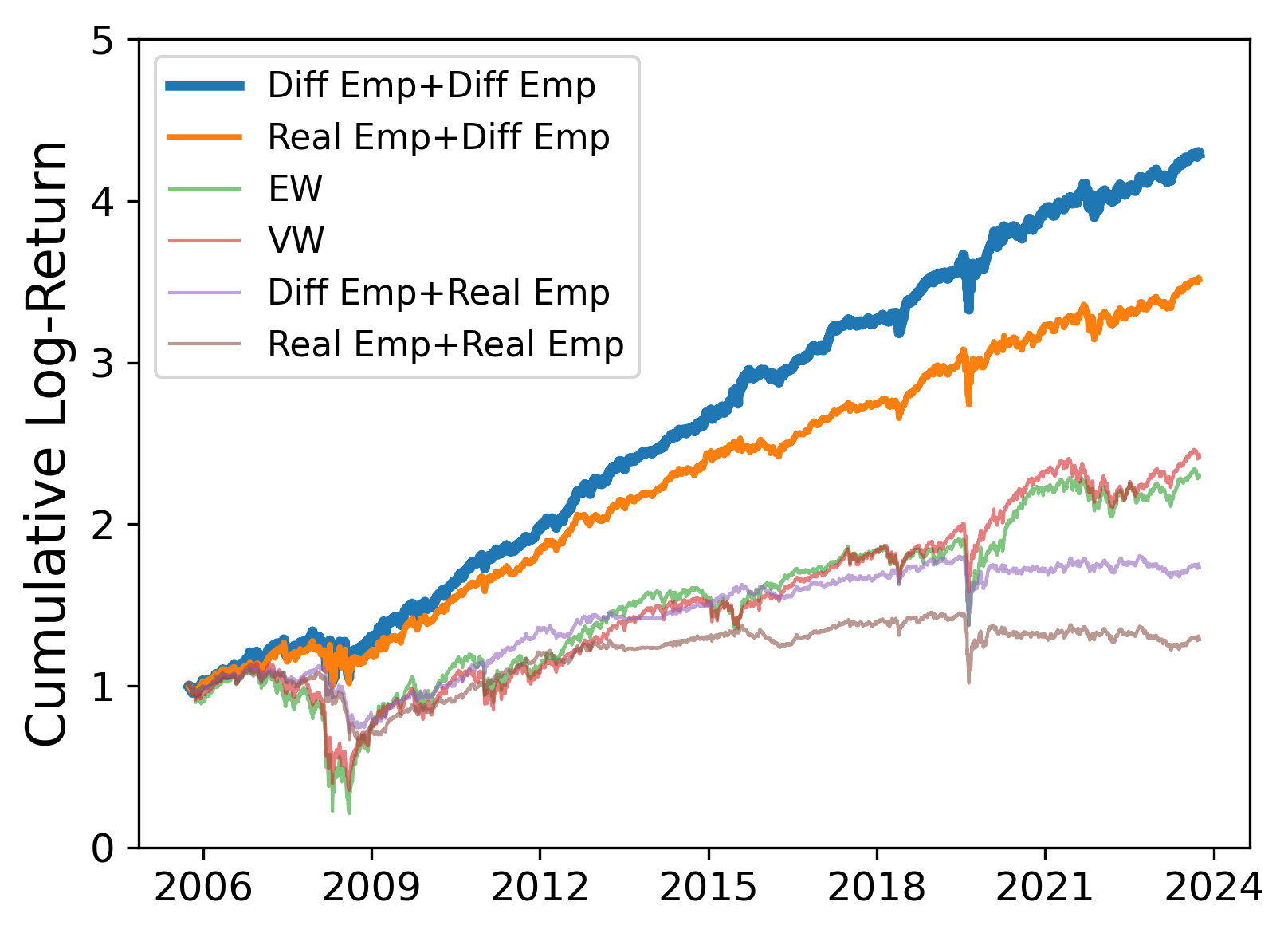}
\end{figure}

\begingroup
\renewcommand{\arraystretch}{0.9}
\setlength{\extrarowheight}{0pt}
\begin{table}[htbp]
\centering
\small
\caption{Performance of different portfolios with and without transaction costs under $\ell_1$-norm constraints.\label{tab: portfolio_performance_metrics_eta3_l1norm}}
{
\begin{tabular*}{\textwidth}{@{\extracolsep{\fill}}l *{6}{c}@{}}
\toprule
Method & Mean & Std & SR & CER & MDD (\%) & TO \\
\midrule
\multicolumn{7}{c}{{\bf Panel A: Without Transaction Costs}} \\
\midrule
\multicolumn{7}{c}{Methods based on real observed data} \\
\midrule
{\tt EW} & 0.106 & 0.206 & 0.516 & 0.043 & 52.437 & {\bf 3.031} \\
{\tt VW} & 0.103 & 0.220 & 0.468 & 0.031 & 57.322 & 3.464 \\
{\tt Real Emp+Real Emp}      & 0.056 & 0.114 & 0.487 & 0.036 & 34.460 & 16.514 \\
{\tt Real BS+Real Emp}       & 0.051 & 0.114 & 0.446 & 0.031 & 34.637 & 16.410 \\
{\tt Real OLSE+Real Emp}     & 0.040 & 0.114 & 0.354 & 0.021 & 35.004 & 16.418 \\
{\tt Real Emp+Real LW}       & 0.057 & 0.114 & 0.501 & 0.038 & 34.369 & 16.250 \\
{\tt Real BS+Real LW}        & 0.052 & {\bf 0.113} & 0.456 & 0.032 & 34.534 & 15.968 \\
{\tt Real OLSE+Real LW}      & 0.042 & 0.114 & 0.370 & 0.023 & 34.863 & 15.952 \\
\midrule
\multicolumn{7}{c}{Methods based on diffusion-generated data} \\
\midrule
{\tt Diff Emp+Diff Emp}      & {\bf 0.239} & 0.164 & {\bf 1.454} & {\bf 0.199} & 28.627 & 21.141 \\
{\tt Diff BS+Diff Emp}       & 0.234 & 0.163 & 1.436 & 0.194 & 29.512 & 20.790 \\
{\tt Diff OLSE+Diff Emp}     & 0.234 & 0.163 & 1.434 & 0.194 & 29.663 & 20.731 \\
{\tt Diff Emp+Diff LW}       & 0.206 & 0.160 & 1.289 & 0.168 & 31.520 & 20.343 \\
{\tt Diff BS+Diff LW}        & 0.202 & 0.159 & 1.274 & 0.164 & 32.292 & 20.081 \\
{\tt Diff OLSE+Diff LW}      & 0.202 & 0.159 & 1.273 & 0.164 & 32.401 & 20.038 \\
\midrule
\multicolumn{7}{c}{Methods based on both real observed data and diffusion-generated data} \\
\midrule
{\tt Real Emp+Diff Emp}      & 0.191 & 0.158 & 1.208 & 0.154 & {\bf 28.414} & 19.440 \\
{\tt Diff Emp+Real Emp}      & 0.066 & 0.115 & 0.574 & 0.046 & 34.225 & 18.574 \\
\midrule
\multicolumn{7}{c}{{\bf Panel B: With Transaction Costs}} \\
\midrule
\multicolumn{7}{c}{Methods based on real observed data} \\
\midrule
{\tt EW} & 0.100 & 0.206 & 0.486 & 0.037 & 53.128 & {\bf 3.031} \\
{\tt VW} & 0.096 & 0.220 & 0.437 & 0.024 & 58.086 & 3.464 \\
{\tt Real Emp+Real Emp}      & 0.022 & 0.115 & 0.196 & 0.003 & 35.023 & 16.514 \\
{\tt Real BS+Real Emp}       & 0.018 & 0.114 & 0.156 & -0.002 & 35.198 & 16.410 \\
{\tt Real OLSE+Real Emp}     & 0.008 & 0.115 & 0.066 & -0.012 & 39.104 & 16.418 \\
{\tt Real Emp+Real LW}       & 0.025 & 0.114 & 0.215 & 0.005 & 34.926 & 16.250 \\
{\tt Real BS+Real LW}        & 0.020 & {\bf 0.113} & 0.174 & 0.000 & 35.089 & 15.968 \\
{\tt Real OLSE+Real LW}      & 0.010 & 0.114 & 0.090 & -0.009 & 39.040 & 15.952 \\
\midrule
\multicolumn{7}{c}{Methods based on diffusion-generated data} \\
\midrule
{\tt Diff Emp+Diff Emp}      & {\bf 0.197} & 0.165 & {\bf 1.190} & {\bf 0.156} & 28.859 & 21.141 \\
{\tt Diff BS+Diff Emp}       & 0.193 & 0.164 & 1.175 & 0.152 & 29.749 & 20.790 \\
{\tt Diff OLSE+Diff Emp}     & 0.192 & 0.164 & 1.174 & 0.152 & 29.896 & 20.731 \\
{\tt Diff Emp+Diff LW}       & 0.165 & 0.160 & 1.029 & 0.126 & 31.709 & 20.343 \\
{\tt Diff BS+Diff LW}        & 0.162 & 0.160 & 1.016 & 0.124 & 32.464 & 20.081 \\
{\tt Diff OLSE+Diff LW}      & 0.162 & 0.160 & 1.015 & 0.124 & 32.570 & 20.038 \\
\midrule
\multicolumn{7}{c}{Methods based on both real observed data and diffusion-generated data} \\
\midrule
{\tt Real Emp+Diff Emp}      & 0.152 & 0.159 & 0.959 & 0.115 & {\bf 28.680} & 19.440 \\
{\tt Diff Emp+Real Emp}      & 0.047 & 0.115 & 0.411 & 0.027 & 34.525 & 18.574 \\
\bottomrule
\end{tabular*}
}
{}
\end{table}
\endgroup

\subsection{Robustness Analysis on Update Frequency}
\label{app:empirical_update_freq}
As a robustness check on the model update frequency, we also evaluate an annual update scheme with a rolling five-year window. Specifically, on May 1 of each year $T$, we reselect and pre-process the stocks, and update the model parameters using training data from May 1 of year $T-5$ to April 30 of year $T$. We test the model on data from May 1 of year $T$ to April 30 of year $T+1$ to evaluate out-of-sample performance. 

For mean-variance portfolios, we report out-of-sample portfolio performance under scenarios without and with transaction costs in Tables~\ref{tab: portfolio_performance_metrics_eta3_annual} and~\ref{tab: portfolio_performance_metrics_eta5_annual} for $\eta = 3$ and $\eta = 5$, respectively. Additionally, we plot the cumulative returns with transaction cost (in log scale) for $\eta = 3$ and $\eta = 5$ in Figure~\ref{fig: example_for_cumulative_return_annual}. {\tt Diff Emp+Diff Emp} outperforms all alternatives, which is consistent with the result observed under quarterly updates.

\begingroup
\renewcommand{\arraystretch}{0.9}
\setlength{\extrarowheight}{0pt}
\begin{table}[htbp]
\centering
\small
\caption{Performance of different portfolios with and without transaction costs for $\eta = 3$ (model updated annually).\label{tab: portfolio_performance_metrics_eta3_annual}}
{
\begin{tabular*}{\textwidth}{@{\extracolsep{\fill}}l *{6}{c}@{}}
\toprule
Method & Mean & Std & SR & CER & MDD (\%) & TO \\
\midrule
\multicolumn{7}{c}{{\bf Panel A: Without Transaction Costs}} \\
\midrule
\multicolumn{7}{c}{Methods based on real observed data} \\
\midrule
{\tt EW}                    & 0.102 & 0.221 & 0.462 & 0.029 & 58.114 & {\bf 3.273} \\
{\tt VW}                    & 0.098 & 0.218 & 0.448 & 0.026 & 61.400 & 3.717 \\
{\tt Real Emp+Real Emp}     & 0.087 & 0.142 & 0.611 & 0.057 & 34.651 & 38.120 \\
{\tt Real BS+Real Emp}      & 0.078 & 0.140 & 0.553 & 0.048 & 31.806 & 37.344 \\
{\tt Real OLSE+Real Emp}    & 0.090 & 0.144 & 0.625 & 0.059 & 35.069 & 38.112 \\
{\tt Real Emp+Real LW}      & 0.085 & 0.134 & 0.635 & 0.058 & 31.475 & 32.143 \\
{\tt Real BS+Real LW}       & 0.076 & {\bf 0.133} & 0.569 & 0.049 & 31.963 & 31.540 \\
{\tt Real OLSE+Real LW}     & 0.086 & 0.136 & 0.632 & 0.058 & 35.529 & 32.417 \\
\midrule
\multicolumn{7}{c}{Methods based on diffusion-generated data} \\
\midrule
{\tt Diff Emp+Diff Emp}     & {\bf 0.189} & 0.192 & {\bf 0.983} & {\bf 0.133} & 42.406 & 17.507 \\
{\tt Diff BS+Diff Emp}      & 0.185 & 0.190 & 0.972 & 0.131 & 42.021 & 17.203 \\
{\tt Diff OLSE+Diff Emp}    & 0.184 & 0.190 & 0.970 & 0.130 & 41.990 & 17.168 \\
{\tt Diff Emp+Diff LW}      & 0.155 & 0.169 & 0.917 & 0.112 & 38.046 & 16.332 \\
{\tt Diff BS+Diff LW}       & 0.152 & 0.168 & 0.906 & 0.110 & 37.908 & 16.115 \\
{\tt Diff OLSE+Diff LW}     & 0.152 & 0.168 & 0.904 & 0.110 & 37.897 & 16.090 \\
\midrule
\multicolumn{7}{c}{Methods based on both real observed data and diffusion-generated data} \\
\midrule
{\tt Real Emp+Diff Emp}     & 0.124 & 0.148 & 0.840 & 0.091 & {\bf 31.057} & 16.752 \\
{\tt Diff Emp+Real Emp}     & 0.113 & 0.167 & 0.676 & 0.071 & 34.043 & 23.360 \\
\midrule
\multicolumn{7}{c}{{\bf Panel B: With Transaction Costs}} \\
\midrule
\multicolumn{7}{c}{Methods based on real observed data} \\
\midrule
{\tt EW}                    & 0.096 & 0.221 & 0.433 & 0.022 & 58.807 & {\bf 3.273} \\
{\tt VW}                    & 0.090 & 0.218 & 0.414 & 0.019 & 62.127 & 3.717 \\
{\tt Real Emp+Real Emp}     & 0.011 & 0.144 & 0.073 & -0.021 & 39.327 & 38.120 \\
{\tt Real BS+Real Emp}      & 0.003 & 0.142 & 0.022 & -0.027 & 40.017 & 37.344 \\
{\tt Real OLSE+Real Emp}    & 0.012 & 0.146 & 0.082 & -0.020 & 45.209 & 38.112 \\
{\tt Real Emp+Real LW}      & 0.021 & 0.136 & 0.153 & -0.007 & {\bf 32.213} & 32.143 \\
{\tt Real BS+Real LW}       & 0.013 & {\bf 0.135} & 0.094 & -0.015 & 33.666 & 31.540 \\
{\tt Real OLSE+Real LW}     & 0.021 & 0.138 & 0.152 & -0.007 & 43.520 & 32.417 \\
\midrule
\multicolumn{7}{c}{Methods based on diffusion-generated data} \\
\midrule
{\tt Diff Emp+Diff Emp}     & {\bf 0.153} & 0.192 & {\bf 0.797} & {\bf 0.098} & 43.651 & 17.507 \\
{\tt Diff BS+Diff Emp}      & 0.150 & 0.191 & 0.788 & 0.096 & 43.267 & 17.203 \\
{\tt Diff OLSE+Diff Emp}    & 0.150 & 0.191 & 0.786 & 0.095 & 43.236 & 17.168 \\
{\tt Diff Emp+Diff LW}      & 0.122 & 0.170 & 0.720 & 0.079 & 38.127 & 16.332 \\
{\tt Diff BS+Diff LW}       & 0.120 & 0.168 & 0.711 & 0.077 & 37.974 & 16.115 \\
{\tt Diff OLSE+Diff LW}     & 0.119 & 0.168 & 0.709 & 0.077 & 37.962 & 16.090 \\
\midrule
\multicolumn{7}{c}{Methods based on both real observed data and diffusion-generated data} \\
\midrule
{\tt Real Emp+Diff Emp}     & 0.090 & 0.148 & 0.608 & 0.057 & 34.729 & 16.752 \\
{\tt Diff Emp+Real Emp}     & 0.019 & 0.170 & 0.111 & -0.024 & 36.913 & 23.360 \\
\bottomrule
\end{tabular*}
}
{}
\end{table}
\endgroup

\begingroup
\renewcommand{\arraystretch}{0.9}
\setlength{\extrarowheight}{0pt}
\begin{table}[htbp]
\centering
\small
\caption{Performance of different portfolios with and without transaction costs for $\eta = 5$ (model updated annually).\label{tab: portfolio_performance_metrics_eta5_annual}}
{
\begin{tabular*}{\textwidth}{@{\extracolsep{\fill}}l *{6}{c}@{}}
\toprule
Method & Mean & Std & SR & CER & MDD (\%) & TO \\
\midrule
\multicolumn{7}{c}{{\bf Panel A: Without Transaction Costs}} \\
\midrule
\multicolumn{7}{c}{Methods based on real observed data} \\
\midrule
{\tt EW}                & 0.102 & 0.221 & 0.462 & -0.020 & 58.114 & {\bf 3.273} \\
{\tt VW}                & 0.098 & 0.218 & 0.448 & -0.021 & 61.400 & 3.717 \\
{\tt Real Emp+Real Emp}      & 0.080 & 0.140 & 0.568 & 0.030 & 32.223 & 38.121 \\
{\tt Real BS+Real Emp}       & 0.074 & 0.140 & 0.525 & 0.024 & 32.181 & 37.243 \\
{\tt Real OLSE+Real Emp}      & 0.061 & 0.142 & 0.426 & 0.010 & 33.213 & 37.553 \\
{\tt Real Emp+Real LW}       & 0.078 & {\bf 0.133} & 0.584 & 0.033 & 31.312 & 32.143 \\
{\tt Real BS+Real LW}        & 0.072 & {\bf 0.133} & 0.542 & 0.028 & 32.301 & 31.452 \\
{\tt Real OLSE+Real LW}       & 0.058 & 0.135 & 0.430 & 0.013 & 33.637 & 31.901 \\
\midrule
\multicolumn{7}{c}{Methods based on diffusion-generated data} \\
\midrule
{\tt Diff Emp+Diff Emp}      & {\bf 0.155} & 0.167 & {\bf 0.925} & {\bf 0.085} & 45.760 & 17.647 \\
{\tt Diff BS+Diff Emp}       & 0.153 & 0.166 & 0.920 & 0.084 & 45.951 & 16.996 \\
{\tt Diff OLSE+Diff Emp}      & 0.152 & 0.166 & 0.918 & 0.084 & 45.937 & 16.970 \\
{\tt Diff Emp+Diff LW}       & 0.140 & 0.156 & 0.901 & 0.080 & 42.809 & 16.440 \\
{\tt Diff BS+Diff LW}        & 0.139 & 0.155 & 0.896 & 0.079 & 42.719 & 15.907 \\
{\tt Diff OLSE+Diff LW}       & 0.139 & 0.155 & 0.894 & 0.079 & 42.709 & 15.892 \\
\midrule
\multicolumn{7}{c}{Methods based on both real observed data and diffusion-generated data} \\
\midrule
{\tt Real Emp+Diff Emp} & 0.124 & 0.147 & 0.848 & 0.071 & 33.164 & 16.887 \\
{\tt Diff Emp+Real Emp} & 0.084 & 0.149 & 0.566 & 0.029 & {\bf 31.205} & 18.639 \\
\midrule
\multicolumn{7}{c}{{\bf Panel B: With Transaction Costs}} \\
\midrule
\multicolumn{7}{c}{Methods based on real observed data} \\
\midrule
{\tt EW}                & 0.096 & 0.221 & 0.433 & -0.026 & 58.807 & {\bf 3.273} \\
{\tt VW}                & 0.090 & 0.218 & 0.414 & -0.029 & 62.127 & 3.717 \\
{\tt Real Emp+Real Emp}      & 0.005 & 0.143 & 0.035 & -0.046 & 39.696 & 38.121 \\
{\tt Real BS+Real Emp}       & -0.001 & 0.142 & -0.005 & -0.051 & 40.889 & 37.243 \\
{\tt Real OLSE+Real Emp}      & -0.014 & 0.144 & -0.100 & -0.066 & 41.295 & 37.553 \\
{\tt Real Emp+Real LW}       & 0.014 & {\bf 0.135} & 0.107 & -0.031 & {\bf 33.287} & 32.143 \\
{\tt Real BS+Real LW}        & 0.009 & {\bf 0.135} & 0.068 & -0.036 & 35.225 & 31.452 \\
{\tt Real OLSE+Real LW}       & -0.006 & 0.136 & -0.042 & -0.052 & 40.965 & 31.901 \\
\midrule
\multicolumn{7}{c}{Methods based on diffusion-generated data} \\
\midrule
{\tt Diff Emp+Diff Emp}      & {\bf 0.128} & 0.167 & {\bf 0.766} & {\bf 0.058} & 47.549 & 17.647 \\
{\tt Diff BS+Diff Emp}       & 0.127 & 0.166 & 0.762 & 0.057 & 47.720 & 16.996 \\
{\tt Diff OLSE+Diff Emp}      & 0.126 & 0.166 & 0.760 & 0.057 & 47.704 & 16.970 \\
{\tt Diff Emp+Diff LW}       & 0.114 & 0.156 & 0.732 & 0.053 & 44.622 & 16.440 \\
{\tt Diff BS+Diff LW}        & 0.113 & 0.156 & 0.727 & 0.053 & 44.518 & 15.907 \\
{\tt Diff OLSE+Diff LW}       & 0.113 & 0.156 & 0.726 & 0.052 & 44.508 & 15.892 \\
\midrule
\multicolumn{7}{c}{Methods based on both real observed data and diffusion-generated data} \\
\midrule
{\tt Real Emp+Diff Emp} & 0.099 & 0.147 & 0.673 & 0.045 & 35.558 & 16.887 \\
{\tt Diff Emp+Real Emp} & 0.043 & 0.149 & 0.288 & -0.013 & 36.045 & 18.639 \\
\bottomrule
\end{tabular*}
}
{}
\end{table}
\endgroup

\begin{figure}[htbp]
\centering
    \caption{Cumulative returns of different portfolios in log scale with transaction cost (model updated annually).\label{fig: example_for_cumulative_return_annual}}
    \subfigure[$\eta=3$.\label{fig:example_for_cumulative_return_eta3_annual}]{
        \includegraphics[width=0.48\linewidth]{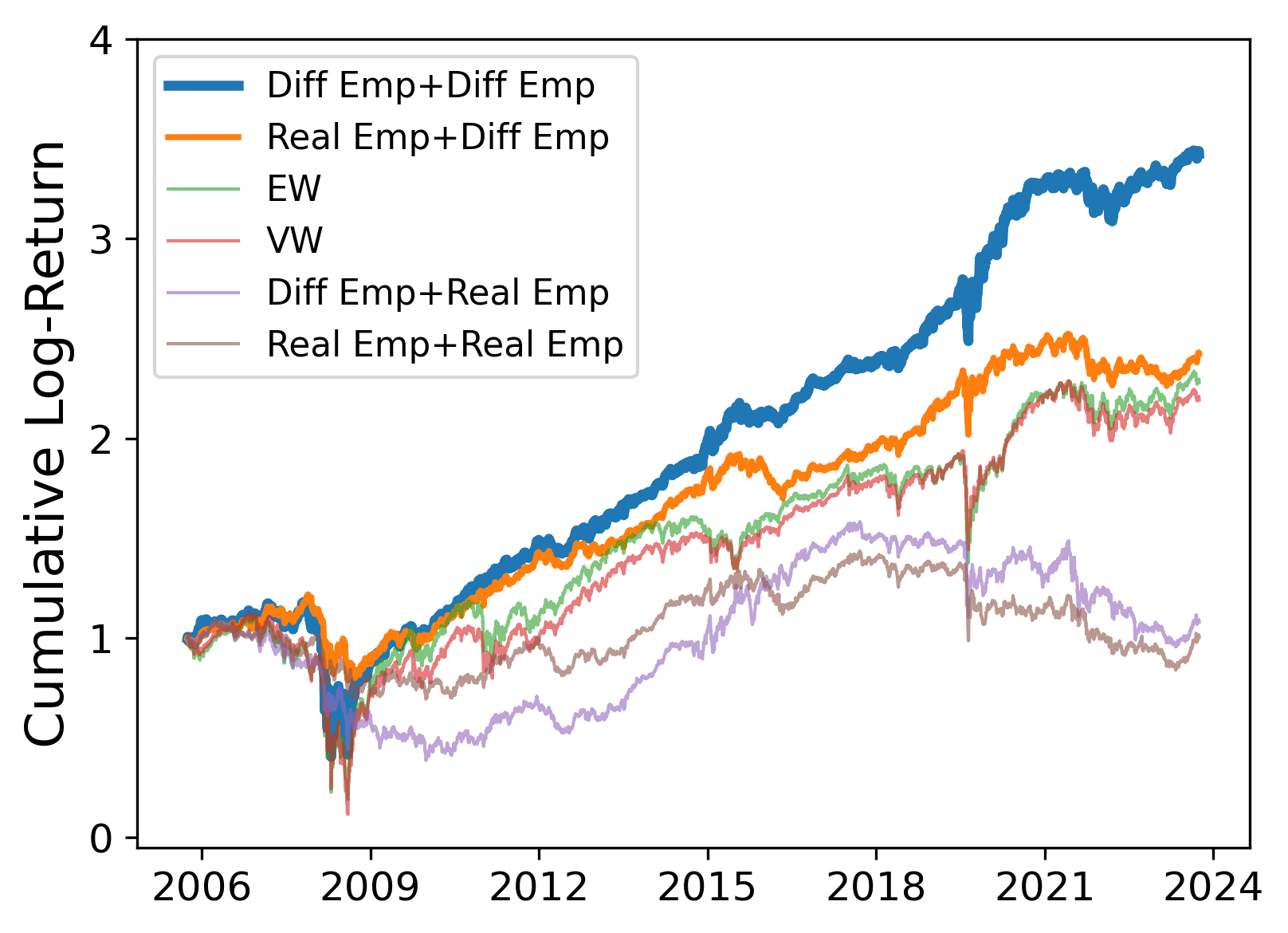}
    }
    \hfill
    \subfigure[$\eta=5$.\label{fig:example_for_cumulative_return_eta5_annual}]{
        \includegraphics[width=0.48\linewidth]{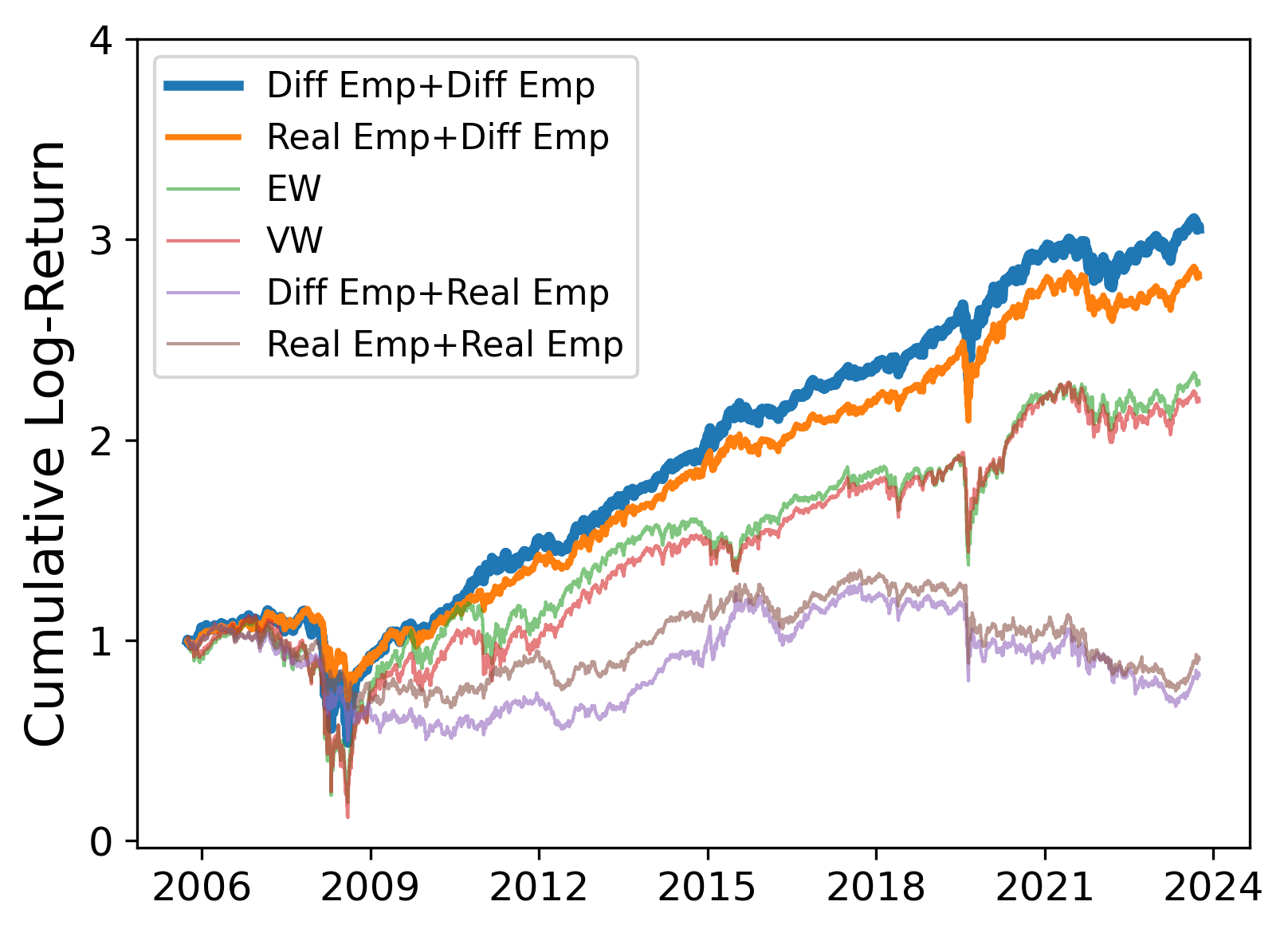}
    }
\end{figure}

Moreover, for factor tangency portfolios, we report Sharpe ratios across varying number of factors in Table~\ref{tab: out_of_sample_sharpe_ratios_annual} and plot the correlation heatmaps between top eight factors estimated using diffusion-based methods and traditional factors in {\tt FF Method} in Figure~\ref{fig:correlation_factors_annual}. The diffusion-based methods exhibit higher Sharpe ratios and outperform all other approaches. The diffusion-generated factors are notably correlated with traditional factors, with Mkt-RF, LT-REV, and MOM as the three dominant factors across all three methods. Overall, the findings are similar to those under quarterly updates.

\begingroup
\renewcommand{\arraystretch}{0.9}
\setlength{\extrarowheight}{0pt}
\begin{table}[htbp]
\centering
\small
\caption
{Out-of-sample Sharpe ratios of factor tangency portfolios (model updated annually). The number of factors is set to be $3$, $5$, $6$, and $8$, respectively.\label{tab: out_of_sample_sharpe_ratios_annual}}
{
\begin{tabular*}{\textwidth}{@{\extracolsep{\fill}}l *{7}{c}@{}}
\toprule
\# Factors & {\tt Diff+PCA} & {\tt Diff+POET} & {\tt Diff+RPPCA} & {\tt FF} & {\tt PCA} & {\tt POET} & {\tt RPPCA} \\
\midrule
$ 3 $ & 1.805 & 1.841 & {\bf 1.985} & 0.648 & 0.402 & 0.872 & 0.631 \\
$ 5 $ & 2.158 & 2.178 & {\bf 2.367} & 0.726 & 0.453 & 0.930 & 1.250 \\
$ 6 $ & 2.322 & 2.339 & {\bf 2.550} & 0.861 & 0.528 & 1.356 & 1.701 \\
$ 8 $ & 2.631 & 2.739 & {\bf 2.810} & 0.881 & 0.673 & 1.463 & 1.892 \\
\bottomrule
\end{tabular*}
}
{}
\end{table}
\endgroup

\begin{figure}[htbp]
    \centering
    \caption{Correlation between the top 8 factors obtained using diffusion-based methods and those from the {\tt FF Method} (model updated annually).\label{fig:correlation_factors_annual}}
    \subfigure[{\tt Diff+PCA Method}]{
        \includegraphics[width=0.31\textwidth]{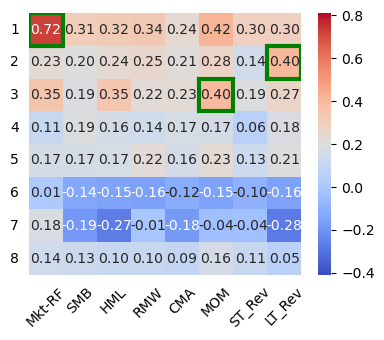}
    }
    \hfill
    \subfigure[{\tt Diff+POET Method}]{
        \includegraphics[width=0.31\textwidth]{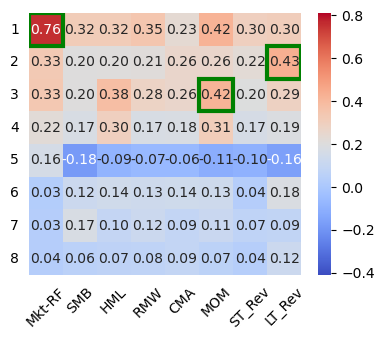}
    }  
    \subfigure[{\tt Diff+RPPCA Method}]{
        \includegraphics[width=0.31\textwidth]{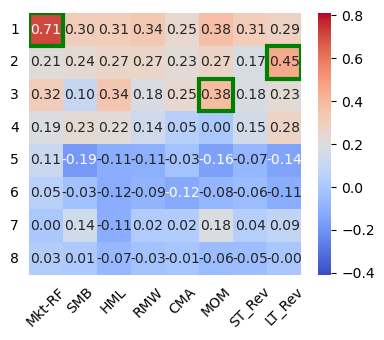}
    }
\end{figure}

\end{document}